\documentclass[11pt,a4paper]{article}
\usepackage{amssymb}
\usepackage{amsthm}
\usepackage{amsmath}
\usepackage[english]{babel} 
\usepackage[a4paper,margin=2.25cm]{geometry}

\usepackage{caption}
\usepackage{subcaption}
\usepackage{url}

\usepackage{graphicx}

\usepackage{microtype}%if unwanted, comment out or use option "draft"
\usepackage{amsmath}
\usepackage{enumerate}
\usepackage{mathtools}
\usepackage{proof}
\usepackage{hyperref}
\usepackage{cleveref}
\usepackage{mdframed}
\mdfdefinestyle{alttight}{innertopmargin=0pt,innerleftmargin=2pt,innerrightmargin=2pt}
\usepackage{multirow}

\usepackage{bussproofs}

\usepackage[dvipsnames]{xcolor}

\usepackage{tikz}
\usetikzlibrary{calc}
\usetikzlibrary{automata,positioning}

  \tikzset{
    invisible/.style={opacity=0},
    visible on/.style={alt={#1{}{invisible}}},
    alt/.code args={<#1>#2#3}{%
      \alt<#1>{\pgfkeysalso{#2}}{\pgfkeysalso{#3}} % \pgfkeysalso doesn't change the path
    },
  }

\tikzset{every state/.style={minimum size=0pt}}

\usepackage{xspace}
\usepackage{mathpartir}
\usepackage{stmaryrd}
\usepackage{colortbl} 
\usepackage[T1]{fontenc}
\usepackage{comment}

\usepackage{float}
\usepackage{thmtools}
\usepackage{thm-restate}
\usepackage{todonotes}
\usepackage{csquotes}

\usepackage{arydshln}
\restylefloat{table}

\newcommand{\ghot}[1]{#1\! \leadsto\! \diamond}
\newcommand{\linecondit}[3]{\ensuremath{(#1)\,\mathbf{?}\,#2\mathbf{:}\,#3}}

\newcommand{\balan}[1]{\ensuremath{\mathsf{balanced}(#1)}}
\newcommand{\subst}[2]{\{^{#1}\!/{\scriptstyle #2}\}}

\newcommand{\lrangle}[1]{\langle #1 \rangle}
\newcommand{\blrangle}[1]{\big\langle #1 \big\rangle}

\newcommand{\relS}{\ensuremath{\mathbin{\mathcal{S}}}}
\newcommand{\relSs}{\ensuremath{\,\mathcal{S}\xspace}}

\newcommand{\brtext}[1]{[\textrm{\small #1}]}

\newcommand{\eltsrule}[1]{{\footnotesize [\textrm{#1}]}}
\newcommand{\trule}[1]{{\footnotesize\brtext{#1}}}
\newcommand{\orule}[1]{{\scriptsize{\brtext{#1}}}}

\newcommand{\bnfis}{\;::=\;}
\newcommand{\bnfbar}{\mathrel{\big|}}
\newcommand{\sbnfbar}{\mathrel{\big|}}

\newcommand{\set}[1]{\{#1\}}
\newcommand{\es}{\emptyset}

\newcommand{\dom}[1]{\mathtt{dom}(#1)}

\newcommand{\freev}[1]{\lrangle{#1}}
\newcommand{\boundv}[1]{(#1)}

\newcommand{\inact}{\mathbf{0}}

\newcommand{\Par}{\;|\;}
\newcommand{\news}[1]{(\nu\, #1)\,}
\newcommand{\newsp}[2]{(\nu\, #1)(#2)}

\newcommand{\varp}[1]{#1}
\newcommand{\rvar}[1]{#1}
\newcommand{\recp}[2]{\mu \rvar{#1}. #2}

\newcommand{\Def}{\sessionfont{def}\ }

\newcommand{\defeq}{\stackrel{\Def}{=}}

\newcommand{\bn}[1]{\mathtt{bn}(#1)}
\newcommand{\fn}[1]{\mathtt{fn}(#1)}
\newcommand{\rfn}[1]{\mathtt{rn}(#1)}

\newcommand{\fv}[1]{\mathtt{fv}(#1)}

\newcommand{\fs}[1]{\mathtt{fs}(#1)}

\newcommand{\subj}[1]{\mathtt{subj}(#1)}

\newcommand{\scong}{\equiv}

\newcommand{\red}{\longrightarrow}

\def\subst#1#2{\{\raisebox{.5ex}{\small$#1$}\! / \mbox{\small$#2$}\}}

\newcommand{\sessionfont}[1]{\mathtt{#1}}
\newcommand{\vart}[1]{\mathsf{#1}}

\newcommand{\shsep}{.}
\newcommand{\outses}{!}
\newcommand{\inpses}{?}
\newcommand{\selses}{\triangleleft}
\newcommand{\brases}{\triangleright}
\newcommand{\dual}[1]{\overline{#1}}
\newcommand{\cat}{,}

\newcommand{\bout}[2]{#1 \outses \freev{#2} \shsep}
\newcommand{\about}[2]{#1 \outses \freev{#2}}
\newcommand{\boutt}[2]{#1 \outses \freev{#2}}
\newcommand{\bbout}[2]{#1 \outses \blrangle{#2} \shsep}
\newcommand{\abbout}[2]{#1 \outses \blrangle{#2}}
\newcommand{\binpt}[2]{#1 \inpses \boundv{#2}}
\newcommand{\abinp}[2]{#1 \inpses \boundv{#2}}
\newcommand{\binp}[2]{#1 \inpses \boundv{#2} \shsep}
\newcommand{\bsel}[2]{#1 \selses #2 \shsep}

\newcommand{\bbra}[2]{#1 \brases \set{#2}}

\newcommand{\tfont}[1]{\mathtt{#1}}
\newcommand{\tsep}{\textcolor{darkgray}{;}}

\newcommand{\chtype}[1]{\lrangle{#1}}

\newcommand{\outtype}{\outses}
\newcommand{\inptype}{\inpses}

\newcommand{\trec}[2]{\mu\vart{#1}.#2}
\newcommand{\tvar}[1]{\vart{#1}}

\newcommand{\tinact}{\tfont{\textcolor{darkgray}{end}}}

\newcommand{\btoutt}[1]{\outtype \lrangle{#1}}
\newcommand{\btinpt}[1]{\inptype (#1) }

\newcommand{\btout}[1]{\outtype \lrangle{#1} \tsep}

\newcommand{\btinp}[1]{\inptype (#1) \tsep}

\newcommand{\btsel}[1]{\oplus \set{#1}}

\newcommand{\btbra}[1]{\& \set{#1}}

\newcommand{\proves}{\vdash}
\newcommand{\hastype}{\triangleright}

\newcommand{\outlts}{\outses}
\newcommand{\inplts}{\inpses}

\newcommand{\bactout}[2]{#1 \outlts \freev{#2}}

\newcommand{\bactinp}[2]{#1 \inplts \freev{#2}}

\makeatletter
\newcommand{\xMapsto}[2][]{\ext@arrow 0599{\Mapstofill@}{#1}{#2}}
\def\Mapstofill@{\arrowfill@{\Mapstochar\Relbar}\Relbar\Rightarrow}
\makeatother

\newcommand{\by}[1]{\mathrel{\xrightarrow{#1}}}
\newcommand{\By}[1]{\mathrel{\xRightarrow{#1}}}

\newcommand{\hby}[1]{\mathrel{\xmapsto{#1}}}
\newcommand{\Hby}[1]{\mathrel{\xMapsto{#1}}}

\definecolor{pcolor}{rgb}{0.36, 0.54, 0.66}
\definecolor{tcolor}{rgb}{0.57, 0.36, 0.51}

\newcommand{\map}[1]{\ensuremath{\textcolor{pcolor}{\llbracket}#1\textcolor{pcolor}{\rrbracket}}}

\newcommand{\dmap}[1]{\textcolor{\colorpro}{\ensuremath{\{\!|\textcolor{black}{#1}|\!\}}}}

\newcommand{\dualof}{\ \mathsf{dual}\ }

\newcommand{\HOp}{\ensuremath{\mathsf{HO}\pi}\xspace}
\newcommand{\HO}{\ensuremath{\mathsf{HO}}\xspace}

\newcommand{\msts}{MSTs\xspace}

\newcommand{\Proc}{\ensuremath{\diamond}}

\newcommand{\appl}[2]{#1\, {#2}}
\newcommand{\abs}[2]{\lambda #1.\,#2}

\newcommand{\lollipop}{\multimap}
\newcommand{\sharedop}{\rightarrow}

\newcommand{\lhot}[1]{#1\! \lollipop\! \diamond}
\newcommand{\shot}[1]{#1\! \sharedop\! \diamond}
\newcommand{\lhotup}[1]{#1^{\lollipop}}
\newcommand{\shotup}[1]{#1^{\sharedop}}

\newcommand{\x}{\ensuremath{x}}

\newcommand{\h}{\ensuremath{h}}

\newtheorem{notation}{Notation}[section]

\newif\ifny\nyfalse
\newif\ifdm\dmtrue
\newif\ifrhu\rhutrue
\newif\ifjp\jptrue
\newif\ifjp\jptrue
\newcommand{\newj}[1]{{#1}}
\newcommand{\newjb}[1]{{#1}}

\newcommand{\AT}[2]{#1 \! : \! #2}

\newcommand{\secref}[1]{\Cref{#1}}
\newcommand{\defref}[1]{\Cref{#1}}

\newcommand{\figref}[1]{\Cref{#1}}
\newcommand{\thmref}[1]{\Cref{#1}}

\newcommand{\exref}[1]{\Cref{#1}}

\newcommand{\lemref}[1]{\Cref{#1}}
\newcommand{\remref}[1]{\Cref{#1}}

\newcommand{\tabref}[1]{\Cref{#1}}

\definecolor{lightgray}{gray}{0.75}

\newcommand{\mapchar}[2]{\ensuremath{[\!\!(#1)\!\!]^{#2}}}
\newcommand{\omapchar}[1]{\ensuremath{[\!\!(#1)\!\!]_{\mathsf{c}}}}

\newcommand{\htrigger}[2]{#1 \hookleftarrow_{\mathtt{H}} #2}

\makeatletter
\newcommand*{\rom}[1]{({\expandafter\romannumeral #1})}
\makeatother

\newcommand{\defas}{\triangleq}

\usepackage{fixme}
\fxsetup{theme=color,mode=multiuser,draft,inline,marginclue}
\definecolor{fxtarget}{rgb}{0.8300,0.1400,0.1400}
\FXRegisterAuthor{dan}{adan}{\colorbox{green!70}{\color{black}DF:}}
\newcommand{\colorpro}{VioletRed}
\newcommand{\colortyp}{JungleGreen}

\newcommand{\B}[3]{\textcolor{\colorpro}{\mathcal{B}^{\textcolor{black}{#1}}_{\textcolor{black}{#2}}\big(\textcolor{black}{#3}\big)}}
\newcommand{\mB}[3]{\textcolor{\colorpro}{\mathsf{B}^{\textcolor{black}{#1}}_{}\big(\textcolor{black}{#3}\big)}}
\newcommand{\V}[3]{\textcolor{\colorpro}{\mathcal{V}^{}_{\textcolor{black}{#2}}\big(\textcolor{black}{#3}\big)}}

\newcommand{\mV}[3]{\textcolor{\colorpro}{\mathsf{V}^{#1}_{#2}\big(\textcolor{black}{#3}\big)}}

\newcommand{\Gt}[1]{\textcolor{\colortyp}{\mathcal{G}(\textcolor{black}{#1})}}
\newcommand{\D}[1]{\textcolor{\colorpro}{\mathcal{D}(\textcolor{black}{#1})}}
\newcommand{\indT}[1]{\textcolor{\colortyp}{[\textcolor{black}{#1}\rangle}}
\newcommand{\indTaux}[2]{\textcolor{\colortyp}{[\textcolor{black}{#1}\rangle^{\star}_{\textcolor{black}{#2}}}}

\newcommand{\Rt}[1]{\textcolor{\colortyp}{\mathcal{R}(\textcolor{black}{#1})}}
\newcommand{\Rts}[3]{\textcolor{\colortyp}{\mathcal{R}^{\star}(\textcolor{black}{#3})}}
\newcommand{\envR}{\Delta_\mu}
\newcommand{\envPropR}{\Phi}
\newcommand{\mD}[1]{\textcolor{\colorpro}{\mathsf{D}(\textcolor{black}{#1})}}
\newcommand{\plen}[1]{\textcolor{\colorpro}{\lbag \textcolor{black}{#1}\rbag }}
\newcommand{\len}[1]{|#1|}
\newcommand{\incrname}[2]{\subst{#1_{#2+1}}{#1_#2}}

\newcommand{\prop}{c}
\newcommand{\apropinp}[1]{\abinp{\prop_{#1}}}
\newcommand{\propinp}[1]{\binp{\prop_{#1}}}

\newcommand{\propinpv}[1]{\binp{\prop{#1}}}
\newcommand{\propoutv}[1]{\bout{\dual{\prop{#1}}}}
\newcommand{\apropinpv}[1]{\abinp{\prop{#1}}}
\newcommand{\apropoutv}[1]{\about{\dual{\prop{#1}}}}

\newcommand{\propout}[1]{\bout{\dual{\prop_{#1}}}}
\newcommand{\propoutt}[1]{\boutt{\dual{\prop_{#1}}}}
\newcommand{\apropout}[1]{\about{\dual{\prop_{#1}}}}
\newcommand{\proprinp}[1]{\binp{\prop^{#1}}}
\newcommand{\proprout}[1]{\bout{\prop^{#1}}}
\newcommand{\aproprout}[1]{\about{\prop^{#1}}}

\newcommand{\apropbout}[1]{\abbout{\dual{\prop_{#1}}}}
\newcommand{\slhotup}[1]{{#1}^{\leadsto}}

\newcommand{\cbpropinp}[2]{\textcolor{orange!40!gray}{\binp{\prop_{#1}}{#2}}}

\newcommand{\cbapropout}[2]{\textcolor{orange!40!gray}{\apropout{#1}{#2}}}

\newcommand{\slhot}[1]{{#1} \leadsto \diamond}
\newcommand{\degree}{l}
\newcommand{\thunk}[1]{\{\!\{#1\}\!\}}
\newcommand{\appthunk}[1]{\mathtt{run}\,{#1}}

\newcommand{\tr}{\mathsf{tr}}

\newcommand{\appendx}[1]{#1}

\newcommand{\mugamma}{\gamma}
\newcommand{\misty}{\textsf{MISTY}\xspace}

\newcommand{\nextn}[1]{\ensuremath{\mathsf{next}(#1)}}

\newcommand{\lin}[1]{\mathsf{lin}(#1)}

\newcommand{\fpn}[1]{\ensuremath{\mathsf{fpn}(#1)}}

\newcommand{\wtd}[1]{\widetilde{#1}}

\definecolor{colorrevision}{rgb}{0,0,255}
{}

\declaretheorem[name=Theorem, numberwithin=section]{thm}
\declaretheorem[name=Lemma, numberwithin=section]{lemm}
\crefname{lemm}{Lemma}{Lemmas}

\newtheorem{theorem}{Theorem}[section]

\newtheorem{proposition}[thm]{Proposition}
\theoremstyle{definition}
\newtheorem{definition}{Definition}[section]
\theoremstyle{remark}
\newtheorem{remark}{Remark}[section]

\theoremstyle{definition}
\newtheorem{example}{Example}[section]

\newcommand{\hidejp}[1]{}

\newcommand{\Dlet}{\mathcal{J}}

\newcommand{\Db}[3]{\Dlet^{#1}_{#2}\big(#3\big)}
\newcommand{\Dset}{\Dlet}

\newcommand{\Cb}[3]{\mathcal{C}^{#1}_{#2}\big(#3\big)}
\newcommand{\Cbs}[3]{\widehat{\mathcal{C}}^{#1}_{#2}\big(#3\big)}
\newcommand{\Vb}[3]{\mathcal{C}^{#1}_{#2}\big(#3\big)}

\newcommand{\vrelate}{\boxtimes}
\newcommand{\valuesrelate}{\bowtie}

\newcommand{\iname}[1]{\ensuremath{\breve{#1}}}
\newcommand{\processrelate}{\diamond}

\newcommand{\rfni}[1]{\ensuremath{\mathsf{rfni}(#1)}}

\newcommand{\fcr}[1]{\ensuremath{\mathsf{cr}(#1)}}

\newcommand{\highlighta}[1]{\colorbox{lightgray!60}{\ensuremath{#1}}}

\newcommand{\indices}[1]{\ensuremath{\mathsf{index}(#1)}}

\newcommand{\initname}[2]{\ensuremath{\subst{#1_{#2}}{#1}}}

\newcommand{\rulename}[2]{\langle \texttt{#1}#2 \rangle}

\newcommand{\Rb}[1]{\ensuremath{\mathcal{R}_{#1}}}

\newcommand{\mstb}{\approx^{\mathtt{M}}}
\newcommand{\hob}{\approx^{\mathtt{H}}}

\newcommand{\mapcharm}[2]{\langle #1 \rangle^{#2}}
\newcommand{\omapcharm}[1]{\langle #1 \rangle}

\newcommand{\mhby}[1]{\hby{#1}_{\texttt m}}
\newcommand{\Mhby}[1]{\Hby{#1}_{\texttt m}}

\newcommand{\horelm}[5]{#1 \proves #2 #3 #4 \proves #5}

\newcommand{\namepass}[1]{\ulcorner #1 \urcorner}
\newcommand{\namepasstyp}[1]{\llceil #1 \rrceil}

\newcommand{\plscheck}[1]{#1}
\newcommand{\subsqn}[1]{\ensuremath{\mathsf{next}(#1)}}
\newcommand{\recprov}[3]{\binp{\prop^{#1}}{#2}\appl{#2}{#3}}
\newcommand{\recprovx}[3]{\binp{\prop^{#1}}{#2}\appl{#2}{#3}}
\newcommand{\recpvar}[1]{#1}

\excludecomment{revisionrem}

\title{A Minimal Formulation of Session Types: \\The Sessions of Trios in Concert}

\newcommand{\thesisalt}[2]{#1}
\usepackage{authblk}
\author{Alen Arslanagi\'{c}}
\author{Jorge A. P\'{e}rez}
\author{Dan Frumin}
\affil{University of Groningen, The Netherlands}
\date{\today}                     %% if you don't need date to appear
\begin{document}

\maketitle
\begin{abstract}
%Session types are a type-based approach to the verification of message-passing programs.
%They specify communication structures essential for program correctness; a session type says what and when should be exchanged through a channel.
%Central to session-typed languages are \emph{sequencing} constructs
%in types and processes that explicitly specify the order of actions in a protocol.
%
%In this paper we study session types without sequencing.
%The resulting framework of \emph{minimal} session types is arguably the simplest form of session types one could conceive.
%In the context of a core process calculus with sessions and higher-order concurrency (abstraction-passing),
%we establish two main technical results. 
%First, we prove that every process $P$ typable with standard session types can be compiled down into a process $\D{P}$ typable with minimal session types. 
%Second, we prove that $P$ and $\D{P}$ are behaviorally equivalent. 
%These results indicate that having sequencing constructs in processes and session types is convenient but \emph{redundant}: only sequentiality in processes is truly indispensable, as it can correctly codify sequentiality in types.
%
%Our developments draw inspiration from work by Parrow on behavior-preserving decompositions of untyped processes.
%By casting Parrow's results in the realm of typed processes, our developments reveal a conceptually simple formulation of session types and a
%principled avenue to the integration of session types into programming languages without sequencing in types.

%%% ====

Session types are a type-based approach to the verification of message-passing programs.
They specify communication structures essential to enforcing program correctness; by relying on \emph{sequencing} constructs, a session type can precisely describe the intended order  of communication actions through a channel.

In this paper we study a fragment of session types that makes a very limited use of sequencing; we call it \emph{minimal} session types. % because this is arguably the simplest formulation of session types one could conceive.
In the context of a core process calculus with sessions and higher-order concurrency,
we establish two technical results. 
First, we prove that every process $P$ typable with standard session types can be compiled down into a process $\D{P}$ typable with minimal session types. 
Second, we prove that $P$ and $\D{P}$ are behaviorally equivalent. 
These results show that having sequencing in both processes and session types is convenient, but that only sequencing in processes is truly indispensable, as it can correctly codify sequencing  in types.

Our developments draw inspiration from work by Parrow on behavior-preserving decompositions of untyped processes using \emph{trios}, i.e., processes with exactly three nested prefixes.
By casting Parrow's approach in the realm of typed processes, our developments reveal a conceptually simple formulation of session types, supported by static and dynamic correctness results.
\end{abstract}

\section{Introduction}  
\label{s:intro}
Session types are a type-based approach to the verification of message-passing programs.
A session type specifies what messages should be exchanged through a channel and in what order.
This makes session types a useful tool to enforce safety and liveness properties related to communication correctness.
Originated in the realm of concurrency theory---and using the $\pi$-calculus as specification language---session types have had a significant impact on the foundations of programming languages~\cite{DBLP:journals/csur/HuttelLVCCDMPRT16}, but also on their practice~\cite{DBLP:journals/ftpl/AnconaBB0CDGGGH16}.
\plscheck{Our goal in this work is to understand to what extent session types can admit simpler formulations and to  clarify the status of such formulations.}
%This foundational question has concrete practical ramifications, as we discuss next.

In session-typed languages, \emph{sequencing} constructs in types and processes specify the intended structure of message-passing protocols.
\plscheck{Sequencing is what distinguishes session types from other type systems, such as simple types~\cite{PierceSangiorgi95} and linear types~\cite{LinearPi}, in which a channel's type specifies exactly one communication action, rather than a series of actions.}
For example, in the session type $S = \btinp{\mathsf{int}} \btinp{\mathsf{int}} \btout{\mathsf{bool}} \tinact$, sequencing (denoted `$;$') allows us to specify a protocol for a channel that \emph{first} receives (?)
two integers, \emph{then} sends (!)
a boolean, and \emph{finally} ends ($\tinact$).
As such, $S$ could be the type of a service that checks for integer equality.

Sequencing in types goes hand-in-hand with sequencing in processes, which is specified using the prefix construct (denoted~`$.$').
The $\pi$-calculus process $P = \binp{s}{x_1} \binp{s}{x_2} \bout{s}{x_1=x_2} \inact$ is an implementation of the equality service: it \emph{first} expects two values on name $s$, \emph{then} outputs a boolean on $s$, and \emph{finally} stops.
We can see that name $s$ in $P$ conforms to the session type~$S$.
Session types can also specify protocols in which sequencing occurs within labeled choices and recursion; these typed constructs are also in close match with their respective process expressions.
This way, e.g., a recursive variant of 
$S$ is $\trec{t}\btinp{\mathsf{int}} \btinp{\mathsf{int}} \btout{\mathsf{bool}} \vart{t}$, which specifies a \emph{persistent} integer equality service. 

Session types were originally developed as a typing discipline for the analysis of binary (two-party) protocols~\cite{honda.vasconcelos.kubo:language-primitives}.
Since then session types have been extended in many directions.
We find, for instance, multiparty session types~\cite{HYC08}, and extensions with dependent types, assertions, and exceptions (cf.~\cite{DL10,DBLP:journals/csur/HuttelLVCCDMPRT16} for surveys).
Building upon sequencing constructs, all these extensions seek to address natural questions on the expressivity and applicability of session types theories.

\plscheck{In this paper, we study a different  direction.
We identify a \emph{fragment} of (binary) session types that makes a very limited use of sequencing constructs.
In our fragment, 
the type $S$ in 
session types such as `$\btout{U} S$' and `$\btinp{U} S$' can only correspond to $\tinact$.
Sequencing is confined to recursive session types, which can only be of the forms `$ \trec{t}{\btout{U} \vart{t}}$' and `$ \trec{t}{\btinp{U} \vart{t}}$'.
That is, we admit sequencing solely for the purpose of specifying recursive protocols with a single action.
Because of this restricted role of sequencing, we call the resulting fragment \emph{minimal session types}}.
 
\plscheck{We study minimal session types on top of \HO, the core process calculus for session-based concurrency studied by Kouzapas et al.~\cite{DBLP:journals/iandc/KouzapasPY19}.
\HO is a very small language, which only supports passing of abstractions (i.e., functions from names to processes) and lacks name-passing and recursion.
Still, \HO is very expressive, because both features can be encoded in \HO in a fully abstract way.
Moreover, \HO has a well-developed theory of behavioral equivalences~\cite{KouzapasPY17}. 
%This combination of minimal number of features and expressivity makes \HO an excellent candidate for our study. 
Being a strict fragment of the typed framework by Kouzapas et al., processes typed with our minimal session types for \HO inherit all the technical properties and behavioral theory from \cite{DBLP:journals/iandc/KouzapasPY19,KouzapasPY17}}.

Leveraging the economical formulation and expressivity of \HO, we investigate whether the limited use of sequencing in minimal session types 
is a too drastic restriction. 
It turns out that it is not:
we show that for every \HO process $P$ typable under standard session types, there is a  \emph{decomposition} of $P$, denoted $\D{P}$, a process that codifies the sequencing information given by the session types (protocols) of $P$ using additional synchronizations.
\Cref{f:idea} illustrates the key idea using the process $P$ and session type $S$ motivated above. 
Because $P$ contains three actions in sequence (as specified by $S$), its decomposition  $\D{P}$ consists of three processes in parallel---each implementing a single action of $P$---as well as of orchestration mechanisms between them: the synchronizations on names $c_2, \ldots, c_5$ ensure that the sequencing in $P$ is preserved and that received names are properly propagated. These three parallel processes are typable with minimal session types, which are obtained by ``slicing'' $S$ (in the figure, these types appear below each process).

\begin{figure}[!t]
\begin{mdframed}
		\begin{center}
		\hspace*{0.2cm} 
		\begin{tikzpicture}[scale=1.75]
		\tikzstyle{ann} = [draw=none,fill=none,right]
		\pgfmathsetmacro{\h}{0.0}
		\pgfmathsetmacro{\a}{-2.5}
		\pgfmathsetmacro{\d}{-1.3}

			\node[draw,dotted](A) at (\h, 0) {$P = \binp{s}{x_1}\binp{s}{x_2}\bout{s}{x_1 = x_2} \inact$};
			\node[draw=none] at (\h,-0.45) {$s : {?({\mathsf{int}})}; {?({\mathsf{int}})};
		 {!\langle {\mathsf{bool}}\rangle};\tinact$};
			\node[draw,dotted] (B) at (\h, \d) {$\D{P}$};
			\draw[->, thick] (\h,-0.575 - 0.2) -- (\h,\d+0.3);
			\node[draw,dotted] (T1) at (\h-3.0,\a) {$\cbpropinp{2}{}\binp{s_1}{x_1}\cbapropout{3}{x_1}$};
			\node at (\h-1.5,\a) {$\parallel$};
			\node[draw,dotted] (T2) at (\h,\a) {$\cbpropinp{3}{y_1}\binp{s_2}{x_2}\cbapropout{4}{y_1, x_2}$}; 
			\node at (\h+1.45, \a) {$\parallel$};
			\node[draw,dotted] (T3) at (\h+3.0,\a) {$\cbpropinp{4}{x_1, x_2}\bout{s_3}{x_1 = x_2}\cbapropout{5}{}$};
			\node at (\h-3.0,\a-0.45) {$s_1:{?({\mathsf{int}})};\tinact$};
			\node at (\h-3.0,\a-0.75) {$\prop_2:{?()};\tinact$};
			\node at (\h,\a-0.45) {$s_2:{?({\mathsf{int}})};\tinact$};
			\node at (\h,\a-0.75) {$\prop_3:{?({\mathsf{int}})};\tinact$};
%			\node at (0,-4.10) {$\prop_4:\textcolor{red}{?}\blue{\mathsf{int}};\tinact$};
			\node at (\h+3.0,\a-0.45) {$s_3:{!\langle{\mathsf{bool}}\rangle};\tinact$};
			\node at (\h+3.0,\a-0.75) {$\prop_4:{?({\mathsf{int}, \mathsf{int}})};\tinact$};
			\draw[->, color=magenta, thick](\h-2.5,\a+0.25) edge[bend left] node [left] {} (\h-1,\a+0.25);
			\draw[->, color=magenta, thick](\h+0.5,\a+0.25) edge[bend left] node [left] {} (\h+2,\a+0.25);
			\node[draw=none] {};
			\draw[->, color=magenta, thick] (B) edge[bend right=22] node [left] {} (T1); 
			%\draw[->, dotted, thick] (B) edge (T2);
			\draw[<-, color=magenta, thick] (B) edge[bend left=22] node [left] {} (T3); 
		\end{tikzpicture}
		\end{center}
		\end{mdframed}
\caption{The process decomposition, illustrated. Arrows in \textcolor{magenta}{magenta} indicate synchronizations orchestrated by the decomposition $\D{P}$ into trios. \label{f:idea}}
		\end{figure}
		
The definition of $\D{P}$ is interesting on its own, as it draws inspiration from a known result by Parrow, who showed that any \emph{untyped} $\pi$-calculus process can be decomposed as a collection of \emph{trios}, i.e., processes with exactly three nested prefixes~\cite{DBLP:conf/birthday/Parrow00}.
As already mentioned, \HO is a higher-order language and so is different from the (untyped, first-order) $\pi$-calculus considered by Parrow. 
Indeed, as we will see, several aspects of our decomposition (and of its properties) take advantage of the higher-order nature of  \HO.

We establish two technical results for $\D{P}$: first, it is well-typed using minimal session types (\emph{static correctness}); second, it is behaviorally equivalent to $P$ (\emph{dynamic correctness}).
These properties ensure that having sequencing in both types and processes is convenient, but that only sequencing at the level of processes is truly indispensable, as it can correctly codify sequencing  in types.

%\textcolor{blue}{In what sense minimal session types are truly minimal? Where do they stand with respect to type systems for the pi-calculus without sequencing? We certainly do not claim that MSTs are the most minimal structure possible, not that they are unheard of. Adapt the narrative based on relative/absolute expressiveness.}

\plscheck{As explained above, we dub our fragment of session types \emph{minimal} to stress the limited use of sequencing with respect to standard presentations of session types, in which sequencing occurs unconstrained. That is, we interpret `minimality' in terms of the role of sequencing in types for channels.  Interestingly, as we discuss later on, the constructs present in our minimal session types correspond to forms already studied in the literature.  One consequence of our technical results is that minimal session types stand between linear types (in which there is no sequencing nor recursion in types) and usual session types. 
Also, our study of the process decomposition into trios complement prior works that seek to formally relate session types with other type systems, such as, generic types~\cite{DBLP:journals/corr/GayGR14} and linear types~\cite{DBLP:conf/concur/DemangeonH11,DBLP:conf/ppdp/DardhaGS12,DBLP:journals/iandc/DardhaGS17}. See  \secref{s:rw} for a discussion on related works.}
%As such, these prior studies concern the \emph{relative expressiveness} of session types, where the expressivity of session types stands with respect to that of some other type system.
%In sharp contrast, we study the \emph{absolute expressiveness} of session types: how session types can be explained in terms of themselves.
%To our knowledge, this is the first study of its kind. %, and may open up new avenues for foundational and practical research.

% \input{contributions_and_outline.tex}

\paragraph{Contributions \& Outline.}
In summary, in this paper we present the following contributions:
\begin{enumerate}
\item We identify \emph{minimal session types} (MSTs): a fragment of standard session types for \HO that makes a very limited use of sequencing  (\Cref{mst:d:mtypesi}).
\item We show how to correctly decompose processes typable with standard session types into processes typable with minimal session types, and how to decompose standard session types into minimal session types. This is a result of static correctness (\Cref{mst:t:decompcore}).
\item We show that the decomposition of a process is behaviorally equivalent to the original process. This is a result of \emph{dynamic correctness}, formalized in terms of \emph{MST bisimulations}, a typed behavioral equivalence that we introduce here (\Cref{mst:t:dyncorr}).
\item We develop optimizations and extensions of our decomposition that bear witness to its robustness.
\end{enumerate}
The rest of the paper is organized as follows.
In \secref{s:lang} we recall the preliminaries on the session type system for \HO, the core process calculus for session-based concurrency on which we base our developments.
In \secref{s:decomp} we present {minimal session types} and the decomposition of well-typed \HO processes into minimal session types processes, accompanied by explanations and examples.
In \secref{ss:dynamic} we show the correctness of the decomposition, by establishing an \emph{MST bisimulation} between an \HO process and its decomposition.
In \secref{s:opt} we examine two optimizations of the decomposition that are enabled by the higher-order nature of our setting.
In \secref{ss:exti} we discuss extensions of our approach to consider constructs for branching and selection.
Finally, in \secref{s:rw} we discuss related work and in \secref{s:concl} we present some closing remarks.
\appendx{The appendix contains omitted definitions and proofs.}

\paragraph{Differences with the conference version.}
An earlier version of this paper was presented at ECOOP 2019~\cite{APV19}.
The current paper revises the conference version, includes further examples, and incorporates a major addition: \secref{ss:dynamic} on dynamic correctness, including the notion of an MST bisimulation and the constructed bisimulation relation,  is completely new to this presentation.

\paragraph{Colors.}
Throughout the paper we use  colors (such as \textcolor{\colorpro}{pink}  and \textcolor{\colortyp}{green}) to improve readability.
However, the usage of colors is not indispensable to follow the paper (black and white suffice).

\section{The Source Language}
\label{s:lang}
We start by recalling the syntax, semantics, and type system for \HO, the higher-order process calculus for session-based concurrency studied by Kouzapas et al.~\cite{DBLP:journals/iandc/KouzapasPY19,KouzapasPY17}. {Our presentation of \HO follows the aforementioned papers, which concern definitions and results for \HOp, the super-calculus of \HO with name-passing, abstraction-passing, and recursion.}

\HO is arguably the simplest language for session types: it supports
passing of abstractions (functions from names to processes)
but does not support name-passing nor process recursion.
Still, \HO is very expressive:  it can encode name-passing, recursion, and polyadic communication via type-preserving encodings that are fully-abstract with respect to  contextual equivalence~\cite{DBLP:journals/iandc/KouzapasPY19}.

\subsection{Syntax and Semantics}
\begin{definition}[\HO processes]
	The syntax of names, variables, values, and \HO processes  is defined as follows:
		\begin{align*}
			n, m  & \bnfis a,b \bnfbar s, \dual{s}
			\qquad \quad 
			u,w   \bnfis n \bnfbar x,y,z
			\qquad \quad 
			V,W   \bnfis {x,y,z} \bnfbar {\abs{x}{P}}
			\\[1mm]
			P,Q
			 & \bnfis
			\bout{u}{V}{P}  \sbnfbar  \binp{u}{x}{P} \sbnfbar
      {\appl{V}{u}}
			 \sbnfbar P\Par Q \sbnfbar \news{n} P
			\sbnfbar \inact
		\end{align*}
\end{definition}
\noindent
We use $a,b,c, \dots$ to range over \emph{shared names}, and $s, \dual{s}, \dots$ to range over \emph{session names}. Shared names are used for unrestricted, non-deterministic interactions; session names are used for linear, deterministic interactions. We write $n, m$ to denote session or shared names, and assume that the sets of session and shared names are disjoint. The \emph{dual} of a name $n$ is denoted $\dual{n}$; we define $\dual{\dual{s}} = s$ and $\dual{a} = a$, i.e., duality is only relevant for session names. Variables are denoted with $x, y, z, \dots$. An abstraction %(or higher-order value)
$\abs{x}{P}$ is a process $P$ with parameter $x$.
\emph{Values} $V,W, \ldots$ include variables and abstractions, but not names.
 
Process
$\appl{V}{u}$
is the application
which substitutes name $u$ on abstraction~$V$. Constructs for
inaction $\inact$,  parallel composition $P_1 \Par P_2$, and
name restriction $\news{n} P$ are standard.
\HO lacks name-passing and recursion, but they are expressible in the language (see \exref{top:ex:np} below).

To enhance readability, we often omit trailing $\inact$'s, so we write, e.g., 
$\about{u}{V}{}$ instead of  $\bout{u}{V}{\inact}$.
Also, we write 
$\bout{u}{}{P}$ and $\binp{u}{}{P}$ whenever the exchanged value is not relevant (cf. \remref{mst:r:prefix}).

Restriction for shared names $\news{a} P$ is as usual; session name restriction $\news{s} P$ simultaneously binds session names $s$ and $\dual{s}$ in $P$.
Functions $\fv{P}$, $\fn{P}$, and $\fs{P}$ denote, respectively, the sets of free
%\jpc{recursion}
variables, names, and session names in $P$, and are defined as expected.
%We assume $V$ in $\bout{u}{V}{P}$ does not include free recursive variables $\rvar{X}$.
If $\fv{P} = \emptyset$, we call $P$ {\em closed}.
 We write $P\subst{u}{y}$ (resp.,\,$P\subst{V}{y}$) for the capture-avoiding substitution of name $u$ (resp.,\, value $V$) for $y$ in process $P$.
We identify processes up to consistent renaming of bound names, writing $\scong_\alpha$ for this congruence.
We shall rely on Barendregt's variable convention, under which free and bound names are different in every mathematical context.

The  operational semantics of \HO is defined in terms of a \emph{reduction relation},
denoted $\red$.
Reduction is closed under \emph{structural congruence},
denoted $\scong$, which is defined as the smallest congruence on processes such that:
	\begin{align*}
		& P \Par \inact \scong P
		\qquad
		P_1 \Par P_2 \scong P_2 \Par P_1
		\qquad
		P_1 \Par (P_2 \Par P_3) \scong (P_1 \Par P_2) \Par P_3
		\\
		& \news{n} \inact \scong \inact
%		\quad
		\qquad
		 P \Par \news{n} Q \scong \news{n}(P \Par Q)
		\ (n \notin \fn{P})
%		\quad
%		\recp{X}{P} \scong P\subst{\recp{X}{P}}{\rvar{X}}
%		\\[1mm]
		\qquad
		P \scong Q \textrm{ if } P \scong_\alpha Q
	\end{align*}
	We assume the expected extension of $\scong$ to values $V$.
The reduction relation expresses the behavior of processes; it is defined as follows:
	\begin{align*}
		\appl{(\abs{x}{P})}{u}   & \red  P \subst{u}{x}
		& \orule{App} &
		\\%[1mm]
		\bout{n}{V} P \Par \binp{\dual{n}}{x} Q & \red  P \Par Q \subst{V}{x}
		& \orule{Pass} &
		\\%[1mm]
% 		\bsel{n}{l_j} Q \Par \bbra{\dual{n}}{l_i : P_i}_{i \in I} & \red Q \Par P_j
% ~~(j \in I)~~
% 		& \orule{Sel} &
% 		\\%[1mm]
%		&&
		P \red P'\Rightarrow  \news{n} P   & \red    \news{n} P'
		& \orule{Res} &
		\\%[1mm]
		P \red P'  \Rightarrow    P \Par Q  & \red   P' \Par Q
		& \orule{Par} &
		\\%[1mm]
%		&&
		P \scong Q \red Q' \scong P'  \Rightarrow  P  & \red  P'
		& \orule{Cong} &
	\end{align*}
%	The rules are largely self-explanatory.
 Rule $\orule{App}$ defines  name application ($\beta$-reduction). Rule $\orule{Pass}$ defines a
shared or session interaction, depending on the nature of $n$. 
Other rules are standard $\pi$-calculus rules.
We write 
$\red^k$ for a $k$-step reduction, and 
$\red^\ast$ for the reflexive, transitive closure of $\red$.

We illustrate \HO processes and their semantics by means of an example.

\begin{example}[Encoding Name-Passing]
  \label{top:ex:np}
  The \HO calculus lacks the name-passing primitives of \HOp. 
  Hence, it cannot express reductions of the form
\begin{align}
\bout{n}{m} P \Par \binp{\dual{n}}{x} Q  \red  P \Par Q \subst{m}{x}
\label{red:ex}
\end{align}
Fortunately,  name-passing can be
encoded in \HO in a fully-abstract way: as shown in \cite{DBLP:conf/esop/KouzapasPY16,DBLP:journals/iandc/KouzapasPY19}, one can use abstraction passing to ``pack'' a name. 
To this end, \figref{f:np} defines the required syntactic sugar, at the level of processes and types.
Then, the reduction \eqref{red:ex} can be mimicked as
\begin{align*}
{\bout{n}{\namepass{m}} P \Par \binp{\dual{n}}{\namepass{x}} Q} =~~ & \bbout{n}{ \abs{z}{\,\binp{z}{x} (\appl{x}{m})} } {P}  \Par \binp{\dual n}{y} \newsp{s}{\appl{y}{s} \Par \about{\dual{s}}{\abs{x}{{Q}}}}
 \\
 \red~~ &  {P}  \Par \newsp{s}{\appl{\abs{z}{\,\binp{z}{x} (\appl{x}{m})}}{s} \Par \about{\dual{s}}{\abs{x}{{Q}}}}
 \\
 \red~~ &  {P}  \Par \newsp{s}{\binp{s}{x} (\appl{x}{m}) \Par \about{\dual{s}}{\abs{x}{{Q}}}}
 \\
 \red~~ &  {P}  \Par \appl{(\abs{x}{{Q}})}{m}
 \\
 \red~~ &  {P}  \Par {Q}\subst{m}{x}
\end{align*}\hspace*{\fill} $\lhd$
\end{example}

\begin{figure}[!t]
  \begin{mdframed}
  \vspace{-4mm}
        \begin{align*}
        \bout{t}{\namepass{\tilde u}}P & \triangleq \about{t}{\lambda z. \abinp{z}{x}.(x\ {\tilde u})}.P\\
        \binp{t}{\namepass{\tilde x}}Q & \triangleq \abinp{t}{y}.\big(\news {z}(y\ z \Par \about{z}{\lambda {\tilde x}. Q})\big)
      \\ \vspace{2mm}
        {\namepass{\wtd S}} &\defas \lhot{(\btinp{\lhot{\namepasstyp{\wtd S}}}\tinact)} \\
        {\namepass{\chtype{\wtd S}}} &\defas 
        \lhot{(\btinp{\lhot{\chtype{\namepasstyp{\wtd S}}}}\tinact)} 
        \\ 
        \namepass{\lhot{\wtd C}} &\defas  \lhot{\namepasstyp{\wtd C}}
        \\ 
        \namepass{\shot{\wtd C}} &\defas  \shot{\namepasstyp{\wtd C}}
        \\
        \namepasstyp{\btout{{U}}S} &\defas 
        {\btout{\namepass{U}}\namepasstyp{S}}
        \\
        \namepasstyp{\btinp{{U}}S} &\defas 
        {\btinp{\namepass{U}}\namepasstyp{S}}
        \\
        \namepasstyp{C_1,\ldots,C_n} &\defas 
        \namepasstyp{C_1}, \ldots, \namepasstyp{C_n} 
      %  \btout{\namepass{S_1}}S_2 &\defas \lhot{(\btinp{\lhot{\namepasstyp{S_1}}}\tinact)} \\
      %   \btout{\namepass{S_1}}S_2 &\defas \lhot{(\btinp{\lhot{\namepasstyp{S_1}}}\tinact)}
      \end{align*}
      \end{mdframed}
  \caption{Encoding name passing in \HO}\label{f:np}\end{figure}

\begin{remark}[Polyadic Communication]
\label{r:poly}
\HO as presented above allows only for \emph{monadic communication}, i.e., the exchange of tuples of values with length 1.
We will find it convenient to use \HO with \emph{polyadic communication}, i.e., the exchange of tuples of values $\widetilde{V} = (V_1,\ldots,V_k)$, with length $\len{\widetilde{V}} = k$.
We will use similar notation for tuples of names and variables, and we will use $\epsilon$ to denote the empty tuple.

In \HO, polyadicity appears in session synchronizations and applications, but not in synchronizations on shared names. This entails having the following reduction rules:
	\begin{align*}
		\appl{(\abs{\widetilde{x}}{P})}{\widetilde{u}}   & \red  P \subst{\widetilde{u}}{\widetilde{x}}
		\\
		\bout{s}{\widetilde{V}} P \Par \binp{\dual{s}}{\widetilde{x}} Q & \red  P \Par Q \subst{\widetilde{V}}{\widetilde{x}}
\end{align*}
where the simultaneous substitutions
 $P\subst{\widetilde{u}}{\widetilde{x}}$
 and $P\subst{\widetilde{V}}{\widetilde{x}}$
 are as expected.
This polyadic \HO can be readily encoded into (monadic) \HO~\cite{KouzapasPY17}; for this reason, by a slight
 abuse of notation we will often write \HO when we actually mean ``polyadic \HO''.
\end{remark}

We discuss two simple examples that illustrate   how \HO can implement mechanisms resembling servers and forms of partial instantiation; these mechanisms shall come in handy later, when defining the process decomposition in \Cref{s:decomp}.

\begin{example}[A Server of a Kind]
	\label{e:server}
	Let $S_a$ denote the process $\binp{a}{x} (\appl{x}{r})$, which receives an abstraction on the shared name $a$ and then applies it to $r$. Consider the following process composition:
	\begin{align*}
	P &= 	\news{r}\news{a}\big(\bout{a}{V} \inact  \Par \bout{a}{W} \inact  \Par S_a 
  \Par \binp{\dual r}{x_1}\binp{\dual r}{x_2}Q \big)
	\\
	V& = \abs{y}{(\bout{y}{V'}S_a\subst{y}{r})}
	\\
	W& = \abs{z}{(\bout{z}{W'}S_a\subst{z}{r})}
		\end{align*}
    \noindent where $V'$ and $W'$ are some unspecified shared values. 
		In $P$, process $S_a$ operates as a server that provides $r$ upon an invocation on  $a$. 
		Dually, the outputs on $a$ are requests to this server. 
		One possible reduction sequence for $P$ is the following:
	\begin{align*}
		P 
		& \red 
		\news{r} \news{a}\big(\bout{a}{W} \inact  \Par \appl{V}{r} \Par 
    \binp{\dual r}{x_1}\binp{\dual r}{x_2}Q  \big) 
		\\
		 & \red 
     \news{r} \news{a}\big(\bout{a}{W} \inact  \Par \bout{r}{V'}S_a \Par 
    \binp{\dual r}{x_1}\binp{\dual r}{x_2}Q \big) 
		\\		
& 		\red 
\news{r} \news{a}\big(\bout{a}{W} \inact  \Par S_a \Par 
    \binp{\dual r}{x_2}Q\subst{V'}{x_1}  \big) = P'
	\end{align*}
        In this reduction sequence, the value $V$ in the first request is instantiated with the name $r$ by consuming a copy of $S_a$ available in $P$.
        However, a copy of the server $S_a$ is restored through the value $V$, after an communication on $r$.
        This way, in $P'$ the exchange of $W'$ on $r$ can take place:
$$
P' \red^\ast 
\news{r}\news{a}\big(S_a \Par Q\subst{V'}{x_1}\subst{W'}{x_2}\big)
$$
% \hspace*{\fill} $\lhd$
\end{example}

\begin{example}[Partial Instantiation] 
	\label{e:partial}
	Let $S_a$ and $S_b$ be servers as defined in the previous 	example: 
	\begin{align*}
		S_a = \binp{a}{x} (\appl{x}{r}) \qquad \qquad 
		S_b = \binp{b}{x} (\appl{x}{v})
	\end{align*}
	Further, let $R$ be a process in which requests to $S_a$ and $S_b$ are nested within abstractions: 
	 \begin{align*}
           R & = \abbout{a}{\abs{y}\about{b}{\abs{z}\appl{V}{(y,z)}}}
         \end{align*}
         \noindent Notice how the polyadic application `$\appl{V}{(y,z)}$' is enclosed in the innermost abstraction.
         Now consider the following composition:
		\begin{align*}
			P = \news{a,b}  R \Par S_a \Par S_b  
		\end{align*}
	The structure of $R$ induces a form of partial instantiation for $y, z$, implemented by combining synchronizations and $\beta$-reductions. To see this, let us inspect one possible reduction chain for $P$:  
		\begin{align*}
			P &\red \news{b} \left(\appl{\abs{y}\about{b}{\abs{z}\appl{V}{(y,z)}}}{r}\right) \Par  S_b 
			\\ 
			&\red \news{b}{\about{b}{\abs{z}\appl{V}{(r,z)}}} \Par  S_b = P'
		\end{align*}
The first request of $R$, aimed to obtain name $r$,  is 
  	realized by the first reduction, i.e., the communication with $S_a$ on name $a$: 
  	the result is the application of the top-level abstraction to $r$. 
  	Subsequently, the application step  substitutes 
  	$y$ with $r$. 
  	Hence, in $P'$,  
  	names in the nested application  are only \emph{partially instantiated}: at this point, we have `$\appl{V}{(r,z)}$'. 
  	
  	  	 Process  $P'$ can then execute the same steps to instantiate  $z$ with name $v$
  	by interacting with $S_b$. 
  	After two reductions, we obtain the fully instantiated application  $\appl{V}{(r,v)}$: 
  			\begin{align*}
			P' & \red \appl{{\abs{z}\appl{V}{(r,z)}}}{v} 
			\red \appl{V}{(r,v)}
		\end{align*}
  
    % \hspace*{\fill} $\lhd$
  	\end{example}

\subsection{Session Types for \HO}
\label{sec:types}

We give essential definitions and properties for the session type system for \HO, following~\cite{KouzapasPY17}.
\begin{definition}[Session Types for \HO]
\label{d:types}
Let us write $\Proc$ to denote the process type.
The syntax of value types $U$, channel types $C$, and session types $S$ for \HO is defined as follows: % (this can be simplified):
	\begin{align*}
%		U & \bnfis 	 {L}
%		\\
U & \bnfis		\shot{C} \bnfbar \lhot{C}
		% \\
    \qquad \qquad 
		C  \bnfis		S \bnfbar  {\chtype{U}}
\\
		S & \bnfis  \btout{U} S \bnfbar \,\btinp{U} S %\bnfbar
%		\\
		%  \bnfbar \btsel{l_i:S_i}_{i \in I} \bnfbar \btbra{l_i:S_i}_{i \in I}
%		\\
		 \bnfbar  \,\trec{t}{S} \bnfbar \vart{t}   \bnfbar 	\tinact 
	\end{align*}
\end{definition}
\noindent
As we have seen, \HO only admits the exchange of abstractions; accordingly,
value types include $\shot{C}$ and $\lhot{C}$, which denote {\em shared} and
{\em linear} higher-order types, respectively. Channel types include session
types and the shared types $\chtype{U}$. 

Session types follow the standard binary session type syntax~\cite{honda.vasconcelos.kubo:language-primitives}, in which sequencing specifies communication structures.
This way, the {\em output type} $\btout{U} S$ describes a session in which first a value of type $U$ is sent, and then the session proceeds as $S$.
Dually, the \emph{input type} $\btinp{U} S$ describes a session in which first a value of type $U$ is received, and then the session proceeds as $S$.
In examples, we often assume basic types (such as $\mathsf{int}$, $\mathsf{bool}$, $\mathsf{str}$) are exchanged in communications.
Session types also include {\em recursive types} $\trec{t}{S}$, in which the variable $\vart{t}$ is assumed to occur guarded in $S$, i.e., types such as $\trec{t}{\vart{t}}$ are not allowed.
In most cases, recursive types will be \emph{tail-recursive}, although instances of \emph{non-tail-recursive} session types will also be relevant (cf.
\Cref{mst:ex:ntrsts}).
Finally, type $\tinact$ is the type of the terminated protocol.

\begin{notation}
As mentioned in the introduction, we shall study session types in which the continuation $S$ in $\btout{U} S$  and $\btinp{U} S$ is always $\tinact$. 
Given this, we may sometimes omit trailing $\tinact$'s and write  
$\btoutt{U}$  and $\btinpt{U}$
rather than
$\btout{U} \tinact$  and $\btinp{U} \tinact$, respectively. 
\end{notation}

In theories of session types  \emph{duality} is a key notion:
implementations derived from dual session types will respect their protocols at run-time, avoiding communication errors.
Intuitively, duality is  obtained by
exchanging $!$ by $?$ (and vice versa), % and  $\oplus$ by $\&$ (and vice versa),
including the fixed point construction.
We write $S \dualof T$ if session types $S$ and $T$ are dual according to this intuition;
the formal definition is coinductive, and given in~\cite{KouzapasPY17} (see also \cite{DBLP:journals/corr/abs-2004-01322}).

We consider shared, linear, and session
\emph{environments}, denoted $\Gamma$, $\Lambda$, and $\Delta$, \newjb{resp.}:
\begin{align*}
		\Gamma  & \bnfis  \emptyset \bnfbar \Gamma \cat \varp{x}: \shot{C} 
		%\bnfbar \Gamma \cat u: \chtype{S} 
		\bnfbar \Gamma \cat u: \chtype{U}
		%\bnfbar \Gamma \cat \rvar{X}: \Delta
\\
		\Lambda  & \bnfis  \emptyset \bnfbar \Lambda \cat \AT{x}{\lhot{C}}
		 \\
		\Delta    & \bnfis   \emptyset \bnfbar \Delta \cat \AT{u}{S}
\end{align*}
%Environment
$\Gamma$ maps variables and shared names to value types; % (see below);
$\Lambda$ maps variables to
%the
linear
%functional
higher-order types. $\Delta$  maps session names to session types. While
$\Gamma$ admits weakening, contraction, and exchange principles, both $\Lambda$
and $\Delta$
%behave linearly: they
are only subject to exchange. The domains of $\Gamma, \Lambda$, and $\Delta$ are
assumed pairwise distinct. We write $\Delta_1\cdot \Delta_2$ to denote the
disjoint union of $\Delta_1$ and $\Delta_2$.

We write $\Gamma\backslash x$ to denote the environment obtained from $\Gamma$ by removing the assignment $x : \shot{C}$, for some $C$.
Notations $\Delta\backslash u$ and $\Gamma\backslash \widetilde{x}$ are defined similarly and have the expected readings.
With a slight abuse of notation, given a tuple of variables $\widetilde x$, we sometimes write $(\Gamma, \Delta)(\widetilde x)$ to denote the tuple of types assigned to the variables in $\widetilde x$ by the environments $\Gamma$ and $\Delta$.

The typing judgements for values $V$ and processes $P$ are denoted
\[
\Gamma; \Lambda; \Delta \proves V \hastype U \qquad \text{and} \qquad\Gamma; \Lambda; \Delta \proves P \hastype \Proc
\]

\begin{figure}[t!]
	\begin{mdframed}
	\[
		\begin{array}{c}
		\inferrule[(Sess)]{}{\Gamma; \emptyset; \set{u:S} \proves u \hastype S} 
			\qquad
			\inferrule[(Sh)]{}{\Gamma \cat u : U; \emptyset; \emptyset \proves u \hastype U}
			\\ \\ 
			\inferrule[(LVar)]{}{\Gamma; \set{x: \lhot{C}}; \emptyset \proves x \hastype \lhot{C}}
							\qquad
			\inferrule[(RVar)]{}{\Gamma \cat \rvar{X}: \Delta; \emptyset; \Delta  \proves \rvar{X} \hastype \Proc}
					\\  \\
			\inferrule[(Abs)]{
				\Gamma; \Lambda; \Delta_1 \proves P \hastype \Proc
				\quad
				\Gamma; \es; \Delta_2 \proves x \hastype C
			}{
				\Gamma\backslash x; \Lambda; \Delta_1 \backslash \Delta_2 \proves \abs{{x}}{P} \hastype \lhot{{C}}
			}
			\quad
			\inferrule[(App)]{
				%\begin{array}{c}
					%U = \hot{C} \lor \shot{C}
					%\\
					{\Gamma; \Lambda; \Delta_1 \proves V \hastype \ghot{C} ~~
					\leadsto\, \in \{\lollipop, \sharedop\} \quad
					\Gamma; \es; \Delta_2 \proves u \hastype C
					}
			%	\end{array}
			}{
				\Gamma; \Lambda; \Delta_1 \cat \Delta_2 \proves \appl{V}{u} \hastype \Proc
			} 
	
			\\  \\
			\inferrule[(Prom)]{
				\Gamma; \emptyset; \emptyset \proves V \hastype 
							 \lhot{C}
			}{
				\Gamma; \emptyset; \emptyset \proves V \hastype 
							 \shot{C}
			} 
			\quad
			\inferrule[(EProm)]{
			\Gamma; \Lambda \cat x : \lhot{C}; \Delta \proves P \hastype \Proc
			}{
				\Gamma \cat x:\shot{C}; \Lambda; \Delta \proves P \hastype \Proc
			}
					\quad
			\inferrule[(End)]{
				\Gamma; \Lambda; \Delta  \proves P \hastype T \quad u \not\in \dom{\Gamma, \Lambda,\Delta}
			}{
				\Gamma; \Lambda; \Delta \cat u: \tinact  \proves P \hastype \Proc
			}
			\\  \\
			\inferrule[(Rec)]{
				\Gamma \cat \rvar{X}: \Delta; \emptyset; \Delta  \proves P \hastype \Proc
			}{
				\Gamma ; \emptyset; \Delta  \proves \recp{X}{P} \hastype \Proc
			}
			\qquad
				\inferrule[(Par)]{
				\Gamma; \Lambda_{i}; \Delta_{i} \proves P_{i} \hastype \Proc \quad i=1,2
			}{
				\Gamma; \Lambda_{1} \cat \Lambda_2; \Delta_{1} \cat \Delta_2 \proves P_1 \Par P_2 \hastype \Proc
			}
			\qquad
					\inferrule[(Nil)]{ }{\Gamma; \emptyset; \emptyset \proves \inact \hastype \Proc}
			\\  \\
			\inferrule[(Req)]{
			\begin{array}{c}
			\Gamma; \Lambda; \Delta_1 \proves P \hastype \Proc
			\quad
			\Gamma; \es; \es \proves u \hastype	\chtype{U}
			\\
			\Gamma; \es; \Delta_2 \proves V \hastype U		
				\end{array}
			}{
				\Gamma; \Lambda; \Delta_1 \cat \Delta_2 \proves \bout{u}{V} P \hastype \Proc
			}
			~~
%			\\ \\
%--
%			~~
			\inferrule[(Acc)]{
				\begin{array}{c}
					\Gamma; \Lambda_1; \Delta_1 \proves P \hastype \Proc
					\quad 
					\Gamma; \emptyset; \emptyset \proves u \hastype \chtype{U} 
					\\
					\Gamma; \Lambda_2; \Delta_2 \proves x \hastype U
					   \end{array}
			}{
				\Gamma\backslash x; \Lambda_1 \backslash \Lambda_2; \Delta_1 \backslash \Delta_2 \proves \binp{u}{x} P \hastype \Proc
			}	
			\\  \\
			\inferrule[(Send)]{
						%\begin{array}{c}
						u:S \in \Delta_1 \cat \Delta_2 %\\
						\quad
				\Gamma; \Lambda_1; \Delta_1 \proves P \hastype \Proc
				\quad
				\Gamma; \Lambda_2; \Delta_2 \proves V \hastype U
				%\end{array}
			}{
				\Gamma; \Lambda_1 \cat \Lambda_2; ((\Delta_1 \cat \Delta_2) \setminus u:S) \cat u:\btout{U} S \proves \bout{u}{V} P \hastype \Proc
			}
			\\  \\
								\inferrule[(Rcv)]{
			%\begin{array}{c}
				\Gamma; \Lambda_1; \Delta_1 \cat u: S \proves P \hastype \Proc
				\quad
				\Gamma; \Lambda_2; \Delta_2 \proves {x} \hastype {U}
			%	\end{array}
			}{
				\Gamma \backslash x; \Lambda_1\backslash \Lambda_2; \Delta_1\backslash \Delta_2 \cat u: \btinp{U} S \vdash \binp{u}{{x}} P \hastype \Proc
			}
\\ \\ 
			% 		\inferrule[(Bra)]{
			% 	 \forall i \in I \quad \Gamma; \Lambda; \Delta \cat u:S_i \proves P_i \hastype \Proc
			% }{
			% 	\Gamma; \Lambda; \Delta \cat u: \btbra{l_i:S_i}_{i \in I} \proves \bbra{u}{l_i:P_i}_{i \in I}\hastype \Proc
			% }
			% \qquad
			%  \inferrule[(Sel)]{
			% 	\Gamma; \Lambda; \Delta \cat u: S_j  \proves P \hastype \Proc \quad j \in I
			% }{
			% 	\Gamma; \Lambda; \Delta \cat u:\btsel{l_i:S_i}_{i \in I} \proves \bsel{u}{l_j} P \hastype \Proc
			% }
			% \\  \\
			\inferrule[(ResS)]{
				\Gamma; \Lambda; \Delta \cat s:S_1 \cat \dual{s}: S_2 \proves P \hastype \Proc \quad S_1 \dualof S_2
			}{
				\Gamma; \Lambda; \Delta \proves \news{s} P \hastype \Proc
			}
			\qquad
			\inferrule[(Res)]{
				\Gamma\cat a:\chtype{S} ; \Lambda; \Delta \proves P \hastype \Proc
			}{
				\Gamma; \Lambda; \Delta \proves \news{a} P \hastype \Proc
			}
			\end{array}
	\]
	\end{mdframed}
	%\vspace{-3mm}
	\caption{Typing Rules for $\HO$. 
	% (including selection and branching constructs).
	%See \appref{app:types} for a full account.
	\label{fig:typerulesmys}}
	%\Hline
	%\vspace{-1mm}
	\end{figure}
	%\myparagraph{Typing System of \HOp}

\noindent
\figref{fig:typerulesmys} shows the typing rules; we briefly describe them and refer the reader to~\cite{KouzapasPY17} for a full account.
The shared type $\shot{C}$ %for shared higher order values $V$
is derived using Rule~\textsc{(Prom)} only
if the value has a linear type with an empty linear
environment.
Rule~\textsc{(EProm)} allows us to freely use a \newj{shared
type variable as linear}.
Abstraction values are typed with Rule~\textsc{(Abs)}.
%The key type for an abstraction is the type for
%the bound variables of the abstraction, i.e.~for
%bound variable type $C$ the abstraction
%has type $\lhot{C}$.
Application typing
is governed by Rule~\textsc{(App)}:  
the type $C$ of an application name $u$
must match the type of the application variable $x$ ($\lhot{C}$ or $\shot{C}$).
%
%A process prefixed with a session send operator $\bout{k}{V} P$
%is typed using rule $\trule{Send}$.
Rules~\textsc{(Req)} and~\textsc{(Acc)} type interaction along shared names; the type of the sent/received object $V$ (i.e., $U$) should match the type of the subject $s$ ($\chtype{U}$).
In Rule~\textsc{(Send)}, the type $U$ of the value $V$ should appear as a prefix in the session type $\btout{U} S$ of $u$.
Rule~\textsc{(Rcv)} is its dual.

To state type soundness, we require two auxiliary definitions on session environments.
First, a session environment $\Delta$ is {\em balanced} (written $\balan{\Delta}$) if whenever $s: S_1, \dual{s}: S_2 \in \Delta$ then $S_1 \dualof S_2$.
Second, we define the reduction relation $\red$ on session environments as:
	\begin{eqnarray*}
			\Delta \cat s: \btout{U} S_1 \cat \dual{s}: \btinp{U} S_2 & \red &
			\Delta \cat s: S_1 \cat \dual{s}: S_2\\%[1mm]
			\Delta \cat s: \btsel{l_i: S_i}_{i \in I} \cat \dual{s}: \btbra{l_i: S_i'}_{i \in I} &\red&
			 \Delta \cat s: S_k \cat \dual{s}: S_k' \ (k \in I)
		\end{eqnarray*}

%We then have:

\begin{theorem}[Type Soundness~\cite{KouzapasPY17}]\label{t:sr}\rm
%	\begin{enumerate}[1.]
%		\item	(Subject Congruence) Suppose $\Gamma; \es; \Delta \proves P \hastype \Proc$.
%			Then $P \scong P'$ implies $\Gamma; \es; \Delta \proves P' \hastype \Proc$.
%
%		\item
%			(Subject Reduction)
			Suppose $\Gamma; \es; \Delta \proves P \hastype \Proc$
			with
			$\balan{\Delta}$.
			Then $P \red P'$ implies $\Gamma; \es; \Delta'  \proves P' \hastype \Proc$
			and $\Delta = \Delta'$ or $\Delta \red \Delta'$
			with $\balan{\Delta'}$.
%	\end{enumerate}
\end{theorem}

\begin{remark}[Typed Polyadic Communication]
When using processes with polyadic communication (cf.~\remref{r:poly}),  
we shall assume the extension of the type system defined in~\cite{KouzapasPY17}.
\end{remark}

\begin{example}[Typing name-passing constructs]
 In \exref{top:ex:np} we recalled how to encode name-passing constructs in \HO; now we show that this translation is typed.
 Following the name-passing encoding from~\cite{DBLP:conf/esop/KouzapasPY16,DBLP:journals/iandc/KouzapasPY19} 
 we define a syntactic sugar for types. 
  % Specifically, for a channel type $C$, we define session types
  %   \begin{align*}
  %     \btout{\wtd{C}}S & \triangleq ~\btout{\lhot{\btinp{\lhot{C}}S}}\\
  %     \btinp{\wtd{C}}S & \triangleq ~\btinp{\lhot{\btinp{\lhot{C}}S}}
  %   \end{align*}
% Then the following typing rules for name-passing are derivable:
The following typing rules for name-passing are derivable:
    \begin{mathpar}
      \inferrule*[left=(SendN)]{
        %	\begin{array}{c}
        \Gamma; \Lambda_1; \Delta_1 \proves P \hastype \Proc
        \quad
        \Gamma; \Lambda_2; \Delta_2 \proves \wtd{b} \hastype \wtd{C}
        % \end{array}
      }{
        \Gamma; \Lambda_1 \cat \Lambda_2; \Delta_1 \cat \Delta_2 \cat 
        t:\btout{\namepass{\wtd{C}}}\tinact \proves \bout{t}{\namepass{\wtd{b}}} P \hastype \Proc
      }
      \and
      \inferrule*[left=(RcvN)]{
        \begin{array}{c}
          \Gamma; \Lambda_1; \Delta_1 \proves P \hastype \Proc
          \quad
          \Gamma; \Lambda_2; \Delta_2 \proves \wtd{x} \hastype \wtd{C}
        \end{array}
      }{
        \Gamma \backslash x; \Lambda_1\backslash \Lambda_2; \Delta_1 \backslash \Delta_2 \cat t: \btinp{\namepass{\wtd{C}}}\tinact \vdash 
        \binp{t}{\namepass{\wtd{x}}} P \hastype \Proc
      }
    \end{mathpar}
\end{example}

% \textcolor{blue}{Next two examples 
% illustrate how typing judgments are inferred following the typing rules 
% from}
\begin{example}[Typing Recursive Servers]
  \label{e:server-types}
Here we show how to type the processes from~\Cref{e:server}. Let us define: 
\begin{align*}
  T = \trec{t}{\btout{U}\vart{t}} \qquad 
  C = \chtype{\lhot{T}} 
\end{align*}
where $U$ is some value type. We recall process $P$ 
from~\Cref{e:server} with the additional typing information on bound names 
$r$ and $a$: 
\begin{align*}
	P &= 	\news{r:T}\news{a : C}\big(\bout{a}{V} \inact  \Par \bout{a}{W} \inact  \Par S_a 
  \Par \binp{\dual r}{x_1}\binp{\dual r}{x_2}Q \big)
	\\
	V& = \abs{y}{(\bout{y}{V'}S_a\subst{y}{r})}
	\\
	W& = \abs{z}{(\bout{z}{W'}S_a\subst{z}{r})}
		\end{align*}
    where $S_a$ stands for $\binp{a}{x} (\appl{x}{r})$. 
    Let us assume 
    there is a shared environment 
    $\Gamma$ under which $V'$ and $W'$ implement type $U$: 
    \begin{align} 
      % \AxiomC{} 
      % \UnaryInfC{
        \label{dt:ex1-asm1}
       \Gamma;\es; \es &\proves
       V' \hastype U \\
       \label{dt:ex1-asm2}
       \Gamma;\es; \es &\proves
       W' \hastype U
      % } 
    \end{align} 
    Also, we assume  that process $\binp{\dual r}{x_1}\binp{\dual r}{x_2}Q$ is well-typed under 
    the following environments: 
    \begin{align}
      \label{dt:ex1-asm3}
      \Gamma \cat a:C; \es ; \dual r: \dual T  \proves 
      \binp{\dual r}{x_1}\binp{\dual r}{x_2}Q \hastype \Proc 
    \end{align} 

Under these assumptions, 
    %Assuming that values $V'$ and $W'$ have type $U$, 
    it holds that the body of process  $P$ correctly implements name $a$ 
    with type $C$, i.e., 
     $$  \Gamma;\es; r : T \proves  
      \news{a : C}
      \big(\bout{a}{V} \inact  \Par \bout{a}{W} \inact  \Par S_a \big)  \hastype \Proc$$
 We detail the corresponding typing derivations: 
    \begin{align}
      \label{dt:ex-2}
      \AxiomC{} 
      \LeftLabel{\scriptsize (LVar)}
      \UnaryInfC{$\Gamma \cat a:C; x:\lhot{T}; \es \proves x \hastype \lhot{T}$}
      \AxiomC{}  
      \LeftLabel{\scriptsize (Sess)}
      \UnaryInfC{$\Gamma \cat a:C;\es;y:T \proves y \hastype T$} 
      \LeftLabel{\scriptsize (App)}
      \BinaryInfC{
        $\Gamma \cat a:C;x:\lhot{T}; y:T \proves  \appl{x}{y}  \hastype \Proc$
      } 
      \DisplayProof 
    \end{align}

    \begin{align}
      \label{dt:ex-5}
      \AxiomC{\eqref{dt:ex-2}} 
      \AxiomC{}
      \LeftLabel{\scriptsize (Sh)}
      \UnaryInfC{ 
        $\Gamma \cat a:C;\es; \es \proves a \hastype \chtype{\lhot{T}}$
        } 
        \AxiomC{} 
        \LeftLabel{\scriptsize (LVar)} 
      \UnaryInfC{$
        \Gamma \cat a:C;x:\lhot{T}; \es \proves x \hastype \lhot{T}
      $}
      \LeftLabel{\scriptsize (Acc)}
      \TrinaryInfC{
        $\Gamma \cat a:C;\es; y:T \proves \binp{a}{x} (\appl{x}{y})  \hastype \Proc$
      }
      \DisplayProof 
    \end{align}

    % \begin{align}
    %   \label{dt:ex-3}
    %   \AxiomC{
    %     \eqref{dt:ex-5} 
    %     % $a:C;\es; y:S \proves S_a\subst{y}{r}$
    %    } 
    %    \AxiomC{} 
    %    \UnaryInfC{
    %     $\Gamma \cat a:C;\es; \es \proves
    %     V' \hastype U$
    %    } 
    %    \LeftLabel{\scriptsize (Send)}
    %    \BinaryInfC{
    %       $\Gamma \cat a:C;\es; y:T \proves
    %       \bout{y}{V'}S_a\subst{y}{r}  \hastype \Proc$
    %     } 
    %     % \AxiomC{\ldots} 
    %     % \UnaryInfC{
    %     %   $\Gamma \cat a:C;\es; y:T \proves \binp{\dual y}{x}Q  \hastype \Proc$
    %     % } 
    %     % % \AxiomC{} 
    %     % \LeftLabel{\scriptsize (Par)}
    %     % \BinaryInfC{
    %     %   $\Gamma \cat a:C;\es; y:T \proves
    %     %   \bout{y}{V'}S_a\subst{y}{r} \Par \binp{\dual y}{x}Q  \hastype \Proc$
    %     % } 
    %     \DisplayProof
    % \end{align}

   \begin{align}
    \label{dt:ex-4}
  % \AxiomC{\eqref{dt:ex-3}} 
  \AxiomC{
    \eqref{dt:ex-5} 
    % $a:C;\es; y:S \proves S_a\subst{y}{r}$
   } 
   \AxiomC{\eqref{dt:ex1-asm1}} 
  %  \UnaryInfC{
  %   $\Gamma \cat a:C;\es; \es \proves
  %   V' \hastype U$
  %  } 
   \LeftLabel{\scriptsize (Send)}
   \BinaryInfC{
      $\Gamma \cat a:C;\es; y:T \proves
      \bout{y}{V'}S_a\subst{y}{r}  \hastype \Proc$
    } 
    \AxiomC{} 
    \LeftLabel{\scriptsize (Sess)}
    \UnaryInfC{
      $\Gamma \cat a:C;\es; y:T \proves y \hastype T$
    } 
    \LeftLabel{\scriptsize (Abs)}
    \BinaryInfC{ 
      $\Gamma \cat a:C;\es;\es \proves V \hastype \lhot{T}$
    }
    \DisplayProof 
   \end{align}

    \begin{align}
      \label{dt:ex-6}
      \AxiomC{} 
      \LeftLabel{\scriptsize (Nil)}
      \UnaryInfC{
        $\Gamma \cat a:C;\es;\es \proves \inact  \hastype \Proc$
      } 
      \AxiomC{} 
      \LeftLabel{\scriptsize (Sh)}
      \UnaryInfC{
        $\Gamma \cat a:C;\es;\es \proves a \hastype \chtype{\lhot{T}}$
      } 
      \AxiomC{
        \eqref{dt:ex-4}
      }
      \LeftLabel{\scriptsize (Req)}
      \TrinaryInfC{$\Gamma \cat a:C;\es;\es \proves \bout{a}{V} \inact  \hastype \Proc$} 
      \DisplayProof 
    \end{align}

    % \begin{align}
    %   \label{dt:ex-5}
    %   \AxiomC{\eqref{dt:ex-2}} 
    %   \AxiomC{}
    %   \LeftLabel{\scriptsize (Sh)}
    %   \UnaryInfC{ 
    %     $\Gamma \cat a:C;\es; \es \proves a \hastype \chtype{\lhot{T}}$
    %     } 
    %     \AxiomC{} 
    %     \LeftLabel{\scriptsize (LVar)} 
    %   \UnaryInfC{$
    %     \Gamma \cat a:C;x:\lhot{T}; \es \proves x \hastype \lhot{T}
    %   $}
    %   \LeftLabel{\scriptsize (Acc)}
    %   \TrinaryInfC{
    %     $\Gamma \cat a:C;\es; r:T \proves \binp{a}{x} (\appl{x}{r})  \hastype \Proc$
    %   }
    %   \DisplayProof 
    % \end{align}
  \noindent In the following derivation tree the  
  right-hand side is shown similarly to \eqref{dt:ex-6}
  using assumption \eqref{dt:ex1-asm2} instead of \eqref{dt:ex1-asm1}: 
\begin{align}
  \label{dt:ex-8}
  \AxiomC{
    %$a:C;\es;\es \proves \bout{a}{V} \inact$
    \eqref{dt:ex-6}
  } 
  % \AxiomC{similar to \eqref{dt:ex-6}} 
  \AxiomC{} 
  \UnaryInfC{
    $\Gamma \cat a:C;\es;\es \proves \bout{a}{W} \inact  \hastype \Proc$
  } 
  \LeftLabel{\scriptsize (Par)}
  \BinaryInfC{
    $\Gamma \cat a:C;\es;\es \proves \bout{a}{V} \inact  \Par \bout{a}{W} \inact
    \hastype \Proc$
    } 
    \DisplayProof 
\end{align}
Finally we have:
    \begin{align*}
      \AxiomC{\eqref{dt:ex-8}} 
        \AxiomC{} 
      \UnaryInfC{$\Gamma \cat a:C;\es; r : T \proves S_a  \hastype \Proc$} 
      \LeftLabel{\scriptsize (Par)}
      \BinaryInfC{
        $\Gamma \cat a:C;\es; r : T \proves
        \bout{a}{V} \inact  \Par \bout{a}{W} \inact  \Par S_a 
      % \Par \binp{\dual r}{x_1}\binp{\dual r}{x_2}Q 
        \hastype \Proc$
      } 
      \AxiomC{\eqref{dt:ex1-asm3}} 
      % \UnaryInfC{$\Gamma \cat a:C;\es; r : T \proves S_a  \hastype \Proc$} 
      \LeftLabel{\scriptsize (Par)}
      \BinaryInfC{
        $\Gamma \cat a:C;\es; r : T \cat \dual r : \dual T \proves
        \bout{a}{V} \inact  \Par \bout{a}{W} \inact  \Par S_a 
      \Par \binp{\dual r}{x_1}\binp{\dual r}{x_2}Q 
        \hastype \Proc$
      } 
      \LeftLabel{\scriptsize (Res)}
      \UnaryInfC{
      $\Gamma;\es; r : T \cat \dual r : \dual T \proves  
      \news{a : C}\big(\bout{a}{V} \inact  \Par \bout{a}{W} \inact  \Par S_a 
      \Par \binp{\dual r}{x_1}\binp{\dual r}{x_2}Q \big) 
      \hastype \Proc$
        } 
        \LeftLabel{\scriptsize (ResS)}
        \UnaryInfC{
          $\Gamma;\es; \es \proves  
          \news{r:T}\news{a : C}\big(\bout{a}{V} \inact  \Par \bout{a}{W} \inact  \Par S_a 
          \Par \binp{\dual r}{x_1}\binp{\dual r}{x_2}Q \big) 
          \hastype \Proc$
        } 
        \DisplayProof 
    \end{align*}
    \noindent Above, we have omitted details of the right-hand 
    side derivation; it the same as \eqref{dt:ex-5} 
    with name $y$ substituted with $r$. 
% \todo[inline]{TODO}
\end{example}

\begin{example}[Typing Nested Abstractions]
  Here we show how to type process $P$ from~\Cref{e:partial}. 
  Let types $C_1$ and $C_2$ be defined as 
  $C_i = \chtype{\lhot{S_i}} $ 
  where $i \in \{1,2\}$ and $S_i$ stands for a tail-recursive type.
%  , 
%  
%  follows: 
%  \begin{align*}
%    C_1 &= \chtype{\lhot{S_1}}  
%    % \qquad 
%      % S_1 = \trec{t}{\btout{U_1}\vart{t}} 
%      \\
%      C_2 &= \chtype{\lhot{S_2}} 
%      %  \qquad 
%      % S_2 = \trec{t}{\btout{U_2}\vart{t}} 
%  \end{align*}
%  \noindent where $S_1$ and $S_2$ are some tail-recursive types. 
  For simplicity, we assume that value $V$ has the following typing: 
  \begin{align}
    \label{dt:ex2-a}
    a:C_1 \cat b:C_2; \es; \es \proves V \hastype 
    \lhot{(S_1,S_2)}
  \end{align}
  The following holds:
  \begin{align} 
    \label{eq:ex2-types-goal}
        \es; \es; \es \proves \news{a : C_1} \news{b:C_2}  R \Par S_a \Par S_b
      \hastype \Proc
  \end{align} 

  \noindent In the following typing derivations, we rely on the following two  typing rules for polyadic elements; they  can be derived from monadic typing rules 
  from~\Cref{fig:typerulesmys} (see~\Cref{r:polyadic-rules} 
  for details):
  \begin{prooftree}
  	\AxiomC{}
  	\LeftLabel{\scriptsize (PolySess)}
  	\UnaryInfC{$\Gamma;\es;\wtd u:\wtd S
  				\proves \wtd u \hastype \wtd S$}
  \end{prooftree}

	\begin{prooftree}
		\AxiomC{$\leadsto \in \{\multimap,\rightarrow \}$}
		\AxiomC{$\Gamma;\Lambda;\Delta_1 \proves V \hastype \wtd C \leadsto \diamond$}
		\AxiomC{$\Gamma;\es;\Delta_2 \proves \wtd u \hastype \wtd C$}
		\LeftLabel{\scriptsize (PolyApp)}
		\TrinaryInfC{$\Gamma;\Lambda;\Delta_1 \cat \Delta_2 \proves
						\appl{V}{\wtd u}$}
	\end{prooftree}

  % \noindent We detail the corresponding typing derivations: 
  \noindent Now, we detail the typing derivations that 
  show \eqref{eq:ex2-types-goal}:  

  \begin{align}
    \label{dt:ex2-t5}
    \AxiomC{\eqref{dt:ex2-a}} 
    \AxiomC{} 
    \LeftLabel{\scriptsize (PolySess)}
    \UnaryInfC{
      $a:C_1 \cat b:C_2; \es; y : S_1 \cat z:S_2 \proves y,z \hastype S_1,S_2$
      } 
    \LeftLabel{\scriptsize (PolyApp)}
    \BinaryInfC{
      $a:C_1 \cat b:C_2; \es; y : S_1 \cat z:S_2 \proves \appl{V}{(y,z)} \hastype \Proc$
    } 
    \DisplayProof 
  \end{align}
  \begin{align}
    \label{dt:ex2-t4}
\AxiomC{\eqref{dt:ex2-t5}} 
    \AxiomC{} 
    \LeftLabel{\scriptsize (Sess)}
    \UnaryInfC{
      $a:C_1 \cat b:C_2; \es;z:S_2 \proves z \hastype S_2$
      } 
    \LeftLabel{\scriptsize (Abs)}
    \BinaryInfC{
      $a:C_1 \cat b:C_2; \es; y : S_1 \proves \abs{z}\appl{V}{(y,z)} \hastype \lhot{S_2}$
      } 
      \DisplayProof 
  \end{align}

  \begin{align}
    \label{dt:ex2-t3}
    \AxiomC{} 
    \LeftLabel{\scriptsize (Nil)}
    \UnaryInfC{$a:C_1 \cat b:C_2; \es; \es \proves \inact \hastype \Proc$} 
    \AxiomC{} 
    \LeftLabel{\scriptsize (Sh)}
    \UnaryInfC{$a:C_1 \cat b:C_2; \es; \es \proves b \hastype \chtype{\lhot{S_2}}$} 
    \AxiomC{
      % $a:C_1 \cat b:C_2; \es; y : S_1 \proves \abs{z}\appl{V}{(y,z)} \hastype \lhot{S_2}$
      \eqref{dt:ex2-t4}
      } 
    \LeftLabel{\scriptsize (Req)}
    \TrinaryInfC{
      $a:C_1 \cat b:C_2; \es; y:S_1 \proves
       \about{b}{\abs{z}\appl{V}{(y,z)}}
       \hastype \Proc$
      } 
      \DisplayProof 
  \end{align}
  \begin{align}
    \label{dt:ex2-t2}
    \AxiomC{
      % $a:C_1 \cat b:C_2; \es; y:S_1 \proves \about{b}{\abs{z}\appl{V}{(y,z)}}$
      \eqref{dt:ex2-t3} 
      } 
      \AxiomC{} 
      \LeftLabel{\scriptsize (Sess)}
    \UnaryInfC{
      $a:C_1 \cat b:C_2; \es; y:S_1 \proves y \hastype S_1$
    } 
    \LeftLabel{\scriptsize (Abs)}
    \BinaryInfC{
      $a:C_1 \cat b:C_2; \es; \es \proves  
      {\abs{y}\about{b}{\abs{z}\appl{V}{(y,z)}}} \hastype \lhot{S_1}$
    }
    \DisplayProof 
  \end{align}

  \begin{align}
    \label{dt:ex2-t1}
    \AxiomC{$a:C_1 \cat b:C_2; \es; \es \proves \inact \hastype \Proc$} 
    \AxiomC{} 
    \LeftLabel{\scriptsize (Sh)}
    \UnaryInfC{$a:C_1 \cat b:C_2; \es; \es \proves a \hastype \chtype{\lhot{S_1}}$} 
    \AxiomC{
      % $a:C_1 \cat b:C_2; \es; \es \proves  
      % {\abs{y}\about{b}{\abs{z}\appl{V}{(y,z)}}}$
      \eqref{dt:ex2-t2}
    }
    \LeftLabel{\scriptsize (Req)}
    \TrinaryInfC{
      $a:C_1 \cat b:C_2; \es; \es \proves  
      \abbout{a}{\abs{y}\about{b}{\abs{z}\appl{V}{(y,z)}}}
      \hastype \Proc$
    } 
    \DisplayProof
  \end{align}
Finally we have:
  \begin{align*}
    \AxiomC{
      % $a:C_1 \cat b:C_2; \es; \es \proves  R$
      \eqref{dt:ex2-t1} 
    } 
    \AxiomC{} 
    \LeftLabel{}
    \UnaryInfC{
      $a:C_1 \cat b:C_2; \es; \es \proves S_a  \hastype \Proc$
    } 
    \LeftLabel{\scriptsize (Par)} 
    \BinaryInfC{
      $a:C_1 \cat b:C_2; \es; \es \proves  R \Par S_a 
      \hastype \Proc$
    } 
    \AxiomC{} 
    \LeftLabel{}
    \UnaryInfC{
      $a:C_1 \cat b:C_2; \es; \es \proves  S_b
      \hastype \Proc$
    }
    \LeftLabel{\scriptsize (Par)} 
    \BinaryInfC{
      $a:C_1 \cat b:C_2; \es; \es \proves   R \Par S_a \Par S_b
      \hastype \Proc$
    } 
    \LeftLabel{\scriptsize (Res)}
    \UnaryInfC{
      $a:C_1; \es; \es \proves  \news{b:C_2}  R \Par S_a \Par S_b
      \hastype \Proc$
    } 
    \LeftLabel{\scriptsize (Res)}
    \UnaryInfC{
      $\es; \es; \es \proves \news{a : C_1} \news{b:C_2}  R \Par S_a \Par S_b
      \hastype \Proc$
    } 
    \DisplayProof 
  \end{align*}
  % In the above typing derivation, we remark that 
  % following judgments
  % \begin{align*}
  %   a:C_1 \cat b:C_2; \es; \es &\proves S_a  \hastype \Proc \\
  %   a:C_1 \cat b:C_2; \es; \es &\proves  S_b
  %     \hastype \Proc
  % \end{align*}
  % \noindent are shown similarly as in~\eqref{dt:ex-5} 
  % from~\Cref{e:server-types} with appropriate 
  % weakening of shared environments and names substitution. 

% \noindent In the above typing derivations, we remark that the  
%   following judgment 
%   \begin{align*}
%     a:C_1 \cat b:C_2; \es; \es \proves S_a  \hastype \Proc 
%   \end{align*}
%   \noindent is shown similarly to~\eqref{dt:ex-5} from~\Cref{e:server-types} 
%   with the substitution of the name $r$ to $y$ 
%   and weakening shared environment with $b : C_2$. 
%   The judgment 
%   \begin{align*}
%     a:C_1 \cat b:C_2; \es; \es &\proves  S_b   \hastype \Proc 
%   \end{align*}
%   \noindent is shown analogously. 
  \noindent In the above typing derivation, we remark that   the  judgments
  \begin{align}
    \label{j:sa}
    a:C_1 \cat b:C_2; \es; \es &\proves S_a  \hastype \Proc \\
    \label{j:sb}
    a:C_1 \cat b:C_2; \es; \es &\proves  S_b
      \hastype \Proc
  \end{align}
  \noindent are shown similarly as in~\eqref{dt:ex-5} 
  from~\Cref{e:server-types}. 
  Indeed, 
   to derive~\eqref{j:sa} reusing
  the derivation tree from~\eqref{dt:ex-5}
  we need to substitute $y$ with $r$ and then weaken 
  the shared environment~\eqref{dt:ex-5} 
  with $b:C_2$ (see~\Cref{lem:weaken}). 
  Similarly, 
  by substituting $a$ with $b$ and $y$ with $r$ 
  in the derivation tree~\eqref{dt:ex-5}
   and 
  then by weakening the shared environment with $a:C_1$ in  
  its conclusion we can obtain~\eqref{j:sb}. 

  % \todo[inline]{TODO}
\end{example} 

\begin{notation}[Type Annotations]
We shall often annotate bound names and variables with their respective type.
We will write, e.g., $\news{s:S}P$ to denote that the type of $s$ in $P$ is $S$.
Similarly for values: we shall write $\abs{u:C}{P}$.
Also, letting $\leadsto \in \{\lollipop, \sharedop\}$, we may write
$\abs{u:\slhotup{C}}{P}$ to denote that the value is linear (if $\leadsto = \lollipop$) or shared (if $\leadsto = \sharedop$).
That is, we write $\abs{u:\slhotup{C}}{P}$ if
$\Gamma ; \Lambda; \Delta \proves \abs{{u}}{P} \hastype \slhot{{C}}$, for some $\Gamma$, $\Lambda$, and $\Delta$.
\end{notation}

Having introduced the core session process language \HO, we now move to detail its type-preserving decomposition into minimal session types.

\section{Decomposing Session-Typed Processes}
\label{s:decomp}

In this section we define the fragment of \emph{minimal session types} and present a \emph{decomposition} of well-typed processes: given a process $P$ typable with (standard) session types, our decomposition returns a process, denoted $\D{P}$, typable with minimal session types.

The definition of  $\D{P}$ follows Parrow's \emph{trios} for the $\pi$-calculus~\cite{DBLP:conf/birthday/Parrow00}.
A {trio} is a process with exactly three sequential prefixes.
Roughly speaking, if $P$ is a process with $k$ sequential actions, then $\D{P}$ will contain $k$ trios running in parallel: each of them will enact exactly one action from $P$.
The decomposition is carefully designed to ensure that trios trigger each other by preserving the sequencing in $P$.

This section is organized as follows.
First, in \Cref{mst:ss:key}, we use examples to discuss some key ideas of the decomposition.
Then, in \Cref{mst:ss:core}, we give the full definitions of minimal session types and the decomposition functions for types and  processes, denoted $\Gt{-}$ and $\D{-}$, respectively.
The former ``slices'' a session type $S$ and returns a list of minimal session types, corresponding to individual actions in $S$; the latter breaks down an \HO process into a parallel composition of processes.
We demonstrate these notions on a number of examples in \Cref{mst:ss:examples}.
Finally, in \Cref{mst:ss:staticcorr} we establish the \emph{static correctness} result (\Cref{mst:t:decompcore}):
if $P$ is well-typed
under session types $S_1, \ldots, S_n$, then $\D{P}$ is typable using the
minimal session types $\Gt{S_1}, \ldots, \Gt{S_n}$.
The issue of dynamic correctness, i.e., the operational correspondence between $P$ and $\D{P}$, is treated separately in \Cref{ss:dynamic}.

\begin{remark}[Color Convention]
  We use colors to differentiate the operations on processes (in \textcolor{\colorpro}{pink}) and on types (in \textcolor{\colortyp}{green}).
  The usage of the colors is for visual aid only, and is not important for the mathematical content of the presented material.
\end{remark}

\subsection{Key Ideas}
\label{mst:ss:key} %\Gt
\newcommand{\processshape}{circle}
\captionsetup[figure]{font=small}
\begin{figure}[!t]
  \begin{mdframed} 
\begin{center} 
  % \xspace
  \hspace*{-30.0pt}
\begin{tikzpicture}
  \tikzstyle{ann} = [draw=none,fill=none,right]
  \pgfmathsetmacro{\h}{2.0}
  \pgfmathsetmacro{\x}{3.0}
  \pgfmathsetmacro{\d}{1.75}
  % \pgfmathsetmacro{\processshape}{circle}
  % \newcommand{\processshape}{circle}

  \node[] (L1) at (-\x*1.5-0.1, 1.5) {Source process $P_1$:}; 

  {\node[draw, \processshape] (P1) at (-\x*1.5, 0.5) {$P_1$};
  % \node at (-1.5,-3) {$\parallel$};
  \node[draw, \processshape] (P2) at (-\x/2.0, 0.5) {$P_2$};
  \node[draw, \processshape] (P3) at (\x/2.0, 0.5) {$P_3$};
  \node[draw, circle] (P4) at (\x*1.5, 0.5) {$\inact$};} 

  {\draw[->, color=black](P1) edge[] node [above] 
    {\footnotesize $u:\abinp{}{\mathsf{str}}$} (P2);}
  {\draw[->, color=black](P2) edge[] node [above] 
    {\footnotesize $u:\abinp{}{\mathsf{int}}$} (P3);}
  {\draw[->, color=black](P3) edge[] node [above] 
    {\footnotesize $u:\about{}{\mathsf{bool}}$} (P4);}

  % {\node[draw=none] at (0,-0.575) 
  % {$u : \btinp{\mathsf{str}}\btinp{\mathsf{int}}\btout{\mathsf{bool}}\tinact$};
  %   % \node[draw](A) at (0, 0) {$P$};
  %   }

% \end{tikzpicture} 
% ~~~{\hbox to \textwidth{\leaders\hbox to 3pt{\hss . \hss}\hfil}}
% \hspace*{-30.0pt}
%   \begin{tikzpicture}
    % \pgfmathsetmacro{\h}{2.0}
    % \pgfmathsetmacro{\x}{3.0}
    % \pgfmathsetmacro{\d}{1.5}
    \tikzstyle{ann} = [draw=none,fill=none,right]
  \node[align=left] (L2) at (-\x*1.25, -0.5) {Decomposed process $\D{P_1}$:};

%    {\node[draw, \processshape] (B) at (0, -1.5) {$\D{P_1}$};
%    % \draw[->, thick] (0,-0.85) -- (0, -1.15);
%    }
    
    \node[draw, \processshape] (T1) at (-\x*1.5, -\d) {$Q_1$};
    \node[draw, \processshape] (T11) at (-\x*1.5, -\d-\h) {$Q'_1$};

    {\draw[->, color=black](T1) edge[] node [left] 
        {\footnotesize $u_1:\abinp{}{\mathsf{str}}$} (T11);}

    \node at (-\x, -\d) {$\parallel$};
    \node[draw, \processshape] (T2) at (-\x/2.0, -\d) {$Q_2$};
    \node[draw, \processshape] (T22) at (-\x/2.0, -\d-\h) {$Q'_2$};

    {\draw[->, color=black](T2) edge[] node [left] 
        {\footnotesize $u_2:\abinp{}{\mathsf{int}}$} (T22);}

    \node at (0, -\d) {$\parallel$};
    \node[draw, \processshape] (T3) at (\x/2.0, -\d) {$Q_3$};
    \node[draw, \processshape] (T33) at (\x/2.0, -\d-\h) {$Q'_3$};
    {\draw[->, color=black](T3) edge[] node [left] 
        {\footnotesize $u_3:\about{}{\mathsf{bool}}$} (T33);}

        \node at (\x, -\d) {$\parallel$};

    \node[draw, \processshape] (T4) at (\x*1.5, -\d) {$Q_4$};

%    \draw[->, dotted, thick] (B) edge[bend right] node [left] {} (T1); 
%    \draw[->, dotted, thick] (B) edge (T2);
%    \draw[->, dotted, thick] (B) edge[] node [left] {} (T3); 
%    \draw[->, dotted, thick] (B) edge[bend left] node [left] {} (T4); 

    % \node at (-2.75,-3.65) {$u_1:\btinp{\mathsf{str}}\tinact$};
    % \node at (0,-3.65) {$u_2:\btinp{\mathsf{int}}\tinact$};
    % \node at (2.75,-3.65) {$u_3:\btout{\mathsf{bool}}\tinact$}; 
    {\draw[->, dashed, color=red, align=left](T11) edge[bend left=25] node [pos=0.05, right] 
      {\footnotesize{$x:\mathsf{str}$}} (T2);}
    {\draw[->, dashed, color=red, align=left](T22) edge[bend left=25] node [pos=0.05, right] 
     {\footnotesize{$x:\mathsf{str}$}, \footnotesize{$y:\mathsf{int}$}} (T3);}
    {\draw[->, dashed, color=red](T33) edge[bend left=20] node [left] {} (T4);}

    \node[draw=none] {};
  \end{tikzpicture}
\end{center} 
\end{mdframed} 
\caption[Our decomposition function 
$\D{-}$, illustrated]{Our decomposition function 
$\D{-}$, illustrated. Nodes 
represent process states, `$\parallel$' represents 
parallel composition of processes, black arrows stand for actions, and red arrows 
indicate synchronizations that preserve the sequentiality of the source process by activating trios and propagating (bound) values.}
\label{mst:fig:decomp}
\end{figure} 

Consider a process $P_1$ that implements the (standard) session type $S=\btinp{\mathsf{str}}\btinp{\mathsf{int}}\btout{\mathsf{bool}}\tinact$ along name $u$.
In   $P_1$, name $u$ is not a single-use resource; rather, it is used several times to implement the communication actions in $S$; \Cref{mst:fig:decomp}~(top) graphically depicts the actions and the corresponding states.

  The decomposition $\D{P_1}$ is illustrated
  in the bottom part of~\Cref{mst:fig:decomp}: it is defined as the parallel composition of four processes 
  $Q_i$ (for $i \in \{1, \ldots, 4\}$).
  Each process $Q_1$, $Q_2$, and $Q_3$ mimic one action of $P_1$ on an indexed name $u_i$, while $Q_4$ simulates the termination   of the session.
  This way, a single name $u$ in $P_1$ is decomposed into a sequence of names $u_1,u_2,u_3$ in $\D{P_1}$.

  The processes $Q_1$, $Q_2$, $Q_3$, and $Q_4$ are composed in parallel, but we would like to retain the same sequentiality of actions on the channels $u_i$ as we have on the channel $u$.
  To that end, each process $Q_i$, with the exception of $Q_1$, does not perform its designated action on $u_i$ until it gets activated by the previous process.
  In turn, after $Q_i$ performs an action on $u_i$ it evolves to a state $Q'_i$, which is responsible for activating the next process $Q_{i+1}$.
  In~\Cref{mst:fig:decomp}, the activations are indicated by red arrows.
  In general, the decomposition orchestrates the activation of sub-processes, following the sequencing prescribed by the session types of the given process. 
  Therefore, assuming a well-typed source process, our decomposition codifies the sequentiality in
  session types into the process level. 

  The activation mechanism includes the \emph{propagation} of values across sub-processes (cf. the labels on red arrows).
  This establishes a flow of values from sub-processes binding them to those that use them (i.e., it makes  variable bindings explicit).  For example, in $P_1$, the Boolean value being sent over as part of the session $S$ might depend on the previously received string and integer values.
  Therefore, both of those values have to be propagated to the process $Q_3$, which is responsible for sending out the Boolean.

  In this example a single name $u : S$ is decomposed into a sequence $\wtd u = (u_1,\ldots,u_n)$: each $u_i \in \wtd u$ is a single-use resource, as prescribed by its minimal session type.
  Such is the case for non-recursive types $S$.
  When $S$ is recursive, the situation is more interesting: each action of $S$ can be repeated many times, and therefore the names $\wtd u$ should be propagated across trios to enable potentially many uses.
  As an example, consider the recursive session type 
  $S=\trec{t}{\btinp{\mathsf{int}}\btout{\mathsf{\mathsf{int}}}}{\vart{t}}$, in which an input and an output actions are repeated indefinitely.
Consider the following process 
  \begin{align*}
    % R=\recp{X}\underbrace{\binp{r}{z}}_{t_1}\underbrace{\bout{r}{-z}}_{t_2}
		% \underbrace{\binp{r}{z}}_{t_3}\underbrace{\bout{r}{z}}_{t_4}X
    R_1=\underbrace{\binp{r}{z}}_{{T_1}}
    \underbrace{\bout{r}{-z}}_{{T_2}}
		\underbrace{\binp{r}{z}}_{{T_3}}
    \underbrace{\bout{r}{z}}_{{T_4}}
    \underbrace{\appl{V}{r}}_{{T_5}}
  \end{align*}
which makes use of the channel $r : S$ and where $V$ has type $\shot{S}$.
\Cref{mst:fig:rec-subfig2} (top) gives the first four actions of $R_1$ and the corresponding sates: the body of type $S$ prescribes two actions on name $r$, performed sequentially in $R_1$ and $R_2$; subsequent actions (enabled in $R_3$
    and $R_4$) correspond to a ``new instance'' of the body of $S$. 
     
    The decomposition $\D{R_1}$, depicted in \Cref{mst:fig:rec-subfig2} (bottom), generates a trio  for each prefix in $R_1$; we denote prefixes with their corresponding trios ${T_1}, \ldots, {T_5}$.
    The type decomposition function on types, $\Gt{-}$, slices $S$ into two \emph{minimal tail-recursive types}: $M_1 = \trec{t}\btinp{\textsf{int}}\vart{t}$ and $M_2 = \trec{t}\btout{\textsf{int}}\vart{t}$.

  In the recursive case, a key idea is that trios that mimic actions prescribed by a recursive session types should reuse names, which should be propagated across trios. 
    This way, for instance, trios $T_1$ and $T_3$ mimic the same (input) action, and so they both should use the same name ($r_1$).
    To achieve this, we devise a mechanism that propagates names with tail-recursive types (such as $(r_1,r_2)$) through the trios.
    These propagation actions are represented by blue arrows in~\Cref{mst:fig:rec-subfig2} (bottom).
    In our example, $T_3$ gathers the complete decomposition of names from preceding trios $(r_1,r_2)$; it mimics an input action on $r_1$ and makes $(r_1,r_2)$ available to future trios (i.e., $T_4$ and $T_5$).
    
    Since the same tail-recursive names can be (re)used infinitely often, we propagate tail-recursive names through the following process.
    All the names ${\wtd r}$ corresponding to the decomposition of a tail-recursive name $r$ are bound in the process
    $$\recprov{r}{x}{\wtd r},$$
    which is similar to the servers discussed in \Cref{e:server}.
    We call these processes \emph{recursive propagators}, and each tail-recursive name in the original process $P$ has a dedicated propagator in $\D{P}$ on the channel $\prop^{r}$.
    Whenever a trio has to perform an action $\alpha(r_i)$ on one of the decomposed tail-recursive names (i.e., a decomposition of an input action `$\binp{r}{y}$' or an output action `$\bout{r}{V}$' on the name $r$), it first has to request the name from the corresponding recursive propagator by performing an output action $\abbout{\prop^r}{N}$, where value $N$ is the abstraction
$$N = \abs{\wtd z} \alpha({z_i}).\big(\apropout{k+1}{\wtd w} \Par \recprov{r}{x}{\wtd z} \big).$$
      A synchronization on $\prop^r$ will result in the reduction:
      \[
        \recprov{r}{x}{\wtd r} \Par \abbout{\prop^r}{N} \red \alpha({r_i}).\big(\apropout{k+1}{\wtd w} \Par 
  \recprov{r}{x}{\wtd r}
      \big). 
      \]
      The resulting process first simulates $\alpha(r)$  and subsequently  reinstates the recursive propagator on $\prop^r$, for the benefit of the other trios requiring access to the names $\wtd r$.
See  \Cref{mst:ex:recdec,mst:ex:rec2} below (Page~\pageref{mst:ex:recdec}) for further illustration of this method.
  
This decomposition strategy handles \HO processes with recursive types which are \emph{simple} and \emph{contractive}.
That is, recursive types of the form $\trec{t}{S}$, where the body $S \neq \vart{t}$ does not itself contain recursive types.
Unless stated otherwise, we consider \emph{tail-recursive} session types such as, e.g., $S = \trec{t}\btinp{\mathsf{int}}\btinp{\mathsf{bool}}\btout{\mathsf{bool}}\vart{t}$.
Non-tail-recursive session types such as $\trec{t}{\btinp{\shot{(\widetilde T,\vart{t})}}\tinact}$,  used in the fully-abstract encoding of \HOp into \HO~\cite{DBLP:conf/esop/KouzapasPY16}, can also be accommodated; see \Cref{mst:ex:ntrsts} below.

\begin{figure}[!t]
  \begin{mdframed} 
        \centering 
        % \begin{center} 
          \hspace*{-20.0pt}
        \begin{tikzpicture}
          \tikzstyle{ann} = [draw=none,fill=none,right]
      \pgfmathsetmacro{\x}{3}
      \pgfmathsetmacro{\h}{2}
      \pgfmathsetmacro{\xr}{2.5}
      \pgfmathsetmacro{\hr}{2}

      \node[] (L1) at (-\xr*2.0-0.175, 6.0) {Source process $R_1$:}; 

      {\node[draw, \processshape] (P1) at (-2*\xr, 5) {$R_1$};
      % \node at (-1.5,-3) {$\parallel$};
      \node[draw, \processshape] (P2) at (-\xr, 5) {$R_2$};
      \node[draw, \processshape] (P3) at (0, 5) {$R_3$};
    
      % \node at (1.5,-3) {$\parallel$};
      
      \node[draw, circle] (P4) at (\xr, 5) {$R_4$}; 

      \node[draw, circle] (P5) at (2*\xr, 5) {$R_5$};
      } 

      {\draw[->, color=black](P1) edge[] node [above] 
    {\footnotesize $r:\abinp{}{\mathsf{int}}$} (P2);}
  {\draw[->, color=black](P2) edge[] node [above] 
    {\footnotesize $r:\about{}{\mathsf{int}}$} (P3);}
  {\draw[->, color=black](P3) edge[] node [above] 
    {\footnotesize $r:\abinp{}{\mathsf{int}}$} (P4);
    \draw[->, color=black](P4) edge[] node [above] 
    {\footnotesize $r:\about{}{\mathsf{int}}$} (P5);
    }

%-\xr*2.0-0.175
    \node[] (L2) at (-\xr*1.75-0.1, \x+1) {Decomposed process $\D{R_1}$:}; 

            \node[draw, \processshape] (T1) at (-2*\xr, \x) {$T_1$};
            \node[draw, \processshape] (T11) at (-2*\xr, \x-\h) {$T'_1$};
        
            {\draw[->, color=black, align=center](T1) edge[] node [pos=0.8, left] 
                {\textcolor{black}{\footnotesize{$r_1:$} 
                \footnotesize{$\abinp{}{\mathsf{int}}$}}} (T11);}
        
            \node at (-1.5*\xr,\x) {$\parallel$};
            \node[draw, \processshape] (T2) at (-\xr,\x) {$T_2$};
            \node[draw, \processshape] (T22) at (-\xr, \x-\h) {$T'_2$};
        
            {\draw[->, color=black, align=center](T2) edge[] node [pos=0.8, left] 
                {\footnotesize{$r_2:$}
                % \\
                \footnotesize{$\about{}{\mathsf{int}}$}} (T22);}
        
            \node at (-0.5*\xr, \x) {$\parallel$};
            \node[draw, \processshape] (T3) at (0, \x) {$T_3$};
            \node[draw, \processshape] (T33) at (0, \x-\h) {$T'_3$};
            {\draw[->, color=black](T3) edge[] node [pos=0.8, left] 
                {{\footnotesize $r_1:$$\abinp{}{\mathsf{int}}$}} (T33);}
        
                \node at (0.5*\xr, \x) {$\parallel$};
        
            \node[draw, \processshape] (T4) at (\xr, \x) {$T_4$};
            \node[draw, \processshape] (T44) at (\xr, \x-\h) {$T'_4$};
      
            {\draw[->, color=black](T4) edge[] node [pos=0.8, left] 
            {{\footnotesize $r_2:$$\about{}{\mathsf{int}}$}} (T44);}
      
            \node at (1.5*\xr, \x) {$\parallel$};
        
            \node[draw, \processshape] (T5) at (2*\xr, \x) {$T_5$};
      
            % % propagators red lines  
            % {\draw[->, color=red](T11) edge[bend left=28] node [left] {} (T2);}
            % {\draw[->, color=red](T22) edge[bend left=28] node [left] {} (T3);}
            % {\draw[->, color=red](T33) edge[bend left=28] node [left] {} (T4);}
            % {\draw[->, color=red](T44) edge[bend left=28] node [left] {} (T5);}
      
            % recursive names propagation 
            % line width=1.2
            \pgfmathsetmacro{\x}{15}
            \pgfmathsetmacro{\pos}{0.65}
            {\draw[->, dashed, color=blue](T11) 
            edge[bend left=\x] node 
              [pos=\pos, left] {\footnotesize{$r_1$,$r_2$}} (T2);}
            {\draw[->, dashed, color=blue](T22) 
            edge[bend left=\x] node [pos=\pos, left] {\footnotesize{$r_1$,$r_2$}} (T3);}
            {\draw[->, dashed, color=blue](T33) 
            edge[bend left=\x] node [pos=\pos, left] {\footnotesize{$r_1$,$r_2$}} (T4);}
            {\draw[->, dashed, color=blue](T44) 
            edge[bend left=\x] node [pos=\pos, left] {\footnotesize{$r_1$,$r_2$}} (T5);}
        
            % \node[draw=none] {};
          \end{tikzpicture}
        % \end{center} 
        \end{mdframed} 
          \caption[Decomposition of processes with recursive session types, illustrated]{Decomposition of processes with recursive session types, illustrated. Dashed blue arrows represent the propagation of tail-recursive names ($r_1$,$r_2$) across trios. 
          }
          \label{mst:fig:rec-subfig2}
\end{figure}

\subsection{The Decomposition}
\label{mst:ss:core}
Here we formally present the decomposition of \HO processes. 
We start introducing some preliminary definitions, including the definition of an auxiliary function, called the \emph{breakdown function}. 

Following Parrow~\cite{DBLP:conf/birthday/Parrow00} we adopt some useful terminology
and notation on trios. 
The \textit{context} of a trio is a tuple of variables
$\widetilde x$, possibly empty, which  makes variable bindings explicit. 
We use
a reserved set of \textit{propagator names} (or simply \emph{propagators}),
denoted with $\prop_k, \prop_{k+1}, \ldots$, to carry contexts and trigger the
subsequent trio. A process with less than three sequential prefixes is called a
\textit{degenerate trio}. Also, a \emph{leading trio} is the one that receives a
context, performs an action, and triggers the next trio; a \emph{control trio}
only activates other trios.

The breakdown function works on both processes and values. The breakdown of
process $P$ is denoted by $\B{k}{\tilde x}{P}$, where $k$ is the index for the
propagators $\prop_k$, and $\widetilde x$ is the context to be received by the
previous trio. Similarly, the breakdown  of a value $V$ is denoted by
$\V{k}{\tilde x}{V}$.

\subsubsection{Minimal Session Types and Decomposing Types}
We start by introducing minimal session types as a fragment of \Cref{d:types}:
\begin{definition}[Minimal Session Types (MSTs)]
\label{mst:d:mtypesi}
The syntax of \emph{minimal session types} for \HO is defined as follows: 
% (this can be simplified):
	\begin{align*}
U & \bnfis		\shot{\widetilde{C}} \bnfbar \lhot{\widetilde{C}}
		\\
		C  & \bnfis		M  \bnfbar  {\chtype{U}}
\\
		% M & \bnfis 	\tinact  \bnfbar  \btout{\widetilde{U}} \tinact \bnfbar \btinp{\widetilde{U}} \tinact
    M & 
    \bnfis 	\mugamma  \bnfbar~  \btout{\widetilde{U}} \mugamma 		
    	\bnfbar~ \btinp{\widetilde{U}} \mugamma \bnfbar \trec{t}{M }	
    	\\
    \mugamma & \bnfis \tinact \bnfbar \vart{t}
	\end{align*}
\end{definition}
The above definition is minimal in its use of sequencing, which is only present in  
  recursive session types such as $\trec{t}{\btout{U}\vart{t}}$ and
$\trec{t}{\btinp{U}\vart{t}}$---these are  tail-recursive session types with exactly one session prefix. 
Clearly, this minimal type structure induces a reduced set of typable \HO processes. 
A type system for \HO with minimal session types can be straightforwardly obtained by specializing the definitions, typing rules, and results summarized in \Cref{sec:types}.

We refer to \HO processes and terms typeable with minimal session types as \emph{MST processes and terms}, respectively.

We now define how to ``slice'' a standard session type into a \emph{list} of
minimal session types.
%Session types in higher-order and shared types are also sliced.
We need the following auxiliary definition.
  \begin{definition}[Predicates on Types and Names]
    Let $C$ be a channel type.
    \begin{itemize}
     \item We write $\tr(C)$  to indicate that $C$ is a tail-recursive session type. 
     \item Given $u:C$, we write $\lin{u}$ if a session type (i.e. $C = S$ for some $S$) that is not tail recursive.
    \end{itemize}
    With a slight abuse of notation, we write $\tr(u)$ to mean $u:C$ and $\tr(C)$ (and similarly for $\neg \tr(u)$).
  \end{definition}

\begin{definition}[Decomposing Session Types] 
  \label{mst:def:typesdecomp}
  Given the session, higher-order, and shared types of \Cref{d:types}, the \emph{type decomposition function}  $\Gt{-}$ is defined using the auxiliary function $\Rt{-}$ as in \Cref{mst:fig:typesdecomp}. 
We write $\len{\Gt{S}}$ to denote the length of  $\Gt{S}$ (and similarly for $\Rt{-}$).
\end{definition}

\begin{figure}[!t]
\begin{mdframed}
\vspace{-4mm}
	\begin{align*}
  \Gt{\btout{U}{S}} &=
  \begin{cases}
  \btoutt{\Gt{U}}   &  \text{if $S = \tinact$} \\
  \btoutt{\Gt{U}}\, ,\Gt{S} & \text{otherwise}  
\end{cases} 
 \\
 \Gt{\btinp{U}{S}} &=
  \begin{cases}
  \btinpt{\Gt{U}}  &  \text{if $S = \tinact$} \\
  \btinpt{\Gt{U}} \, , \Gt{S} & \text{otherwise}
  \end{cases}   
  \\
  \Gt{\trec{t}{S}} & = 
  \begin{cases} 	
    \Rt{S} & \text{if $\tr(\trec{t}{S})$} \\ 
    \trec{t}{\Gt{S}} & \text{if $\neg\tr(\trec{t}{S})$ and $\Gt{S}$ is a singleton}
 \end{cases} 
 \\
  \Gt{\tinact} &= \tinact 
  \\
   \Gt{\vart{t}} & = \vart{t} 
  \\
  \Gt{\lhot{C}} &= \lhot{\Gt{C}} 
  \\
  \Gt{\shot{C}} &= \shot{\Gt{C}}
  \\
  % \Gt{S_1,\ldots,S_n} &=
  % \Gt{S_1},\ldots,\Gt{S_n} 
  % \\
    \Gt{\chtype{U}} &= \chtype{\Gt{U}}
   \\
   \Rt{\vart{t}} & = \epsilon
  \\
    \Rt{\btout{U}S} &= \trec{t}{\btout{\Gt{U}}} \tvar{t}, \Rt{S}
   \\
   \Rt{\btinp{U}S} & = \trec{t}{\btinp{\Gt{U}}} \tvar{t}, \Rt{S} 
  % & 
  % \Rt{\btout{U}S} = \trec{t}{\btout{\Gt{U}}} \tvar{t}, \Rt{S}
  % \Rt{\btinp{U}S} = \trec{t}{\btinp{\Gt{U}}} \tvar{t}, \Rt{S} 
  % \\ 
  % \Rts{}{s}{\btinp{U}S} &= \Rts{}{s}{S}
  % \\
  % \Rts{}{s}{\btout{U}S} &= \Rts{}{s}{S}
\end{align*}
\end{mdframed}	
\caption[Decomposing session types into minimal session types]{Decomposing session types into minimal session types (\Cref{mst:def:typesdecomp})   \label{mst:fig:typesdecomp}}
\end{figure}
\noindent
The decomposition is self-explanatory; 
intuitively, if a session type $S$ contains $k$ input/output actions, the list $\Gt{S}$ will contain $k$ minimal session types.
For a tail recursive $\trec{t}{S}$, $\Gt{\trec{t}{S}}$ is a list of minimal recursive session types, obtained
 using the auxiliary function $\Rt{-}$ on   $S$: if $S$ has $k$ prefixes
 then the list  $\Gt{\trec{t}{S}}$ will contain $k$ minimal recursive session types.

 We illustrate \Cref{mst:def:typesdecomp} with three examples.
\begin{example}[Decomposition a Non-recursive Type]
Let $S = \btinp{\mathsf{int}} \btinp{\mathsf{int}} \btout{\mathsf{bool}} \tinact$ be the session type given in \Cref{s:intro}.
Then $\Gt{S}$ denotes the list   $\btinpt{\mathsf{int}}  \, ,  \btinpt{\mathsf{int}}   \, , \btoutt{\mathsf{bool}} $. 
 \hspace*{\fill} $\lhd$
\end{example}
\begin{example}[Decomposing a Recursive Type]
  \label{mst:ex:rtype}
  Let $S = \trec{t}S'$ be a recursive session type, with
  $S'=\btinp{\mathsf{int}}\btinp{\mathsf{bool}}\btout{\mathsf{bool}}\vart{t}$.
  By \Cref{mst:def:typesdecomp}, since $S$ is tail-recursive,  $\Gt{S} = \Rt{S'}$.
  Further, $\Rt{S'} = \trec{t}\btinp{\mathsf{\Gt{int}}} \vart{t},
  \Rt{\btinp{\mathsf{bool}}\btout{\mathsf{bool}}\vart{t}}$. By definition of
  $\Rt{-}$, we obtain $$\Gt{S} = \trec{t}\btinp{\mathsf{int}} \vart{t},~
  \trec{t}\btinp{\mathsf{bool}} \vart{t},~ \trec{t}\btout{\mathsf{bool}}
  \vart{t}, \Rt{t}$$ (using $\Gt{\mathsf{int}} = \mathsf{int}$ and
  $\Gt{\mathsf{bool}} = \mathsf{bool}$). Since $\Rt{\vart{t}} = \epsilon$, we
  obtain $$\Gt{S} = \trec{t}\btinp{\mathsf{int}} \vart{t},~
  \trec{t}\btinp{\mathsf{bool}} \vart{t},~ \trec{t}\btout{\mathsf{bool}}
  \vart{t}$$  	
  \hspace*{\fill} $\lhd$
\end{example}

In addition to tail-recursive types that are handled by $\Rt{-}$, we need to
support non-tail-recursive types of form $\trec{t}{\btinp{\shot{(\widetilde
T,\vart{t})}}\tinact}$ that are essential for the encoding of recursion in \HOp 
 into \HO. 
The following example illustrates such a decomposition. 
\begin{example}[Decomposing a Non-tail-recursive Type]
  \label{mst:ex:ntrsts}
  Let 
  $S =
  \trec{t}{\btinp{\shot{(\btinp{\mathsf{str}}\btout{\mathsf{str}}\tinact,
  \vart{t})}}\tinact}$ 
  be a non-tail-recursive type. We obtain the following decomposition: 
  \begin{align*} 
  \Gt{S} &= \trec{t}{\Gt{\btinp{\shot{(\btinp{\mathsf{str}}
  \btout{\mathsf{str}}\tinact, \vart{t})}}\tinact}} \\
  &= \trec{t}{{?({\Gt{\shot{(\btinp{\mathsf{str}}
  \btout{\mathsf{str}}\tinact,\vart{t})}}})}} 
  \\
  &= 
  \trec{t}{\btinpt{{{{\shot{(\btinpt{\mathsf{str}}, 
  \btoutt{\mathsf{str}},\vart{t})}}}}}} = M 
  \end{align*} 
  \noindent We can see that we have generated 
  minimal non-tail-recursive type $M$. 
  \hspace*{\fill} $\lhd$
\end{example}

Now, we illustrate the encoding of \HOp recursive
processes into \HO from~\cite{DBLP:conf/esop/KouzapasPY16} using the 
non-tail-recursive type  $S$ given in the above example.
\begin{example}[Encoding Recursion]
  \label{mst:e:en-rec}
  Consider the process
  %	\begin{align*}
    $P = \recp{X}\binp{a}{m}\bout{a}{m}X$, which contains recursion and so it  is not an $\HO$ process.  
    Still, $P$ can be encoded into $\HO$ as follows~\cite{DBLP:conf/esop/KouzapasPY16}:
    \begin{align*}
    \map{P} = \binp{a}{m}\bout{a}{m} \news{s}
    {(\appl{V}{(a,s)} \Par \about{\dual s}{V})}
    \end{align*}
    where the value $V$ is an abstraction that potentially reduces to $\map{P}$: 
    \begin{align*}
    V = \abs{(x_a,y_1)}\binp{y_1}{z_x}\binp{x_a}{m}
    \bout{x_a}{m}\news{s}
    {(\appl{z_x}{(x_a,s)} \Par \bout{\dual s}{z_x}\inact)} 
    \end{align*}
  
  As detailed in~\cite{DBLP:conf/esop/KouzapasPY16}, this encoding relies on non-tail-recursive types. 
  In particular, the bound name $s$ in $\map{P}$ is typed 
  with the following type, discussed above in~\Cref{mst:ex:ntrsts}:
  \begin{align*}
    S = \trec{t}{\btinp{\shot{(\btinp{\mathsf{str}}\btout{\mathsf{str}}\tinact,
\vart{t})}}\tinact} 
  \end{align*}

    We compose $\map P$ with an appropriate client process to illustrate the encoding of recursion. Below $R$ stands for some unspecified process such that $a \in \rfn{R}$:
    \begin{align*}
   % & \\
   \map{P} \Par \bout{a}{W} \binp{a}{b}R  & \red^2 \news{s} 
    {(\appl{V}{(a,s)} \Par \about{\dual s}{V})} \Par R \\
    &  \red  \news{s} {(\binp{s}{z_x}\binp{a}{m}
      \bout{a}{m}\news{s'}
      {(\appl{z_x}{(a,s')} \Par \about{\dual {s'}}{z_x})} \Par \about{\dual s}{V})} \Par R \\
    & \red \binp{a}{m}
      \bout{a}{m}\news{s'}
      {(\appl{V}{(a,s')} \Par \about{\dual {s'}}{V})} \Par R
      \\ 
      & = \map{P} \Par R
    \end{align*}
\end{example} 

\subsubsection{Decomposing Processes}
As we have seen, each session type $S$ is decomposed into $\Gt{S}$, a list  of minimal session types.
Accordingly, given an assignment $s : S$, we decompose $s$ into a series of names, one for each action in $S$.
We use \emph{indexed names} to formalize the names used by minimally typed processes.
Formally, an indexed name is a pair $(n, i)$ with $i \in \mathbb{N}$, which we denote as $n_i$.
We refer to processes with indexed names  as \emph{indexed processes}.

The decomposition of processes is defined in \Cref{mst:def:decomp}, and it relies on a breakdown function, denoted $\B{k}{\tilde x}{-}$, which operates on indexed processes.
Before we dive into those functions we present some auxiliary definitions.

\paragraph{Preliminaries.}
To handle the unfolding of recursive types, we shall use the following auxiliary function, which decomposes guarded recursive types, by first ignoring all the actions until the recursion.
\begin{definition}[Decomposing an Unfolded Recursive Type]
Let $S$ be a session type. The function $\Rts{}{s}{-}$: is defined as follows
  \begin{align*}
    \Rts{}{s}{\trec{t}{S}} & = \Rt{S}
    \\
    \Rts{}{s}{\btinp{U}S} & = \Rts{}{s}{S} 
    \\
    \Rts{}{s}{\btout{U}S} & = \Rts{}{s}{S}
  \end{align*} 
\end{definition}
\begin{example}
   Let $T = \btinp{\mathsf{bool}}\btout{\mathsf{bool}}S$ be a derived unfolding
    of  $S$ from \Cref{mst:ex:rtype}. Then, by \Cref{mst:def:typesdecomp},
    $\Rts{}{s}{T}$ is the list of minimal recursive types obtained as follows:
    first, $\Rts{}{s}{T} = \Rts{}{s}{\btout{\mathsf{bool}}\trec{t}S'}$ and after
    one more step, $\Rts{}{s}{\btout{\mathsf{bool}}\trec{t}S'} =
    \Rts{}{s}{\trec{t}S'}$. Finally, we have $\Rts{}{s}{\trec{t}S'} = \Rt{S'}$.
    We get the same list of minimal types as in \Cref{mst:ex:rtype}: $\Rts{}{s}{T}
    = \trec{t}{\btinp{\mathsf{int}}\vart{t}},
    \trec{t}{\btinp{\mathsf{bool}}\vart{t}}, \trec{t}{\btout{\mathsf{bool}}
    \vart{t}}$. 
    \hspace*{\fill} $\lhd$
\end{example}

Given an unfolded recursive session type $S$, the auxiliary function $\indT{S}$ returns the position of the top-most prefix of $S$ within its body.
  \begin{definition}[Index function]
  \label{mst:def:indexfunction}
  Let $S$ be an (unfolded) recursive session type. The function $\indT{S}$ is defined as follows:  
  \begin{align*}
    \indT{S} = \begin{cases}
     \indTaux{S'\subst{S}{\vart{t}}}{0} & \text{if} \ S =\trec{t}{S'} \\
     \indTaux{S}{0} & \text{otherwise}
   \end{cases}
  \end{align*}
  \noindent where $ \indTaux{S}{l}$:
  \begin{align*}
  \indTaux{\trec{t}{S}}{l} & = \len{\Rt{S}} - l + 1
    \\
 \indTaux{\btout{U}S}{l}  & = \indTaux{S}{l+1}
  \\
   \indTaux{\btinp{U}S}{l} & = \indTaux{S}{l+1}
  \end{align*}
  \end{definition}
    
  \begin{example}
  \label{mst:ex:fs}
      Let  
      $S' = \btinp{\mathsf{bool}}\btout{\mathsf{bool}}S$ where 
      $S$ is as in \Cref{mst:ex:rtype}.
      Then $\indT{S'} = 2$ since the top-most prefix of $S'$ (`$\btinp{\mathsf{bool}}$') is 
      the second prefix in the body of $S$. 
      \hspace*{\fill} $\lhd$
  \end{example}
\noindent
In order to determine the required number of propagators ($\prop_k, \prop_{k+1}, \ldots$) required in the breakdown of processes and  values, we  define the \emph{degree} of a process:
\begin{definition}[Degree of a Process]
  \label{mst:def:sizeproc}
	Let $P$ be an \HO process.
	The \emph{degree} of $P$, denoted $\plen{P}$, is defined as
	follows:
	$$ %\plen
	\plen{P} =
	\begin{cases}
    \plen{Q} + 1 & \text{if $P = \bout{u_i}{V}Q$ or $P=\binp{u_i}{y}Q$}
    \\
	% \len{V} + \len{Q} + 1 & \text{if $P = \bout{u_i}{V}Q$}
	% \\
	% \len{Q} + 1 & \text{if $P =\bout{u_i}{y}Q$ or $P=\binp{u_i}{y}Q$}
	% \\
	% 0 & \text{if $P = \appl{V}{u_i}$}
	% \\
	\plen{P'} & \text{if $P = \news{s:S}P'$}
	\\
	\plen{Q} + \plen{R} + 1 & \text{if $P = Q \Par R$}
	\\
	1 & \text{if $P = \appl{V}{u_i}$ or $P = \inact$}
	\end{cases}
	$$

\end{definition}

\noindent
We define an auxiliary function that ``initializes'' the indices of a tuple of
names, for turning a regular process into an indexed process.

\begin{definition}[Initializing an indexed process]
  \label{mst:d:counterinit}
Let $\widetilde u = (a,b,s,s',r,r',\ldots)$ be a finite tuple of names.
We shall write $\mathsf{init}(\widetilde u)$ to denote the tuple of indexed names
$(a_1,b_1,s_1,s'_1,r_1,r_1',\ldots)$.
\end{definition}

  \begin{definition}[Subsequent index substitution]
    \label{mst:d:nextn}
    Let $n_i$ be an indexed name. We define $\nextn{n_i} = \linecondit{\lin{n_i}}{\incrname{n}{i}}{\{\}}$.
  \end{definition}

\begin{remark}\label{mst:r:prefix}
Recall that we write
`$\apropinp{k}{}$' and  `$\apropout{k}{}$'
to denote input and output prefixes in which the value communicated along $\prop_k$ is not relevant.
While `$\apropinp{k}{}$' stands for `$\apropinp{k}{x}$', 
`$\apropout{k}{}$' stands for `$\apropout{k}{\abs{x}{\inact}}$'.
Their corresponding minimal types are 
$\btinpt{\shot{\tinact}}$ 
and 
$\btoutt{\shot{\tinact}}$, 
which are 
denoted by
$\btinpt{-}$ and  $\btoutt{-}$, respectively.
\end{remark}

Given a typed process $P$, we write $\rfn{P}$ to denote the set of free names of
$P$ whose types are recursive. As mentioned above, for each $r \in \rfn{P}$ with
$r : {S}$ we shall rely on a  control trio of the form
$\binp{\prop^r}{x}\appl{x}{\widetilde r}$, where $\widetilde{r} = r_1, \ldots,
r_{\len{\Gt{S}}}$.

\begin{definition}[Decomposition of a Process]
	\label{mst:def:decomp}
	Let $P$ be a closed \HO process with $\widetilde u = \fn{P}$ and $\widetilde v = \rfn{P}$.
  The \emph{decomposition} of $P$, denoted $\D{P}$, is
defined as:
$$
  \D{P} = \news{\widetilde \prop}\news{\widetilde \prop_r}\Big(
  \propout{k}{} \inact \Par \B{k}{\epsilon}{P\sigma} \Par \prod_{r \in \tilde{v} } 
  \recprov{r}{x}{\wtd r}
  \Big)
  $$
  \noindent where: $k >0$;
    $\widetilde \prop = (\prop_k,\ldots,\prop_{k+\plen{P}-1})$;
  $\widetilde{\prop_r} = \bigcup_{r\in \tilde{v}}\prop^r$;
    $\sigma = \subst{\mathsf{init}(\widetilde u)}{\widetilde u}$.
\end{definition}
\noindent
% As we can see, the decomposition of $P$ consists of three parts: the control trios for recursive names, the breakdown of the process with the initialized names, and the control trio for kick-starting the breakdown.
Notice that when $\rfn{P} = \emptyset$, then 
$
  \D{P} = \news{\widetilde \prop}(
  \propout{k}{} \inact \Par \B{k}{\epsilon}{P\sigma})
  $.
We now discuss the {breakdown} of process $P$, denoted $\B{k}{\tilde x}{P}$.

\paragraph{The Breakdown Function.}
Given  a context $\widetilde x$ and a  $k>0$, the breakdown of an indexed
process $P$, denoted $\B{k}{\tilde x}{P}$, is defined recursively on the structure of processes. The definition of $\B{k}{\tilde x}{-}$ relies on an auxiliary
breakdown function on values, denoted  $\V{k}{\tilde x}{-}$.
When $V = y$, then the breakdown function is simply the identity: 
$\V{}{\tilde x}{y} = y$.

  \begin{table}[!t]
\begin{tabular}{ |l|l|l|}
  \rowcolor{gray!25}
  \hline
  $P$ &
    \multicolumn{2}{l|}{
  \begin{tabular}{l}
    \noalign{\smallskip}
    $\B{k}{\tilde x}{P}$
    \smallskip
  \end{tabular}
}  \\
  \hline

$\bout{u_i}{V}{Q}$ &
    \begin{tabular}{l}
      \noalign{\smallskip}
      $\bullet~ \neg\tr(S)$: %& 
      \\
      \quad $\propinp{k}{\wtd x}
			\bbout{u_i}{\V{k+1}{\tilde y}{V\sigma}}
			\apropout{k+1}{\wtd w}  \Par \B{k+1}{\tilde w}{Q\sigma}$ 
      \smallskip 
      \\
      \hdashline 
      \noalign{\smallskip}
      $\bullet~ \tr(S)$:  \\
      \quad $\propinp{k}{\wtd x}
      		\abbout{\prop^u}{N_V}   \Par 
        \B{k+1}{\tilde w}{Q}$
        \\
        \quad where:
        \\
        \qquad $N_V = \abs{\wtd z}
      		{\bbout{z_{\indT{S}}}{\V{k+1}{\tilde y}{V}}}{}$ 
          \\
          \qquad \qquad \qquad \qquad \quad 
          $\big(\apropout{k+1}{\wtd w} \Par 
          \recprov{u}{x}{\wtd z}\big) 
      $
     \smallskip
  \end{tabular}
  &
  %side-conditions
  \begin{tabular}{l}
    \noalign{\smallskip}
    $u_i:S$ \\
    $\widetilde y = \fv{V}$\\
    $\widetilde w = \fv{Q}$ \\
    % $\degree = \len{V}$ \\
    $\sigma = \nextn{u_i}$ \\
    $\wtd z = (z_1,\ldots,z_{\len{\Rts{}{s}{S}}})$
    \smallskip
  \end{tabular}
  \\
\hline 
$\binp{u_i}{y}Q$
&

  \begin{tabular}{l}
      \noalign{\smallskip}
      $\bullet~ \tr(S)$: 
      \\
      \quad $\propinp{k}{\wtd x}\binp{u_i}{y} \apropout{k+1}{\wtd w}
     \Par \B{k+1}{\tilde w}{Q\sigma}$ %& (if $\neg\tr(S)$)
     \smallskip 
      \\
      \hdashline 
      \noalign{\smallskip} 
         $\bullet~ \neg\tr(S)$: 
      \\ 
     \quad $\propinp{k}{\wtd x}\abbout{\prop^u}
     {N_y} 
     \Par \B{k+1}{\tilde w}{Q}$ %& (if $\tr(S)$)
     \\
     \quad where:
     \\
     \qquad $N_y = \abs{\wtd z}{\binp{z_{\indT{S}}}{y}$
     \\
     \qquad \qquad \qquad \qquad \quad 
     $\big(\apropout{k+1}{\wtd w} \Par
     \recprov{u}{x}{\wtd z}} \big)$
      \smallskip
  \end{tabular}
  &
  %side-conditions
  \begin{tabular}{l}
    \noalign{\smallskip}
   $u_i:S$ \\
	$\wtd w = \fv{Q} $ \\
  $\sigma = \nextn{u_i}$ \\
  $\wtd z = (z_1,\ldots,z_{\len{\Rts{}{s}{S}}})$
    \smallskip
  \end{tabular}
    \\
  \hline

    $\appl{V}{(\wtd r, u_i)}$ &
  \begin{tabular}{l}
    \noalign{\smallskip}
    $\propinp{k}{\widetilde x}\overbracket{\prop^{r_1}!\big\langle 
      \lambda \widetilde z_1. \prop^{r_2}!\langle\lambda \widetilde z_2.\cdots. 
   \prop^{r_n}!\langle \lambda \widetilde z_n.}^{n = |\tilde r|} 
   Q \rangle \,\rangle \big\rangle$ \\
  where:
  \\
   \quad $Q = \appl{\V{k+1}{\tilde x}{V}}{(\widetilde z_1,\ldots,
   \widetilde z_n, \widetilde m)}$
    \smallskip
  \end{tabular}
  &
  %side-conditions
  \begin{tabular}{l}
  	\noalign{\smallskip}
  	$u_i:C$ \\
    $\forall r_i \in \widetilde r.(r_i: S_i \wedge \mathsf{tr}(S_i) \wedge$
    \\
    \quad $\wtd{z_i} = (z^i_1,\ldots,z^i_{\len{\Rts{}{s}{S_i}}}))$\\
    $\wtd m = (u_i, \ldots, u_{i+\len{\Gt{C}}-1})$
    \smallskip
  \end{tabular}
  \\
  \hline
  $\news{s:C}{P'}$ &
  \begin{tabular}{l}
    \noalign{\smallskip}
     $\bullet \neg\tr(C):$
     \\
 	$\news{\widetilde{s}:\Gt{C}}{\,\B{k}{\tilde x}
 			{P'\sigma}}$ %& (if $\neg\tr(C)$)
       \smallskip 
       \\
       \hdashline
       \noalign{\smallskip}
       $\bullet ~\tr(C):$
       \\
   $\news{\widetilde{s}:\Gt{C}}
 		\news{c^s}
    \recprov{s}{x}{\wtd s}
    % \binp{\prop^s}{b}(\appl{b}{\wtd s}) 
    \Par$
    % & (if $\tr(C)$)
     \\ 
 		\qquad 
 		$\news{c^{\bar{s}}}
    \recprov{\bar s}{x}{\wtd {\dual s}}
      \Par$ 
 		$\B{k}{\tilde x}{P'\sigma}$
    \smallskip
  \end{tabular}
  &
  %side-conditions
  \begin{tabular}{l}
    \noalign{\smallskip}
    $\wtd x = \fv{P'}$ \\
    $\wtd{s} = (s_1,\ldots,s_{\len{\Gt{C}}})$ \\
    $\wtd {\dual{s}} = (\dual{s_1},\ldots,\dual{s_{\len{\Gt{C}}}})$ \\
    $\sigma=\subst{s_1 \dual{s_1}}{s \dual{s}}$
    \smallskip
  \end{tabular}
  \\
  \hline
  $Q \Par R$ &
  \begin{tabular}{l}
    \noalign{\smallskip}
    $\propinp{k}{\wtd x} \propout{k+1}{\wtd y}
    \apropout{k+\degree+1}{\wtd w}   \Par$
    \\ 
    $\B{k+1}{\tilde y}{Q} \Par \B{k+\degree+1}{\tilde w}{R}$
    \smallskip
  \end{tabular}
  &
  %side-conditions
  \begin{tabular}{l}
    \noalign{\smallskip}
    $\wtd y  = \fv{Q}$ \\ 
    $\wtd w = \fv{R}$ \\
    $\degree = \plen{Q}$
    \smallskip
  \end{tabular}
      \\
  \hline

  $\inact$ &
\begin{tabular}{l}
  \noalign{\smallskip}
  $\propinp{k}{}\inact$
 \smallskip
\end{tabular}
&
\begin{tabular}{l}
	\noalign{\smallskip}
%$\widetilde x = \es$ 
\smallskip
\end{tabular}
%$\widetilde x = \epsilon$
%side-conditions
  \\
  \hline

\rowcolor{gray!25}
$V$ &
  \multicolumn{2}{l|}{
\begin{tabular}{l}
  \noalign{\smallskip}
  $\V{}{\tilde x}{V}$
  \smallskip
\end{tabular} } 
\\
\hline
$y$ &
\begin{tabular}{l}
  \noalign{\smallskip}
  $y$
  \smallskip
\end{tabular}
&
%side-conditions\\
\begin{tabular}{l}
  \noalign{\smallskip}
%  $\widetilde x = \fv{V}$ \\
%  $\widetilde{y} = (y_1,\ldots,y_{\len{\Gt{C}}})$
  \smallskip
\end{tabular}
\\
\hline

% $\abs{(\widetilde y z):\slhotup{(\wtd S C)}}P$ 
$\abs{(\wtd y z)}P$ 
& 
\begin{tabular}{l}
  \noalign{\smallskip}
$\abs{(\widetilde{y^1},\ldots,\widetilde{y^n}, \widetilde z):
\slhotup{(\widetilde{M})}}{N}$ 
\\
\smallskip
  where:
  \\
  $\widetilde{M} = (\Gt{S_1},\ldots,\Gt{S_n}, \Gt{C})$
  \\
         $N = \news{\widetilde \prop} \news{\widetilde \prop_r}
     \prod_{i \in \len{\widetilde y}}
     (\recprov{y_i}{x}{\wtd y^i})
     \Par$ 
     \\
\quad \quad \qquad \qquad \qquad  \quad  
     $\apropout{1}{\widetilde x}
 \Par$ 
%  \\
% \quad \quad \qquad \qquad \qquad  \qquad  
% \ \  
$\B{1}{\tilde x}{P \subst{z_1}{z}}$
\smallskip 
\end{tabular}
&
%side-conditions\\
\begin{tabular}{l}
  % \noalign{\smallskip}
  \noalign{\smallskip}
  $\wtd y z : \wtd S C$ \\
  $\forall y_i \in \widetilde y.(y_i: S_i \wedge \mathsf{tr}(S_i) \wedge$\\
   \quad $\widetilde{y^i} = (y^i_1,\ldots,y^i_{\len{\Gt{S_i}}}))$\\
   $\widetilde z = (z_1,\ldots,z_{\len{\Gt{C}}})$ \\
   $\wtd \prop = (\prop_1,\ldots,\prop_{\plen{P}})$ \\ 
$\widetilde{\prop_r} = \bigcup_{r\in \tilde{y}}\prop^r$ 
  \smallskip
\end{tabular}
\\
\hline
\end{tabular}
\caption{The breakdown function for processes and values. \label{mst:t:bdowncore}}
\end{table}
%\end{center}
 
The breakdown function relies on type information, in two ways.
First, names are decomposed based on their session types.
Second, for most constructs the shape of decomposed process depends on whether the associated session type is tail-recursive or not.
The definition of the breakdown function is given in \Cref{mst:t:bdowncore}.
Next, we describe each of the cases of the definition.
In \Cref{mst:ss:examples} (Page~\pageref{mst:ss:examples}) we develop several examples.

\paragraph{Output:}
The decomposition of $\bout{u_i}{V}Q$ is arguably the most interesting case, as both the sent value $V$ and the continuation $Q$ have to be decomposed. 
We distinguish two cases:
\begin{itemize}
  \item If $\neg\tr(u_i)$ then  $u_i$ is linear or shared, and then we have:
\begin{align*} 
  \B{k}{{\tilde x}}{\bout{u_i}{V}Q} =
    \propinp{k}{\widetilde x}
        \bbout{u_i}{\V{}{\tilde y}{V\sigma}}
        \apropout{k+1}{\wtd w}  \Par
  \B{k+1}{\tilde w}{Q\sigma}
  \end{align*} 

This decomposition consists of a leading trio that
mimics an output action in parallel with the breakdown of  $Q$.
The context  $\wtd x$ must include the free variables of $V$ and $Q$, which are 
denoted $\wtd y$ and $\wtd w$, respectively. These tuples are not
necessarily disjoint: variables with shared types can appear free in both $V$
and $Q$. The value $V$ is then broken down with parameters $\wtd
y$ and $k+1$; the latter serves to consistently generate propagators for the
trios in the breakdown of $V$, denoted $\V{}{\tilde y}{V\sigma}$ (see below). 
The substitution $\sigma$  increments the index of session
names; it is applied to both $V$ and $Q$ before they are broken down. 
By taking $\sigma = \nextn{u_i}$ we 
distinguish two cases (see \Cref{mst:d:nextn}): 
\begin{itemize}
\item If name $u_i$ is linear (i.e., it has a session type) then its future
occurrences are renamed into $u_{i+1}$, and $\sigma = \subst{u_{i+1}}{u_i}$; 
\item Otherwise, if $u_i$ is {shared}, then $\sigma = \{\}$.
\end{itemize}
Note that if $u_i$ is linear then it appears either in $V$ or $Q$  and $\sigma$
affects only one of them. The last prefix   activates the
breakdown of  $Q$ with its corresponding context $\wtd w$.

\smallskip

In case $V  = y$, the same strategy applies; because $\V{k}{\tilde y}{y\sigma} = y$, we have: 
\begin{align*} 
\B{k}{\tilde x}{\bout{u_i}{y}Q} =
\binp{c_k}{\widetilde x} \bout{u_i}{y} \apropout{k+1}{\wtd w}
\Par \B{k+1}{\tilde w}{Q\sigma}
\end{align*} 
Notice that variable $y$ is not propagated further if it does not appear free in $Q$. 

\item 
If $\tr(u_i)$ then $u_i$ is tail-recursive and then we have: 

\begin{align*}
  \B{k}{\tilde x}{\bout{u_i}{V}Q}=~~&\propinp{k}{\wtd x}
        \abbout{\prop^u}{N_V} \Par
      \B{k+1}{\tilde w}{Q} \\
     \text{where:}~~  
 N_V =~~& \abs{\wtd z}
        {\bbout{z_{\indT{S}}}{\V{}{\tilde y}{V}}}
        \big(\apropout{k+1}{\wtd w} \Par 
        \recprov{u}{x}{\wtd z}
    \big) 
\end{align*}

The decomposition consists of a leading trio that mimics the output action
running in parallel with the breakdown of  $Q$. After receiving the context
$\widetilde x$, the leading trio sends an abstraction $N_V$ along $\prop^u$, which 
performs several tasks. First, $N_V$ collects the sequence of names $\tilde u$;
then, it mimics the output action of $P$ along one of such names ($u_{\indT{S}}$) and
triggers the next trio, with context $\wtd w$; finally, it reinstates the
server on $\prop^u$ for the next trio that uses~$u$.  
Notice that indexing is not relevant in this case. 

In case $V = y$, we have $\V{}{\tilde y}{y\sigma}=y$ and $\plen{y}=0$, hence: 
\begin{align*} 
  \B{k}{\tilde x}{\bout{u_i}{y}Q} = 
  \propinp{k}{\wtd x} \abbout{\prop^u}
  {\abs{\wtd z}{\bout{z_{\indT{S}}}{y}}
  \big(\apropout{k+1}{\wtd w} \Par 
  \recprov{u}{x}{\wtd z}
  % \binp{c^u}{b}(\appl{b}{\wtd z}) 
  \big) } 
  \Par \B{k+1}{\tilde w}{Q}
\end{align*} 
\end{itemize}

\paragraph{Input:}
To decompose a process $\binp{u_i}{y}Q$ we distinguish two cases, as before: 
\rom{1}  name $u_i$ is linear or shared or \rom{2} tail-recursive.
In case~\rom{1}, the breakdown   is defined as follows: 
\begin{align*} 
  \B{k}{\tilde x}{\binp{u_i}{y}Q} =
  \propinp{k}{\wtd x}\binp{u_i}{y} \apropout{k+1}{\wtd w}
      \Par \B{k+1}{\tilde w}{Q\sigma}
\end{align*} 
\noindent where $\wtd w = \fv{Q}$.
A leading trio mimics the input action and possibly extends the context with the
received variable $y$. The substitution $\sigma$ is defined as in the output
case.

In case~\rom{2}, when $u_i$ has tail-recursive session type $S$, the decomposition is as in the output case: 
\begin{align*} 
\B{k}{\tilde x}{\binp{u_i}{y}Q}=\propinp{k}{\wtd x}\abbout{\prop^u}
      {\abs{\wtd z}{\binp{z_{\indT{S}}}{y}
      \big(\apropout{k+1}{\wtd w} \Par 
      \recprov{u}{x}{\wtd z}
      % \binp{\prop^u}{b}(\appl{b}{\wtd z}) 
      \big)}} 
      \Par \B{k+1}{\tilde w}{Q} 
\end{align*} 
%In this case we need to receive $\tilde u$ using one of the received names ($u_{f(S)}$) as a subject for the input action.

\paragraph{Application:}
For simplicity we consider the breakdown of applications of the form $\appl{V}{(\wtd r, u_i)}$, where
every $r_i \in \wtd r$ is such that $\tr(r_i)$ and only  $u_i$ is such that $\neg\tr(u_i)$.
The general case (involving different orders in names and multiple names with non-recursive types) is similar. We have:  
\begin{align*}
 \B{k}{\tilde x}{\appl{V}{(\widetilde r, u_i)}} = &
 \propinp{k}{\widetilde x}\overbracket{\prop^{r_1}!\big\langle 
        \lambda \widetilde z_1. \prop^{r_2}!\langle\lambda 
        \widetilde z_2.\cdots. 
	 	\prop^{r_n}!\langle \lambda \widetilde z_n.}^{n = |\tilde r|} 
	 	\appl{\V{}{\tilde x}{V}}{(\widetilde z_1,\ldots,
	 	\widetilde z_n, \widetilde m)} \rangle \,\rangle \big\rangle
\end{align*}

Let us first discuss how names in $(\wtd r, u_i)$ are decomposed using types.
Letting $|\tilde
r| = n$ and $i \in \{1,\ldots,n\}$, for each $r_i \in \widetilde r$ (with
$r_i:S_i$) we generate a sequence $\wtd
z_i=(z^i_1,\ldots,z^i_{\len{\Rts{}{s}{S_i}}})$ as in the output case. 
We decompose name $u_i$ (with $u_i:C$) as $\wtd m = (u_i,\ldots,u_{i+\len{\Gt{C}}-1})$. 

The decomposition first receives a context $\wtd x$ for value $V$: we break down $V$
   with $\wtd x$ as a context since these variables need to be propagated to the
   abstracted process.
Subsequently, an output on $\prop^{r_1}$ sends a value containing $n$ abstractions that occur nested
within output prefixes---this is similar to the mechanism for partial instantiation shown in 	\Cref{e:partial}.
For each $j \in \{1,\ldots,n-1\}$, each abstraction
binds $\widetilde z_j$ and sends the next abstraction along $\prop^{r_{j+1}}$.
The innermost abstraction abstracts over $\widetilde z_n$ and encloses
the process $\appl{\V{}{\tilde x}{V}}{(\widetilde z_1,\ldots, \widetilde z_n,
\widetilde m)}$, which effectively mimics the  application. 
This abstraction nesting binds all variables $\widetilde z_i$, the 
decompositions of all tail-recursive names ($\wtd r$). 
%This structure can be seen as an encoding of partial application: by virtue of a
%single synchronization on $\prop^{r_i}$ part of variables (i.e., $\widetilde
%z_i$) will be instantiated.   

The breakdown of a value application of the form 
$\appl{y}{(\widetilde r, u_i)}$ results into the following specific case:
%\begin{align*}
$$
	 \B{k}{\tilde x}{\appl{y}{(\widetilde r, u_i)}} = 
	 \propinp{k}{\widetilde x}\overbracket{\prop^{r_1}!\big\langle 
        \lambda \widetilde z_1. \prop^{r_2}!\langle\lambda \widetilde z_2.\cdots. 
	 	\prop^{r_n}!\langle \lambda \widetilde z_n.}^{n = |\tilde r|} 
	 	\appl{y}{(\widetilde z_1,\ldots,
	 	\widetilde z_n, \widetilde m)} \rangle \,\rangle \big\rangle
$$ 
%\end{align*}

\paragraph{Restriction:}
The decomposition of $\news{s:C}{P'}$ depends on $C$:
  \begin{itemize}
    \item If $\neg\tr(C)$  then     \begin{align*} 
\B{k}{\tilde x}{\news{s:C}{P'}}
=
\news{\widetilde{s}:\Gt{C}}{\,\B{k}{\tilde x}
	{P'\sigma}}
    \end{align*} 
By construction, $\widetilde x = \fv{P'}$. Similarly as in the decomposition of
$u_i$ into $\widetilde m$ discussed above, we use the type $C$ of $s$ to obtain
the tuple $\widetilde s$ of length $\len{\Gt{C}}$. 
We initialize the index of
$s$ in $P'$ by applying the substitution $\sigma$. This substitution depends on
$C$: if it is a shared type then $\sigma = \subst{s_1}{s}$; otherwise, if $C$ is
a session type, then $\sigma = \subst{s_1\dual{s_1}}{s\dual{s}}$.

    \item Otherwise, if $\tr(C)$ then we have: 
    \begin{align*}
    \B{k}{\tilde x}{\news{s:C}{P'}} = 
    \news{\widetilde{s}:\mathcal{R}(S)}
          \news{c^s}
          \recprov{s}{x}{\wtd s}
          \Par 
       \news{c^{\bar{s}}}
       \recprov{\bar{s}}{x}{\wtd{\dual s}}
       \Par
         \B{k}{\tilde x}{P'}
    \end{align*}
    We decompose $s$ into $\wtd{s} = (s_1,\ldots,s_{\len{\Gt{S}}})$ and
    $\dual{s}$ into $\widetilde{\dual{s}} =
    (\dual{s_1},\ldots,\dual{s_{\len{\Gt{S}}}})$. Notice
    that as $\tr(C)$ we have  
    $C \equiv \trec{t}S$, therefore $\Gt{C} = \Rt{S}$.  
     The breakdown introduces two servers in parallel with the breakdown of
      $P'$; they provide names for $s$ and $\dual{s}$ along $\prop^s$
      and $\prop^{\dual{s}}$, respectively. The server on $\prop^s$ (resp.
      $\prop^{\dual{s}}$) receives a value   
      and applies it to the sequence $\widetilde s$ (resp.
      $\widetilde{\dual{s}}$). We restrict over $\widetilde s$ and propagators
      $\prop^s$ and $\prop^{\dual{s}}$. 
  \end{itemize}

\paragraph{Composition:}
 The breakdown of a process $Q \Par R$ is as follows:
$$
\B{k}{\tilde x}{Q \Par R}
=
 \propinp{k}{\widetilde x} \propout{k+1}{\wtd y}
    \apropout{k+\degree+1}{\wtd w}   \Par
\B{k+1}{\tilde y}{Q} \Par \B{k+\degree+1}{\tilde w}{R}
$$
A control trio triggers the breakdowns of $Q$ and $R$; it does not mimic any
 action of the source process. The tuple $\wtd y \subseteq \wtd x$ (resp. $\wtd
 w \subseteq \wtd x$) collects the free variables in  $Q$ (resp. $R$). To avoid
 name conflicts, the trigger for the breakdown of  $R$ is $\dual
 {c_{k+\degree+1}}$, with $\degree = \plen{Q}$.

\paragraph{Inaction:}
To breakdown $\inact$, we define a degenerate trio with only one input prefix
that receives a context that by construction will always be empty (i.e., $\widetilde x = \epsilon$, cf. \Cref{mst:r:prefix}):
$$
\B{k}{\tilde x}{\inact}
=
	\propinp{k}{}\inact
$$

\paragraph{Value:}
For simplicity, let us consider   values of the form 
$V=\abs{(\widetilde y, z):\slhotup{(\widetilde S, C)}}P$, 
where $\tr(y_i)$ holds for every $y_i \in \widetilde y$ and $\neg \tr(z)$, and 
$\leadsto \in \{\lollipop, \sharedop\}$.
The general case is defined similarly. 
We have: 
\begin{align*}
\V{}{\tilde x}{\abs{(\widetilde y, z):\slhotup{(\widetilde S, C)}}P} & =  
\abs{(\widetilde{y^1},\ldots,\widetilde{y^n}, \widetilde z):
  	\slhotup{(\widetilde{M})}}{N}  \quad \text{where:}
	\\
	\widetilde{M} & =\Gt{S_1},\ldots,\Gt{S_n}, \Gt{C}
	\\
  	N & = \news{\widetilde \prop}
    \news{\wtd \prop_r}
   \prod_{i \in \len{\widetilde y}}
  %  (\binp{\prop^{y_i}}{b}
  %  (\appl{b}{\widetilde y^i})) 
   \recprov{y_i}{x}{\wtd y^i}
   \Par \apropout{1}{\widetilde x}
   \Par \B{1}{\tilde x}{P \subst{z_1}{z}}
\end{align*}  
 
Every $y_i$ (with $y_i : S_i$) is
decomposed into $\widetilde y^i=(y_1,\ldots,y_{\len{\Gt{S_i}}})$. We use
$C$ to decompose $z$ into $\widetilde{z}$. 
We abstract over $\wtd
y^1,\ldots,\wtd y^n, \wtd z$; 
%  $\widetilde{z}$; 
 the body of the abstraction (i.e. $N$) is the composition of recursive names
 propagators, the control trio, and the breakdown of $P$, with name index
 initialized with the substitution $\subst{z_1}{z}$.  
 For every $y_i \in \widetilde y$ there is a server
$\recprov{y_i}{x}{\wtd y^i}$
as a subprocess in the
abstracted composition---the rationale for these servers is as in previous cases.
  We restrict the
 propagators 
 $\wtd \prop = (\prop_1,\ldots,\prop_{\plen{P}})$: this
 enables us to type the value in a shared environment when $\leadsto =
 \rightarrow$. 
% Variable $z$ is decomposed as in \Cref{mst:t:bdowncore}.
% The breakdown is similar to the (monadic) shared value given in \Cref{mst:t:bdowncore}. 
 Also, we restrict special propagator names $\wtd \prop_r
= \bigcup_{r\in \tilde{v}}\prop^r$. 

\subsection{The Decomposition by Example}
\label{mst:ss:examples}
We illustrate the decompositions by means of several examples.

\subsubsection{Decomposing Processes with Non-Recursive Names}
\begin{example}
\label{mst:e:decomp-proc-types}
  Consider process $P = \news{u} (Q \Par R)$ 
 whose body implements end-points of channel $u$ with
 session type $S=\btinp{U}\btinp{\textsf{bool}}\tinact$, 
  with $U=\lhot{(\btinp{\textsf{bool}}\tinact)}$, 
 	where: 
  \begin{align*}
    Q &= \binp{u}{x}
    \overbrace{\binp{u}{y}
    \news{s}\big( \appl{x}{\dual s} \Par \about{s}{y} \big)}^{Q'} 
    \\
    R &= \bout{\dual u}{V}\bout{\dual u}{\textsf{true}}\inact 
    \\ 
    V &= \abs{z}\binp{z}{b}\inact 
  \end{align*}
  The process $P$ reduces as follows: 
\begin{align*}
  P &\red 
  \news{u}\left( \binp{u}{y}
    \news{s}\big( \appl{V}{\dual s} \Par \about{s}{y} \big)  
    \Par 
      \bout{\dual u}{\textsf{true}}\inact\right)
  \red 
    \news{s}\big( \appl{V}{\dual s} \Par \about{s}{\textsf{true}} \big)  
    \\
    &
    \red 
    \news{s}\big( \binp{\dual s}{b}\inact  \Par \about{s}{\textsf{true}} \big)  
    = P' 
\end{align*}
  \noindent 
  By~\Cref{mst:def:decomp} we have that the decomposition of $P$ is as 
  follows: 
  \begin{align*}
    \D{P}=
	\news{\prop_1, \ldots, \prop_{10}}(\apropout{1}{} \Par 
	\B{1}{\epsilon}{P\sigma}) 
  \end{align*}
  where  
  $\sigma = \subst{u_1 \dual u_1}{u \dual u}$. We have: 
  \begin{align*}
  	\B{1}{\epsilon}{P\sigma} &= 
	\news{u_1,u_2} \propinp{1}{}\propout{2}{}\apropout{8}{}
	\Par \B{2}{\epsilon}{Q\sigma} \Par \B{8}{\epsilon}{R\sigma} 
  \end{align*}
  \noindent The breakdowns of sub-processes $Q$ and $R$ are 
  as follows: 
  \begin{align*}
	\B{2}{\epsilon}{Q\sigma} &= 
	\propinp{2}{}\binp{u_1}{x}\apropout{3}{x} 
	\Par 
  \B{3}{\epsilon}{Q'\sigma'} 
  \\
  \B{3}{x}{Q'\sigma'} &=
  \propinp{3}{x}\binp{u_2}{y}\apropout{4}{x,y} 
	\Par 
  \B{4}{}{\news{s}\big( \appl{x}{\dual s} \Par \about{s}{y} \big)}
  \\
  \B{4}{x,y}{\news{s}\big( \appl{x}{\dual s} \Par \about{s}{y} \big)} &= 
  \news{s_1}
	\big( 
	\propinp{4}{x,y}\propout{5}{x}\apropout{6}{y}  
	\Par 
	\propinp{5}{x}\appl{x}{\dual s_1} 
	\Par 
	\propinp{6}{y}\bout{s_1}{y}\apropout{7}{} 
	\Par \propinp{7}{}\inact 
	\big) 
	\\
	\B{8}{\epsilon}{R\sigma} &= 
	\propinp{8}{}\bout{\dual u_1}{\V{}{\epsilon}{V}}\apropout{9}{} 
	\Par 
  \B{9}{\epsilon}{\bout{u_2}{\textsf{true}}\inact} 
	\\
  \B{9}{\epsilon}{\bout{u_2}{\textsf{true}}\inact} 
  &=
  \propinp{9}{}\bout{\dual u_2}{\textsf{true}}\apropout{10}{} \Par 
	\propinp{10}{}\inact 
  \\
	\V{}{\epsilon}{V} &= \abs{z_1}
	(\news{\prop^V_1,\prop^V_2} 
	\apropoutv{^V_1}{} \Par \propinpv{^V_1}{}
	\binp{z_1}{b}\apropoutv{^V_2}{}
	\Par \propinpv{^V_2}{}\inact)
	  \end{align*}
	\noindent where $\sigma'= \subst{u_2 \dual u_2}{u \dual u}$. 
	By $\Gt{-}$ from~\Cref{mst:def:typesdecomp} we decompose $S$ into 
	$M_1$ and $M_2$ given as follows: 
	\begin{align*}
		M_1&=\btinp{\Gt{U}}\tinact=\btinp{{U}}\tinact 
		\\
		M_2&=\btinp{\textsf{bool}}\tinact
	\end{align*} 
	\noindent Above we may notice that 
	$\Gt{U}=U$. 
%	$M_1=\btinp{\Gt{U}}\tinact=\btinp{{U}}\tinact$ and 
%	$M_2=\btinp{\textsf{bool}}\tinact$. 
	We remark that $\D{P}$ accordingly   
	implements indexed names $u_1,u_2$ typed with 
	$M_1,M_2$, 
	respectively. 
	
    Let us inspect the reductions of $\D{P}$. 
    First, there are three synchronizations on $\prop_1,\prop_2$, and $\prop_8$: 
    \begin{align*}
      \D{P} &\red 
      \news{\prop_2,\ldots,\prop_{10}}
      \news{u_1,u_2} 
      \propout{2}{}\apropout{8}{}
      \Par \B{2}{\epsilon}{Q\sigma} \Par \B{8}{\epsilon}{R\sigma}
      \\ 
      &\red^2  
       \news{\prop_3,\ldots, \prop_7,\prop_9, \prop_{10}}
       \highlighta{\binp{u_1}{x}}\apropout{3}{x} 
	\Par 
  \B{3}{\epsilon}{Q'\sigma'} 
  \\
  & 
  \qquad 
  \Par 
  \highlighta{\bout{\dual u_1}{\V{}{\epsilon}{V}}}\apropout{9}{} 
	\quad \Par 
	\B{9}{\epsilon}{\bout{u_2}{\textsf{true}}\inact}
	= D^1 
    \end{align*}
    After reductions on propagators, $D^1$ is able to mimic the original synchronization  
    on channel $u$ (highlighted above). It is followed by 
    two administrative reductions on $\prop_3$ and $\prop_9$: 
    \begin{align*}
      D^1 &\red 
       \news{\prop_3,\ldots, \prop_7,\prop_9, \prop_{10}}
    \apropout{3}{\V{}{\epsilon}{V}} 
	\Par 
  \propinp{3}{x}\binp{u_2}{y}\apropout{4}{x,y} 
	\Par 
  \B{4}{}{\news{s}\big( \appl{x}{\dual s} \Par \about{s}{y} \big)}
  \Par 
  \\
  & \qquad 
  \apropout{9}{} 
	\Par 
	\propinp{9}{}\bout{\dual u_2}{\textsf{true}}\apropout{10}{} \Par 
	\propinp{10}{}\inact 
  \\
  & \red^2 
   \news{\prop_4,\ldots, \prop_7, \prop_{10}}
  \highlighta{\binp{u_2}{y}}\apropout{4}{\V{}{\epsilon}{V},y} 
	\Par 
	\\
	& 
	\quad 
	\news{s_1}
	\big( 
	\propinp{4}{x,y}\propout{5}{x}\apropout{6}{y}  
	\Par 
	\propinp{5}{x}\appl{x}{\dual s_1} 
	\Par 
	\propinp{6}{y}\bout{s_1}{y}\apropout{7}{} 
	\Par \propinp{7}{}\inact 
	\big) 
  \Par 
  \\
  & \qquad 
  \highlighta{\bout{\dual u_2}{\textsf{true}}}\apropout{10}{} \Par 
	\propinp{10}{}\inact = D^2 
    \end{align*}
    Similarly, $D^2$ can mimic the next synchronization of 
    the original process on name $u_2$. Following up on that, 
    syncronization on $\prop_{10}$ takes place: 
   \begin{align*}
    D^2 &\red^2 
    \news{\prop_4, \ldots, \prop_7} 
   \apropout{4}{\V{}{\epsilon}{V},\textsf{true}} 
	\Par 
	\\
	& 
	\quad 
	\news{s_1}
	\big( 
	\propinp{4}{x,y}\propout{5}{x}\apropout{6}{y}  
	\Par 
	\propinp{5}{x}\appl{x}{\dual s_1} 
	\Par 
	\propinp{6}{y}\bout{s_1}{y}\apropout{7}{} 
	\Par \propinp{7}{}\inact 
	\big) = D^3 
   \end{align*}   
   Now, we can see that the next three reductions on $\prop_4$, 
   $\prop_5$, and $\prop_6$ 
   appropriately propagate values $\V{}{\epsilon}{V}$ 
   and $\textsf{true}$ to the breakdown of sub-processes. 
   Subsequently, value $\V{}{\epsilon}{V}$ is applied to 
   name $\dual s_1$: 
   \begin{align*}
   	D^3 
  &\red 
   \news{\prop_5, \ldots, \prop_7} 
  \news{s_1}
	\big( 
	\propout{5}{\V{}{\epsilon}{V}}\apropout{6}{\textsf{true}}  
	\Par 
	\propinp{5}{x}\appl{x}{\dual s_1} 
	\Par 
	\propinp{6}{y}\bout{s_1}{y}\apropout{7}{} 
	\Par \propinp{7}{}\inact 
	\big)  
 \\
 & 
 \red^2 
 \news{\prop_7} 
 \news{s_1}
	\big( 
	\appl{\V{}{\epsilon}{V}}{\dual s_1} 
	\Par \bout{s_1}{\textsf{true}}\apropout{7}{} 
	\Par \propinp{7}{}\inact 
	\big) 
  \\
  & 
  \red 
  \news{\prop_7}
  \news{s_1}
	(\news{\prop^V_1,\prop^V_2} 
	\apropoutv{^V_1}{} \Par \propinpv{^V_1}{}
	\binp{s_1}{b}\apropoutv{^V_2}{}
	\Par \propinpv{^V_2}{}\inact)
  \Par 
  \bout{s_1}{\textsf{true}}\apropout{7}{} 
	\Par \propinp{7}{}\inact = D^4
   \end{align*}
   
%   Finally, we may notice that process $D^4$ 
%   is equivalent to a decomposition of $P'$. 
\noindent 
   Finally, after syncronization on $\prop^V_1$ 
   we reach the process $D^5$ that is clearly able to simulate $P'$, and its internal communication on the channel $s$: 
   \begin{align*}
   D^4 \red 
   \news{\prop_7}
   	\news{s_1}
	(\news{\prop^V_2} 
	\binp{s_1}{b}\apropoutv{^V_2}{}
	\Par \propinpv{^V_2}{}\inact)
  \Par 
  \bout{s_1}{\textsf{true}}\apropout{7}{} 
	\Par \propinp{7}{}\inact = D^5 
   \end{align*}

  $\hfill \lhd$
\end{example}

\begin{example}[Breaking Down Name-Passing]
  \label{mst:ex:bdnp}
  Consider the following process $P$, in which a channel $m$ is passed, through which a boolean value is sent back: 
  \begin{align*}
    P = \news{u}(\bout{u}{\namepass{m}} \abinp{\dual m}{b}  \Par
      \binp{\dual u}{\namepass{x}} \about{x}{\textsf{true}} )
  \end{align*}
  After expanding the syntactic sugar of name-passing, we get a process $P = \news{u} (Q \Par R)$, where 
  \begin{align*}
    & Q = \bout{u}{V}
      \binp{\dual m}{y}\news{s}{(\appl{y}{s} \Par
      \about{\dual s}{\abs{b}{\inact}}
      )} & V = \abs{z}\binp{z}{x}(\appl{x}{m})\\
    & R = \binp{\dual u}{y}\news{s}{(\appl{y}{s} \Par
      \about{\dual s}{W})} & W =\abs{x}{\about{x}{W'}} \text{ with }W' = \abs{z}\binp{z}{x}(\appl{x}{\textsf{true})}
  \end{align*}
  Note that to mimic the name-passing synchronization, we require exactly four reduction steps:
  \begin{align}
  \label{mst:ex:reduction-chain}
  &P \red^4  \map { \abinp{\dual m}{b}  \Par
  \about{m}{\textsf{true}}  } \red^4 \inact
  \end{align}
  We will now investigate the decomposition of $P$ and its reduction chain.
  First, we use \Cref{mst:def:sizeproc} to compute $\plen{Q} = 6$, and similarly, $\plen{R} = 5$.
  Therefore, $\plen{{P}} = 12$.
  Following \Cref{mst:def:decomp}, we see that $\sigma = \subst{m_1\dual{m_1}}{m\dual m}$, which we silently apply.
  Taking $k=1$, the breakdown of $P$ and its subprocesses is shown in \Cref{mst:t:examp}.
  
  \begin{table}[t!]
  \begin{mdframed}
  \vspace{-4mm}
  	  \begin{align*}
     \D{P} & = \news{\prop_1,\ldots, \prop_{12}} \Big( 
       \apropout{1}{}   \Par 
    \news{u_1}\big( \propinp{1}{} \propout{2}{} \apropout{8}{}   \Par 
    \B{2}{\epsilon}{Q} \Par \B{8}{\epsilon}{R} \big) \Big)
    \\
     \B{2}{\epsilon}{Q} & = \propinp{2}{} \bout{u_1}{\V{}{\epsilon}{V}}
    \apropout{3}{}  \Par \propinp{3}{} \binp{\dual{m_1}}{y} \apropout{4}{y}  
    \Par 
    \\ 
    &   \qquad \news{s_1}(\propinp{4}{y}\propout{5}{y}\apropout{6}{}  \Par 
    \propinp{5}{y}(\appl{y}{s_1}) \Par 
    \propinp{6}{}\bout{\dual {s_1}}{\V{}{\epsilon}{\abs{b}{\inact}}}
    \apropout{7}{}  \Par \apropinp{7}{}  )
    \\
    \B{8}{\epsilon}{R} & = 
    \propinp{8}{}\binp{\dual {u_1}}{y}\apropout{9}{y}  \Par  
    \\
    & \qquad \news{s_1}\big( \propinp{9}{y}\propout{10}{y}\apropout{11}{} 
    \Par \propinp{10}{y}(\appl{y}{s_1}) \Par 
    \propinp{11}{}\bout{\dual{s_1}}{\V{}{\epsilon}{W}}\apropout{12}{}  \Par 
    \apropinp{12}{}  \big) 
    \\
    \V{}{\epsilon}{V} & =
     \abs{z_1}{
       \news{\prop^V_1, \prop^V_2}
       (\apropoutv{^V_1}{}  \Par 
    \propinpv{^V_1}{} \binp{z_1}{x} \apropoutv{^V_2}{x}   \Par 
    \propinpv{^V_2}{x} (\appl{x}{m_1}) )}
    \\
     \V{}{\epsilon}{\abs{b}{\inact}} & = 
     \abs{b_1}{
       \news{\prop^b_1}
       (\apropoutv{^b_1}{}  \Par 
      \apropinpv{^b_1}{})} 
      \\
     \V{}{\epsilon}{W} & = 
     \abs{x_1}{
       \news{\prop^W_{1}, \prop^W_{2}}
       ( \apropoutv{^W_1}{}  
        \Par \propinpv{^W_1}{} \bout{x_1}{\V{}{\epsilon}{W'}} 
       \apropoutv{^W_2}{}   \Par \apropinpv{^W_2}{}  )} 
     \\
     \V{}{\epsilon}{W'} &  = \abs{z_1}
     { \news{\prop^{W'}_{1}, \prop^{W'}_{2}}
       (\apropoutv{^{W'}_{1}}{}  
       \Par \propinpv{^{W'}_{1}}{} \binp{z_1}{x} 
       \apropoutv{^{W'}_{2}}{x}  \Par
        \propinpv{^{W'}_{2}}{x} (\appl{x}{\textsf{true}}) )}
  \end{align*}
    \end{mdframed}
  \caption[Example: Decomposition of processes]{The decomposition on processes discussed in \Cref{mst:ex:bdnp}. \label{mst:t:examp}}
  \end{table}
  In \Cref{mst:t:examp} we have omitted substitutions that have no effect and trailing $\inact$s.
  The first interesting action appears after synchronizations on $\prop_1$, $\prop_2$, and $\prop_{8}$.
  At that point, the process will be ready to mimic the first action that is performed by $P $, i.e., $u_1$ will send $\V{}{\epsilon}{V}$, the breakdown of $V$, from the breakdown of $Q$ to the breakdown of $R$.
  Next, $\prop_{9}$ and $\prop_{10}$ will synchronize, and $\V{}{\epsilon}{V}$ is passed further along, until $s_1$ is ready to be applied to it in the breakdown of $R$.
  At this point, we know that ${P} \red^{7} \news{\wtd \prop}P'$, where $\widetilde \prop = (\prop_3,\ldots,\prop_{12})$, and
  % where $\widetilde c = (c_3,\ldots,c_{10},c_{15},\ldots,c_{19})$, and 
  \begin{align*}
   P' = &~
  \apropout{3}{} \Par \propinp{3}{} \binp{ \dual{m_1}}{y} 
  \apropout{4}{y}  
  \\
  & \Par  
  \news{s_1}(\propinp{4}{y}\propout{5}{y}\apropout{6}{} \Par 
  \propinp{5}{y}\appl{y}{s_1} \Par 
  \propinp{6}{}\bout{\dual {s_1}}
  {\V{}{\epsilon}{\abs{b}{\inact}}}\apropout{7}{} 
   \Par \apropinp{7}{})\\
  & \Par
  \news{s_1}\big(  \appl{\V{}{\epsilon}{V}}{s_1} \Par 
  \bout{\dual{s_1}}{\V{}{\epsilon}{W}}\apropout{12}{} \Par 
  \apropinp{12}{} \big) 
  \end{align*} 
  After $s_1$ is applied, the trio guarded by $c_3$ will be activated, where
  $z_1$ has been substituted by $s_1$. Then $\dual{s_1}$ and $s_1$ will
  synchronize, and the breakdown of $W$ is passed along. Then $c_4$ and $c_{19}$
  synchronize, and now $m_1$ is ready to be applied to $\V{}{\epsilon}{W}$,
  which was the input for $c_4$ in the breakdown of $V$. After this application,
  $\prop_3$ and $\prop{^W_1}$ can synchronize with their duals, and we know that
  $\news{\widetilde c}P' \red^8 \news{\widetilde{c}'}P''$, where
  % $\widetilde{c}'=(c_6,\ldots,c_{10},c_{16},c_{17},c_{18})$, and
  $\widetilde{c}'=(\prop_4,\ldots,\prop_{7},\prop{^{W'}_2})$, and
  \begin{align*}
  P'' = & ~\binp{\dual{m_1}}{y} \apropout{4}{y}   
  \Par
   \bout{m_1}{\V{}{\epsilon}{W'}} \apropoutv{^{W'}_2}{}  \Par 
   \apropinpv{^{W'}_2}{}
  \\
  & \Par \news{s_1}(\propinp{4}{y}\propout{5}{y}\apropout{6}{} \Par 
  \propinp{5}{y}\appl{y}{s_1} \Par 
  \propinp{6}{}\bout{\dual {s_1}}{\V{}{\epsilon}{\abs{b}{\inact}}}
  \apropout{7}{}  
  \Par \apropinp{7}{} )
  \end{align*}
  Remarkably, $P''$ is standing by to mimic the encoded exchange of value \textsf{true}. Indeed, the decomposition of the four-step reduced process in \eqref{mst:ex:reduction-chain} will reduce in three steps to a process that is equal (up to $\scong_\alpha$) to the process we obtained here. This strongly suggests  a tight operational correspondence between a process and its decomposition, which we will explore in \Cref{ss:dynamic}. \hspace*{\fill} $\lhd$
  \end{example}

% examples with the recursion
% \input{recursion.tex}

% We now explain how to decompose processes whose names are typed with recursive
% types.

\subsubsection{Decomposing Processes with Recursive Names}
Next, we illustrate the decomposition of processes involving names with
tail-recursive types. Recall process $R_1$, which  we used in \Cref{mst:ss:key} to
motivate the need for recursive propagators: 
\begin{align*} 
  R_1={\binp{r}{z}}_{}
  {\bout{r}{-z}}_{}
  {\binp{r}{z}}_{}
  {\bout{r}{z}}_{}
  {\appl{V}{r}}_{}
\end{align*} 

Two following examples illustrate the low-level workings of the propagation
mechanism of the decomposition 
in~\Cref{mst:fig:rec-subfig2}. 
The first example illustrates how the 
propagation of recursive names works in the case of  
input and output actions on names with recursive types (the ``first part'' of
$R_1$). The second example shows how an application where a value is applied to
a tuple of names with recursive names is broken down (the ``second part'' of
$R_1$). 
% The second example
% shows how to break down an application where a value is applied to a tuple of
% names with recursive names (the ``second part'' of $R$). 

\begin{example}[Decomposing Processes with Recursive Names (I)]
\label{mst:ex:recdec}    
Let $P = \binp{r}{x}\bout{r}{x}P'$ be a process where $r$ has type
$S=\trec{t}{\btinp{\mathsf{int}}\btout{\mathsf{\mathsf{int}}}}{\vart{t}}$ and $r
\in \fn{P'}$. 
By \Cref{mst:def:decomp} we have: 
\begin{align*}
	\D{P} = \news{\wtd \prop}\news{\prop^r} 
	\big( 
\apropout{1}{} \Par  \B{1}{\epsilon}{P} \Par 
\recprov{r}{x}{(r_1,r_2)}
% \proprinp{r}{b}{\appl{b}{(r_1,r_2)}} 
\big)
\end{align*}
\noindent where $\wtd \prop = (\prop_1, \ldots, \prop_{\len{P}})$. 
 The control trio in the parallel composition provides a
decomposition of $r$ on name $\prop^r$, which is \emph{shared}.
The decomposition $\B{1}{\epsilon}{P}$ is 
defined as follows: 
\begin{align*}
	\B{1}{\epsilon}{P} & = 
\apropout{1}{} \Par  \propinp{1}{} \aproprout{r}{N_1}  \Par 
	\propinp{2}{y} \aproprout{r}{N_2} 
	\Par \B{3}{\epsilon}{P'}
\big)
\\
  N_1 & = \abs{(z_1, z_2)}\binp{z_1}{x}
  \recprov{r}{x}{(z_1, z_2)}
  % \propout{2}{x}\proprinp{r}{b}{
  % \appl{b}{(z_1,z_2)}}	 
  \\
  N_2 & = \abs{(z_1, z_2)}\bout{z_2}{x}\propout{3}{}
  \recprov{r}{x}{(z_1,z_2)}
  % \proprinp{r}{b}{\appl{b}{(z_1,z_2)}}	
\end{align*}

\noindent Each trio in $\B{1}{\epsilon}{P}$ that mimics some action on $r$
requests the sequence $\tilde r$ from the server on $\prop^r$. We can see that
this request is realized by a higher-order communication: trios send
abstractions ($N_1$ and $N_2$) to the server; these abstractions contain
further actions of trios and it will be applied to the sequence $\tilde r$.
Hence, the formal arguments for these values are meant to correspond to 
$\tilde r$.

After two reductions (the trio activation on $\prop_1$ 
and the communication on $\prop^r$), we have: 
\begin{align*}
  \D{P} \red^2 
  \news{\prop_2, \ldots, \prop_{\plen{P}}}\, 
  \binp{r_1}{x}\propout{2}{x}
  \recprov{r}{x}{(r_1,r_2)}
  % \proprinp{r}{b}{
  % \appl{b}{(r_1,r_2)}}	 
  \Par 
  \propinp{2}{y} \aproprout{r}{N_2}   
	\Par \B{3}{\epsilon}{P'} = P_1 
\end{align*}
% Hence, the formal arguments for these values are meant to correspond to 
% $\tilde r$. 
By synchronizing with the top-level server on $\prop^r$, the bound names
in $N_1$ are instantiated with $r_1,r_2$. 
Now, the first trio in $P_1$ is able
to mimic the action on $r_1$ that is followed by the activation of the next
trio on $\prop_2$. Then, the server on $\prop^r$ gets reinstantiated making names $r_1,r_2$
available for future trios. 
The break down of the output action follows the same pattern. 
% for the
%     benefit of future trios mimicking actions on $r$. 
    \hspace*{\fill} $\lhd$ 

% The server on name $\prop^r$ will appropriately instantiate these names.
		% Notice that all names in $\tilde r$ are propagated, even if a trio in
    % abstraction only use some of them. For instance, $N_1$ only uses $r_1$,
    % whereas $N_2$ uses $r_2$. After simulating an action on $r_i$ and activating
    % the next trio, these values reinstate the server on $\prop^r$ for the
    % benefit of future trios mimicking actions on $r$. \hspace*{\fill} $\lhd$

% The server on name $\prop^r$ will appropriately instantiate these names.
% 		Notice that all names in $\tilde r$ are propagated, even if the abstractions
% 		only use some of them. For instance, $N_1$ only uses $r_1$, whereas $N_2$
% 		uses $r_2$. After simulating an action on $r_i$ and activating the next
% 		trio, these values reinstate the server on $\prop^r$ for the benefit of
% 		future trios mimicking actions on $r$. 
% 		\hspace*{\fill} $\lhd$

% \todo[inline]{Double-check updated example}
\begin{revisionrem} 
To define $\B{1}{\epsilon}{P}$ in a compositional way, names
$(r_1,r_2)$ should be provided to its first trio; they cannot be known
beforehand. To this end, we introduce a new  control trio that will hold these
names:
\begin{align*} 
    \proprinp{r}{b}{\appl{b}{(r_1, r_2)}}
\end{align*}
where the \emph{shared} name $\prop^r$ provides a decomposition of the
(recursive) name $r$. The intention is that each name with a recursive type $r$
will get its own dedicated propagator channel $\prop^r$. Since there is only one
recursive name in $P$, its decomposition will be of the following form:
\begin{align*}
	\D{P} = \news{\tilde \prop}\news{\prop^r} 
	\big( \proprinp{r}{b}{\appl{b}{(r_1,r_2)}} \Par
\propout{1}{}\inact \Par  \B{1}{\epsilon}{P} \big)
\end{align*}
The new control trio can be seen as a server that provides names: each trio that
mimics some action on $r$ should request the sequence $\tilde r$ from the server
on $\prop^r$. This request will be realized by a higher-order communication:
trios should send an abstraction to the server; such an abstraction will contain
further actions of a trio and it will be applied to the sequence $\tilde r$.
Following this idea, we may refine the definition of $\D{P}$ by expanding
$\B{1}{\epsilon}{P}$:
\begin{align*}
	\D{P} = \news{\tilde \prop}\news{\prop^r} 
	\big( \proprinp{r}{b}{\appl{b}{(r_1,r_2)}} \Par
\propout{1}{}\inact \Par  \propinp{1}{} \proprout{r}{N_1} \inact \Par 
	\propinp{2}{y} \proprout{r}{N_2} \inact  
	\Par \B{3}{\epsilon}{P'}
\big)
\end{align*}
The trios involving names with recursive types have now a different shape. After
being triggered by a previous trio, rather than immediately mimicking an action,
they will send an abstraction to the server available on $\prop^r$. The
abstractions $N_1$ and $N_2$ are defined as follows:
\begin{align*}
N_1 = \abs{(z_1, z_2)}\binp{z_1}{x}\propout{2}{x}\proprinp{r}{b}{
\appl{b}{(z_1,z_2)}}	 
\quad
N_2 = \abs{(z_1, z_2)}\bout{z_2}{x}\propout{3}{}
\proprinp{r}{b}{\appl{b}{(z_1,z_2)}}	
\end{align*}
Hence, the formal arguments for these values are meant to correspond to $\tilde
		r$. The server on name $\prop^r$ will appropriately instantiate these names.
		Notice that all names in $\tilde r$ are propagated, even if the abstractions
		only use some of them. For instance, $N_1$ only uses $r_1$, whereas $N_2$
		uses $r_2$. After simulating an action on $r_i$ and activating the next
		trio, these values reinstate the server on $\prop^r$ for the benefit of
		future trios mimicking actions on $r$. 
		\hspace*{\fill} $\lhd$
\end{revisionrem} 
\end{example}

% CR begin
\begin{example}%[Decomposition of Application to Recursive Names]
  [Decomposing Processes with Recursive Names (II)]
  \label{mst:ex:rec2}
  % \todo[inline]{Update this example!}
  %Now we will consider the decomposition of an application with 
  %multiple arguments involving recursive names.  
  %This example is particularly relevant since the encoding 
  %of the recursion construct of \HOp relies 
  %on application to tuples. 
  Let 
  $S = \trec{\vart{t}}\btinp{\mathsf{int}}\btout{\mathsf{int}} \vart{t}$
  and
  $T = \trec{\vart{t}}\btinp{\mathsf{bool}}\btout{\mathsf{bool}} \vart{t}$, and
  define $Q = \appl{V}{(u,v)}$ as a process where $u : S$ and $v : T$, where $V$
  is some value of type $\shot{(S, T)}$. 
  By \Cref{mst:def:decomp}, 
  the  decomposition of $Q$ is as in the previous example, except that now there
   are two servers, one for $u$ and one for $v$:
  \begin{align*}
    \D{Q} & = \news{\prop_1 \tilde \prop}\news{\prop^u \prop^v} \big( 
      \recprov{u}{x}{(u_1, u_2)} \Par 
    % \binp{\prop^u}{b}\appl{b}{(u_1,u_2)} \Par
    \recprov{v}{x}{(v_1, v_2)} \Par  
    % \binp{\prop^v}{b}\appl{b}{(v_1,v_2)} \Par 
    \apropout{1}{} \Par 
    \B{1}{\epsilon}{Q}
    \big)	
\\
    \B{1}{\epsilon}{Q} & = \propinp{1}{}\abbout{\prop^u}{\abs{(x_1,x_2)}{
    \aproprout{v}{\abs{(y_1,y_2)}{\appl{\V{}{\epsilon}{V}}{(x_1,x_2,y_1,y_2)}}}
	}}
  \end{align*}
  
  \noindent with $\tilde \prop = (\prop_2,\ldots,\prop_{\plen{Q}})$. Process $Q$
  is broken down in such a way that it communicates with both servers to collect
    $\tilde u$ and $\tilde v$.  
  To  this end, $\B{1}{\epsilon}{Q}$ is a process in which abstractions are
  nested using output prefixes and whose innermost process is an application.
  After successive communications with multiple servers this innermost
  application will have collected all names in $\tilde u$ and $\tilde v$. 
  
  % We apply this idea
  % to breakdown $Q$:
  % \begin{align*}
  %   \B{1}{\epsilon}{Q} = \propinp{1}{}\bbout{\prop^u}{\abs{(x_1,x_2)}{
  %   \proprout{v}{\abs{(y_1,y_2)}{\appl{\V{}{\epsilon}{V}}{(x_1,x_2,y_1,y_2)}}}
	% \inact}}\inact
  % \end{align*}

  Observe that we use two nested outputs, one for each name with recursive types
  in $Q$. 
  %Putting it all together we obtain the decomposition of $P$ as follows: 
  %\begin{align*}
  %	\D{P} = \news{\tilde \prop}\news{(\prop^u, \prop^v)} \big( 
  %	\binp{\prop^u}{b}\appl{b}{(u_1,u_2)} \Par 
  %	\binp{\prop^v}{b}\appl{b}{(v_1,v_2)} \Par 
  %	\propout{1}{}\inact \Par 
  %	\B{1}{\epsilon}{P}
  %	\big)	
  %\end{align*}
  We now look at the reductions of $\D{Q}$ to analyze how the communication of nested abstractions allows us to collect all name sequences needed. After the first reduction along $\prop_1$ we have: 
  \begin{align*}
    \D{Q} \red & \news{\tilde \prop}\news{\prop^u \prop^v} \big( 
      \recprov{u}{x}{(u_1, u_2)} \Par 
    % \binp{\prop^u}{b}\appl{b}{(u_1,u_2)} \Par 
    \recprov{v}{x}{(v_1, v_2)} \Par 
    % \binp{\prop^v}{b}\appl{b}{(v_1,v_2)} \Par  
    \\
    & \qquad \qquad \qquad  \abbout{\prop^u}{\abs{(x_1,x_2)}{
    \aproprout{v}{\abs{(y_1,y_2)}{\appl{\V{}{\epsilon}{V}}{(x_1,x_2,y_1,y_2)}}}}} \big)
   = R^1
  \end{align*}
  From $R^1$ we have a synchronization along  name $\prop^u$: 
  \begin{align*}
    R^1 \red & \news{\tilde \prop}\news{\prop^u \prop^v} \big( 
    \appl{(\abs{(x_1,x_2)}{
    \aproprout{v}{\abs{(y_1,y_2)}{\appl{\V{}{\epsilon}{V}}{(x_1,x_2,y_1,y_2)}}}})}{(u_1,u_2)} \Par 
    \recprov{v}{x}{(v_1, v_2)}
    %  \binp{\prop^v}{b}\appl{b}{(v_1,v_2)} 
     \big)
   = R^2
  \end{align*}
  Upon receiving the value, the server applies it to $(u_1,u_2)$, thus 
  obtaining the following process: 
  \begin{align*}
    R^2 \red & \news{\tilde \prop}\news{\prop^u \prop^v} \big( 
      \aproprout{v}{\abs{(y_1,y_2)}{\appl{\V{}{\epsilon}{V}}{(u_1,u_2},y_1,y_2)}} \Par 
      \recprov{v}{x}{(v_1, v_2)}
    % \binp{\prop^v}{b}\appl{b}{(v_1,v_2)}
    \big)  = R^3
  \end{align*}
  Up to here, we have partially  instantiated name variables of a value 
  with the sequence $\tilde u$. Next, the first trio in $R^3$ can communicate with the server on name $\prop^v$:  
  \begin{align*}
    R^3 \red & \news{\tilde \prop}\news{\prop^u \prop^v} \big( \appl{\abs{(y_1,y_2)}{\appl{\V{2}{\epsilon}{V}}{(u_1,u_2,y_1,y_2)}}}{(v_1,v_2)}\big)  
    \\
    \red & \news{\tilde \prop}\news{\prop^u \prop^v} 
	\big( {\appl{\V{}{\epsilon}{V}}{(u_1,u_2,v_1,v_2)}}\big) 
  \end{align*}
  This completes the instantiation of name variables with appropriate sequences of names with recursive types.
  At this point, $\D{Q}$ can proceed to mimic the application in $Q$. 

  % \todo[inline]{Double-check updated example.}
  \hfill $\lhd$
  \end{example}

  \begin{example}[Breakdown of Recursion Encoding]
  We recall process $\map{P}$ from~\Cref{mst:e:en-rec}:
  \begin{align*}
    \map{P} &= \binp{a}{m}\bout{a}{m} \news{s}
    {(\appl{V}{(a,s)} \Par \about{\dual s}{V})}  \\ 
  V & = \abs{(x_a,y_1)}\binp{y_1}{z_x}\binp{x_a}{m}
    \bout{x_a}{m}\news{s}
    {(\appl{z_x}{(x_a,s)} \Par \bout{\dual s}{z_x}\inact)}
    \end{align*}
    \noindent Here, bound name $s$ is typed with $S$, from~\Cref{mst:ex:ntrsts}, defined as: 
    \begin{align*}
      S = \trec{t}{\btinp{\shot{(\btinp{\mathsf{str}}\btout{\mathsf{str}}\tinact,
  \vart{t})}}\tinact} 
    \end{align*} 
  % \noindent 
    We now analyze $\D{\map{P}}$ and its reduction chain. 
    By \Cref{mst:def:sizeproc}, we have $\plen{\map P} = 7$.
    Then, we choose $k=1$ and observe that $\sigma =
    \subst{a_1\dual{a_1}}{a\dual a}$. Following \Cref{mst:def:decomp}, we get:
    \begin{align*}
    &\D{\map P} = \news{\prop_1,\ldots,\prop_7}
    \news{\prop^a}(
      \recprov{a}{x}{(a_1, a_2)}
      % \binp{\prop^a}{b}\appl{b}{(a_1,a_2)} 
      \Par 
    \apropout{1}{}  \Par \B{1}{\epsilon}{\map P\sigma})\\
    &\B{1}{\epsilon}{\map P} =
    \propinp{1}{} \about{\prop^a}{\abs{(z_1,z_2)}\binp{z_1}
      {m}\propout{2}{m}
      \recprov{a}{x}{(z_1, z_2)}
      % \binp{\prop^a}{b}\appl{b}{(z_1,z_2)}
      }
    \\
    &\qquad \qquad 
    \Par
    \propinp{2}{m}
    \about{\prop^a}{\abs{(z_1,z_2)}\bout{z_2}
      {m}\propout{3}{}
      \recprov{a}{x}{(z_1, z_2)}
      % \binp{\prop^a}{b}\appl{b}{(z_1,z_2)}
      }
     \\
    &\qquad \qquad \Par \news{s_1}\big(\propinp{3}{} \propout{4} {} 
      \apropout{5}{} 
      \Par
      \propinp{4}{}\about{\dual{\prop^a}}
      {\abs{(z_1,z_2)}\appl{\V{5}{\epsilon}{V}}{(z_1,z_2,s_1)}}
    \\
    &\qquad \qquad \qquad \quad 
    \Par \propinp{5}{}\bout{\dual s_1}{\V{6}{\epsilon}{V}}\apropout{7}{} 
      \Par \apropinp{7}{}\big)
    \end{align*}

    In accordance with~\Cref{mst:ex:ntrsts}, the type of $s_1$ in the decomposed process is
    \begin{align*}
      M = \trec{t}{\btinpt{{{{\shot{(\btinpt{\mathsf{str}}, 
  \btoutt{\mathsf{str}},\vart{t})}}}}}}.
    \end{align*}
    
    The decomposition relies twice on $\V{}{\epsilon}{V}$, the breakdown of value $V$, which we give below.
    For this, we observe that $V$
    is an abstraction of a process $Q$ with $\len Q = 7$. 
    We also $\alpha$-convert the process abstracted in $\V{}{\epsilon}{V}$ 
    renaming bound propagators $\prop_1,\ldots,\prop_7$ to 
    $\prop^V_1,\ldots,\prop^V_7$ to avoid name clashes. 
    \begin{align*}
    \V{}{\epsilon}{V} & = \abs{(x_{a_1},x_{a_2},y_1)}
    {	\news{\prop{^V_1},\ldots,\prop{^V_7}} \big (
       \apropoutv{^V_1}{}  \Par \B{1}{\epsilon}{Q}
      \subst{\prop{^V_1},\ldots,\prop{^V_7}}
      {\prop{_1},\ldots,\prop{_7}}} \Par 
      \recprovx{x_{a}}{x}{(x_{a_1},x_{a_2})} 
      % \binp{\prop^{x_{a}}}{b}\appl{b}{(x_{a_1},x_{a_2} )} 
      \big) 
      \\
    \B{1}{\epsilon}{Q} & = \propinpv{_1}{}\binp{y_1}{z_x}
    \apropout{2}{z_x} 
    \\
    & \quad
    \Par
    \propinp{2}{z_x}\about{\prop^a}
    {\abs{(z_1,z_2)}{\binp{z_1}{m}\propout{3}{z_x,m}
      \recprov{a}{x}{(z_1, z_2)}
        % \binp{\prop^a_2}{b}
        % \appl{b}{(z_1,z_2)}}
        }}
        \\
    & \quad \Par
     \propinp{3}{z_x}\about{\prop^a}
    {\abs{(z_1,z_2)}{\bout{z_2}{m}} \propout{4}{z_x}
    \recprov{a}{x}{(z_1, z_2)}
    % \binp{\prop^a}  {b}\appl{b}{(z_1,z_2)}
    }   
      \\
    & \quad \Par \news{s_1}\big(\propinp{4}{x_z} \propout{5}{z_x}
      \apropout{6}{z_x} \Par  
      \\
    &\quad \qquad  \qquad \propinp{5}{z_x}
    \about{\prop^a}{\abs{(z_1,z_2)}\appl{z_x}{(z_1,z_2,s_1)}}
     \Par \propinp{6}{z_x} \bout{\dual{s_1}}{z_x}
    \apropout{7}{} \Par \apropinp{7}{}\big)
    \end{align*}
    We follow the reduction chain on $\D{\map P}$ until it is ready to mimic the first action with channel $a$, which is an input. First, $\prop_1$ will synchronize, after which $\prop^a$ sends the abstraction to which then $(a_1,a_2)$ is applied. We obtain $\D{\map P} \red^3 \news{c_2,\ldots, c_7,c^a}P'$, where
    \begin{align*}
      & P' = \binp{a_1}
      {m}\propout{2}{m}
      \recprov{a}{x}{(a_1, a_2)}
      % \binp{\prop^a}{b}\appl{b}{(a_1,a_2)}
      \\
      &\qquad \qquad \Par
      \propinp{2}{m}
      \about{\prop^a}{\abs{(z_1,z_2)}\bout{z_2}
        {m}\propout{3}{}
        \recprov{a}{x}{(z_1, z_2)}
        % \binp{\prop^a}{b}\appl{b}{(z_1,z_2)}
        }
      \\
      &\qquad \qquad \Par \news{s_1}\big(\propinp{3}{} \propout{4} {} 
      \propout{5}{} 
      \Par
      \propinp{4}{}\about{\dual{\prop^a}}
      {\abs{(z_1,z_2)}\appl{\V{}{\epsilon}{V}}{(z_1,z_2,s_1)}} \Par
      \\
      &\qquad \qquad \qquad \qquad  \propinp{5}{}\bout{\dual s_1}
      {\V{}{\epsilon}{V}}\apropout{7}{} 
      \Par \apropinp{7}{}\big)
    \end{align*} 
  \noindent  Note that this process is awaiting an input on channel $a_1$, after which
  $c_2$ can synchronize with its dual. At that point, $c^a$ is ready to receive
  another abstraction that mimics an input on $a_1$. This strongly suggests a
  tight operational correspondence between a process $P$ and its decomposition
  in the case where $P$ performs higher-order recursion.
  $\hfill \lhd$
  \end{example}

\subsection{Static Correctness}
\label{mst:ss:staticcorr}
Having presented and illustrated our decomposition, we may now state its technical results.
Given an environment $\Delta = \Delta_1, \Delta_2$, below we write $\Delta_1 \circ \Delta_2$ to indicate the split of $\Delta$ into a $\Delta_1$ containing non-recursive names and a $\Delta_2$ containing recursive names. 

We extend the decomposition function $\Gt{-}$ to typing
environments in the obvious way.
We rely on the following notation.
Given a tuple of names $\widetilde{s} = s_1, \ldots, s_n$ and a tuple of
(session) types $\widetilde{S} = S_1, \ldots, S_n$ of the same length, we write
$\widetilde{s}:\widetilde{S}$ to denote a list of typing assignments $s_1:S_1,
\ldots, s_n:S_n$.

\begin{definition}[Decomposition of Environments] \label{def:typesdecompenv}
  Let $\Gamma$, $\Lambda$, and $\Delta$ be typing environments.
  We define $\Gt{\Gamma}$, $\Gt{\Lambda}$, and $\Gt{\Delta}$ inductively as follows:
  \begin{align*}
        \Gt{\Delta \cat u_i:S} &= \Gt{\Delta},(u_i,\ldots,u_{i+\len{\Gt{S}}-1}) : \Gt{S} \\
      \Gt{\Gamma \cat u_i:\chtype{U}} &= \Gt{\Gamma} \cat u_i : \Gt{\chtype{U}} \\
    \Gt{\Gamma \cat x:U} &= \Gt{\Gamma} \cat x:\Gt{U} \\
    \Gt{\Lambda \cat x:U} &= \Gt{\Lambda} \cat x:\Gt{U} \\
\Gt{\emptyset} &=  \emptyset
  \end{align*}
\end{definition}

\begin{restatable}[]{lemm}{typerec}
  \label{mst:t:typecore}
  Let $P$ be an indexed \HO process and $V$ be a value.
  \begin{enumerate}
  \item
  If $\Gamma;\Lambda;\Delta \circ  \envR \proves P \hastype \Proc$ then
  $\Gt{\Gamma_1}\cat \envPropR;\es;\Gt{\Delta} \cat \Theta \proves \B{k}{\tilde x}{P} \hastype \Proc$, 
  where: 
  \begin{itemize}
  \item $k > 0$
  \item $\widetilde r = \text{dom}(\envR)$
  \item $\envPropR = \prod_{r \in \tilde r} c^r:\chtype{\lhot{\Rts{}{s}{\envR(r )}}}$
  \item $\widetilde x = \fv{P}$ 
  \item $\Gamma_1=\Gamma \setminus \widetilde x$ 
  \item $\text{dom}(\Theta) =
  \{\prop_k,\ldots,\prop_{k+\plen{P}-1}\} \cup
  \{\dual{\prop_{k+1}},\ldots,\dual{\prop_{k+\plen{P}-1}}\}$
  \item $\Theta(\prop_k)=
  \btinpt{U_1, \ldots, U_n} $, where $(\Gt{\Gamma}\cat\Gt{\Lambda})(\widetilde x) = (x_1 : U_1, \ldots, x_n : U_n)$
  \item $\balan{\Theta}$
  \end{itemize}

  \item If $\Gamma;\Lambda;\Delta \circ  \envR \proves V \hastype 
  \lhot{\wtd T}$ then
  $\Gt{\Gamma} \cat \envPropR; \Gt{\Lambda};\Gt{\Delta} \proves 
  \V{k}{\tilde x}{V} \hastype \lhot{\Gt{\wtd T}}$, 
  where: 
  \begin{itemize}
  	\item 
$\wtd x = \fv{V}$ 
\item 
  $\envPropR = \prod_{r \in \tilde r} c^r:\chtype{\lhot{\Rts{}{s}{\envR(r)}}}$
    \end{itemize}

\end{enumerate}
\end{restatable} 

  \begin{proof}
    By mutual induction on the structure of $P$ and $V$. 
    \appendx{See~\Cref{mst:app:typecore} for details.}
    \end{proof}

Using the above lemma we can prove our static correctness result, which explains how our decomposition induces minimal session types.

\begin{restatable}[Static Correctness]{thm}{decompcore}
  \label{mst:t:decompcore}
	Let $P$ be a closed \HO process (i.e. $\fv{P} = \emptyset$) with $\widetilde u = \fn{P}$.
  % and $\widetilde v
	% = \rfn{P}$. 
	% \\
	If $\Gamma;\es;\Delta \circ  \envR \proves P \hastype \Proc$, 
	%where $\envR$ only involves recursive session types,  
	then $\Gt{\Gamma
	\sigma};\es;\Gt{\Delta\sigma}  \cat \Gt{\envR \sigma} \proves \D{P} \hastype
	\Proc$, where $\sigma = \subst{\mathsf{init}(\widetilde u)}{\widetilde u}$.
\end{restatable}

\begin{proof}
Directly from the definitions, using \Cref{mst:t:typecore}.
\appendx{See~\Cref{mst:app:decompcore} for details.}
\end{proof}

\section{Dynamic Correctness}
\label{ss:dynamic}

%%%%%%%%%%%%%%%%%%%%%%%%%%%%%%%%%%%%%%%%%%%%%%%%%%%% 
%%%% OPERATIONAL CORRESPONDENCE
%%%%%%%%%%%%%%%%%%%%%%%%%%%%%%%%%%%%%%%%%%%%%%%%%%%%
% \subsection{Dynamic Correctness}

% \subsection{Dynamic Correctness: Overview}
% \subsubsection{Dynamic Correctness: Overview}

\newcommand{\tni}{t_1}

In this section, we establish the dynamic correctness of our decomposition, stated in terms of a typed behavioral equivalence.
More specifically, we would like to show that any typed process $P$ is equivalent to its decomposition $\D{P}$.
But how do we even state it formally?
Both $P$ and $\D{P}$ are typed $\HO$ processes (as any minimally typed process is also an $\HO$ process), so we can consider compare them as $\HO$ terms inside the $\HO$ type system.
The conventional notion of typed equivalence for $\HO$ processes is contextual equivalence, which is given a local characterization in terms of \emph{higher-order bisimulations}~\cite{KouzapasPY17}.
In our case, however, contextual equivalence is not the right choice: contextual equivalence applies to processes of the \emph{same type}, whereas the process $P$ and its decomposition $\D{P}$ have different types and typing contexts.
Instead of using contextual equivalence, we generalize the notion of higher-order bisimilarity to a notion that we call \emph{MST bisimilarity}, which relates processes of (potentially) different types.

This section is organized as follows.
In \Cref{mst:sec:hobisim} we recall the notion of higher-order bisimulation, used for characterizing behavioral equivalence in \HO, and discuss its limitations for our purposes.
We use higher-order bisimulation as a basis to give a formal definition of MST bisimulation in \Cref{mst:sec:mstbisim}, which we will use as a notion of behavioral equivalence for comparing $P$ and $\D{P}$.
In order to show that our decomposition is correct, in \Cref{mst:sec:relationS} we exhibit a bisimulation relation $\relS$ which relates a process and its decomposition, containing a number of intermediate pairs, working from a motivating example in \Cref{mst:sec:motexample}.
Finally, in \Cref{mst:sec:opcorrproof} we show that $\relS$ is indeed an MST bisimulation.

\subsection{Behavioral Equivalence in \HO and its Limitations}
\label{mst:sec:hobisim}
Let us begin by recalling the notion of \HO bisimulation, defined in~\cite{KouzapasPY17} to characterize contextual equivalence of \HO processes.

\begin{definition}[Definition 17 in \cite{KouzapasPY17}]
  A typed relation $\Re$ is an {\em \HO bisimulation} if 
  for all $\horelm{\Gamma_1;\Lambda_1;\Delta_1}{P_1}{\ \Re \ }
  {\Gamma_2;\Lambda_2;\Delta_2}{Q_1}$, 
  \begin{enumerate}[1)]
  \item 
    Whenever 
    $\horelm{\Gamma_1;\Lambda_1;\Delta_1}{P_1}
    {\hby{\news{\widetilde{m_1}} \bactout{n}{V_1}}}
    {\Lambda'_1;\Delta'_1}{P_2}$ 
    then there exist
    $Q_2$, $\Delta'_2$, and $\Lambda'_2$ such that 
    $\horelm{\Gamma_2;\Lambda_2;\Delta_2}{Q_1}
    {\Hby{\news{\widetilde{m_2}}\bactout{n}{V_2}}}
    {\Lambda'_2;\Delta_2'}{Q_2}$
    where, for a fresh $t$,
    \[
      \Gamma_1; \Lambda_1; \Delta''_1 \proves 
      {\newsp{\widetilde{m_1}}{P_2 \Par \htrigger{t}{V_1}}}
      \ \Re\ 
      \Gamma_2; \Lambda_2; \Delta''_2\proves 
      {\newsp{\widetilde{m_2}}{Q_2 \Par \htrigger{t}{V_2}}}
    \]

  \item	
    Whenever  
    $\horelm{\Gamma_1;\Lambda_1;\Delta_1}{P_1}{\hby{\ell}}
    {\Lambda_1'; \Delta_1'}{P_2}$, with  
    $\ell$ not an output, then there exist 
    $Q_2$, $\Lambda_2'$, and $\Delta'_2$ such that 
    $\horelm{\Gamma_2;\Lambda_2;\Delta_2}{Q_1}
    {\Hby{\hat{\ell}}}{\Lambda_2';\Delta_2'}{Q_2}$
    and
    $\horelm{\Gamma_1;\Lambda_1';\Delta_1'}{P_2}{\ \Re \ }
    {\Gamma_2;\Lambda_2';\Delta_2'}{Q_2}$.
  \item	The symmetric cases of 1, 2.                 
  \end{enumerate}
  The largest such bisimulation is called \emph{\HO bisimilarity}, denoted by $\hob$.
\end{definition}
There are two points worth highlighting in this definition.
Firstly, the labeled transition system $\hby{\ell}$ used in the definition of $\hob$ is what is called the \emph{refined transition system}, different from the standard labeled transition system for the higher-order $\pi$-calculus.
The idea behind the refined transition system is that we want to disallow arbitrary inputs $P \hby{x(V)} P'$; having to consider such transitions in the definition of bisimilarity is undesirable, because it involves input of an arbitrary (higher-order) value $V$, making the definition very much non-local and ensuring that the bisimulations are very large.
As it turns out, due to the typed nature of the system, it suffices to consider inputs of the processes of a very particular kind---\emph{characteristic values}, defined based on the type.

Secondly, because the inputs are restricted in the refined LTS, there is some price to pay in the handling of the outputs.
If an output action ${P_1}\hby{\news{\widetilde{m_1}} \bactout{n}{V_1}}{P_2}$ is matched by an output action ${Q_1}\Hby{\news{\widetilde{m_2}} \bactout{n}{V_2}}{Q_2}$, then we need to ensure that that the output processes $V_1$ and $V_2$ are somehow related.
We have to ensure this in the output clause, because on the receiving end transitions inputing values $V_1$ or $V_2$ might not even be considered.
To that extent, we package the values $V_1$ or $V_2$ in \emph{trigger processes} (denoted $\htrigger{t}{V_1}$ and $\htrigger{t}{V_2}$), which are defined based on the typing.
We then make them part of the processes that are considered at the ``next step'' of the bisimulation.

This notion of \HO bisimilarity works for processes of the same type.
For our case, we need to compare processes of different but related types.
To that extent we make several changes to the definition above.
Firstly, during the decomposition a single name $x$ in a source process is decomposed into a sequence of names $x_1, \dots, x_k$ in the target process.
So in the definition of MST bisimilarity we match an action on a name $x$ with an action on an \emph{indexed} name $x_i$.
Secondly, such discrepancy between names might arise in input and output values.
This also needs to be considered as part of the definition.
For this, we need to accommodate the difference between characteristic values and trigger processes for MST and \HO bisimilarities.
In the next subsection we work out the details sketched above.

\subsection{MST Bisimilarity}
\label{mst:sec:mstbisim}
In this section we define a generalized version of \HO bisimilarity allowing for comparing MST and \HO process terms.
Our goal is to define \emph{MST bisimilarity} (denoted $\mstb$), a typed behavioral equivalence, which we give in \Cref{mst:d:fwb}.
To define $\mstb$, we require some auxiliary definitions, in particular:
\begin{itemize}
\item A refined LTS on typed processes (\Cref{mst:def:mlts});
\item A relation $\valuesrelate$ on values (\Cref{mst:def:valuesrelate}) and on names (\Cref{mst:def:namesrelate});
\item A revised notion of trigger processes  (\Cref{mst:def:trigger}).
\end{itemize}

% \subsubsection{Auxiliary Definitions for MST bisimilarity}
\paragraph{Refined LTS and characteristic values.}
The idea behind defining the refined LTS is to restrict the input of arbitrary processes (values) and make the transition system image-finite (modulo names).

The \emph{refined LTS} for \HO is defined in \cite{KouzapasPY17} in three layers.
First comes the \emph{untyped LTS} $P \by{\ell} P'$, which describes reductions of untyped processes in the usual style of the LTS semantics for $\pi$-calculus.
Secondly, there is a notion  of the \emph{environmental LTS} $(\Gamma_1; \Lambda_1; \Delta_1) \by{\ell} (\Gamma_2; \Lambda_2; \Delta_2)$, which describes reductions of typing environments.
This LTS describes the way a typing context can evolve in accordance with its session types.
On top of these layers there are notions of  \emph{refined environmental LTS} and  \emph{refined LTS for processes}.
The former restricts the environmental LTS to inputs on characteristic values, as we discussed in \Cref{mst:sec:hobisim}.
Finally, the refined LTS for processes restricts the untyped LTS to those actions which are supported by the refined environmental LTS.

We follow this approach for defining the refined LTS for MST processes.
Both the untyped LTS for processes and the environmental LTS for MST processes coincides with the same LTSs for \HO (or, to be more precise, with its restriction to minimal session types).
It remains, then, to define the refined environmental LTS for MST processes, with the idea that the
refined LTS restricts inputs to the inputs on \emph{minimal characteristic values} and \emph{minimal trigger values}.
\begin{definition}[Minimal trigger value]
  \label{mst:d:mtv}
  Given a value type $\slhot{C}$ and fresh (indexed) name $\tni$, the \emph{minimal trigger
  value}  on $\tni$ of type $\slhot{\Gt{C}}$ is defined as the abstraction
  $$\abs{\wtd x}\binp{\tni}{y}\appl{y}{\wtd x}$$ where $\wtd x =
  (x_1,\ldots,x_{\len{\Gt{C}}})$. 
\end{definition}
%
%\todo[inline]{J: This is an indexed version of the definition in ACTA, right?}
%\todo[inline, color=blue!30]{A: Yes}
\begin{definition}[Minimal characteristic values]
  \label{mst:d:mcv}
  Let $u$ be a name and $i > 0$.
  We define  $\mapcharm{-}{u}_i$ and $\omapcharm{-}$ on types as follows.
  \thesisalt{
  \[
    \begin{array}{rclcrcl}
      \mapcharm{\btinp{L} S}{u}_i
      &\defeq&
      \binp{u_i}{x} (\bout{\tni}{\namepass{u_{i+1}, 
      \ldots,u_{i+\len{\Gt{S}}}}}
       \inact \Par 
      \mapcharm{L}{x}_i)
      % \tni 
      &\qquad &
      \omapcharm{S}  & \defeq &  \wtd s ~~ (\len{\wtd s} = 
      \len{\Gt{S}}, \wtd s \textrm{ fresh})
      \\
      \mapcharm{\btout{L} S}{u}_i
      &\defeq&
      \bout{u_i}{\omapcharm{L}} 
      \bout{\tni}{\namepass{u_{i+1}, \ldots,u_{i+\len{\Gt{S}}}}} \inact
      &&
      \omapcharm{\chtype{L}} &\defeq& a_1 ~~ (a_1 \textrm{ fresh})
      \\
      \mapcharm{\tinact}{u}_i
      & \defeq &
      \inact	 
      &&
      \omapcharm{\shot{C}} &\defeq&  \abs{(x_1,\ldots,x_{\len{\Gt{C}}})}{\mapcharm{C}{x}_1}
      \\
      \mapcharm{ \trec{t}{S} }{u}_i &\defeq& \mapcharm{S \subst{\tinact}{\vart{t}} }{u}_i
  
      &&
      \omapcharm{\lhot{C}} & \defeq &  \abs{(x_1,\ldots,x_{\len{\Gt{C}}})}{\mapcharm{C}{x}_1}
      \\
      \mapcharm{\chtype{L}}{u}_i
      &\defeq&
       \bout{u_1}{\omapcharm{L}} \bout{\tni}{\namepass{u_1}} \inact 
      \\
      \mapcharm{\shot{C}}{x}_i
      &\defeq& 
      \appl{x}{\omapcharm{C}}
      \\
      \mapcharm{\lhot{C}}{x}_i
      &\defeq &
      \appl{x}{\omapcharm{C}}
    \end{array}
    \]
    }
    {
      \begin{align*}
        \mapcharm{\btinp{L} S}{u}_i
        &\defeq
        \binp{u_i}{x} (\bout{\tni}{\namepass{u_{i+1}, 
        \ldots,u_{i+\len{\Gt{S}}}}}
         \inact \Par 
        \mapcharm{L}{x}_i) 
        \\
        \mapcharm{\btout{L} S}{u}_i
      &\defeq
      \bout{u_i}{\omapcharm{L}} 
      \bout{\tni}{\namepass{u_{i+1}, \ldots,u_{i+\len{\Gt{S}}}}} \inact
      \\
      \mapcharm{\tinact}{u}_i
      & \defeq 
      \inact	
        \\
        \mapcharm{ \trec{t}{S} }{u}_i &\defeq \mapcharm{S \subst{\tinact}{\vart{t}} }{u}_i
        \\
        \mapcharm{\chtype{L}}{u}_i
        &\defeq
         \bout{u_1}{\omapcharm{L}} \bout{\tni}{\namepass{u_1}} \inact 
         \\
         \mapcharm{\shot{C}}{x}_i
         &\defeq
         \appl{x}{\omapcharm{C}}
         \\
         \mapcharm{\lhot{C}}{x}_i
         &\defeq 
         \appl{x}{\omapcharm{C}}
      % second column
        \\
        \omapcharm{S}  & \defeq  \wtd s ~~ (\len{\wtd s} = 
        \len{\Gt{S}}, \wtd s \textrm{ fresh})
        \\
        \omapcharm{\chtype{L}} &\defeq a_1 ~~ (a_1 \textrm{ fresh})
        \\
        \omapcharm{\shot{C}} &\defeq  \abs{(x_1,\ldots,x_{\len{\Gt{C}}})}{\mapcharm{C}{x}_1}
        \\
        \omapcharm{\lhot{C}} & \defeq  \abs{(x_1,\ldots,x_{\len{\Gt{C}}})}{\mapcharm{C}{x}_1}
      \end{align*}
    %   \[
    % \begin{array}{rclcrcl}

    %   \\
  
    %   \\
    %   \mapcharm{\shot{C}}{x}_i
    %   &\defeq& 
    %   \appl{x}{\omapcharm{C}}
    %   \\
    %   \mapcharm{\lhot{C}}{x}_i
    %   &\defeq &
    %   \appl{x}{\omapcharm{C}}
    % \end{array}
    % \]
    } 
    \noindent where $\tni$ is a fresh (indexed) name.
    % \thesisalt{
    % where $\tni$ is a fresh (indexed) name.}
    % {where $\tni$ is a fresh (indexed) name.}
    % {where $\wtd u^S=$}
    In this definition we use name-passing constructs, as outlined in \Cref{top:ex:np}.
\end{definition}
\begin{definition}[Refined environmental LTS]
  The refined LTS, denoted $\mhby{~\ell~}$, is defined 
  on top of the environmental LTS using the following rules: 
  \begin{mathpar}
		\inferrule[\eltsrule{MTr}] {
			(\Gamma_1; \Lambda_1; \Delta_1) \by{\ell} (\Gamma_2; \Lambda_2; \Delta_2)
			\and
			\ell \not= \bactinp{n}{V}
                      }{
			(\Gamma_1; \Lambda_1; \Delta_1) \mhby{~\ell~} (\Gamma_2; \Lambda_2; \Delta_2)
    }
  \end{mathpar}
  \begin{mathpar}
		\inferrule[\eltsrule{MRcv}]{
			(\Gamma_1; \Lambda_1; \Delta_1) \by{\bactinp{n}{V}} (\Gamma_2; \Lambda_2; \Delta_2)
                        \and      
		  \left(V \scong \omapcharm{L}\right)
			 \vee \left(V  \scong \abs{{\wtd x}}{\binp{\tni}{y} (\appl{y}{{\wtd x}})}\right)
                         \quad\textrm{ {\small with $\tni$ fresh}}
		}{
			(\Gamma_1; \Lambda_1; \Delta_1) \mhby{\bactinp{n}{V}} (\Gamma_2; \Lambda_2; \Delta_2)
		}
  \end{mathpar}
  \noindent where $\abs{{\wtd x}}{\binp{\tni}{y} (\appl{y}{{\wtd x}})}$ is 
  a minimal trigger value of type $\Gt{C}$ (\Cref{mst:d:mtv}).%  
\end{definition}
Finally, the refined LTS for MST processes is just a combination of the untyped LTS with the refined environmental LTS:
\begin{definition}[Refined LTS]
  \label{mst:def:mlts}
  The environmental refined LTS extends to the typed refined LTS on processes.
  We write $\horelm{\Gamma_1;\Lambda_1;\Delta_1}{P_1}{\mhby{~\ell~}}{\Lambda'_1;\Delta'_1}{P_2}$
  when
  \begin{itemize}
  \item $P_1 \by{\ell} P_2$, and
  \item $(\Gamma_1; \Lambda_1; \Delta_1) \mhby{~\ell~} (\Gamma_2; \Lambda_2; \Delta_2)$.
  \end{itemize}
\end{definition}
We write $\Mhby{\ell}$ for the weak version of the transition $\mhby{\ell}$.
Notice that while the untyped LTS and the non-refined environmental LTS coincide
with that of \HO, the refinement that we impose on the environmental LTS is
different from its \HO counterpart. Specifically in Rule~\eltsrule{MRcv} we take
special care to use minimal
characteristic processes $\omapcharm{-}$, instead of general \HO
characteristic process $\omapchar{-}$ 
as defined in~\cite{KouzapasPY17}.

\paragraph{Relating trigger and characteristic values.}
As we mentioned earlier, the notion of bisimulation that we consider requires matching transitions of the source \HO term with the transitions of the target MST term.
However, the two transitions might differ on the inputs of characteristic values.
We accommodate for that difference by establishing a relation between the trigger and characteristic values of \HO and MST.
\begin{definition}
  \label{mst:def:valuesrelate}
  We define the relation $\valuesrelate$ between \HO processes and indexed processes inductively as:
$$
\frac{
\len{\tilde x} = \len{\Gt{C}}
}{
\abs{x:C}{\binp{t}{y}\appl{y}{x}}
\valuesrelate
\abs{\tilde x: \Gt{C}}\binp{\tni}{y}\appl{y}{\tilde x}
}
\qquad
\frac{}{
 \omapchar{\slhot{C}} \valuesrelate \omapcharm{\slhot{C}}{}}
$$
\noindent where $\abs{{\wtd x}:\Gt{C}}{\binp{\tni}{y} (\appl{y}{{\wtd x}})}$ is 
  a minimal trigger value of type $\slhot{\Gt{C}}$  (\Cref{mst:d:mtv}) 
  and  $\omapchar{-}$ denotes the characteristic values
  defined in~\cite{KouzapasPY17}.
  We write $\abs{x:C}{\binp{\tni}{y}\appl{y}{x}}$ to mean that value
 $\abs{x}{\binp{\tni}{y}\appl{y}{x}}$ is of type $ \slhot{C}$. 
\end{definition}

\paragraph{Trigger processes and MST bisimilarity.}
Before we give the definition of MST bisimilarity, we establish the following notations:
\begin{definition}[Indexed name]
  \label{mst:def:indexedname}
  Given a name $n$, we write $\iname{n}$ to either denote 
  $n$ or any indexed name $n_i$, with $i > 0$.
\end{definition}

\begin{definition}[Trigger process]
  \label{mst:def:trigger}
 Given a value $V$, a trigger process for a fresh (indexed) name 
 $\tni$ is defined as: 
 \[
   \htrigger{\tni}{V} \defeq \binp{\tni}{\wtd{x}}(V\ \wtd{x})
 \]
\noindent 
where $\len{\wtd x} = \len{\wtd C}$ for $V : \slhot{\wtd C}$.
\end{definition}
\begin{lemm}
  If $\Gamma; \Lambda; \Delta \proves V \hastype \slhot{\wtd C}$,
  then $\Gamma; \Lambda; \Delta, \tni : \btinpt{\wtd{C}} \proves \htrigger{\tni}{V} \hastype \Proc$.
\end{lemm}

Finally, we are ready to formally define \textsf{MST} bisimilarity. 
\begin{definition}[MST Bisimilarity]
  \label{mst:d:fwb}
	A typed relation $\Re$ is an {\em  MST bisimulation} if 
	for all $\horelm{\Gamma_1;\Lambda_1;\Delta_1}{P_1}{\ \Re \ }
  {\Gamma_2;\Lambda_2;\Delta_2}{Q_1}$, 
	\begin{enumerate}[1)]
		\item 
				Whenever 
				$\horelm{\Gamma_1;\Lambda_1;\Delta_1}{P_1}
        {\hby{\news{\widetilde{m_1}} \bactout{n}{V_1}}}
        {\Lambda'_1;\Delta'_1}{P_2}$ 
        then there exist
				$Q_2$, $\Delta'_2$, and $\Lambda'_2$ such that 
				$\horelm{\Gamma_2;\Lambda_2;\Delta_2}{Q_1}
        {\Mhby{\news{\widetilde{m_2}}\bactout{\iname{n}}{V_2}}}
        {\Lambda'_2;\Delta_2'}{Q_2}$
        where, for a fresh $t$,
				\[
					\Gamma_1; \Lambda_1; \Delta''_1 \proves 
           {\newsp{\widetilde{m_1}}{P_2 \Par \htrigger{t}{V_1}}}
	 				\ \Re\ 
					\Gamma_2; \Lambda_2; \Delta''_2\proves 
          {\newsp{\widetilde{m_2}}{Q_2 \Par \htrigger{\iname t}{V_2}}}
				\]
                                
       \item Whenever $\horelm{\Gamma_1; \Lambda_1; \Delta_1}{P_1}
       {\hby{\abinp{n}{V_1}}}{\Lambda_1';\Delta_1'}{P_2}$
       then there exist  $Q_2$, $\Lambda_2'$, and $\Delta_2'$ such that 
       $\horelm{\Gamma_2;\Lambda_2;\Delta_2}{Q_1}{\Mhby{\abinp{\iname{n}}{V_2}}}
       {\Lambda'_2, \Delta_2'}{Q_2}$ 
       where 
       $V_1 \valuesrelate V_2$
       and 
       $\horelm{\Gamma_1;\Lambda_1';\Delta_1'}{P_2}{\ \Re \ }
       {\Gamma_2;\Lambda_2';\Delta_2'}{Q_2}$, 
      %  $V_2 \in \Cb{}{}{V_1}$, 
        
     \item	The symmetric cases of 1 and 2.
	\end{enumerate}
	The largest such bisimulation is called \emph{MST bisimilarity}, denoted by $\mstb$.
\end{definition}
In all clauses,
we use the refined LTS (\Cref{mst:def:mlts}) and rely on notation
$\iname{n}$ (\Cref{mst:def:indexedname}). In the output clause, we use the
triggers (\Cref{mst:def:trigger}). In the input clause, we use the
relation $\valuesrelate$ on values (\Cref{mst:def:valuesrelate}).

We discuss differences between MST bisimilarity and higher-order bisimilarity as defined in~\cite{KouzapasPY17}.
First, an action in $P_1$ must be matched by an action on an indexed name in $Q_1$, and refined LTS actions in $P_1$ are matched by minimal refined
LTS actions in $Q_1$  (\Cref{mst:def:valuesrelate}). As a consequence of the latter,
in the input case the observed values are not identical but related by $\valuesrelate$ (\Cref{mst:def:valuesrelate}). 
In other words, whenever $P_1$
receives a trigger or a characteristic value, then $Q_1$ should receive their minimal
counterparts (\Cref{mst:d:mtv,mst:d:mcv}). Further, as names could be
indexed on the right-hand side, the typing environments could differ for open processes, so the MST bisimilarity assumes different typing environments on both sides.

\subsection{The Bisimulation Relation} 
\label{mst:sec:relationS}
Our goal is to complement our static correctness result (\Cref{mst:t:decompcore}) by  
proving the following statement about the  decomposition of processes (\Cref{mst:def:decomp}):
\begin{theorem}
\label{mst:t:dyncorr}
  Let $P$ be an \HO process such that $\Gamma;\Delta;\Lambda \proves P \hastype \Proc$. 
  We have 
  $$\horelm{\Gamma;\Lambda;\Delta}{P}{\ \mstb \ }
  {\Gt{\Gamma};\Gt{\Lambda};\Gt{\Delta}}{\D{P}}$$ 
% $P \mstb \Db{\tilde W}{\tilde x}{P}$.
\end{theorem}
%\todo[inline]{J: Please revise and complete the statement above, adding assumptions (typability of $P$?) as appropriate}
To show that $P$ and $\D{P}$ are MST-bisimilar, we provide a concrete bisimulation relation $\relS$  that contains $(P,\D{P})$.
Defining $\relS$ to be just the set of such pairs is, however, not going to work; instead, the relation $\relS$ should also contain pairs corresponding to ``intermediate'' states in which the process and its decomposition may get ``desynchronized''.
Before we give the concrete definition of $\relS$ we look at an example, illustrating the need for such intermediate pairs.

\subsubsection{A Motivating Example}
\label{mst:sec:motexample}

Consider the following process: 
\begin{align*}
  P_1 = \binp{u}{t}\binp{v}{x}
  \news{s : S}(\bout{u}{x}\inact \Par \appl{t}{s}  \Par \bout{\dual s}{x}\inact) 
	\Par \bout{\dual v}{V}\inact
\end{align*}
  \noindent where $u : \btinp{\chtype{U_t}}\btout{U_{V}}\tinact$ and $v:S$ with   
  $S = {\btinpt{U_V}}$, $U_t = \shot{S}$, and 
  $U_V$ is some shared value type, {i.e.} $U_V = \shot{S_V}$, for some 
  session type $S_V$. 
  Further, $V$ is some value, such that $V = \abs{y:S_V}R$. 
 
  Thus, $P_1$ is typed using the typing of its constituents:
  \begin{prooftree}
    \AxiomC{$\es; \es; \dual{v} : \dual{S}
      \fCenter\proves \bout{\dual v}{V}\inact \hastype \Proc$}     
    \noLine
    \UnaryInfC{$\es; \es; u : \btinp{\chtype{U_t}}\btout{U_{V}}\tinact, v:S \fCenter\proves 
     \binp{u}{t}\binp{v}{x}
      \news{s : S}(\bout{u}{x}\inact \Par \appl{t}{s}  \Par \bout{\dual s}{x}\inact)
     \hastype \Proc$}
    \UnaryInfC{$\es; \es; u: \btinp{\chtype{U_t}}\btout{U_V}\tinact, v:S, \dual{v} :
     \dual{S} 
     \fCenter\proves P_1 \hastype \Proc$}
  \end{prooftree}
  \noindent The decomposition of $P_1$ is as follows: 
  \begin{align*}
    \D{P_1} & = \news{\wtd{\prop}} \left( \apropout{1}{} \Par 
    \B{1}{\epsilon}{P_1}\right) \\
    & = \news{\wtd{\prop}} \big( \apropout{1}{} \Par
      \apropinp{1}{}. \apropout{2}{}.\apropout{11}{}  \\
    & \quad\qquad \Par \apropinp{2}{}.\binp{u_1}{t}\apropout{3}{t}
      \Par \apropinp{3}{t}.\binp{v_1}{x}\apropout{4}{t,x} \\
    & \quad\qquad\qquad \Par \news{s_1}
    (\apropinp{4}{t,x}.
    \apropout{5}{x}.\apropout{6}{t, x} \Par 
    \propinp{5}{x}\bout{u_2}{x}\apropout{6}{} \Par \apropinp{6}{} \Par 
    \\ 
    & \qquad\qquad\qquad 
    \Par \propinp{7}{t,x}\propout{8}{t}\apropout{9}{x} \Par 
    \apropinp{8}{t}.\appl{t}{s_1}
      \Par \apropinp{9}{x}.\bout{\dual{s_1}}{x}\apropout{10}{} \Par 
      \apropinp{10}{}) \\
    & \quad\qquad \Par \apropinp{11}{}.\bout{\dual{v_1}}{\V{12}{\epsilon}{V}}
      \apropout{12}{} \Par \apropinp{12}{}
 \big),
  \end{align*}
  where $\wtd \prop = c_1,\ldots,c_{12}$. 
  Let us write $Q_1$ for the decomposition $\D{P_1}$.

% \begin{figure} 
% 	\includegraphics[scale=0.35]{ho-bisimulation3}
% 	\caption{Transitions}
% 	\end{figure}

  % \begin{figure}
  %   \begin{tikzpicture}[shorten >=0.5pt,node distance=1.60cm,on grid,auto] 
  %     \node[state] (p_1) {$P_1$};
  %     \node[state] (p_2) [below=2cm  of p_1] {$P_2$};
  %     \node[state] (p_3) [below=2cm of p_2] {$P_3$};
  %     \node[state](p_4) [below left=2 cm of p_3] {$P_4$}; 
  %     \node[state](p_5) [below right=2 cm of p_3] {$P_5$}; 
  %     \node[state](p_6) [below=2 cm of p_4] {$P_6$}; 
  %     \node[state](p_7) [below=2 cm of p_5] {$P_7$}; 
  %     \node[state] (p_8) [below right=2cm of p_6]{$P_8$};

  %     \path[->] 
  %     (p_1) edge node {test} (p_2)
  %     (p_2) edge node {test} (p_3)

  %     (p_3) edge node {test} (p_4)
  %     (p_3) edge node {test} (p_5)
  %     (p_4) edge node {test} (p_6)
  %     (p_5) edge node {test} (p_7)
  %     (p_7) edge node {test} (p_8)
  %     (p_6) edge node {test} (p_8)
  %     ;
  %   \end{tikzpicture}
  % \end{figure}

  \begin{figure}[!t]
    \begin{mdframed}
      \center 
      \vspace*{0.5cm} 
    \begin{tikzpicture}[shorten >=0.5pt,node distance=1.60cm,on grid,auto] 
      \node[state] (p_1)  {$P_1$};
      \node[state] (p_2) [right=2.5cm  of p_1] {$P_2$};
      \node[state] (p_3) [right=2.5cm of p_2] {$P_3$};
      \node[state, fill=blue!20](p_4) [above right=3.0 cm of p_3] {$P_4$}; 
      \node[state](p_5) [below right=3 cm of p_3] {$P_5$}; 
      \node[state, fill=blue!20](p_6) [right=3.0 cm of p_4] {$P_6$}; 
      \node[state](p_7) [right=3.0 cm of p_5] {$P_7$}; 
      \node[state, fill=blue!20] (p_8) [below right=3.0 cm of p_6]{$P_8$};

      \path[->] 
      (p_1) edge node {\scriptsize{$\abinp{u}{\textcolor{blue}{V_c}}$}} (p_2)
      (p_2) edge node {\scriptsize{$\tau$}} (p_3)

      (p_3) edge node {\scriptsize{$\about{u}{V}$}} (p_4)
      (p_3) edge node {\scriptsize{$\tau$}} (p_5)
      (p_4) edge node {\scriptsize{$\tau$}} (p_6)
      (p_5) edge node {\scriptsize{$\tau$}} (p_7)
      (p_5) edge node {\scriptsize{$\about{u}{V}$}} (p_6)
      (p_7) edge node {\scriptsize{$\about{u}{V}$}} (p_8)
      (p_6) edge node {\scriptsize{$\tau$}} (p_8)
      ;
%    \end{tikzpicture}
%    \vspace*{1.0cm}
%    \\
%    \begin{tikzpicture}[shorten >=0.5pt,node distance=1.60cm,on grid,auto]
    \pgfmathsetmacro{\x}{1.25} 
    \pgfmathsetmacro{\xs}{0.5} 
    \pgfmathsetmacro{\a}{0.2}

      \node[state] (q_1) [below=6cm of p_1, xshift=0cm] {$Q_1$};
      \node[state] (q1_1) [below=\x cm  of q_1, xshift=\xs cm] {$Q'_1$}; 
      \node[state] (q1_2) [below=\x cm  of q1_1, xshift=\xs cm] {$Q''_1$};
      
      \node[state] (q_2) [right=2.5cm  of q_1] {$Q_2$};
      \node[state] (q2_1) [below=\x cm  of q_2, xshift=\xs cm] {$Q'_2$}; 
      \node[state] (q2_2) [below=\x cm  of q2_1, xshift=\xs cm] {$Q''_2$};

      \node[state] (q_3) [right=2.5cm of q_2] {$Q_3$};
       \node[state] (q3_1) [below=\x cm  of q_3, xshift=\xs cm] {$Q'_3$}; 
      \node[state] (q3_2) [below=\x cm  of q3_1, xshift=\xs cm] {$Q''_3$};

      \node[state, fill=blue!20](q_4) [above right=5.5 cm of q3_2, xshift=-2cm] {$Q_4$}; 
        \node[state, fill=blue!20] (q4_1) [below=\x cm  of q_4, xshift=\xs cm] {$Q'_4$}; 
      \node[state, fill=blue!20] (q4_2) [below=\x cm  of q4_1, xshift=\xs cm] {$Q''_4$};

      \node[state](q_5) [below right=5.5 cm of q3_2, xshift=-2cm] {$Q_5$}; 
       \node[state] (q5_1) [above=\x cm  of q_5, xshift=\xs cm] {$Q'_5$}; 
      \node[state] (q5_2) [above=\x cm  of q5_1, xshift=\xs cm] {$Q''_5$};

      \node[state, fill=blue!20](q_6) [right=3 cm of q_4] {$Q_6$}; 
       \node[state, fill=blue!20] (q6_1) [below=\x cm  of q_6, xshift=\xs cm] {$Q'_6$}; 
      \node[state, fill=blue!20] (q6_2) [below=\x cm  of q6_1, xshift=\xs cm] {$Q''_6$};

      \node[state](q_7) [right=3 cm of q_5] {$Q_7$}; 
       \node[state] (q7_1) [above=\x cm  of q_7, xshift=\xs cm] {$Q'_7$}; 
      \node[state] (q7_2) [above=\x cm  of q7_1, xshift=\xs cm] {$Q''_7$};
      
      \node[state, fill=blue!20] (q_8) [right=7.5cm of q3_2]{$Q_8$};

      \path[->] 
      (q1_2) edge [pos=0.9, left] node {\scriptsize{$\abinp{u_1}{\textcolor{blue}{V^m_c}}$}} (q_2)
      (q_1) edge node {\scriptsize{$\tau$}} (q1_1)
      (q1_1) edge node {\scriptsize{$\tau$}} (q1_2)

      (q2_2) edge [pos=0.8, left] node {\scriptsize{$\tau$}} (q_3)
      (q_2) edge node {\scriptsize{$\tau$}} (q2_1)
      (q2_1) edge node {\scriptsize{$\tau$}} (q2_2)

      (q3_2) edge [pos=0.8, left] node {\scriptsize{$\about{u_2}{\V{}{}{V}}$}} (q_4)
      (q_3) edge node {\scriptsize{$\tau$}} (q3_1)
      (q3_1) edge node {\scriptsize{$\tau$}} (q3_2)
      
      (q3_2) edge node {\scriptsize{$\tau$}} (q_5)
      (q4_2) edge node {\scriptsize{$\tau$}} (q_6)
      (q_4) edge node {\scriptsize{$\tau$}} (q4_1)
      (q4_1) edge node {\scriptsize{$\tau$}} (q4_2)

      (q5_2) edge node {\scriptsize{$\tau$}} (q_7)
      (q_5) edge node {\scriptsize{$\tau$}} (q5_1)
      (q5_1) edge node {\scriptsize{$\tau$}} (q5_2)

      (q5_2) edge [pos=0.3, left] node {\scriptsize{$\about{u_2}{\V{}{}{V}}$}} (q_6)
      (q7_2) edge [pos=0.8, left] node {\scriptsize{$\about{u_2}{\V{}{}{V}}$}} (q_8)
      (q_7) edge node {\scriptsize{$\tau$}} (q7_1)
      (q7_1) edge node {\scriptsize{$\tau$}} (q7_2)
      
      (q6_2) edge node {\scriptsize{$\tau$}} (q_8)
      (q_6) edge node {\scriptsize{$\tau$}} (q6_1)
      (q6_1) edge node {\scriptsize{$\tau$}} (q6_2)
      
      % PAIRS 
      (p_1) edge [line width=\a mm, -, dotted] node {} (q_1)
      (p_2) edge [line width=\a mm, -, dotted] node {} (q_2)
      (p_3) edge [line width=\a mm, -, dotted] node {} (q_3)
      (p_4) edge [line width=\a mm, -, dotted] node {} (q_4)
      (p_5) edge [line width=\a mm, -, dotted] node {} (q_5)
      (p_6) edge [line width=\a mm, -, dotted] node {} (q_6)
      (p_7) edge [line width=\a mm, -, dotted] node {} (q_7)
      (p_8) edge [line width=\a mm, -, dotted] node {} (q_8)

      ;
    \end{tikzpicture}
     \vspace*{0.5cm} 
  \end{mdframed}
  \caption[Transitions of $P_1$ and $Q_1 = \D{P_1}$ in \Cref{mst:sec:motexample}.]{Transitions of $P_1$ and $Q_1 = \D{P_1}$ in \Cref{mst:sec:motexample}. 
  The \textcolor{blue}{blue} nodes represent processes that contain characteristic values 
  and trigger processes induced by the bisimilarites defined in~\cite{KouzapasPY17}.}
\label{mst:fig:exampletransitions}  
  \end{figure}
	
  We wish to show $P_1 ~\mstb~ Q_1$. For this, we must exhibit a relation $\mathcal{S}$ included  in $\mstb$ such that $(P_1,  \D{P_1}) \in  \relS$.
  To illustrate the notions required to define the additional pairs, we consider possible transitions of $P_1$ and $Q_1$, denoted schematically in \Cref{mst:fig:exampletransitions}.
  First, let us consider a possible (refined) transition of $P_1$, an input on $u$ of a characteristic value:
\begin{align*}
	P_1 \by{\bactinp{u}{\textcolor{blue}{V_C}}} 
              \binp{v}{x}
  \news{s : S}(\bout{u}{x}\inact \Par \appl{\textcolor{blue}{V_C}}{s}  
  \Par \bout{\dual s}{x}\inact) 
	\Par \bout{\dual v}{V}\inact = P_2 
\end{align*}
	\noindent where $\textcolor{blue}{V_C} = \omapchar{U_t} = 
  \abs{y:S} \binp{y}{x'}(\bout{}{}\inact \Par \appl{x'}{s'})$ is the  \emph{characteristic value} of $U_t$.%
\footnote{We use \textcolor{blue}{blue} to denote characteristic values and trigger processes that do no occur in the original process, but  which are induced by the bisimilarities defined in~\cite{KouzapasPY17}.}
  Process $Q_1$ can weakly match this input action on the indexed name $u_1$.
  This input does not involve $\textcolor{blue}{V_C}$ but the 
  \emph{minimal} characteristic value of type $U_t$ (\Cref{mst:d:mcv}).
We have: 
  \begin{align*}
    Q_1 & \by{\tau} Q'_1 \by{\tau} Q''_1 \by{\bactinp{u_1}{\textcolor{blue}{V^m_C}}}  
      \news{\wtd \prop_{\bullet}}\apropout{3}{\textcolor{blue}{V^m_C}}
      \Par 
      \apropinp{3}{t}.\binp{v_1}{x}\apropout{4}{t,x} \Par \apropout{11}{}\\
      & \quad\qquad\qquad \Par \news{s_1}
      (\apropinp{4}{t,x}.
      \apropout{5}{x}.\apropout{7}{t, x} \Par 
      \propinp{5}{x}\bout{u_2}{x}\apropout{6}{} \Par \apropinp{6}{} \Par 
    \\ 
    & \qquad\qquad\qquad 
    \Par \propinp{7}{t,x}\propout{8}{t}\apropout{9}{x} \Par 
    \apropinp{8}{t}.\appl{t}{s_1}
      \Par \apropinp{9}{x}.\bout{\dual{s_1}}{x}\apropout{10}{} \Par 
      \apropinp{10}{}) \\
    & \quad\qquad\qquad \Par \apropinp{11}{}.\bout{\dual{v_1}}{\V{12}{\epsilon}{V}}
      \apropout{12}{} \Par \apropinp{12}{} = Q_2 
  \end{align*}

	\noindent where  $\textcolor{blue}{V^m_C} = \omapcharm{U_t} = 
	\abs{(y_1)}\binp{y_1}{x'}(\bout{\tni}{}\inact \Par \appl{x'}{\wtd s'})$,  
  with $y_1 : \btinpt{S}$, $\len{\wtd s'} = \len{\Gt{S_V}}$, and $\wtd \prop_{\bullet} = \prop_3,\ldots,\prop_{12}$. 

  Hence, we should have $P_2 ~\relS~ Q_2$. 
Observe that $Q_2$ is not exactly the decomposition of $P_2$.
First, $\textcolor{blue}{V^m_C}$ is not the breakdown of $\textcolor{blue}{V_C}$.
Second, $\textcolor{blue}{V^m_C}$ is not at the same position in $Q_2$ as $\textcolor{blue}{V_C}$; the later being in the application position and the former being pushed through several propagators.
  Therefore, the relation
  $\relSs$ needs to (1)  relate $\textcolor{blue}{V_C}$ and
  $\textcolor{blue}{V^m_C}$ and (2) account for the fact that a value related to
  $\textcolor{blue}{V_C}$ and thus it needs to be propagated (as in $Q_2$). 
  To address the first point, we establish a relation $\vrelate$ between characteristic values and their minimal counterparts.
  For the second point, we record this fact by ``decomposing'' the process as $P_2 = P'_2\subst{V_C}{t}$, and propagating the information  about this substitution when computing the set of processes that are related to $P_2$.

  The same considerations we mentioned also apply to the value $V$, which is 
  transmitted internally, via a synchronization:
  \begin{align*}
		P_2 \by{\tau}  
			\news{s}(\bout{u}{V}\inact \Par \appl{\textcolor{blue}{V_C}}{s}  
						\Par \bout{\dual s}{V}\inact) = P_3 
  \end{align*}

  \noindent Value $V$ transmitted in $P_2$ should be related to its corresponding 
  breakdown $\V{12}{\epsilon}{V}$, which should be propagated through 
  the decomposition: 
  \begin{align*}
    Q_2 & \by{\tau} Q'_2 \by{\tau} Q''_2 \By{\tau}  
    \news{\wtd \prop_{\bullet \bullet}}
    \apropout{4}{\textcolor{blue}{V^m_C},
    \V{11}{\epsilon}{V}}  \\
    & 
    \quad\qquad\qquad \Par \news{s_1}
      (\apropinp{4}{t,x}.
      \apropout{5}{x}.\apropout{7}{t, x} \Par 
      \propinp{5}{x}\bout{u_2}{x}\apropout{6}{} \Par \apropinp{6}{} \Par 
    \\ 
    & \qquad\qquad\qquad 
    \Par \propinp{7}{t,x}\propout{8}{t}\apropout{9}{x} \Par 
    \apropinp{8}{t}.\appl{t}{s_1}
      \Par \apropinp{9}{x}.\bout{\dual{s_1}}{x}\apropout{10}{} \Par 
      \apropinp{10}{}) \Par 
      \\
      & \qquad\qquad\qquad  
       \Par \apropout{12}{} \Par \apropinp{12}{} = Q_3
  \end{align*}
  \noindent where $\wtd \prop_{\bullet \bullet}=\prop_4,\ldots, \prop_{10}, \prop_{12}$. 

  Now, in $P_3$ we can observe the output of $V$ along  $u$: 
  \begin{align*}
		P_3 \by{\bactout{u}{V}}
			\news{s}(\inact \Par \appl{\textcolor{blue}{V_C}}{s}  
						\Par \bout{\dual s}{x}\inact) = P_4 
	\end{align*}
Process $Q_3$ mimics this action by sending the process
  $\V{12}{\epsilon}{V}$ 
  along name $u_2$: 
  \begin{align*}
    Q_3 &\By{\bactout{u_2}{\V{12}{\epsilon}{V}}}   
    \news{\wtd \prop_{*}}
    \apropout{7}{\textcolor{blue}{V^m_C},
      \V{12}{\epsilon}{V}} \Par 
     \apropout{6}{} \Par \apropinp{6}{} \Par 
    \\ 
    & \qquad\qquad\qquad 
    \Par \propinp{7}{t,x}\propout{8}{t}\apropout{9}{x} \Par 
    \propinp{8}{t}\appl{t}{s_1}
      \Par \propinp{9}{x}\bout{\dual{s_1}}{x}\apropout{10}{} \Par 
      \apropinp{10}{})  = Q_4 
  \end{align*}
  \noindent where $\wtd \prop_{*} = \prop_6, \ldots, \prop_{10}$. 
  Following the definition of higher-order bisimilarity, 
  we should have: 
	\begin{align*}
		P_4 \parallel \textcolor{blue}{\htrigger{t'}{V}}
		 ~\relS~ 
     Q_4 \parallel \textcolor{blue}{\htrigger{\tni'}{\V{11}{\epsilon}{V}}}
	\end{align*}
for a fresh $t'$, where we have used `$\parallel$' (rather than `$|$') to denote process composition: we find it convenient to highlight those sub-processes in parallel that originate from trigger and characteristic processes. 

  We can see that the trigger process for   $V$ on the left-hand side should be matched  with a trigger process for the \emph{breakdown} of $V$ on the 
  right-hand side. 
  Moreover, the definition of trigger processes should be generalized to 
  polyadic values, as $\V{12}{\epsilon}{V}$ could be polyadic
  (see \Cref{mst:def:trigger}). 

  Let us briefly consider how $	P_4 \parallel \textcolor{blue}{\htrigger{t'}{V}}$ 
  evolves after due to the synchronization in sub-process 
  $\appl{\textcolor{blue}{V_c}}{s}$ within $P_4$: 
  \begin{align*}
    P_4 \parallel \textcolor{blue}{\htrigger{t'}{V}} \by{\tau}  
    \news{s} 
    (
    \textcolor{blue}
	{\binp{s}{x'}(\about{t}{} \Par \appl{x'}{s'})} 
						\parallel  \bout{\dual s}{V}\inact)   
    \parallel \textcolor{blue}{\htrigger{t'}{V}}  
    =  P_6 \parallel \textcolor{blue}{\htrigger{t'}{V}} 
  \end{align*}
  
\noindent We can see that $Q_4$ can mimic this synchronization after a few administrative 
reductions on propagators: 
  \begin{align*} 
    Q_4 &\By{\tau}   
    \news{\prop_9 \prop_{10}}
   \apropout{9}{\V{12}{\epsilon}{V}}
    \Par 
    \textcolor{blue}{\binp{s_1}{x'}(\about{\tni}{} \Par \appl{x'}{\wtd s'})} 
      \Par \propinp{9}{x}\bout{\dual{s_1}}{x}\apropout{10}{} \Par 
      \apropinp{10}{}) 
      \parallel 
      \textcolor{blue}{\htrigger{\tni'}{\V{12}{\epsilon}{V}}} 
      \\ 
      & \qquad \qquad 
      = Q_6 \parallel  \textcolor{blue}{\htrigger{\tni'}{\V{12}{\epsilon}{V}}} 
  \end{align*}

\noindent Therefore, we need to have: 
\begin{align*}
 	P_6 \parallel \textcolor{blue}{\htrigger{t'}{V}} 
	~\relS~
	 Q_6 \parallel  \textcolor{blue}{\htrigger{\tni'}{\V{12}{\epsilon}{V}}} 
\end{align*}
To ensure that this pair is in $\relS$, we introduce an auxiliary  relation, denoted $\processrelate$ (\Cref{mst:def:valuesprelation}), which allows us to account for the sub-processes that originate from 
characteristic values or trigger processes (in $\textcolor{blue}{\text{blue}}$).
We need to account for them separately, because one of them is not the decomposition of the other.
We thus decree:
\begin{align*}
 	\textcolor{blue}{\binp{s}{x'}(\about{t}{} \Par \appl{x'}{s'})}
	& \processrelate 
	 \textcolor{blue}{\binp{s_1}{x'}(\about{\tni}{} \Par \appl{x'}{\wtd s'})}  
	 \\
	 \textcolor{blue}{\htrigger{t'}{V}}
	& \processrelate 
	 \textcolor{blue}{\htrigger{\tni'}{\V{12}{\epsilon}{V}}}
\end{align*}
%\noindent and 
%\begin{align*}
% 	\textcolor{blue}{\htrigger{t}{V}}
%	\processrelate 
%	 \textcolor{blue}{\htrigger{t}{\V{12}{\epsilon}{V}}}  
%\end{align*}

\noindent Next, the synchronization on $s$ in  $P_6$ is mimicked by $Q_6$ with a 
synchronization on $s_1$: 
\begin{align*}
	P_6 \parallel \textcolor{blue}{\htrigger{t}{V}} 
   &\by{\tau}
   (\bout{t}{}\inact \Par \appl{V}{s'})   
     \parallel  \textcolor{blue}{\htrigger{t'}{V}}  = P_8 
    \parallel \textcolor{blue}{\htrigger{t'}{V}}  
    \\
     Q_6 \parallel \textcolor{blue}{\htrigger{\tni'}{\V{12}{\epsilon}{V}}}  &\By{\tau}   
    \news{\prop_{10}}
    (\about{t_1}{} \Par \appl{\V{12}{\epsilon}{V}}{\wtd s'}) 
      \Par \apropout{10}{} \Par 
      \apropinp{10}{}) = 
      Q_8 \parallel \textcolor{blue}{\htrigger{\tni'}{\V{12}{\epsilon}{V}}} 
\end{align*}

\noindent  Finally, we can see that after the output on the trigger name $t$ 
there is an application that activates $R$, the body of $V$:  
\begin{align*}
	P_8 &\by{\bactout{t}{}}  \appl{V}{s'} \by{\tau} R\subst{s'}{y} \\ 
	Q_8 & \by{\bactout{t}{}}  \appl{\V{12}{\epsilon}{V}}{\wtd s'} 
	\by{\tau} 
	\news{\wtd \prop_{**}} \apropout{12}{} \Par \B{12}{\epsilon}{R} \subst{\wtd s'}{\wtd y} 
	\equiv \D{R\subst{s'}{y}}
\end{align*}
We reached the point where we relate process $R\subst{s'}{y}$ with its 
decomposition $\D{R\subst{s'}{y}}$. Hence, the remaining pairs in $\relS$ are obtained 
in the same way. 

% Therefore, it remains to establish $R\subst{s'}{y} ~\relS~ \D{R\subst{s'}{y}}$.
% We again reached the point where the process on the right-hand side 
% is exactly the decomposition of the left-hand side process. 

% At this point we again have to related a process with its decomposition, which we do in the same manner.
% which is a synchronized pair.

\paragraph{Key insights.}
We summarize some key insights from the example: 
  \begin{itemize} 
  \item A received value can either be a pure value or a characteristic value.
    In the former case, the pure value has to be related to its decomposition, but in the later case the value should be related to an MST characteristic value of the same type.
    We define the relation $\vrelate$ on values  to account for this (\Cref{mst:def:vrelate}). 
    
  \item Trigger processes mentioned in the output case of MST bisimilarity should be matched with their minimal counterparts, and the same applies to processes originating from such trigger processes.
    The relation $\processrelate$ accounts for this (see \Cref{mst:def:valuesprelation}).

    \item Any value in process $P$ could have been previously received. The definition
    of $\relS$ takes this into account by explicitly relating processes with
    substitutions (see \Cref{mst:d:relation-s}). That is, for $P$, it relates
    $P'\subst{\tilde W}{\tilde x}$ such that $P'\subst{\tilde W}{\tilde x}=P$.
    Here, the substitution $\subst{\tilde W}{\tilde x}$ records values that should be propagated. 
  \end{itemize} 

\subsubsection{The relation $\relS$}

In this section we give the definition of the relation $\relS$ (\Cref{mst:d:relation-s}), following the insights gathered from the example. 
More specifically, we define
\begin{itemize}
\item a relation $\vrelate$ on values, which includes the relation $\valuesrelate$ from \Cref{mst:def:valuesrelate}, (\Cref{mst:def:vrelate});
\item a relation $\processrelate$ on processes, for relating characteristic and trigger processes with their MST counterparts, (\Cref{mst:def:valuesprelation});
\item a set $\Cb{\tilde W}{\tilde x}{P}$ of processes \emph{correlated} to a process $P\subst{\tilde W}{\tilde x}$, (\Cref{mst:t:tablecd}).
\end{itemize}

Because  we will be working extensively with indexed processes, we will use   the following function, which  returns a set of all valid indexing substitutions for a list of names.
\begin{definition}[Indexed names substitutions]
  \label{mst:d:indexedsubstitution}
  Let 
  $\wtd u = (a,b, r, \dual r,  
  r', \dual r', s, \dual s, s',  \dual{s}', \ldots)$
  be a finite tuple of names, 
    where $a, b, \ldots$ denote shared names, $r, \dual r,  
    r', \dual r', \ldots$ denote tail-recursive names , and $s, \dual s,  s',
    \dual{s}', \ldots$ denote linear (non tail-recursive names). 
    We write $\indices{\wtd u}$ to denote
  \begin{align*}
    \indices{\wtd u} = \subst{
      a_1, b_1,  r_1, \dual r_1,  
      r'_1, \dual r'_1, s_i, \dual s_i, s'_j, \dual{s}'_j, \ldots}
      {a,b,  r, \dual r,  
      r', \dual r', s, \dual s, s', \dual{s}', \ldots : 
  i,j,\ldots > 0}
  \end{align*} 
\end{definition}
Any substitution  $\sigma \in \indices{\fn{P}}$ turns an \HO process $P$ into an indexed process $P\sigma$.

\paragraph{Correlated values.}
The main ingredient in defining the relation $\relS$ is the the set $\Cb{\tilde W}{\tilde x}{P}$, which contains processes \emph{correlated} to process $P$ with a substitution $\subst{\tilde W}{\tilde x}$.
The substitution, as discussed above, denotes previously received values, and we assume that $\fv{P}=\wtd x$.
Essentially,
$\Cb{-}{-}{-}$ computes a breakdown of $P\subst{\tilde W}{\tilde x}$ in parallel with an activating trio, that mimics the original actions of $P$ up to transitions on propagators.
The activating trio propagates not the original values ${\tilde W}$, but the values related to ${\tilde W}$.
To do that we introduce the set $\Cb{-}{-}{V}$ of correlated values and the relation $\vrelate$ on values, which are defined mutually recursively in the three following definitions.

\begin{definition}[Broken down values]
  \label{mst:def:valuesset}
  Given a value $V$, the set $\Cb{}{}{V}$ is defined as follows:
   $$\Cb{}{}{V} = \bigcup \big\{\Cb{\tilde W}{\tilde x}{V'} : V = V' \subst{\tilde W}{\tilde x} \mbox{ and $V'$ is not a variable}\big\}$$
We extend $\Cb{}{}{-}$ to work on a list of values $\wtd V$ component-wise,
that is: 
$$\Cb{}{}{V_1,\ldots,V_n} = \{B_1,\ldots,B_n : B_i \in \Cb{}{}{V_i} \ \text{for} \ i \in 1\ldots n\}.$$
\end{definition}
\noindent This way, the elements in $\Cb{}{}{V}$ differ in the propagated values $\wtd W$. Consider the following example:
\begin{example} Let $V = \abs{y}\bout{y}{V_1}\bout{y}{V_2}\inact$. 
  There are four possibilities of $V'$, $\wtd W$, and $\wtd x$ 
  such that $V = V'\subst{\tilde W}{\tilde x}$. 
  That is, 
  \begin{itemize}
      \item 
  $V = V^1\subst{V_1 V_2}{x_1 x_2}$ 
  where $V_1 = \abs{y}\bout{y}{x_1}\bout{y}{x_2}\inact$
  \item 
  $V = V^2\subst{V_1}{x_1}$ where $V^2 = \abs{y}\bout{y}{x_1}\bout{y}{V_2}\inact$
  \item 
 $V = V^3\subst{V_2}{x_2}$ where 
  $V^3 = \abs{y}\bout{y}{V_1}\bout{y}{x_2}\inact$ 
  
  \item Finally, we can take the identity substitution $\wtd W = \epsilon$ and 
  $\wtd x = \epsilon$.
  
\end{itemize}
  Thus, we have 
  $\Cb{}{}{V} = \big\{\Cb{V_1 V_2}{x_1 x_2}{V^1},\ \Cb{V_1}{x_1}{V^2},\
  \Cb{V_2}{x_2}{V^3},\ \Cb{\epsilon}{\epsilon}{V} \big\}$. 
\end{example}
\begin{definition}
  \label{mst:def:valuessetsubst}
  Given a value $V$, the set $\Cb{\tilde W}{\tilde x}{V}$, where  $\fn{V} = \wtd x$ is defined as follows:
  $$\Cb{\tilde W}{\tilde x}{V} = \big\{\V{k}{\tilde x}{V}\subst{\tilde B}{\tilde x}\mid \wtd W \vrelate \wtd B \big\}.$$
\end{definition}
\begin{definition}[Relating values]
  \label{mst:def:vrelate}
  The relation $\vrelate$ on  values (with indexed names) is defined as follows:
$$
V_1 \vrelate V_2 \iff
\thesisalt{
\begin{cases}
  \exists V'_1,\,\sigma \in \indices{\fn{V'_1}}.\ V_1 = V'_1\sigma \wedge V'_1 \valuesrelate V_2 & \mbox{ if $V_1$ is a characteristic or a trigger value}\\
  V_2 \in \Cb{}{}{V_1} & \mbox{ otherwise. }
\end{cases}
}
{
  \begin{cases}
    \exists V'_1,\,\sigma \in \indices{\fn{V'_1}}.\ V_1 = V'_1\sigma \wedge V'_1 \valuesrelate V_2 & 
    \ 
    \begin{tabular}{l}
    \text{if $V_1$ is a characteristic} 
    % \\ \\ 
    \\ 
    \text{or a trigger value}
    \end{tabular} 
    \\
    V_2 \in \Cb{}{}{V_1} & \ 
    \begin{tabular}{l}
      \text{otherwise. }
    \end{tabular} 
  \end{cases}
}
$$
%  	$$
% \frac{V_2 \in \Cb{}{}{V_1}}{V_1 \vrelate V_2}
%         \qquad
%         \frac{V_1 \valuesrelate V_2}{V_1 \vrelate V_2}
%         $$
where $\valuesrelate$ is the relation from \Cref{mst:def:valuesrelate}.
\end{definition}
Thus, in the definition of $\Cb{\tilde W}{\tilde x}{V}$, the value $V$ is related to the triggered break down values with $\wtd B$ substituted for $\wtd x$ such that $\wtd W \vrelate \wtd B$.

Additionally, 
to define $\Cb{\tilde W}{\tilde x}{-}$ for processes, 
we have to observe the behavior of processes enclosed in the received trigger and characteristic values. 
Further, we have to observe the behavior of 
trigger processes of shape $\htrigger{t}{V}$. 
For this we need to define a relation $\processrelate$ on processes that contains pairs
$$(\mapchar{C}{x},~\mapcharm{C}{x}_1), 
(\binp{t}{y}\appl{y}{x},~\binp{\tni}{y}\appl{y}{\wtd x}), 
(\htrigger{t}{V},~\htrigger{\tni}{W})$$ 
 where $x:C$ and $\len{\wtd x} = \len{\Gt{C}}$  and $V \vrelate W$. 

Before we define $\processrelate$ we need the following auxiliary definition:
\begin{definition}[Relating names]
  \label{mst:def:namesrelate}
  We define $\processrelate$ as the relation on names defined as
  $$
  \frac{}
  {\epsilon \processrelate \epsilon} 
  \qquad 
  \frac{\Gamma;\Lambda;\Delta \proves n_i \hastype C}
  {n_i \processrelate (n_i,\ldots,n_{i+\len{\Gt{C}}-1})} 
  \qquad 
   \frac{\tilde n \processrelate \tilde m_1 \ 
  \quad  n_i \processrelate \tilde m_2}
  {\tilde n,n_i \processrelate \tilde m_1, \tilde m_2} 
  $$
  \noindent where $\epsilon$ denotes the empty list. 
\end{definition}
\noindent Now, we are ready to relate processes, modulo indexed names
(cf. \Cref{mst:def:indexedname}), using the relation $\processrelate$ defined as follows: 
\begin{definition}[$\processrelate$ Indexed process relation]
We define the relation $\processrelate$ as
  \label{mst:def:valuesprelation}
  \begin{mathpar}
		\inferrule[\eltsrule{IPApp}] {
      V \vrelate W \and 
      x_i \processrelate \tilde x
		}{
			\appl{V}{x_i} \processrelate \appl{W}{\tilde x}
    } 
    \qquad 
    \inferrule[\eltsrule{IPPar}] {
      P \processrelate P' \and Q \processrelate Q'
		}{
			P \Par Q \processrelate P' \Par Q'
    }
    % \qquad 
    \\
    \inferrule[\eltsrule{IPInact}] {
    }{\inact \processrelate \inact} 
    \qquad
  \inferrule[\eltsrule{IPNews}] { P \processrelate P' \and \tilde m_1 \processrelate \tilde m_2 
    }{\news{\tilde m_1}P \processrelate \news{\tilde m_2}P' } 
    \\
    \inferrule[\eltsrule{IPSnd}]
    {P\sigma \processrelate P' \and V\sigma \vrelate W \and \sigma = \subsqn{n_i}}
    {\bout{n_i}{V}P \processrelate \bout{n_i}{W}P'}
    \qquad 
    \inferrule[\eltsrule{IPRcv}] {
      P\sigma \processrelate P' \and \sigma = \subsqn{n_i}
    }{\binp{n_i}{y}P \processrelate \binp{n_i}{y}P'}  
  %  \qquad
  % \inferrule[\eltsrule{IPNews}] { P \processrelate P' \and \tilde m_1 \processrelate \tilde m_2 
  %   }{\news{\tilde m_1}P \processrelate \news{\tilde m_2}P' } 
% 
  \end{mathpar}
\end{definition}
We can now show the property that we wanted, namely that: the bodies of trigger values and minimal trigger values (\Cref{mst:d:mtv}) are related; the bodies of characteristic values and minimal characteristic values (\Cref{mst:d:mcv}) are related; and that the trigger processes and minimal trigger processes 
(\Cref{mst:def:trigger}) are related, with appropriate name substitutions.
\begin{lemm}
We have:
  $$\big\{(\mapchar{C}{x}\subst{x_i,\tni}{x,t},\ 
  \mapcharm{C}{x}_i),~ 
  (\binp{\tni}{y}\appl{y}{x} \subst{x_i}{x},\
  \binp{\tni}{y}\appl{y}{\wtd x}),~ (\htrigger{t}{V}\sigma,\ 
  \htrigger{\tni}{W}) \big\}
  \subset \processrelate$$
  where $i,j>0$, $x:C$, $\wtd x = (x_i, \ldots,x_{i + \len{\Gt{C}}-1})$, 
   $\sigma \in \indices{\wtd u}$,
   $\wtd u = \fn{V}$, and 
   $V\sigma \vrelate W$. 
\label{mst:lemm:processrelate-triggers}
\end{lemm}
\begin{proof}[Proof (Sketch)]
  We may notice that $\processrelate$ relates process
  up to incremented indexed names and values related by $V\sigma \vrelate W$
  for some $\sigma$. More precisely, free names as subject of actions are indexed and incremented 
  accordingly in 
  a related process, and names as objects of output actions 
  are broken down in a related process, by $V\sigma \valuesrelate W$  
  when $V\sigma = m_i$, that is $m_i \valuesrelate \tilde m$ 
  where $m_i:C$ and $\tilde m = (m_i, \ldots,m_{i + \len{\Gt{C}}-1})$. 

  For the first pair $(\mapchar{C}{x}\subst{x_i, t_1}{x, t},~\mapcharm{C}{x}_i)$ by 
  inspection of 
  \Cref{mst:def:indexedname} we can observe that $\mapcharm{C}{x}_i$
  is essentially $\mapchar{C}{x}$ with its subject names indexed and 
  incremented (starting with 
  index $i$) and objects names broken down. Thus, it is contained in 
  $\processrelate$. 
  Similarly, $(\binp{t}{y}\appl{y}{x}\subst{x_i}{x},~
  \binp{\tni}{y}\appl{y}{\wtd x})$ is contained by observing that 
  $x_i \processrelate \wtd x$. 
  Finally, for $(\htrigger{t}{V}\sigma,~ \htrigger{\tni}{W})$,
  by \Cref{mst:def:trigger}, $V \sigma \vrelate W$. 
  % and since $\iname n$ denotes $n$ as well.  
% \qed
\end{proof}

\paragraph{Correlated Processes.}
Finally, we can use the introduced notions to define the set $\Cb{-}{-}{-}$ of correlated processes.
As mentioned, the set $\Cb{\tilde W}{\tilde x}{P}$ contains processes correlated to process $P$ with a substitution $\subst{\tilde W}{\tilde x}$.
The definition of $\Cb{-}{-}{-}$ is given in \Cref{mst:t:tablecd}.
Before looking into the details, we first describe how the $\Cb{-}{-}{-}$ is used.

We introduce auxiliary notions for treating 
free (tail-recursive) names in processes.
\begin{definition}[Auxiliary Notions]
  \label{mst:def:rec-providers}
  
%\begin{definition}[Propagators of $P$]
%    \label{mst:d:fpn}
  Let $P$ be an \HO process. 
\begin{itemize}

\item      We write $\fpn{P}$ to denote the set of free propagator names in $P$.

\item 
  We define $\mathtt{rfv}(P)$ to denote free tail-recursive names in 
  values in $P$. 

\item
  We define $\fcr{P}$ to denote free names of form $\prop^r$ in $P$.
\item 
  \label{mst:d:rfni}

  We define $\rfni{P}$ such that $r \in \rfni{P}$ if and only if 
  $(r_i, \ldots, r_{j}) \subseteq \rfn{P}$ for some $i, j > 0$. 
  
  	\item 
	Given  $r:S$ and $\wtd r = (r_1,\ldots,r_{\len{\Gt{S}}})$, 
  we write $\Rb{\tilde {v}}$ to denote the process 
  $$\Rb{\tilde {v}} = \prod_{r \in \tilde{v} } 
  \recprov{r}{x}{\wtd r}
  % \binp{\prop^r}{b}\appl{b}{\wtd r} 
  $$

  \end{itemize}
\end{definition}
\begin{definition}[Relation $\relS$]
  \label{mst:d:relation-s}
Let $P\subst{\tilde W}{\tilde x}$ be a well-typed process 
such that $\fn{P} \cap \fn{\wtd W}=\emptyset$, and let the
$\mathcal{C}$-set be as in \Cref{mst:t:tablecd}. 
We define the relation $\relS$ as follows: 
\begin{align*}
  \relS &= \big\{\big(P\subst{\tilde W}{\tilde x}, 
  \news{\wtd \prop_r} \news{\wtd \prop} R\big) :\ 
  R \in \Cb{\tilde W\sigma}{\tilde x}{P\sigma} 
  \\  
  & \qquad \text{with }
    \wtd u = \fn{P\subst{\tilde W}{\tilde x}},\ \sigma \in \indices{\wtd u},\ 
   \wtd \prop_r = \fcr{R},\ \wtd \prop = \fpn{R} \big\} 
\end{align*}
\end{definition}

Now we describe the definition of $\Cb{-}{-}{-}$ in \Cref{mst:t:tablecd}.
Essentially, $\Cb{-}{-}{-}$ computes a breakdown of $P\subst{\tilde W}{\tilde x}$ in parallel with an activating trio, that mimics the original actions of $P$ up to transitions on propagators.
This is done with the help of $\Db{-}{-}{-}$ (also given in \Cref{mst:t:tablecd}), which computes a closure of a process with respect to $\tau$-transitions on propagators.

To define the $\mathcal{C}$-set we distinguish processes that do not appear in the given process, but that are composed in parallel by the clauses of MST bisimilarity (\Cref{mst:d:fwb}).
For this we use the following notions:
\begin{definition}[Trigger Collections]
 \label{mst:d:triggerscollection}
 We let $H, H'$ to range over \emph{trigger collections}: processes of the form $P_1 \Par \cdots \Par P_n$ (with $n \geq 1$), 
 where each $P_i$ is a trigger process or a process that originates from a trigger or from a characteristic value. 
\end{definition}

\begin{example}
 Let $H_1 = \htrigger{t}{V} \Par \mapchar{C}{u} \Par 
 \bout{t'}{n}\inact$  
 where $t,t',u,n$ are channel names, 
 $V$ is a value, and
 $C$ a channel type. Then,
 we could see that $\bout{t'}{n}\inact$ 
 originates from a characteristic value. Thus, $H_1$ is a trigger collection.
\end{example}

Notice that we write $P$ to denote a ``pure'' process that is not composed with a trigger collection.
For processes with trigger collections, the following notation is relevant:

\begin{definition}[Process in parallel with a trigger or a characteristic process]
   \label{mst:def:parallel}
We write 
   $P \parallel Q$ to stand for $P \Par Q$ 
   where either $P$ or $Q$ is a trigger collection.
\end{definition}

\thesisalt{
  \begin{table}[!t]
% \begin{table}[!t]
    \begin{tabular}{ |l|l|l|}
      \rowcolor{gray!25}
      \hline
      $P$ &
        \multicolumn{2}{l|}{
      \begin{tabular}{l}
        \noalign{\smallskip}
        $\Cb{\tilde W}{\tilde x}{P}$
        \smallskip
      \end{tabular} 
    } 
         \\
      \hline
    $Q_1 \parallel Q_2$ 
    &
      \begin{tabular}{l}
      \noalign{\smallskip}
         $\big\{ R_1 \parallel R_2: 
         R_1 \in \Cb{\tilde W_1}{\tilde y}{Q_1}, 
         R_2 \in \Cb{\tilde W_2}{\tilde w}{Q_2}
         \big\}$
         \smallskip 
      \end{tabular}
      &
      %side-conditions
      \begin{tabular}{l}
        \noalign{\smallskip}
        $\wtd y = \fv{Q_1}$, $\wtd w = \fv{Q_2}$ \\
        $\subst{\tilde W}{\tilde x} = 
        \subst{\tilde W_1}{\tilde y} \cdot \subst{\tilde W_2}{\tilde w}$ 
        \smallskip
      \end{tabular}
      \\
     \hline
    $\news{m:C}\,Q$~
    &
      \begin{tabular}{l}
      \noalign{\smallskip}
         $\left\{\news{\wtd m : \Gt{C}} \news{\tilde \prop^m}{\,R} : 
          R \in 
         \Cb{\tilde W}{\tilde x}{Q\sigma} \right\}$
         \smallskip 
        %  \\
        %  $\big\{\news{\wtd s : \mathcal{R}({\trec{t}{S}})}
        %           {\news{\prop^s}
        %           \news{\prop^{\dual s}}R} :
        %            R \in \Cb{\tilde W}{\tilde x}{Q\sigma} \big\}$ 
        %            \smallskip 
      \end{tabular}
      &
      %side-conditions
      \begin{tabular}{l}
        \noalign{\smallskip}
        $\widetilde m = (m_1,\ldots,m_{\len{\Gt{C}}})$
        \\
        $\sigma = \subst{m_1 \dual{m_1}}{m \dual{m}}$
        \\
        % $\wtd \prop^m = \prop^m \cdot \prop^{\dual m}$
        $\tilde \prop^m = \linecondit{\tr(C)}{\prop^m \cdot \prop^{\dual m}}{\epsilon}$
        \smallskip
      \end{tabular}
      \\
     \hline
    $Q$
    &
      \begin{tabular}{l}
      \noalign{\smallskip}
      $\big\{ \Rb{\tilde v} \Par \apropout{k}{\wtd B}  \Par \B{k}{\tilde x}{P} \}$
      \\  
              % $\cup \big\{ \Rb{\tilde v \setminus n} \Par 
              % \propout{k}{\wtd B}
              % % \binp{\prop^n}{b}(\appl{b}{\wtd n})
              % \Rb{n}
              % \Par \B{k}{\tilde x}{P} : n \in \rfn{P} \big\}$ 
              % \\ 
              $\cup$ 
            $\big\{ \Rb{\tilde v \setminus \tilde r} \Par R: 
            R \in \Db{\tilde W}{\tilde x}{P},~ 
            \wtd r = \rfni{R} \big\}$
            % $\big\{ \Rb{\tilde v} \Par \apropout{k}{\wtd B}  \Par \B{k}{\tilde x}{P},$
            % \\  
            %         $\Rb{\tilde v \setminus r} \Par \propout{k}{\wtd B}
            %         \binp{\prop^r}{b}(\appl{b}{\wtd r})
            %         \Par \B{k}{\tilde x}{P} \big\}$ 
            %         \\ 
            %         $\cup$ 
            %       $\big\{ \Rb{\tilde v \setminus \tilde r} \Par R: 
            %       R \in \Db{\tilde W}{\tilde x}{P},~ 
            %       \wtd r = \rfni{R} \big\}$
         \smallskip 
      \end{tabular}
      &
      %side-conditions
      \begin{tabular}{l}
        \noalign{\smallskip}
        $\wtd W \vrelate \wtd B$ \\ 
        % $r \in \rfn{P}$, 
        $\wtd v = \rfn{P\subst{\tilde W}{\tilde x}}$ 
        % $\wtd v = \rfn{P}$ 
        % \\
        % $R_3 \in $ \\ 
        % $\wtd r = \rfni{R_3}$ 
        % \\
        % $\wtd \prop_{i} = \fpn{R_i}, i \in \{1,2,3\}$
          % $\wtd W \vrelate \wtd B$ \\
          % $\wtd \prop = \fpn{R}$
        \smallskip
      \end{tabular}
      \\
     \hline
    $H$
    &
      \begin{tabular}{l}
      \noalign{\smallskip}
      % $\big\{ H' : H\subst{\tilde W}{\tilde x} 
      % \processrelate H'  \big\}$
      $\big\{ \Rb{\tilde v} \parallel H' : H\subst{\tilde W}{\tilde x} 
      \processrelate H'  \big\}$
      % $\{ H' : H\subst{\tilde W}{\tilde x} \processrelate H'\}$
        %  $\{ \news{\wtd \prop}H' : H\subst{\tilde W}{\tilde x} \processrelate H', 
        %  \wtd \prop = \fpn{H'} \}$
         \smallskip 
      \end{tabular}
      &
      %side-conditions
      \begin{tabular}{l}
        \noalign{\smallskip}
        $\tilde v= \rfn{H\subst{\tilde W}{\tilde x}}$
        % $\tilde v= \rfn{H}$
        \smallskip
      \end{tabular}
      \\
      \hline
      \hline
      \rowcolor{gray!25}
      $P$ &
        \multicolumn{2}{l|}{
      \begin{tabular}{l}
        \noalign{\smallskip}
        $\Db{\tilde W}{\tilde x}{P}$
        \smallskip
      \end{tabular}
    }

    \\
      \hline

    $\bout{u_i}{V_1}Q$
    &
      \begin{tabular}{ll}
      \noalign{\smallskip}
      $\bullet~\lnot$$\tr(C)$: &
      \\
            \quad  $\big\{\bout{u_i}{V_2}
             \apropout{k}{\widetilde B_2} \Par 
        \B{k}{\tilde z}{Q\sigma} \big\}$ 
        % $\quad  V_1 \sigma \subst{\tilde W_1}{\tilde y} \vrelate V_2, 
        % \wtd W_2 \vrelate \wtd B_2 \}$
        & %(\text{if } $\lnot$$\tr(C)$)
        \smallskip
        \\
        \hdashline 
        \noalign{\smallskip}
        $\bullet~\tr(C)$: & 
        \\
              \quad $\big\{ \abbout{\prop^u}{M^{\tilde B_2}_{V_2}}   \Par$ 
              $\B{k}{\tilde w}{Q},\ 
               \appl{M^{\tilde B_2}_{V_2}}{\wtd u} \Par \B{k}{\tilde w}{Q},$ 
               & %(\text{if } $\tr(C)$)
               \\ 
              \qquad $ {\bbout{u_{\indT{S}}}{V_2}}{} 
              (\apropout{k} {\wtd B_2} \Par 
              \recprovx{u}{x}{\wtd u}
              % \binp{\prop^u}{b}(\appl{b}{\wtd u})
          ) 
          \Par \B{k}{\tilde w}{Q} \big\}$ 
          % \\
              % $: V_1 \subst{\tilde W_1}{\tilde y} \vrelate V_2,\ 
              % \wtd W_2 \vrelate \wtd B_2 \}$
              \smallskip
              \\
              \quad where:
              \\
              \quad\quad$M^{\tilde B}_V = \abs{\wtd z}
                {\bbout{z_{\indT{S}}}{V}}{}$  
                $(\apropout{k}{\wtd B} \Par 
                \recprovx{u}{x}{\wtd z}
                % \binp{\prop^u}{b} (\appl{b}{\wtd z})
            )$ 
            \smallskip
      \end{tabular}
      &
      %side-conditions
      \begin{tabular}{l}
        \noalign{\smallskip}
        $\wtd y = \fv{V_1}$, $\wtd w = \fv{Q}$ \\
        $\subst{\tilde W}{\tilde x} = 
        \subst{\tilde W_1}{\tilde y} \cdot \subst{\tilde W_2}{\tilde w}$ \\
        $\sigma = \subsqn{u_i}$ \\
        $V_1 \sigma \subst{\tilde W_1}{\tilde y} \vrelate V_2$
        \thesisalt{,}{\\} 
        $\wtd W_2 \vrelate \wtd B_2$ \\
        % $\wtd v = \rfn{\bout{u}{V}{Q}\subst{\tilde W}{\tilde x}}$ \\
        % $\wtd v = \rfn{P\subst{\tilde W}{\tilde x}}$ \\
        $\wtd z = (z_1,\ldots, z_{\len{\Rts{}{s}{S}}})$ \\
        $\wtd u = (u_1,\ldots, u_{\len{\Rts{}{s}{S}}})$ 
        % $\sigma = \begin{cases}
        %   \incrname{u}{i} & \text{if } u_i:S \\
        %   \epsilon & \text{otherwise}
        %   \end{cases}$ 
        \smallskip
      \end{tabular}
        \\
      \hline

    $\binp{u_i}{y}Q$
    &
      \begin{tabular}{ll}
      \noalign{\smallskip}
      $\bullet~\neg$$\tr(C)$: & 
      \\
             \quad $\big\{\binp{u_i}{y}\apropout{k}{\wtd By} 
             \Par 
             \B{k}{\tilde xy}{Q\sigma} \big\}$
             & %(\text{if } $\neg$$\tr(C)$)
        \smallskip 
        \\
        \hdashline 
        \noalign{\smallskip}
         $\bullet~\tr(C)$: & \\
                \quad $\big\{  \abbout{\prop^u}
                {M^{\tilde B}_y} 
                \Par \B{k}{\tilde xy}{Q},\
                 \appl{M^{\tilde B}_y}{\wtd u} \Par \B{k}{\tilde xy}{Q},$ 
                 & %(\text{if } $\tr(C)$)
                 \\
             \qquad  ${ \binp{u_{\indT{S}}}{y}
             (\apropout{k}{\wtd By} \Par 
             \recprovx{u}{x}{\wtd u}
              % \binp{\prop^u}{b}(\appl{b}{\wtd z})
              )
              \Par \B{k}{\tilde xy}{Q}} \big\}$ 
              % \\
                % $ : \wtd W \vrelate \wtd B \}$
                \\
                \quad where:
                \\
                \qquad $M^{\tilde B}_y = 
                \abs{\widetilde z}{\binp{z_{\indT{S}}}{y}
                (\apropout{k}{\wtd By} \Par 
                \recprovx{u}{x}{\wtd z}
                % \binp{\prop^u}{b}(\appl{b}{\wtd z})
                )}$
                \smallskip 
      \end{tabular}
      &
      %side-conditions
      \begin{tabular}{l}
        \noalign{\smallskip}
        $\wtd W \vrelate \wtd B$ \\
        $\sigma = \subsqn{u_i}$ \\
        % $\wtd v = \rfn{P\subst{\tilde W}{\tilde x}}$ \\
        % $\wtd v = \rfn{\binp{u_i}{y}Q\subst{\tilde W}{\tilde x}}$ \\
        $\wtd z = (z_1,\ldots,z_{\len{\Rts{}{s}{S}}})$ \\
        $\wtd u = (u_1,\ldots,u_{\len{\Rts{}{s}{S}}})$ 
        % $\sigma = \begin{cases}
        %   \incrname{u}{i} & \text{if } u_i:S \\
        %   \epsilon & \text{otherwise}
        %   \end{cases}$ 
        \smallskip 
      \end{tabular}
      \\
      \hline
    $\appl{V_1}{(\wtd r, u_i)}$
    &
      \begin{tabular}{l}
        \noalign{\smallskip}
        $\big \{   \overbracket{\prop^{r_{l}}!\big\langle 
        \lambda \wtd z_{l}. \prop^{r_{{l}+1}}!\langle\lambda \wtd z_{{l}+1}.\cdots. 
     \prop^{r_n}!\langle \lambda \widetilde z_n.}^{|\tilde r| - {l} +1} 
     Q_l \rangle \,\rangle \big\rangle, $  \\
     \quad $  \lambda \wtd z_l.
     \overbracket{
         \prop^{r_{l+1}}!\langle\lambda \wtd z_{l+1}.\cdots. 
     \prop^{r_n}!\langle \lambda \wtd z_n.}^{|\tilde r| - l} 
     Q_l \rangle \, \big\rangle \ {\wtd r_l},$ \\
     \qquad $: 1 \leq l \leq n,\ V_1\subst{\tilde W}{\tilde x} \vrelate V_2 \big \} $ \\
    $\cup\, \{\appl{V_2}{\wtd r_1, \ldots, \wtd r_n, \wtd m}
    : V_1\subst{\tilde W}{\tilde x} \vrelate V_2 \}$ \\
    where:
    \\
     \quad $Q_l = \appl{V_2}{(\wtd r_1,\ldots,\wtd r_{l-1}, \wtd z_{l}, \ldots, 
     \wtd z_n, \wtd m)}$
    \smallskip
      \end{tabular}
      &
      %side-conditions
      \begin{tabular}{l}
        \noalign{\smallskip}
        $\forall r_i \in \widetilde r.(r_i: S_i \wedge \mathsf{tr}(S_i) \wedge$\\
          \thesisalt{\qquad}{\quad} $\wtd{z_i} = (z^i_1,\ldots,z^i_{\len{\Rts{}{s}{S_i}}}), $ \\
          \thesisalt{\qquad}{\quad} $\wtd{r_i} = (r^i_1,\ldots,r^i_{\len{\Rts{}{s}{S_i}}}) )$\\
          % $\wtd v = \rfn{V\subst{\tilde W}{\tilde x}}$ \\
          $u_i : C$ \\ 
          $\wtd m = (u_i, \ldots, u_{i+\len{\Gt{C}}-1})$
          \smallskip
      \end{tabular}
        \\
      \hline
    $Q_1 \Par Q_2$
    &
      \begin{tabular}{l}
      \noalign{\smallskip}
         $\big\{ \propout{k}{\wtd B_1}
         \apropout{k+l}{\wtd B_2} \Par \B{k}{\tilde y}{Q_1}
          \Par 
         \B{k+l}{\tilde z}{Q_2} \big\} $ \\
        %  $\quad \wtd W_1 \vrelate \wtd B_1$, $\wtd W_2 \vrelate \wtd B_2  \}$ \\ 
         $\cup$ \\
         $ \big\{(R_1 \Par R_2) : R_1 \in
         \Cb{\tilde W_1}{\tilde y}{Q_1}, R_2 \in 
         \Cb{\tilde W_2}{\tilde z}{Q_2} \big\}$ 
         \smallskip 
      \end{tabular}
      &
      %side-conditions
      \begin{tabular}{l}
        \noalign{\smallskip}
          $l = \plen{Q_1}$
          \thesisalt{,}{\\} 
          $\wtd W_1 \vrelate \wtd B_1$, $\wtd W_2 \vrelate \wtd B_2$ \\
        $\wtd y = \fv{Q_1}$, $\wtd z = \fv{Q_2}$  \\
        $\subst{\tilde W}{\tilde x} = \subst{\tilde W_1}{\tilde y} \cdot \subst{\tilde W_2}{\tilde z}$
        \smallskip 
      \end{tabular}
      \\ 
      \hline 
    $\inact$
    &
      \begin{tabular}{l}
      \noalign{\smallskip}
         $\inact$ 
         \smallskip 
      \end{tabular}
      &
      %side-conditions
      \begin{tabular}{l}
        \noalign{\smallskip}
        \smallskip 
      \end{tabular}
      \\
      \hline
    \end{tabular}
    
    \caption{The sets $\Cb{\tilde W}{\tilde x}{P}$ and $\Db{\tilde W}{\tilde x}{P}$. \label{mst:t:tablecd}}
    % \end{table} 
  \end{table}}
{\begin{table}[!t]
  \resizebox{1.01\textwidth}{!}{
  }
    \end{table}}

Now we can describe all the cases in the definitions of 
the $\Dset$-set and the $\mathcal{C}$-set in
\Cref{mst:t:tablecd} (Page~\pageref{mst:t:tablecd}).
Observe that the second and third columns in \Cref{mst:t:tablecd}  are closely
related: the third column lists side conditions for the definitions in the
second column. 
Note that in each case we assume the substitution $\rho = \subst{\tilde W}{\tilde x}$. 
We start with the cases for $\Cb{\tilde W}{\tilde x}{P}$:
\begin{description}
\item[Parallel with a trigger collection:]
The $\mathcal{C}$-set of $Q_1 \parallel Q_2$ is defined as:
\begin{align*}
 \{ R_1 \parallel R_2: 
R_1 \in \Cb{\tilde W_1}{\tilde y}{Q_1},\  
R_2 \in \Cb{\tilde W_2}{\tilde w}{Q_2}
\}
\end{align*}
\noindent By \Cref{mst:def:parallel},  either $Q_1$ or $Q_2$ is a trigger collection. 
Notice that a composition $Q_1 \Par Q_2$ (where both $Q_1$ and $Q_2$ are ``pure'')
 is handled by $\Db{}{}{-}$, see below.
We treat $Q_1 \parallel Q_2$  compositionally: we split the
substitution into parts concerning $Q_1$ and $Q_2$, i.e.,  
$\subst{\tilde W}{\tilde x} =
\subst{\tilde W_1}{\tilde y} \cdot \subst{\tilde W_2}{\tilde w}$
 such that $\wtd y
= \fv{Q_1}$ and $\wtd w = \fv{Q_2}$, and relate it to a parallel composition whose
components come from a corresponding $\mathcal{C}$-set.
\item[Restriction:] The $\mathcal{C}$-set of $\news{m:C}Q$
is inductively defined as: 
\begin{align*}
  \left\{\news{\wtd m : \Gt{C}}{\,R} : 
     \news{\tilde \prop^m} R \in 
     \Cb{\tilde W}{\tilde x}{Q\sigma} \right\}
\end{align*}
  \noindent where $\sigma = \subst{m_1 \dual{m_1}}{m \dual{m}}$ and $\wtd m =
  (m_1,\ldots,m_{\len{\Gt{C}}})$ is the decomposition of $m$ under $C$. The
  elements are processes from the $\mathcal{C}$-set of $Q$ with names $\wtd m$
  restricted. In the case when restricted name $m$ is a tail-recursive 
  then we also restrict the special propagator names $\prop^m$ and $\prop^{\dual m}$
  which appear in $R$.
  Notice that the processes of the form $\news{m}(Q_1
  \parallel Q_2)$, which are induced by the output clause of MST bisimilarity, are treated in this case in the definition of $\Cb{}{}{-}$.

    \item[Pure process:] The $\mathcal{C}$-set of a pure process $Q$ is defined as
    follows: 
    \begin{align*}
      \begin{tabular}{l}
      $\big\{ \Rb{\tilde v} \Par \apropout{k}{\wtd B} 
       \Par \B{k}{\tilde x}{Q} : \wtd W \vrelate \wtd B \big\} \cup
            \big\{ \Rb{\tilde v \setminus \tilde r} \Par R: 
            R \in \Db{\tilde W}{\tilde x}{Q},~ 
            \wtd r = \rfni{R} \big\}$
      \end{tabular}
  \end{align*}
%    \noindent where $\wtd B$ is such that $\wtd W \vrelate \wtd B$.  
	\noindent where $\wtd v = \rfn{Q\subst{\tilde W}{\tilde x}}$. 
	The elements in the first set are
    essentially the decomposition of $Q$ (without restrictions of recursive
    propagators, which are handled in $\relS$) up to 
    different possibilities of values $\wtd B$ that are $\vrelate$-related 
    to $\wtd W$ (see \Cref{mst:def:vrelate}). Here, we remark that 
   $\Rb{\tilde v}$ is 
    recursive name providers for all tail-recursive names of $Q$
    and $\wtd W$ (by $\wtd v = \rfn{Q\subst{\tilde W}{\tilde x}}$). 
    The second set contains elements of the $\Dset$-set of $Q$ in
    parallel with $\Rb{\tilde v \setminus \tilde r}$ where $\wtd r =
    \rfni{R}$. 
    By \Cref{mst:d:rfni} we can see that $\rfni{R}$ denotes tail-recursive names
    already gathered in $R$ by communications that consumed
     $\Rb{\tilde r}$ : thus, we have $\Rb{\tilde v \setminus \tilde r}$ 
     as providers at top level.

     In this sense, the processes from the second set can be seen as reducts of the processes from the first set.
     For example, if we examine the $\mathcal{C}$-set corresponding to the process $P_2$ from \Cref{mst:fig:exampletransitions}, we note that the process $Q_2$ belongs to the first set, and the processes $Q'_2$ and $Q''_3$ belong to the second set.

   \item[Trigger collection:] The $\mathcal{C}$-set of a trigger collection $H$
     contains its minimal counterparts, defined using the relation  $\processrelate$ (\Cref{mst:def:valuesprelation}): 
    \begin{align*}
    \big\{ \Rb{\tilde v} \parallel H' : H\subst{\tilde W}{\tilde x} 
  \processrelate H'  \big\}. 
%      \big\{ H' : H\subst{\tilde W}{\tilde x} 
%      \processrelate H'  \big\}.
      % \{ H' : H\subst{\tilde W}{\tilde x} \processrelate H'\} 
    \end{align*}
    \noindent where $\wtd v = \rfn{H\subst{\tilde W}{\tilde x}}$.
    In this case we do not use the information on the substitution $\subst{\tilde W}{\tilde x}$,
    because the substitution information is needed for values that are, or were, propagated.
    However, because $H$ is a trigger collection, it will only contain propagators as part of values.
    The substitutions related the propagators in values are already handled by the relation $\vrelate$, invoked by $\processrelate$.
    As in the case with pure processes, the process $\Rb{\tilde r}$ is the recursive names provider for the tail-recursive names of $H$.
  \end{description}
% \end{itemize}
\bigskip
\noindent We now discuss the cases for $\Db{\tilde W}{\tilde x}{P}$:
\begin{description}
  \item[Output:] 
  The  $\Dset$-set of $\bout{u_i}{V_1}Q$ depends on whether 
  (i)~$u_i$ is linear or shared name (i.e., $\neg \tr(u_i)$) or 
  (ii)~$u_i$ is a tail-recursive name (i.e., $\tr(u_i)$). 
  % (i)~ $u_i$ is linear or shared name or 
  % (ii)~ $u_i$ is a tail-recursive name. 
  In sub-case (i) $\Dset$-set is defined as follows: 
  % The $\Dset$-set of $\bout{u_i}{V_1}Q$ contains decomposed
  % processes that can immediately mimic the action of the given process: 
    \begin{align*}
      \begin{tabular}{l}
        \noalign{\smallskip}
               $\big\{\bout{u_i}{V_2}
               \apropout{k}{\widetilde B_2} \Par 
          \B{k}{\tilde z}{Q\sigma}:$ 
          $V_1\sigma \subst{\tilde W_1}{\tilde y} \vrelate V_2,\ 
          \wtd W_2 \vrelate \wtd B_2 \big\}$
          \smallskip
      \end{tabular}
    \end{align*}
    \noindent where $\sigma = \subsqn{u_i}$. By the definition,
     the substitution $\sigma$ depends on whether $u_i$ is linear or shared: in
    the former case, we use a substitution that increments $u_i$; in the latter
    case we use an identity substitution. 
    We split $\wtd W$ into $\wtd W_1$ and $\wtd W_2$, associated to the emitted value
    $V_1$ and the continuation $Q$, respectively.

    Instead of the emitted value $V_1$ we consider values $V_2$ that are $\vrelate$-related to $V_1\sigma\subst{\tilde W_1}{\tilde y}$.
    This way, we uniformly handle cases when (i)~$V_1$ is a pure value, (ii)~variable, and (iii)~a characteristic value.
    In particular, if $V_1$ is a pure value, the set $\Cb{\tilde W_1}{\tilde y}{V_1\sigma}$ is included in all the values $\vrelate$-related to $V_1\sigma\subst{\tilde W_1}{\tilde y}$.

    Further, the propagator $\prop_k$ actives the next trio with the values $\wtd B_2$
    such that $\wtd W_2 \vrelate \wtd B_2$: as $\wtd W_2$ denotes 
    previously received values, we take a context of $\vrelate$-related values.
    Again, received values could be either trigger and characteristic values
    (required to be observed by MST bisimilarity, cf. \Cref{mst:d:fwb}) or
    pure values originated from internal actions. Again, by $\vrelate$
    (\Cref{mst:def:vrelate}) we account for both cases.%  This explains why
    % $\valuesrelatev$ relies on  \Cref{mst:def:valuesset}: we enumerate over all
    % possible substitutions for a value, because we do not know the exact
    % substitutions for  values exchanged in internal actions. 

    In sub-case (ii), when $u_i$ is a tail-recursive name, the elements are as follows: 
    \begin{align*}
      &\big\{ \abbout{\prop^u}{M^{\tilde B_2}_{V_2}}   \Par
          \B{k}{\tilde w}{Q},\ 
           \appl{M^{\tilde B_2}_{V_2}}{\wtd u} \Par \B{k}{\tilde w}{Q},\ 
           {\bbout{u_{\indT{S}}}{V_2}}{} 
          (\apropout{k} {\wtd B_2} \Par 
          \recprovx{u}{x}{\wtd u}
          % \binp{\prop^u}{b}(\appl{b}{\wtd u})
          ) 
      \Par \B{k}{\tilde w}{Q} \\
          &\qquad : V_1 \subst{\tilde W_1}{\tilde y} \vrelate V_2,\ 
          \wtd W_2 \vrelate \wtd B_2 \big\} \\
          &  \text{where }  
          M^{\tilde B}_{V} = \abs{\wtd z}
          {\bbout{z_{\indT{S}}}{V}}{}  
          (\apropout{k}{\wtd B} \Par 
          \recprovx{u}{x}{\wtd z}
          % \binp{\prop^u}{b}(\appl{b}{\wtd z})
          )
    \end{align*} 

    The first element is a process obtained by the activation from the preceding
    trio. The second element is a result of a communication of the first element
    with top-level provider $\Rb{u_i}$ (\Cref{mst:def:rec-providers}) on channel  $\prop^u$.
    By this synchronization, the decomposition of recursive name $u$, that is 
    $\wtd u$, is gathered in application $\appl{M^{\tilde B_2}_{V_2}}{\wtd u}$.  
%    
%    the application
%    of $M^{\tilde B_2}_{V_2}$ to $\wtd u$, 
%    the decomposition of name $u_i$. 
%    
    Finally, the third
    element represents the result of the application: it is a process ready to
    mimic the original output action on $u_{\indT{S}}$. Differently from sub-case
    (i), here we do not have to increment index of $u_i$ in
    $Q$ and $V_1$ as indices of recursive names are obtained based on the type
    $S$, that is $\indT{S}$. 

  \item[Input:] 
  % The $\Dset$-set of $\binp{u_i}{y}Q$ is defined as similarly as 
  % in the output case; here we just need to expand the context with $y$: 
  The $\Dset$-set of $\binp{u_i}{y}Q$ depends on whether 
  (i)~$u_i$ is linear or shared name (i.e., $\neg \tr(u_i)$) or 
  (ii)~$u_i$ is a tail-recursive name (i.e., $\tr(u_i)$). 
  In both sub-cases $\Dset$-set is defined similarly to 
  the output case, with only one caveat: we need to expand the context for the continuation with a newly received value $y$. 
  The $\Dset$-set in sub-case (i) is defined as follows: 
   \begin{align*}
    \begin{tabular}{l}
      \noalign{\smallskip}
             $\big\{\binp{u_i}{y}\apropout{k}{\wtd By} 
             \Par 
             \B{k}{\tilde xy}{Q\sigma} :$ 
             $\wtd W \vrelate \wtd B \big\}$
        \smallskip 
      \end{tabular}
  \end{align*}
\noindent where $\sigma = \subsqn{u_i}$.
  The $\Dset$-set in sub-case (ii) is defined as follows: 
  \begin{align*}
    \begin{tabular}{l} 
      \noalign{\smallskip}
     $\big\{  \abbout{\prop^u}
            {M^{\tilde B}_y} 
            \Par \B{k}{\tilde xy}{Q},\
             \appl{M^{\tilde B}_y}{\wtd u} \Par \B{k}{\tilde xy}{Q},\ 
             { \binp{u_{\indT{S}}}{y}
             (\apropout{k}{\wtd By}
          \Par 
          \recprovx{u}{x}{\wtd z}
          % \binp{\prop^u}{b}(\appl{b}{\wtd z})
          )
          \Par \B{k}{\tilde xy}{Q}} : \wtd W \vrelate \wtd B \big\}$ 
          % \\
            % $ : \wtd W \vrelate \wtd B \}$
            \\
            where:
            \\
            $M^{\tilde B}_y = 
            \abs{\widetilde z}{\binp{z_{\indT{S}}}{y}
            (\apropout{k}{\wtd By} \Par 
            \recprovx{u}{x}{\wtd z}
            % \binp{\prop^u}{b}(\appl{b}{\wtd z})
            )}.$
            \smallskip 
    \end{tabular} 
  \end{align*} 
  The elements of the set represent steps of obtaining 
  name $u_{\indT{S}}$, along which the original action is mimicked, 
%  along which the original action is mimicked 
%  by 
%  gathering a full decomposition of recursive name $u$
   by synchronizing with the
  top-level provider $\Rb{u_i}$, obtained in the corresponding $\mathcal{C}$-set.

  \item[Application:] The $\Dset$-set of ${\appl{V_1}{(\wtd r, u_i)}}$
    where  $\wtd r$ are tail-recursive names, is a union of two sets as follows: 
    \begin{align*}
              & \big\{   \overbracket{\prop^{r_{l}}!\big\langle 
              \lambda \wtd z_{l}. \prop^{r_{{l}+1}}!\langle\lambda \wtd z_{{l}+1}.\cdots. 
          \prop^{r_n}!\langle \lambda \widetilde z_n.}^{|\tilde r| - {l} +1} 
          Q_l \rangle \,\rangle \big\rangle,\ 
           \\
             &
            \quad (\lambda \wtd z_l.
          \overbracket{
              \prop^{r_{l+1}}!\langle\lambda \wtd z_{l+1}.\cdots. 
          \prop^{r_n}!\langle \lambda \wtd z_n.}^{|\tilde r| - l} 
               Q_l \rangle \big\rangle) \ {\wtd r_l} 
          : 1 \leq l \leq n,
           ~~V_1\subst{\tilde W}{\tilde x} \vrelate V_2 \big\} 
          \\
          & \cup 
        %   \\
        %  & 
         \big\{\appl{V_2}{\wtd r_1, \ldots, \wtd r_n, \wtd m}
          : V_1\subst{\tilde W}{\tilde x} \vrelate V_2 \big\} \\
          & \text{where:} \\
          & \qquad Q_l = \appl{V_2}{(\wtd r_1,\ldots,\wtd r_{l-1}, \wtd z_{l}, \ldots, 
          \wtd z_n, \wtd m)}
    \end{align*}
    The first set contains intermediate processes emerging while collecting recursive
    names using synchronizations with  
    recursive name providers. We can see that the body of the inner-most abstraction,
    $Q_l$, is an application of $V_2$ (such that 
    $V_1\subst{\tilde W}{\tilde y} \vrelate V_2$) 
    to partially instantiated recursive names:  
    $l$ denotes that decompositions of first $l-1$ recursive names are
    retrieved. The final tuple in arguments of $Q_l$, $\wtd m = (u_i,\ldots,
    u_{i+\len{\Gt{C}} - 1})$, is a full decomposition of non-recursive   
    (linear or shared) name $u_i$.
    Just like in the previous cases, by taking   $V_2$ as  a 
    $\vrelate$-related value to $V_1\subst{\tilde W}{\tilde y}$, 
    we uniformly handle all the three possibilities for $V_1$ (pure value, variable,  
  and characteristic value).

    In the first set, the first element is a process is ready to send an
    abstraction to an appropriate name provider, in order to retrieve the decomposition of
    $l$-th recursive name.
    The second element is a process that results from a
    communication of the first element with a provider: an application which
    will instantiate $l$-th recursive name in $Q_l$. 
    % We can see the body of the inner-most abstraction, $Q_l$, 
    % is an application where 
    % recursive names are partially instantiated: $l$ denotes that 
    % decompositions of 
    % first $l-1$ recursive names are already retrieved. 
    Finally, the second set contains application processes in which
    the decompositions of all $n$ recursive names are gathered, 
    % $\wtd m =
    % (u_i,\ldots, u_{i+\len{\Gt{C}} - 1})$ is the decomposition of name $u_i$,
    and it is ready to mimic the silent action (application reduction) of the original
    process.

  \item[Parallel composition:] The $\Dset$-set of $Q_1 \Par Q_2$ is defined using two sets:
  \begin{align*}
  \begin{tabular}{l}
    \noalign{\smallskip}
    $\big\{ \propout{k}{\wtd B_1}
    \apropout{k+l}{\wtd B_2} \Par \B{k}{\tilde y}{Q_1} \Par \B{k+l}{\tilde z}{Q_2} : $ 
    $\wtd W_1 \vrelate \wtd B_1$, $\wtd W_2 \vrelate \wtd B_2  \big\}$ \\ 
    $\cup$ \\
    $ \big\{(R_1 \Par R_2) : R_1 \in
    \Cb{\tilde W_1}{\tilde y}{Q_1}, R_2 \in 
    \Cb{\tilde W_2}{\tilde z}{Q_2} \big\}$ 
       \smallskip 
    \end{tabular}
  \end{align*}

  \noindent The first set contains a control trio that is ready to activate the
  decomposition of the two components in parallel.
  Just like in the other cases, the control trio propagates values that are $\vrelate$-related to ${\tilde W_1}$ and ${\tilde W_2}$.
  In order to close the set with
  respect to the $\tau$-actions on propagators, the second set  contains the
  composition of processes drawn from the $\mathcal{C}$-sets of $Q_1$ and
  $Q_2$, with appropriate substitutions. 
\end{description}

\subsection{Proving Operational Correspondence}
\label{mst:sec:opcorrproof}
Recall that we aim to establish \Cref{mst:t:dyncorr}. 
To that end, we prove that $\relS$ (\Cref{mst:d:relation-s}) is an MST bisimulation, by establishing two results:
\begin{itemize}
\item \Cref{mst:lemma:mstbs2} %(Page~\pageref{lemma:mstbs2}) 
covers the case in which the given process performs an action,
which is matched by an action of the decomposed process.
In terms of operational correspondence (see, e.g.,~\cite{DBLP:journals/iandc/Gorla10}), this establishes \emph{completeness} of the decomposition.
\item \Cref{mst:lemma:mstbs3} %(Page~\pageref{lemma:mstbs3}) 
covers the converse direction, in which the decomposed process performs an action, which is matched by the initial process.
  This established the \emph{soundness} of the decomposition.
\end{itemize}
  
For proving both operational completeness and soundness, we will need the following result.
Following Parrow~\cite{DBLP:conf/birthday/Parrow00},  we refer to prefixes 
that do not correspond to prefixes of the original process, i.e. prefixes on propagators $\prop_i$,
as \emph{non-essential prefixes}. 
Then the relation $\relS$ is closed under reductions that involve non-essential prefixes.
\begin{restatable}[]{lemm}{tauclosed}
  \label{l:c-prop-closed}
  Given an indexed process $P_1\subst{\tilde W}{\tilde x}$, 
  the set $\Cb{\tilde W}{\tilde x}{P_1}$ is closed under $\tau$-transitions on
  non-essential prefixes. That is, 
   if $R_1 \in \Cb{\tilde W}{\tilde x}{P_1}$ 
  and $R_1 \by{\tau} R_2$ is inferred from the actions on non-essential prefixes, 
  then $R_2 \in \Cb{\tilde W}{\tilde x}{P_1}$.
\end{restatable}
\begin{proof} 
  By the induction on the structure of $P_1$. 
  See~\Cref{mst:app:tauclosed} for more details. 
\end{proof}

\paragraph{Operational Completeness.}
% \subsubsection{Operational Completeness}

 We first consider transitions using the unrestricted and untyped LTS; in \Cref{mst:lemma:mstbs2} we will consider transitions with the refined LTS.

 \begin{restatable}[]{lemm}{mstbs}
   \label{lemma:mstbs1}
  Assume  $P_1\subst{\tilde W}{\tilde x}$ is a process such that   
  $\Gamma_1; \Lambda_1; \Delta_1 \proves 
  P_1\subst{\tilde W}{\tilde x} \hastype \Proc$ 
   with $\balan{\Delta_1}$
   and $P_1\subst{\tilde W}{\tilde x} \relS Q_1$. 
   \begin{enumerate}
     \item 
 				Whenever 
         ${P_1\subst{\tilde W}{\tilde x}}
         {\by{\news{\wtd m_1} \bactout{n}{V_1}}}{P_2}$ 
         , such that $\dual n \not\in \fn{P_1\subst{\tilde W}{\tilde x}}$, 
 				then there exist
 				$Q_2$ and $V_2$ such that 
         ${Q_1}{\By{\news{\wtd m_2}\bactout{\iname n}{V_2}}}{Q_2}$
         % where $V_1 \vrelate V_2$
 				and, for a fresh $t$,
 				\[
          {\newsp{\wtd m_1}{P_2 \parallel \htrigger{t}{V_1}}}
          \subst{\tilde W}{\tilde x}
 	 				 \relS
            {\newsp{\wtd m_2}{Q_2 \parallel \htrigger{t_1}{V_2}}}
         \]
     \item	
       Whenever ${P_1\subst{\tilde W}{\tilde x}}
       {\by{\abinp{n}{V_1}}}{P_2}$
       , such that $\dual n \not\in \fn{P_1\subst{\tilde W}{\tilde x}}$,
       then there exist $Q_2$,  $V_2$, and $\sigma$ such that 
       ${Q_1}{\By{\abinp{\iname n}{V_2}}}{Q_2}$
 		where $V_1 \sigma \vrelate V_2$
       and
       ${P_2}{ \relS }{Q_2}$, 
       \item	
        Whenever ${P_1}{\by{\tau}}{P_2}$ 
        then there exists $Q_2$ such that 
 				${Q_1}{\By{\tau}}{Q_2}$
 				and
         ${P_2}{\relS }{Q_2}$. 
   \end{enumerate} 
 \end{restatable}

 \begin{proof} 
  By transition induction. See~\Cref{mst:app:mstbs1} for more details. 
 \end{proof}

The following statement builds upon the previous one to address the case of the
typed LTS (\Cref{mst:def:mlts}):
% TYPED BISIMULATION PROOF 
\begin{restatable}[]{lemm}{mstbstyped}
 \label{mst:lemma:mstbstyped}
 Assume $P_1 \subst{\tilde W}{\tilde x}$ is a 
%  well-formed 
 process and $P_1\subst{\tilde W}{\tilde x} \relS Q_1$.
 \begin{enumerate}
   \item 
				Whenever 
       $\horelm{\Gamma_1;\Lambda_1;\Delta_1}
       {P_1\subst{\tilde W}{\tilde x}}
       {\by{\news{\widetilde{m_1}} \bactout{n}{V_1}}}
       {\Lambda'_1;\Delta'_1}{P_2}$ 
				then there exist
				$Q_2$, $V_2$, $\Delta_2'$, and $\Lambda_2'$ such that 
       $\horelm{\Gamma_2;\Lambda_2;\Delta_2}{Q_1}
       {\By{\news{\widetilde{m_2}}\bactout{\iname{n}}{V_2}}}
       {\Lambda'_2;\Delta_2'}{Q_2}$
				and, for a fresh $t$,
				\[
         % \Gamma_1; \Lambda_1; \Delta''_1 \proves 
        {\newsp{\widetilde{m_1}}{P_2 \parallel \htrigger{t}{V_1}}}
        \subst{\tilde W}{\tilde x}
	 				 \relS
          {\newsp{\widetilde{m_2}}{Q_2 \parallel \htrigger{\tni}{V_2}}}
       \]
   \item	
     Whenever 
     $\horelm{\Gamma_1; \Lambda_1; \Delta_1}{P_1\subst{\tilde W}{\tilde x}}
      {\by{\abinp{n}{V_1}}}{\Lambda_1';\Delta_1'}{P_2}$
     then there exist $Q_2$,  $V_2$, $\sigma$, $\Lambda'_2$, and $\Delta'_2$ such that 
%      ${Q_1}{\By{\abinp{\iname n}{V_2}}}{Q_2}$
$\horelm{\Gamma_2;\Lambda_2;\Delta_2}{Q_1}{\By{\abinp{\iname{n}}{V_2}}}
      {\Lambda'_2, \Delta_2'}{Q_2}$ 
		where $V_1 \sigma \vrelate V_2$
     and
     ${P_2}{ \relS }{Q_2}$, 

     \item	
      Whenever 
       $\horelm{\Gamma_1;\Lambda_1;\Delta_1}{P_1\subst{\tilde W}{\tilde x}}{\by{\tau}}
       {\Lambda_1'; \Delta_1'}{P_2}$
      then there exist $Q_2$, $\Lambda_2'$, and $\Delta_2'$ such that 
%				${Q_1}{\By{\tau}}{Q_2}$
				$\horelm{\Gamma_2;\Lambda_2;\Delta_2}{Q_1}
       {\By{\tau}}{\Lambda_2';\Delta_2'}{Q_2}$
				and
       ${P_2}{\relS }{Q_2}$. 
 \end{enumerate} 
\end{restatable}
\begin{proof} 
  The proof uses results of \Cref{lemma:mstbs1}.
  We consider the first case, the other two being similar.

  By the definition of the typed LTS we have: 
    \begin{align}
      &\Gamma_1; \Lambda_1; \Delta_1 \proves P_1\subst{\tilde W}{\tilde x} 
      \label{mst:pt:cr-typed-1}
      \\
      &(\Gamma_1; \es; \Delta_1)  \by{\news{\widetilde{m}} \bactout{n}{V}}
      (\Gamma_1; \es; \Delta_2)
      \label{mst:pt:cr-typed-2}
    \end{align}
    By \eqref{mst:pt:cr-typed-2} we further have 
    \begin{align*}
      \AxiomC{
        \begin{tabular}{c}
          $\Gamma \cat \Gamma'; \Lambda'; \Delta' \proves V \hastype U$ \\
      $\Delta'\backslash (\cup_j \Delta_j) \subseteq (\Delta \cat n: S)$
        \end{tabular}} 
      \AxiomC{\begin{tabular}{c}
        $\Gamma'; \es; \Delta_j \proves m_j  \hastype U_j$ \\ 
        $\Gamma'; \es; \Delta_j' \proves \dual{m}_j  \hastype U_j'$
      \end{tabular}} 
      \AxiomC{
        \begin{tabular}{c}
          $\dual{n} \notin \dom{\Delta}$ \\
          $\Lambda' \subseteq \Lambda$
        \end{tabular}
      } 
      \LeftLabel{\scriptsize \trule{SSnd}}
      \TrinaryInfC{
        $(\Gamma; \Lambda; \Delta \cat s: \btout{U} S)
        \by{\news{\widetilde{m}} \bactout{n}{V}}
        (\Gamma \cat \Gamma'; \Lambda\backslash\Lambda';
         (\Delta \cat n: S \cat \cup_j \Delta_j') \backslash \Delta')$
      }
      \DisplayProof 
    \end{align*}
    \noindent By 
    \eqref{mst:pt:cr-typed-1} 
    and the condition $\dual{n} \notin \dom{\Delta}$ 
    we have $\dual n \not\in \fn{P_1\subst{\tilde W}{\tilde x}}$. 
    Therefore, we can apply Item 1 of \Cref{lemma:mstbs1}. 
\end{proof} 

Finally, we are in a position to address the case of the
refined typed LTS (\Cref{mst:def:mlts}):
\begin{restatable}[]{lemm}{mstbs2}
  \label{mst:lemma:mstbs2}
     Assume $P_1\subst{\tilde W}{\tilde x}$ is a 
    %  well-formed 
    process 
      and $P_1\subst{\tilde W}{\tilde x} \relS Q_1$.
      \begin{enumerate}
        \item 
             Whenever 
            $\horelm{\Gamma_1;\Lambda_1;\Delta_1}
            {P_1\subst{\tilde W}{\tilde x}}
            {\hby{\news{\widetilde{m_1}} \bactout{n}{V_1}}}
            {\Lambda'_1;\Delta'_1}{P_2}$ 
             then there exist
             $Q_2$, $V_2$, $\Delta_2'$, and $\Lambda_2'$ such that 
            $\horelm{\Gamma_2;\Lambda_2;\Delta_2}{Q_1}
            {\Mhby{\news{\widetilde{m_2}}\bactout{\iname{n}}{V_2}}}
            {\Lambda'_2;\Delta_2'}{Q_2}$
             and, for a fresh $t$,
             \[
             {\newsp{\widetilde{m_1}}{P_2 \parallel \htrigger{t}{V_1}}}
             \subst{\tilde W}{\tilde x}
                 \relS
               {\newsp{\widetilde{m_2}}{Q_2 \parallel \htrigger{\tni}{V_2}}}
            \]
        \item	
          Whenever 
          $\horelm{\Gamma_1; \Lambda_1; \Delta_1}{P_1\subst{\tilde W}{\tilde x}}
          {\hby{\abinp{n}{V_1}}}{\Lambda_1';\Delta_1'}{P_2}$ then there exist
          $Q_2$,  $V_2$, $\Lambda'_2$, and $\Delta'_2$ such that
          $\horelm{\Gamma_2;\Lambda_2;\Delta_2}{Q_1}{\Mhby{\abinp{\iname{n}}{V_2}}}
          {\Lambda'_2, \Delta_2'}{Q_2}$ where $V_1 \valuesrelate V_2$
          and
          ${P_2}{ \relS }{Q_2}$, 
     
          \item	
           Whenever 
            $\horelm{\Gamma_1;\Lambda_1;\Delta_1}
            {P_1\subst{\tilde W}{\tilde x}}{\hby{\tau}}
            {\Lambda_1'; \Delta_1'}{P_2}$
           then there exist $Q_2$, $\Lambda_2'$, and $\Delta_2'$ such that 
     %				${Q_1}{\By{\tau}}{Q_2}$
             $\horelm{\Gamma_2;\Lambda_2;\Delta_2}{Q_1}
            {\Mhby{\tau}}{\Lambda_2';\Delta_2'}{Q_2}$
             and
            ${P_2}{\relS }{Q_2}$. 
      \end{enumerate}

  \end{restatable}
  \begin{proof}
    By case analysis of the transition label $\ell$. 
    It uses results of \Cref{mst:lemma:mstbstyped}. 
    We consider two cases: 
    \rom{1} $\ell \equiv \abinp{n}{V_1}$ and 
    \rom{2} $\ell \not\equiv \abinp{n}{V_1}$.  
    \begin{itemize}
      \item[\rom{1}] Case $\ell \equiv \abinp{n}{V_1}$. This case concerns Part (2) of the lemma. 
      In this case we know $P_1 = \binp{n}{y}Q$.  
      We have the following 
transition inference tree: 
\begin{align}
	\AxiomC{} 
	\LeftLabel{\scriptsize $\rulename{Rv}{}$}
  \UnaryInfC{$(\binp{n}{y}P_2) \subst{\tilde W}{\tilde x} \by{\abinp{n}{V_1}} 
  P_2$}
  \label{mst:pt:reflts-inp1}
  \DisplayProof  	
\end{align}
\begin{align}
	\AxiomC{\eqref{mst:pt:reflts-inp1}} 
  \AxiomC{$V_1 \scong \omapchar{U}
  \vee V_1  \scong \abs{{x}}{\binp{t}{y} (\appl{y}{{x}})}
  \textrm{ {\small $t$ fresh}} $} 
  \LeftLabel{\scriptsize $\rulename{RRcv}{}$}
  \BinaryInfC{$(\binp{n}{y}P_2) \subst{\tilde W}{\tilde x} \hby{\abinp{n}{V_1}} 
  P_2$}
  \DisplayProof  	
  \label{mst:pt:reflts-inp2}
\end{align}

From \eqref{mst:pt:reflts-inp1} and \Cref{mst:lemma:mstbstyped} we know that there exist $Q_2$, and $V_2$ such that 
$Q_1 \By{\abinp{\iname n}{V_2}}Q_2$ and $P_2 \relS Q_2$ where $V_1\sigma \vrelate V_2$.
Since $V_1$ is a characteristic or a trigger value, we have $V_1 \valuesrelate V_2$ and that $V_2$ is  a minimal characteristic or a trigger value.
Hence, $Q_1 \Mhby{\abinp{\iname n}{V_2}}Q_2$ using the Rule~\textsc{MTr} (\Cref{mst:def:mlts}).

\item Case $\ell \not\equiv \abinp{n}{V}$. This case concerns Parts (1) and (3) of the lemma. 
We only consider the first part, when $\ell \equiv \news{\wtd m_1}\about{n}{V_1}$, since the other part is similar.

We apply \Cref{lemma:mstbs1} to obtain $Q_2$ such that 
      ${Q_1}{\Hby{\news{\widetilde{m_2}}\bactout{\iname n}{V_2}}}{Q_2}$, and, for a fresh $t$,
      % and, for a fresh $t$,
%
      \[
       {\newsp{\widetilde{m_1}}{P_2 \parallel \htrigger{t}{V_1}}}
       \subst{\tilde W}{\tilde x}
         \ \relS\
         {\newsp{\widetilde{m_2}}{Q_2 \parallel \htrigger{\tni}{V_2}}}.
      \]
Since we are dealing with an output action, we can immediately conclude that
      ${Q_1}{\Mhby{\news{\widetilde{m_2}}\bactout{\iname n}{V_2}}}{Q_2}$. 
    \end{itemize}
\end{proof}

\paragraph{Operational Soundness.}
% \subsubsection{Operational Soundness}
For the proof of operational soundness we follow the same strategy of stratifying it into three lemmas.

  \begin{restatable}[]{lemm}{mstbsrev}
    \label{mst:lemma:mstbs3}
% \begin{lemma}
% \label{mst:lemma:mstbs3}
  Assume $P_1\subst{\tilde W}{\tilde x}$ is a 
  % well-formed 
  process and $P_1\subst{\tilde W}{\tilde x} \relS Q_1$.
  \begin{enumerate}
    \item 
        Whenever 
        ${Q_1}{\by{\news{\wtd{m_2}}\bactout{n_i}{V_2}}}{Q_2}$ 
        , such that $\dual n_i \not\in \fn{Q_1}$, 
				then there exist 
				$P_2$ and $V_2$  such that  
        ${P_1\subst{\tilde W}{\tilde x}}{\by{\news{\wtd{m_2}} 
        \bactout{n}{V_2}}}{P_2}$ 
				and, for a fresh $t$,
				\[
         {\newsp{\wtd{m_1}}{P_2 \parallel \htrigger{t}{V_1}}}
         \subst{\tilde W}{\tilde x}
	 				 \relS 
           {\newsp{\wtd{m_2}}{Q_2 \parallel \htrigger{\tni}{V_2}}}.
        \]
        \item	
        Whenever 
        ${Q_1}{\by{\abinp{n_i}{V_2}}}{Q_2}$
        , such that $\dual n_i \not\in \fn{Q_1}$, 
        there exist $P_2$, $V_2$, and $\sigma$ such that 
        ${P_1\subst{\tilde W}{\tilde x}}
        {\by{\abinp{n}{V_1}}}{P_2}$ 
        where $V_1\sigma \vrelate V_2$
        and
        ${P_2}{\relS }{Q_2}$. 

    \item	
      Whenever 
      $Q_1{\by{\tau}}{~Q_2}$
      either (i) $P_1\subst{\tilde W}{\tilde x} \relS Q_2$ or (ii)
       there exists $P_2$ such that 
      ${P_1}{\by{\tau}}{P_2}$
      and
      ${P_2}{\relS}{Q_2}$. 
  \end{enumerate} 
% \end{lemma}
\end{restatable}
\begin{proof}[Proof (Sketch)] 
  By transition induction.
  See~\Cref{mst:app:mstbs3} for more details. 
\end{proof} 
\begin{lemm}
\label{mst:lemma:mstbs5}
  Assume $P_1\subst{\tilde W}{\tilde x}$ is a 
  % well-formed 
  process and $P_1\subst{\tilde W}{\tilde x} \relS Q_1$.
  \begin{enumerate}
   \item 
				Whenever 
       $\horelm{\Gamma_2;\Lambda_2;\Delta_2}
       {Q_1}
       {\by{\news{\widetilde{m_2}} \bactout{\iname{n}}{V_2}}}
       {\Lambda'_2;\Delta'_2}{Q_2}$ 
				then there exist
				$P_2$, $V_1$, $\Delta_1'$, and $\Lambda_1'$ such that 
       $\horelm{\Gamma_1;\Lambda_1;\Delta_1}{P_1\subst{\tilde W}{\tilde x}}
       {\By{\news{\widetilde{m_1}}\bactout{n}{V_1}}}
       {\Lambda'_1;\Delta_1'}{P_2}$
				and, for a fresh $t$,
				\[
         % \Gamma_1; \Lambda_1; \Delta''_1 \proves 
        {\newsp{\widetilde{m_1}}{P_2 \parallel \htrigger{t}{V_1}}}
        \subst{\tilde W}{\tilde x}
	 				 \relS
          {\newsp{\widetilde{m_2}}{Q_2 \parallel \htrigger{t_1}{V_2}}}
       \]
   \item	
     Whenever 
     $\horelm{\Gamma_2; \Lambda_2; \Delta_2}{Q_1}
      {\by{\abinp{\iname{n}}{V_2}}}{\Lambda_2';\Delta_2'}{Q_2}$
     then there exist $P_2$,  $V_1$, $\sigma$, $\Lambda'_1$, and $\Delta'_1$ such that 
%      ${Q_1}{\By{\abinp{\iname n}{V_2}}}{Q_2}$
$\horelm{\Gamma_1;\Lambda_1;\Delta_1}{P_1\subst{\tilde W}{\tilde x}}{\By{\abinp{{n}}{V_1}}}
      {\Lambda'_1, \Delta'_1}{P_2}$
		where $V_1 \sigma \vrelate V_2$
     and
     ${P_2}{ \relS }{Q_2}$, 

     \item	
      Whenever 
       $\horelm{\Gamma_2;\Lambda_2;\Delta_2}{Q_1}{\by{\tau}}
       {\Lambda_2'; \Delta_2'}{Q_2}$
      then 
      either (i) $P_1\subst{\tilde W}{\tilde x} \relS Q_2$ 
      or (ii) 
      there exist $P_2$, $\Lambda_1'$, and $\Delta_1'$ such that 
%				${Q_1}{\By{\tau}}{Q_2}$
				$\horelm{\Gamma_1;\Lambda_1;\Delta_1}{P_1\subst{\tilde W}{\tilde x}}
       {\by{\tau}}{\Lambda_1';\Delta_1'}{P_2}$
				and
       ${P_2}{\relS }{Q_2}$. 
  \end{enumerate} 
\end{lemm}
\begin{lemm}
\label{mst:lemma:mstbs4}
  Assume $P_1\subst{\tilde W}{\tilde x}$ is a 
  % well-formed 
  process and $P_1\subst{\tilde W}{\tilde x} \relS Q_1$.
  \begin{enumerate}
   \item 
				Whenever 
       $\horelm{\Gamma_2;\Lambda_2;\Delta_2}
       {Q_1}
       {\hby{\news{\widetilde{m_2}} \bactout{\iname{n}}{V_2}}}
       {\Lambda'_2;\Delta'_2}{Q_2}$ 
				then there exist
				$P_2$, $V_1$, $\Delta_1'$, and $\Lambda_1'$ such that 
       $\horelm{\Gamma_1;\Lambda_1;\Delta_1}{P_1\subst{\tilde W}{\tilde x}}
       {\Hby{\news{\widetilde{m_1}}\bactout{n}{V_1}}}
       {\Lambda'_1;\Delta_1'}{P_2}$
				and, for a fresh $t$,
				\[
         % \Gamma_1; \Lambda_1; \Delta''_1 \proves 
        {\newsp{\widetilde{m_1}}{P_2 \parallel \htrigger{t}{V_1}}}
        \subst{\tilde W}{\tilde x}
	 				 \relS
          {\newsp{\widetilde{m_2}}{Q_2 \parallel \htrigger{t_1}{V_2}}}
       \]
   \item	
     Whenever 
     $\horelm{\Gamma_2; \Lambda_2; \Delta_2}{Q_1}
      {\hby{\abinp{\iname{n}}{V_2}}}{\Lambda_2';\Delta_2'}{Q_2}$
     then there exist $P_2$,  $V_1$,$\Lambda'_1$, and $\Delta'_1$ such that 
%      ${Q_1}{\By{\abinp{\iname n}{V_2}}}{Q_2}$
$\horelm{\Gamma_1;\Lambda_1;\Delta_1}{P_1\subst{\tilde W}{\tilde x}}{\Hby{\abinp{{n}}{V_1}}}
      {\Lambda'_1, \Delta'_1}{P_2}$
		where $V_1 \valuesrelate V_2$
     and
     ${P_2}{ \relS }{Q_2}$, 

     \item	
      Whenever 
       $\horelm{\Gamma_2;\Lambda_2;\Delta_2}{Q_1}{\hby{\tau}}
       {\Lambda_2'; \Delta_2'}{Q_2}$
      then 
      either (i) $P_1\subst{\tilde W}{\tilde x} \relS Q_2$ 
      or (ii) 
      there exist $P_2$, $\Lambda_1'$, and $\Delta_1'$ such that 
%				${Q_1}{\By{\tau}}{Q_2}$
				$\horelm{\Gamma_1;\Lambda_1;\Delta_1}{P_1\subst{\tilde W}{\tilde x}}
       {\hby{\tau}}{\Lambda_1';\Delta_1'}{P_2}$
				and
       ${P_2}{\relS }{Q_2}$. 
  \end{enumerate} 
\end{lemm}

\paragraph{Summary.}
Together, 
% \Cref{mst:lemma:mstbs2, mst:lemma:mstbs3} 
\Cref{mst:lemma:mstbs2} and \Cref{mst:lemma:mstbs3} 
imply that $\relS$ is an MST-bisimilarity.
In summary, we have shown \Cref{mst:t:dyncorr}, i.e., that for any typed process $P$, we have that
  $$\horelm{\Gamma;\Lambda;\Delta}{P}{\ \mstb \ }
  {\Gt{\Gamma};\Gt{\Lambda};\Gt{\Delta}}{\D{P}}.$$

In this section we have defined a notion of MST bisimilarity, following the notion \HO bisimilarity for non-minimal processes.
Following the strategy of Parrow in the untyped setting, we defined a relation $\relS$ containing all pairs $(P, \D{P})$, which we proved to be an MST bisimulation.

%%% Local Variables:
%%% mode: latex
%%% TeX-master: "mst"
%%% End:

\section{Optimizations of the Decomposition}
\label{s:opt}
In this section we discuss two optimizations that can be applied to the decomposition process.
These optimizations simplify the structure of the trios and the nature of the underlying communication discipline.

The first optimization replaces trios in the decomposition with \emph{duos} (i.e., processes with two sequential prefixes).
The decomposition in \Cref{s:decomp} follows Parrow's approach in that it converts a process into a parallel composition of trios.
The use of trios seems to be necessary in (plain) $\pi$-calculus; in our first optimization we show that, by exploiting the higher-order nature of communications in \HO, the trios can be replaced by duos. 

The second optimization replaces polyadic communications (sending and receiving several values at once) with monadic communications (sending and receiving only a single value per prefix).
In the decomposition, we use polyadic communications in order to propagate dependencies through sub-processes.
We show that the use of monadic communication prefixes is sufficient for that task.

\paragraph{From Trios to Duos.}
In the first optimization we replace trios with \emph{duos}, i.e., processes with two sequential prefixes.
This optimization is enabled   by the higher-order nature of  $\HO$.
In the translation we make of \emph{thunk processes}, i.e., inactive processes that can be activated upon reception.
We write $\thunk{P}$ to stand for the thunk process $\abs{x:\chtype{\shot{\tinact}}}P$, for a fresh $x \not\in \fn{P}$.
We write $\appthunk{\thunk{P}}$ to denote the application of a thunk to a (dummy) name of type $\shot{\tinact}$.
This way, we have a reduction $\appthunk{\thunk{P}} \red P$.

The key idea behind replacing trios with duos is to transform a trio like $$\binp{c_k}{\widetilde x}\bout{u}{V}\about{c_{k+1}}{\widetilde z}{}$$
into the composition of two duos, the second one being a ``control'' duo:
\begin{equation}
  \binp{c_k}{\widetilde x}\abbout{c_{0}}{\thunk{\bout{u}{V}\about{c_{k+1}}{\widetilde z}}} \Par \binp{c_{0}}{y}{(\appthunk{y})}
  \label{mst:eq:duo}
\end{equation}
The first action (on $c_k$) is as before; the two remaining prefixes (on $u$ and $c_{k+1}$) are encapsulated into a thunk.
This thunk is sent via an additional propagator  (denoted $c_0$) to the control duo that activates it upon reception.
Because of this additional propagator, this transformation involves minor modifications in the definition of the degree function $\plen{-}$ (cf.
\Cref{mst:def:sizeproc}).

In some cases, the breakdown function in \Cref{mst:ss:core} already produces duos.
Breaking down input and output prefixes and parallel composition involves proper trios; following the scheme illustrated by \eqref{mst:eq:duo}, we can define a map $\dmap{-}$ to transform these trios into duos:
	\begin{align*} 
			\dmap{\propinp{k}{\widetilde x}\bout{u_i}{V}
			\apropout{k+1}{\widetilde z}{}} & =
			\propinp{k}{\widetilde x}\apropbout{k+1}{\thunk{\bout{u_i}{V}\apropout{k+2}{\widetilde z}}} \Par \propinp{k+1}{y}{(\appthunk{y})} 
			\\
			\dmap{\propinp{k}{\widetilde x}\binp{u_i}{y}
			\apropout{k+1}{\widetilde x'}{}} & =  
			\propinp{k}{\widetilde x}\apropbout{k+1}{\thunk{\binp{u_i}{y}\apropout{k+2}{\widetilde x'}}} \Par \propinp{k+1}{y}{(\appthunk{y})} 
			\\
			\dmap{\propinp{k}{\widetilde x}\propout{k+1}{\widetilde y}
			\apropout{k+\degree+1}{\widetilde z}} &= 
			\propinp{k}{\widetilde x}\apropbout{k+1}{
			\thunk{
			\propout{k+2}{\widetilde y}  
			\apropout{k+\degree+2}{\widetilde z}}}  \Par 
			\propinp{k+1}{y}{(\appthunk{y})} 
	\end{align*}

	In breaking down prefixes involving tail-recursive names  
  (\Cref{mst:t:bdowncore})
  we encounter trios of the following form: 
  \begin{align*}
    \B{}{}{\binp{u_i}{y}Q} = 
    \propinp{k}{\wtd x}\abbout{\prop^u}{N_y} \Par \B{k+1}{\tilde w}{Q} 
    \quad \text{where} 
    \quad N_y = \abs{\wtd z}{\binp{z_{\indT{S}}}{y}
    \big(\apropout{k+1}{\wtd w} \Par 
    \recprovx{u}{x}{\wtd z}
    %  \binp{\prop^u}{b}(\appl{b}{\wtd z}
     \big))
    }
  \end{align*}
  \noindent Here we can see that the top-level process is a duo and that only
  $N_y$ packs a proper trio. By applying the same idea we can translate
  this trio into the following composition of duos: 
  \begin{align*}
    \dmap{\binp{z_{\indT{S}}}{y}
    \big(\apropout{k+1}{\wtd w} \Par 
    \recprovx{u}{x}{\wtd z}
     \big)} = 
     \binp{z_{\indT{S}}}{y}\apropout{k+1}
     {\thunk{\apropout{k+2}{\wtd w} \Par 
     \recprovx{u}{x}{\wtd z}
     }
     }  
     \Par  
     \propinp{k+1}{y}{\appthunk{y}}
  \end{align*}
  \noindent This is the idea behind the breakdown of a process starting  with an input prefix; the breakdown of a process with an output prefix follows the same lines.

\newcommand{\dummyv}{}
\newcommand{\dummyvar}{}

\paragraph{From Polyadic to Monadic Communication.}
\begin{figure}[!t]
  \begin{mdframed} 
\begin{center} 
  % \xspace
  \hspace*{-30.0pt}
\begin{tikzpicture}
  \tikzstyle{ann} = [draw=none,fill=none,right]
  \pgfmathsetmacro{\h}{2.0}
  \pgfmathsetmacro{\x}{3.0}
  \pgfmathsetmacro{\d}{1.75}
  % \pgfmathsetmacro{\processshape}{circle}
  % \newcommand{\processshape}{circle}

  \node[] (L1) at (-\x*1.5-0.1, 1.5) {Source process $P_1$:}; 

  {\node[draw, \processshape] (P1) at (-\x*1.5, 0.5) {$P_1$};
  % \node at (-1.5,-3) {$\parallel$};
  \node[draw, \processshape] (P2) at (-\x/2.0, 0.5) {$P_2$};
  \node[draw, \processshape] (P3) at (\x/2.0, 0.5) {$P_3$};
  \node[draw, circle] (P4) at (\x*1.5, 0.5) {$\inact$};} 

  {\draw[->, color=black](P1) edge[] node [above] 
    {\footnotesize $u:\abinp{}{\mathsf{str}}$} (P2);}
  {\draw[->, color=black](P2) edge[] node [above] 
    {\footnotesize $u:\abinp{}{\mathsf{int}}$} (P3);}
  {\draw[->, color=black](P3) edge[] node [above] 
    {\footnotesize $u:\about{}{\mathsf{{{bool}}}}$} (P4);}

    \tikzstyle{ann} = [draw=none,fill=none,right]
  \node[align=left] (L2) at (-\x*1.2, -0.5) {Monadic decomposition $\mD{P_1}$:};

    \node[draw, \processshape] (T1) at (-\x*1.5, -\d) {$Q_1$};
    \node[draw, \processshape] (T11) at (-\x*1.5, -\d-\h) {$Q'_1$};

    {\draw[->, color=black](T1) edge[] node [below left] 
        {\footnotesize $u_1:\abinp{}{\mathsf{str}}$} (T11);}

    \node at (-\x, -\d) {$\parallel$};
    \node[draw, \processshape] (T2) at (-\x/2.0, -\d) {$Q_2$};
    \node[draw, \processshape] (T22) at (-\x/2.0, -\d-\h) {$Q'_2$};

    {\draw[->, color=black](T2) edge[] node [below left] 
        {\footnotesize $u_2:\abinp{}{\mathsf{int}}$} (T22);}

    \node at (0, -\d) {$\parallel$};
    \node[draw, \processshape] (T3) at (\x/2.0, -\d) {$Q_3$};
    \node[draw, \processshape] (T33) at (\x/2.0, -\d-\h) {$Q'_3$};
    {\draw[->, color=black](T3) edge[] node [below left] 
        {\footnotesize $u_3:\about{}{\mathsf{{{bool}}}}$} (T33);}

        \node at (\x, -\d) {$\parallel$};

    \node[draw, \processshape] (T4) at (\x*1.5, -\d) {$Q_4$};

    {\draw[->, dashed, color=red, align=left](T11) edge[bend left=25] 
    node [pos=0.05, right] 
      {} (T2);}
      {\draw[->, dashed, color=blue, align=left](T11) 
      edge[bend left=11] node [pos=0.05, right] 
      {\footnotesize{$x:\mathsf{str}$}} (T3);}
    {\draw[->, dashed, color=red, align=left](T22) edge[bend left=25] 
    node [pos=0.05, right] {\footnotesize{}} (T3);}
     {\draw[->, dashed, color=blue, align=left](T22) 
     edge[bend left=5] node [pos=0.05, right] 
     {\footnotesize{$y:\mathsf{int}$}} (T3);}
    {\draw[->, dashed, color=red](T33) edge[bend left=20] node [left] {} (T4);}

    \node[draw=none] {};
  \end{tikzpicture}
\end{center} 
\end{mdframed} 
\caption[Our monadic decomposition function 
$\mD{-}$, illustrated]{Our monadic decomposition function 
$\mD{-}$, illustrated. 
As in \Cref{mst:fig:decomp}, 
nodes 
represent process states, `$\parallel$' represents 
parallel composition of processes, black arrows stand for actions, and red arrows 
indicate synchronizations that preserve the sequentiality of the source process;  
also,  blue arrows indicates synchronizations that propagate (bound) values.}
\label{mst:fig:decomp-monadic}
\end{figure} 
Our second optimization replaces polyadic communications, used for the propagators, with monadic communications.
Recall that propagators in $\B{}{}{-}$ serve two purposes: they \rom{1} encode sequentiality by properly activating trios and \rom{2} propagate bound values.
By separating propagators along those two roles, we can we can dispense with polyadic communication in the breakdown function.

We define \emph{monadic breakdown}, $\mB{}{}{-}$, and \emph{monadic decomposition}, $\mD{-}$, which use two \emph{kinds} of propagators: \rom{1} propagators for only activating trios of form $\prop_k$ (where $k > 0$ is an index) and \rom{2} for propagating bound values of form $\prop_x$ (where $x$ is some variable).
We depict the mechanism of the monadic breakdown in~\Cref{mst:fig:decomp-monadic}.
The main idea is to establish a direct link between trio that binds the variable $x$ and trios that make use of $x$ on propagator channel $\prop_x$.
Thus, propagators on $\prop_k$ only serve to activate next trios: they do so by \emph{receiving} an abstraction that contains the next trio.

Formally, we define a \emph{monadic decomposition}, $\mD{P}$, that simplifies \Cref{mst:def:decomp} as follows:
  \begin{align*}
     \mD{P} =  \news{\widetilde c}\big(
      \apropout{k}{\dummyv} \Par \mB{k}{}{P\sigma}\big)
  \end{align*}
  \noindent where $k > 0$, $\widetilde \prop = (\prop_k,\dots,\prop_{k+\plen{P}-1})$, 
  % $\dummyv=\abs{x:\shotup{\chtype{\tinact}}}\inact$, 
  and the initializing substitution $\sigma  = \subst{\indices{\wtd u}}{\wtd u}$ is the same as in \Cref{mst:def:decomp}.

      The monadic break down function $\mB{k}{}{-}$, given in~\Cref{mst:fig:monadic-breakdown}, simplifies the one in~\Cref{mst:t:bdowncore} by using only one parameter, namely $k$.
      In \Cref{mst:fig:monadic-breakdown} we use $\sigma$ to denote the subsequent substitution $\subsqn{u_i}$, the same as in \Cref{mst:t:bdowncore},
      and use $\wtd m$ to denote the breakdown $(u_i,\ldots,u_{i+\len{\Gt{C}}-1})$ of the name $u_i$.

\begin{figure}[!t]
  % \resizebox{1.01\textwidth}{!}
  % {
\begin{mdframed}
  % \resizebox{1.01\textwidth}{!}
  % {
\begin{align*}
V_{x} &= \abs{y}\binp{y}{z}\bout{\prop_x}{x}
  \news{s}(\appl{z}{s} \Par \about{\dual s}{z})
  \\
W^{\leadsto}_x & = \begin{cases}
        \apropbout{x}{x} & \text{if} \ \leadsto = \multimap  \\ 
                     \news{s}(\appl{V_{x}}{s} \Par \about{\dual s}{V_{x}})  & \text{if} \ \leadsto = \rightarrow 
     \end{cases}  
     \vspace{2mm}
     \\
\mB{k}{}{\binp{u_i}{x:\slhot{C}}Q} & = 
\news{\prop_x}\big( 
\propinp{k}{\dummyvar}\binp{u_i}{x}(\apropout{k+1}{\dummyvar}
\Par W^{\leadsto}_x)  
\big) 
    % \news{\prop_x}\big(\apropbout{k}{\thunk{\binp{u_i}{x}\propinp{k+1}{y}({W^{\leadsto}_x \Par (\appthunk{y})})}} 
    \Par 
    % & \\ 
    \mB{k+1}{}{Q\sigma}\big) 
\\
\mB{k}{}{\bout{u_i}{x}Q} & = 
\propinp{k}{\dummyvar}\propinp{x}{x}\bout{u_i}{x}
\apropout{k+1}{\dummyvar}
% \apropbout{k}{\thunk{\binp{\prop_x}{x} 
%     \bout{u_i}{x}\propinp{k+1}{y}(\appthunk{y})}} 	
    \Par 		\mB{k+1}{}{Q\sigma} 
\\
\mB{k}{}{\bout{u_i}{V}Q} & = 
\propinp{k}{\dummyvar}\bout{u_i}{\mV{}{}{V\sigma}}
\apropout{k+1}{\dummyvar}
\Par 
\mB{k+1}{}{Q\sigma} 
% \apropbout{k}{\thunk{\bout{u_i}{\mV{k+1}{}{V\sigma}}
%     \propinp{k+1}{y}(\appthunk{y})}}
% \Par 
% \mB{k+1}{}{Q\sigma} 
\\
\mB{k}{}{\appl{x}{u_i}} & = 
\propinp{k}{\dummyvar}\propinp{x}{x}\appl{x}{\wtd m}
% \apropbout{k}{\thunk{\propinp{x}{x}(\appl{x}{\wtd m})}}
\\ 
\mB{k}{}{\appl{V}{u_i}} & = 
                          \propinp{k}{}{{\appl{\mV{}{}{V}}{\wtd m}}}
\\
\mB{k}{}{\news{s}{P'}} & = \news{\widetilde s}\mB{k}{}{P'\sigma}	
\\
\mB{k}{}{Q\Par R} & = 
\propinp{k}{\dummyvar}\propout{k+1}{\dummyvar}\apropout{k+\plen{Q}+1}{\dummyvar}
\Par \mB{k+1}{}{Q} \Par \mB{k+\plen{Q}+1}{}{R}  	\\
\mV{}{}{y} &= y \\
\mV{}{}{\abs{y:\slhotup{C}}P} &= 
\abs{(y_1,\ldots,y_{\len{\Gt{C}}}):{\slhotup{\Gt{C}}}}
{\news{\prop_1,\ldots,\prop_{\plen{P}}}
\big(\apropout{1}{\dummyv}
\Par
\mB{1}{}{P \subst{y_1}{y}}}\big)
\end{align*}
  % }
\end{mdframed}	
  % }
\caption{Monadic breakdown of processes and values  
\label{mst:fig:monadic-breakdown}}
\end{figure}

% The breakdown function $\mB{k}{}{-}$ propagates values using thunks and a dedicated propagator $\prop_x$ for each variable $x$.
% {\color{blue} 
The breakdown function $\mB{k}{}{-}$ uses 
propagators $\prop_k$ ($k > 0$) for encoding sequentiality and 
dedicated propagators $\prop_x$ for each variable $x$.
As propagators $\prop_k$ now only serve to encode sequentiality,
only dummy values are being communicated along these channels (see~\Cref{mst:r:prefix}).

Let us describe the breakdown of a process with an input prefix, as it illustrates the key points common to all the other cases. 
The breakdown $\mB{k}{}{\binp{u_i}{x}Q}$ consists of a trio in parallel  
with the breakdown of the continuation $\mB{k+1}{}{Q\sigma}$ 
with name $\prop_x$ restricted. The trio is first activated 
on $\prop_k$. This is followed by 
the prefix that mimics original input action on indexed name $u_i$. 
Upon receiving value $x$, two things will happen in parallel.
First, the next trio will be activated on name $\prop_{k+1}$. Second,
the value $x$ received on $u_i$ is propagated further by the dedicated process $W^{\leadsto}_x$.

The specific mechanism of propagation depends on whether a received value is linear ($\leadsto = \multimap$) or shared ($\leadsto = \rightarrow$).
In the former case, we simply propagate a value along the \emph{linear} name $\prop_x$ once.
In the later case, we cannot propagate the value only once, because a shared variable can be used in multiple trios.
Thus, $W^{\rightarrow}_x$ implements a recursive mechanism that repeatedly sends a value on the \emph{shared} name $\prop_x$.
The recursion is encoded in the same way as in \Cref{mst:e:en-rec}: action $\about{\prop_x}{x}$ is enclosed in value $V$ that gets appropriately duplicated upon a synchronization.

The breakdown function for values, $\mV{}{}{-}$, is 
accordingly changed to invoke $\mB{1}{}{-}$
for breaking down a function body. 

For simplicity, we defined the decomposition of the output process using a subprocess with four prefixes.
Alternatively, we could have used a decomposition that relies on two trios, by introducing abstraction passing as in the previous section.

Let us illustrate the monadic breakdown by the means of an example: 
\begin{example}[Monadic Decomposition]
We again consider process $P = \news{u} (Q \Par R)$ as 
in~\Cref{mst:e:decomp-proc-types} 
where: 
\begin{align*}
  Q &= \binp{u}{x}
  \overbrace{\binp{u}{y}
  \news{s}\big( \appl{x}{\dual s} \Par \about{s}{y} \big)}^{Q'} 
  \\
  R &= \bout{\dual u}{V}\bout{\dual u}{\textsf{true}}\inact 
  \\ 
  V &= \abs{z}\binp{z}{w}\inact 
\end{align*}
Let us recall the reductions of $P$: 
\begin{align*}
P &\red 
\binp{u}{y}
  \news{s}\big( \appl{V}{\dual s} \Par \about{s}{y} \big)  
  \Par 
 \bout{\dual u}{\textsf{true}}\inact 
\red 
  \news{s}\big( \appl{V}{\dual s} \Par \about{s}{\textsf{true}} \big)  
  \\
  &
  \red 
  \news{s}\big( \binp{\dual s}{w}\inact  \Par \about{s}{\textsf{true}} \big)  = P' 
\end{align*}
The monadic decomposition of $P$ is as follows: 
\begin{align*}
\mD{P} =  \news{\prop_1,\ldots,\prop_{10}}
\news{u_1,u_2}
\big(\apropout{1}{\dummyv}
\Par \mB{1}{}{P\sigma}\big)
\end{align*}
\noindent where $\sigma=\subst{u_1 \dual u_1}{u \dual u}$.
We have: 
\begin{align*}
  \mB{1}{}{P\sigma} &= 
  \propinp{1}{\dummyvar}\propout{2}{\dummyvar}\apropout{8}{\dummyvar}
  \Par \mB{2}{}{Q\sigma} \Par \mB{8}{}{R\sigma} 
\end{align*}
\noindent where: 
\begin{align*}
  \mB{2}{}{Q\sigma} &= 
  \news{\prop_x}\big( 
    \propinp{2}{\dummyvar}\binp{u_1}{x}(\apropout{3}{\dummyvar}
    \Par 
    \apropout{x}{x}
    % W^{\leadsto}_x)  
    \big) \Par 
        \mB{3}{}{Q'\sigma'}\big) 
  \\
  \mB{3}{}{Q'\sigma'} &= 
  \news{\prop_y}\big( 
    \propinp{3}{\dummyvar}\binp{u_2}{y}
    (\apropout{4}{\dummyvar}
    \Par 
    % \news{s}(\appl{V_{y}}{s} \Par \about{s}{V_{y}})
    W_y 
    % \apropout{y}{y}
    % W^{\leadsto}_x)  
    \big) 
    \Par  
\mB{4}{}{
  \news{s}
  \big( 
    \appl{x}{\dual s} \Par \about{s}{y} \big)}
    \big) 
  \\
  \mB{4}{}{
  \news{s}\big( \appl{x}{\dual s} \Par \about{s}{y} \big)} &= 
  \news{s_1} 
  \propinp{4}{\dummyvar}\propout{5}{\dummyvar}\apropout{6}{\dummyvar}
  \Par 
  \propinp{5}{\dummyvar}\propinp{x}{x}\appl{x}{\dual s_1}
  \Par 
  \\
  & \qquad 
  \propinp{6}{\dummyvar}\propinp{y}{y}\bout{s_1}{y}
  \apropout{7}{\dummyvar}
  \Par \propinp{7}{\dummyvar}\inact 
  \\
  \mB{8}{}{R\sigma} &= 
  \propinp{8}{\dummyvar}\bout{\dual u_1}{\mV{}{}{V}}\apropout{9}{\dummyvar}
  \Par  \mB{9}{}{\bout{\dual u_2}{\textsf{true}}\inact} 
\\
      \mB{9}{}{\bout{\dual u_2}{\textsf{true}}\inact} 
      &= 
      \propinp{9}{\dummyvar}\bout{\dual u_2}{\textsf{true}}\apropout{10}{\dummyvar}
      \Par \propinp{10}{\dummyvar}\inact 
\\
\mV{}{}{V} &=
\abs{z_1}
\news{\prop^V_1, \prop^V_2}
\apropoutv{^V_1}{} 
  % \binp{\prop^v_1}{z}(\appthunk{z}) 
  \Par 
  \\
  & \qquad 
   \news{\prop_w}
   \propinpv{^V_1}{z}
   \binp{z_1}{w}
   (
    % \about{\dual \prop^v_2}{} 
    \apropoutv{^V_2}{}
    \Par
  W_w 
   ) 
   \Par \binp{\prop^V_2}{}\inact 
\end{align*}
\noindent 
where $W_x = \news{s}(\appl{V_{x}}{s} \Par \about{\dual s}{V_{x}})$ 
with $V_{x} = \abs{y}\binp{y}{z}\bout{\prop_x}{x}
\news{s}(\appl{z}{\dual s} \Par \about{s}{z})$. 
% $$
% % \binp{y}{z}
% \bout{\prop_y}{\textsf{true}}
% \news{s}(\appl{V_y\subst{\textsf{true}}{y}}{\dual s} \Par \about{s}{V_y})
% $$
We may observe that $\mD{P}$ correctly implements 
$u_1$ and $u_2$ typed with \msts $M_1$ and $M_2$ (resp.) 
as given in~\Cref{mst:e:decomp-proc-types}. 

\noindent 
Now, we inspect the reductions of $\mD{P}$. First
we have three reductions on propagators:
\begin{align*}
  \mD{P} &\red 
  \news{\prop_2,\ldots,\prop_{10}}
  \news{u_1,u_2} 
  \propout{2}{}\apropout{8}{}
  \Par \mB{2}{\epsilon}{Q\sigma} \Par \mB{8}{\epsilon}{R\sigma}
  \\ 
  &\red^2  
   \news{\prop_3,\ldots, \prop_7,\prop_9, \prop_{10}}
   \news{\prop_x}\big( 
   \highlighta{\binp{u_1}{x}}(\apropout{3}{\dummyvar}
    \Par 
    \apropout{x}{x}
    % W^{\leadsto}_x)  
    \big) \Par 
        \mB{3}{}{Q'\sigma'}\big) 
\\
& 
\qquad 
\Par 
     \highlighta{\bout{\dual u_1}{\mV{}{}{V}}}\apropout{9}{} 
\Par 
\mB{9}{\epsilon}{\bout{u_2}{\textsf{true}}\inact}
= D^1 
\end{align*}

\noindent Now, the synchronization on $u_1$ can take a place in $D^1$ (on the prefixes highlighted above). 
We can see that value $\mV{}{}{V}$ received on $u_1$ can be 
propagated along $\prop_x$ to a trio using it. 
Following up on that, propagators $\prop_3$ and $\prop_9$ 
are synchronized.
\begin{align*}
D^1 &\red 
\news{\prop_3, \ldots, \prop_7,\prop_9, \prop_{10}}
\news{\prop_x}
\big( 
\apropout{3}{} \Par \about{\prop_{x}}{\mV{}{}{V}} \Par 
\\
& \qquad 
  \Par 
  \news{\prop_y}\big( 
    \propinp{3}{\dummyvar}\binp{u_2}{y}(\apropout{4}{\dummyvar}
    \Par 
    W_y
    % \apropout{y}{y}
    % W^{\leadsto}_x)  
    \big) 
    \Par  
\mB{4}{}{
  \news{s}
  \big( 
    \appl{x}{\dual s} \Par \about{s}{y} \big)}
    \big)\big) 
% \mB{3}{}{Q'\sigma'} 
\Par \apropout{9}{} \Par 
\mB{9}{}{\bout{\dual u_2}{\textsf{true}}\inact} 
\\
&\red^2 
\news{\prop_4, \ldots, \prop_7,\prop_{10}}
\news{\prop_x}
\big( \about{\prop_{x}}{\mV{}{}{V}} \Par 
\\
& \qquad 
  \Par 
  \news{\prop_y}\big( 
    \highlighta{\binp{u_2}{y}}(\apropout{4}{\dummyvar}
    \Par 
    W_y
    % \apropout{y}{y}
    % W^{\leadsto}_x)  
    \big) 
    \Par  
\mB{4}{}{
  \news{s}
  \big( 
    \appl{x}{\dual s} \Par \about{s}{y} \big)}
    \big)\big) 
% \mB{3}{}{Q'\sigma'} 
\Par 
\highlighta{\bout{\dual u_2}{\textsf{true}}}\apropout{10}{\dummyvar}
\Par \propinp{10}{\dummyvar}\inact 
  = D^2
\end{align*} 

\noindent Similarly, $D^2$ can mimic the synchronization on 
name $u_2$. Again, this is followed by synchronizations on propagators. 
\begin{align*}
D^2 &\red 
\news{\prop_4, \ldots, \prop_7,\prop_{10}}
\news{\prop_x}
\big( \about{\prop_{x}}{\mV{}{}{V}} 
  \Par 
  \news{\prop_y}
  \big( 
   \apropout{4}{\dummyvar}
   \Par 
   W_y\subst{\textsf{true}}{y}
   % \news{s}(\appl{V_{x}}{s} \Par \about{s}{V_{x}})
    % \apropout{y}{\textsf{true}}
    \Par  
\mB{4}{}{
  \news{s}
  \big( 
    \appl{x}{\dual s} \Par \about{s}{y} \big)}
    \big)\big) 
\\
& \qquad 
\Par 
\apropout{10}{\dummyvar}
\Par \propinp{10}{\dummyvar}\inact 
\\ 
  & \red^4  
  \news{\prop_7}
\news{\prop_x}
\big( \about{\prop_{x}}{\mV{}{}{V}} 
  \Par 
  \news{\prop_y}\big(  
    W_y\subst{\textsf{true}}{y}
    % \apropout{y}{\textsf{true}}
    \Par  
    \news{s_1} 
   \propinp{x}{x}\appl{x}{\dual s_1}
    \\
    & \qquad    \Par  
    \propinp{y}{y}\bout{s_1}{y}
    \apropout{7}{\dummyvar}
    \Par \propinp{7}{\dummyvar}\inact 
    \big)
    \big) 
% \apropout{10}{\dummyvar}
    = D^3 
\end{align*}

\noindent The subprocess $W_y\subst{\textsf{true}}{y}$ is dedicated to providing the value $\textsf{true}$ on a shared name $\prop_y$.
Specifically, it reduces as follow
Its reductions are as follows:
\begin{align*}
  W_y\subst{\textsf{true}}{y} \red^2 \bout{\prop_y}{\textsf{true}} W_y\subst{\textsf{true}}{y}
\end{align*}
% \begin{align*}
%   W_y\subst{\textsf{true}}{y} &= 
%   \news{s}(\appl{V_{y}\subst{\textsf{true}}{y}}{\dual s} 
%   \Par \about{s}{V_{y}\subst{\textsf{true}}{y}}) 
%   \\
%   &\red 
%   \news{s}
%   \binp{\dual s}{z}\bout{\prop_y}{y}
% \news{s'}(\appl{z}{\dual s'} \Par \about{s'}{z}) 
% \Par 
% \about{s}{V_{y}\subst{\textsf{true}}{y}} 
%   % \news{s}(\appl{V_{x}}{\dual s} \Par \about{s}{V_{x}}) 
%   \\
%   &\red 
%  \bout{\prop_y}{\textsf{true}}
% \news{s'}(\appl{V_y\subst{\textsf{true}}{y}}{\dual s'} \Par 
% \about{s'}{V_y\subst{\textsf{true}}{y}})
% \equiv 
% \bout{\prop_y}{\textsf{true}} W_y\subst{\textsf{true}}{y}
% \end{align*}
\noindent 
In this example, the shared value received on $y$ is used only once; 
in the general case, 
a process could use a shared value multiple times: thus there could   
be multiple 
trios requesting the shared value on $\prop_y$. 

With this information, we have the following reductions of the decomposed process:
\begin{align*}
  D^3 &\red^2  
  \news{\prop_7}
\news{\prop_x}
\big( \about{\prop_{x}}{\mV{}{}{V}} 
  \Par 
  \news{\prop_y}\big(  
    \bout{\prop_y}{\textsf{true}} W_y\subst{\textsf{true}}{y}
    % W_y\subst{\textsf{true}}{y}
    \Par  
    \news{s_1} 
   \propinp{x}{x}\appl{x}{\dual s_1}
    \\
    & \qquad    \Par  
    \propinp{y}{y}\bout{s_1}{y}
    \apropout{7}{\dummyvar}
    \Par \propinp{7}{\dummyvar}\inact 
    \big)
    \big) = D^4 
\end{align*}

\noindent In $D^4$  a value
for $x$ 
is requested on name $\prop_x$ before it is applied 
to name $\dual s_1$. Similarly, a value for $y$ is 
gathered by the communication on $\prop_y$. 
These values are retrieved in two reductions steps as follows: 
\begin{align*}
  D^4 &\red^2 
  \news{\prop_7}
  \news{s_1} 
  \appl{\mV{}{}{V}}{\dual s_1} 
  \Par 
  \bout{s_1}{\textsf{true}}\apropout{7}{}
  \Par 
  \propinp{7}{}\inact 
  \Par 
  \news{\prop_y}W_y\subst{\textsf{true}}{y}
  = D^5 
\end{align*}
\noindent 
% We can see above that after 
% the synchronization on $\prop_y$, sub-process 
% $W_y\subst{\textsf{true}}{y}$ will again  be able to provide shared value $\textsf{true}$ on name $\prop_y$ in two steps.
% In this example, shared value received for $y$ is used only once, 
% but in general case, 
% a process could use shared value multiple times: thus there would  
% be multiple 
% trios requesting the shared value on $\prop_y$. 
We remark that 
$\news{\prop_y}W_y\subst{\textsf{true}}{y}$  
reduces  to $\news{\prop_y}\propout{y}{\textsf{true}}W_y\subst{\textsf{true}}{y}$ which is behaviorally 
equivalent to the inactive process.

% We remark that in this example value $\textsf{true}$ that is received for variable $y$ is used only once, However, in general case
% a process could use shared value multiple times: thus there would  
% be multiple 
% trios requesting the shared value. 
% These requests can be fulfilled, as 
% we can see above that after 
% the synchronization on $\prop_y$, sub-process 
% $W_y\subst{\textsf{true}}{y}$ will again  be able to provide shared value $\textsf{true}$ on name $\prop_y$ in two steps. 

Next, the application of the value is followed by the 
synchronization on propagator $\prop^V_1$: 
\begin{align*}
D^5 
&\red 
\news{\prop_7}\news{s_1}
\news{\prop^V_1,\prop^V_2}
\about{\prop^V_1}{\dummyv} 
  % \binp{\prop^v_1}{z}(\appthunk{z}) 
  \Par 
   \news{\prop_w}
   \propinpv{^V_1}{} 
   \binp{\dual s_1}{w}
   (
    \apropoutv{^V_2}{}
  %  (\about{\dual \prop^v_2}{} 
   \Par
   W_w 
  %  \news{s}(\appl{V_{w}}{s} \Par \about{s}{V_{w}})
  %  \apropout{b}{b}
   )
   \Par \apropinpv{^V_2}{}\inact 
   \\
   & \qquad \Par 
   \bout{s_1}{\textsf{true}}\apropout{7}{}
\Par 
\propinp{7}{}\inact 
\Par 
\news{\prop_y}W_y\subst{\textsf{true}}{y}
\\
  &\red
  \news{\prop_7}\news{s_1}
  \news{\prop^V_2}
    % \binp{\prop^v_1}{z}(\appthunk{z}) 
    % \Par 
      \news{\prop_w}
      \binp{\dual s_1}{w}(
        \apropoutv{^V_2}{}
        % \about{\dual \prop^V_2}{}
         \Par
      W_w) \Par 
      % \binp{\prop^v_2}{}
      \apropinpv{^V_2}{}
      \inact 
      \\
      & \qquad \Par 
      \bout{s_1}{\textsf{true}}\apropout{7}{}
  \Par 
  \propinp{7}{}\inact 
  \Par 
  \news{\prop_y}W_y\subst{\textsf{true}}{y}
= D^6 
\end{align*}
Here, we can see that $D^6$ can simulate $P'$, and its internal communication on the channel~$s$.
\end{example}

\section{Extension with Labeled Choice} 
\label{ss:exti}

In this section we discuss how to extend our approach to include sessions with \emph{selection} and \emph{branching}---the constructs used to express deterministic choices.
Forgoing formal proofs, we illustrate by examples how to harness the expressive power of abstraction-passing to decompose these constructs at the process level.
First, we demonstrate how to break down selection and branching constructs in absence of recursion in \Cref{mst:ss:selbranch}.
Then, in \Cref{mst:ss:selbranchrec} we explore the interplay of recursion and labeled choice, as it requires special attention.
Finally, in \Cref{mst:ss:selbranchproof} we sketch how the operational correspondence proof can be adapted to account for branching and selection.

Let us briefly recall the labeled choice constructs in \HO, following~\cite{DBLP:conf/esop/KouzapasPY16}.
On the level of processes, selection and branching are modeled using labeled choice:
% We extend the grammar of the processes accordingly:
\begin{align*}
  P,Q
   & \bnfis
   \ldots \sbnfbar 
  \bsel{u}{l} P \sbnfbar \bbra{u}{l_i:P_i}_{i \in I} 
\end{align*}
The process $\bsel{u}{l}P$ selects the label $l$ on channel $u$ and then proceeds as $P$.
The process $\bbra{u}{l_i: P_i}_{i \in I}$ receives a label on the channel $u$ and proceeds with the continuation branch $P_i$ based on the received label.
Selection and branching constructs can synchronize with each other, as represented in the operational semantics by the following reduction rule:
\begin{align*}
  \bsel{u}{l_j} Q \Par \bbra{\dual{u}}{l_i : P_i}_{i \in I} & \red 
    Q \Par P_j
\qquad~~(j \in I)~~\orule{Sel}
\end{align*}

At the level of types, selection and branching are represented with the following types:
\begin{align*}
  S & \bnfis		
       \ldots 
       \bnfbar \btsel{l_i:S_i}_{i \in I} \bnfbar \btbra{l_i:S_i}_{i \in I}
\end{align*}
The {\em selection type} $\btsel{l_i:S_i}_{i \in I}$ and the {\em branching type} $\btbra{l_i:S_i}_{i \in I}$ are used to type, respectively,
the selection and branching process constructs.
Note the implicit sequencing in the sessions involving selection and branching: the exchange of a label $l_i$ precedes the execution of one of the stipulated protocol $S_i$.
The typing rules for type-checking branching and selection processes are given in~\Cref{mst:fig:tr-selbra}.
\begin{figure}[t!]
\begin{mdframed}
\vspace{-4mm}
	\[
                        \inferrule[(Sel)]{
				\Gamma; \Lambda; \Delta \cat u: S_j  \proves P \hastype \Proc \quad j \in I
			}{
				\Gamma; \Lambda; \Delta \cat u:\btsel{l_i:S_i}_{i \in I} \proves \bsel{u}{l_j} P \hastype \Proc
			}
                        \qquad
                        \inferrule[(Bra)]{
				 \forall i \in I \quad \Gamma; \Lambda; \Delta \cat u:S_i \proves P_i \hastype \Proc
			}{
				\Gamma; \Lambda; \Delta \cat u: \btbra{l_i:S_i}_{i \in I} \proves \bbra{u}{l_i:P_i}_{i \in I}\hastype \Proc
			}		
	\]
	\end{mdframed}
  \caption[Typing rules for selection and branching]{Typing rules for selection and branching. \label{mst:fig:tr-selbra}}
\end{figure} 

Given these process constructs and types, what are the minimal versions of the session types with labeled choice?
We do not consider branching and selection as atomic actions as their purpose is to make a choice of a stipulated protocol.
In other words, it is not meaningful to type a channel with branching type in which all protocols are $\tinact$.
Thus, we extend the  syntax of minimal session types (\Cref{mst:d:mtypesi}) with 
  branching and selection constructs, as follows: 
  \begin{align*}
		M & \bnfis 	
    \ldots \bnfbar \btsel{l_i : M_i}_{i \in I} \bnfbar \btbra{l_i : M_i}_{i \in I}
	\end{align*}
  \noindent  That is, \msts also include branching and selection 
  types with \msts nested in branches. 

  Next we explain our strategy for extending the breakdown function to  
  account for selection and branching.

\subsection{Breaking Down Selection and Branching}
\label{mst:ss:selbranch}

Notice that in a branching process $\bbra{u}{l_i:P_i}_{i \in I}$ each subprocess
$P_i$ can have a different session with a different degree.
Abstraction-passing allows to uniformly handle these kinds of processes.
We extend the breakdown function in \Cref{mst:def:typesdecomp} to selection and branching as follows:
\begin{align*} 
  \Gt{\btbra{l_i : S_i}_{i \in I}} &= \btbra{l_i:
  \btoutt{\lhot{\Gt{S_i}}}}_{i \in I}
  \\
  \Gt{\btsel{l_i : S_i}_{i \in I}} &=
  \btsel{l_i:\btinpt{\lhot{\Gt{\dual {S_i}}}}}_{i \in I} 
\end{align*} 
This decomposition follows the intuition that branching and selection correspond to the input and output of labels, respectively. 
For example, in the case of branching, once a particular branch $l_i$ has been selected, we would like to input names on which to provide sessions from the branch $\Gt{S_i}$.
In our higher-order setting, we do not input or output names directly.
Instead, we send out an abstraction of the continuation process, which binds those names.
It is then the job of the (complementary) selecting process to activate that abstraction with the names we want to select.

To make this more concrete, let us consider the decomposition of branching and selection at the level of processes through the following extended example.
\begin{example}%[Breaking down Selection and Branching]
  \label{mst:ex:selbra}
Consider a mathematical server $Q$ that offers clients two operations: addition and negation of integers. 
The server uses name $u$ to implement the following session type: 
  \begin{align*}
    S = \btbra{\textsf{add}: 
    \underbrace{\btinp{\textsf{int}} \btinp{\textsf{int}} 
    \btout{\textsf{int}} \tinact}_{S_{\textsf{add}}}
    ~,\ 
    \textsf{neg} : 
    \underbrace{\btinp{\textsf{int}} \btout{\textsf{int}} \tinact}_{S_{\textsf{neg}}} 
    }
  \end{align*}
  The branches have session types with different lengths: one receives two integers and sends
	over their sum, the other has a single input of an integer followed by an
	output of its negation. 
	Let us consider a possible implementation for the server $Q$ and for a client $R$ that {selects} the first branch to add integers 16 and 26:
  \begin{align*}
     Q & \defas \bbra{u}{\textsf{add} : 
    {Q_{\textsf{add}}},\ 
    \textsf{neg} : 
    {Q_{\textsf{neg}}}} 
    & 
    R \defas \bsel{\dual{u}}{\textsf{add}} 
    \bout{\dual {u}}{\mathbf{16}} \bout{\dual u}{\mathbf{26}} \binpt{\dual u}{r} 
    \\ 
        Q_{\textsf{add}} & \defas  \binp{u}{a}\binp{u}{b}\boutt{u}{a+b}   
        \\  
        Q_{\textsf{neg}} & \defas \binp{u}{a}\boutt{u}{-a} 
    \end{align*}
The composed process $P \defas \news{u}(Q \Par R)$ can reduce as follows: 
  \begin{align*}
    P & ~\red~ 
    \news{u} 
    (\binp{u}{a}\binp{u}{b}\about{u}{a+b} \Par 
    \bout{\dual {u}}{\mathbf{16}} \bout{\dual u}{\mathbf{26}} \binpt{\dual u}{r} ) \\
    & ~\red^2~ 
    \news{u} (\about{u}{\mathbf{16+26}}  \Par \binpt{\dual u}{r} )
    = P'
  \end{align*}
  Let us discuss the decomposition of $P$. First, the decomposition  of $S$ is the minimal session type $M$, defined as follows:
  \begin{align*}
    M = \Gt{S} =
        \btbra{
          &\textsf{add}:
        \btoutt{\lhot{\big(\btinpt{\textsf{int}}, \btinpt{\textsf{int}}, \btoutt{\textsf{int}}  \big)}},
        \\ 
        &\textsf{neg}:
        \btoutt{\lhot{\big(\btinpt{\textsf{int}}, \btoutt{\textsf{int}} \big)}}} 
  \end{align*}
  Following \Cref{mst:def:decomp}, we decompose $P$ as follows: 
	\begin{align*}
	\D{P} = 
  \news{\prop_1 \ldots \prop_7} \big( \apropout{1}{}  \Par \news{u_1} (\propinp{1}{} \propout{2}{} \propoutt{3}{}   \Par \B{2}{\epsilon}{Q\sigma_2} \Par \B{3}{\epsilon}{R\sigma_2}) \big)
	\end{align*}
	where $\sigma_2=\subst{u_1\dual{u_1}}{u\dual u}$. 
  The breakdown of the server process $Q$, which implements the branching, is as follows: 
	\begin{align*}
	\B{2}{\epsilon}{Q\sigma_2} = 
    \propinp{2}{} u_1 \brases \{  \textsf{add}: & 
    ~u_1 
    \outses 
    \big\langle 
      % V begin 
      \underbrace{
       \abs{(y_1, y_2, y_3)}{
		\news{\prop^V_{1} \ldots \prop^V_4} \apropoutv{^V_1}{} 
		\Par \B{1}{\epsilon}{ Q_{\textsf{add}} \subst{y_1}{u}}}
    \sigma_V 
    % \subst{\prop^v_{1}, \ldots, \prop^v_4}{\prop_1,\ldots,\prop_4}
      }_{V}
    % V end 
 \big\rangle,  \\
	\textsf{neg}: & ~u_1 \outses \big\langle 
    % W begin 
    \underbrace{
        \abs{(y_1, y_2)}{
            \news{\prop^W_1 \ldots \prop^W_{3}} \apropoutv{^W_1}{} 
		\Par \B{1}{\epsilon}{ Q_{\textsf{neg}} \subst{y_1}{u}}}
    \sigma_W 
    }_{W}
    % W end 
  \big\rangle  \}
	\end{align*}
  \noindent where: 
  \begin{align*}
    % &V =  \abs{(y_1, y_2, y_3)}{
		% \news{\prop_{1} \ldots \prop_4} \apropout{1}{} 
		% \Par \B{1}{\epsilon}{ Q_{\textsf{add}} \subst{y_1}{u}}}
    % \\ 
    \B{1}{\epsilon}{ Q_{\textsf{add}} \subst{y_1}{u}} &= 
    \propinp{1}{}\binp{y_1}{a}\apropout{2}{a} 
    \Par 
    \propinp{2}{a}\binp{y_2}{b}\apropout{3}{a,b} 
    \Par 
    \propinp{3}{a,b}\bout{y_3}{a+b}\apropout{4}{} \Par 
    \apropinp{4}{} 
    \\
    \B{1}{\epsilon}{ Q_{\textsf{neg}} \subst{y_1}{u}} &= 
    \propinp{1}{}\binp{y_1}{a}\apropout{2}{a} \Par 
    \propinp{2}{a}\bout{y_2}{-a}\apropout{3}{} \Par \apropinp{3}{} 
  \end{align*}
  \noindent with $\sigma_V=\subst{\prop^V_{1}, \ldots, \prop^V_4}{\prop_1,\ldots,\prop_4}$ 
  and $\sigma_W = \subst{\prop^W_{1}, \prop^W_{2}, \prop^W_3}{\prop_1,\prop_2,\prop_3}$.
  In process $\B{2}{\epsilon}{Q\sigma_2}$, name $u_1$ implements the minimal session type $M$. 
  Following the common trio structure, the first prefix awaits activation on 
  $\prop_2$. The next prefix mimics the
  branching action of $Q$ on $u_1$. 
  Then, each branch consists of the output of an abstraction along $u_1$. 
  This output does not have a counterpart in $Q$; it is meant to synchronize with process $ \B{3}{\epsilon}{R\sigma_2}$, the breakdown of the corresponding selection process (see below). 
  
  The abstractions sent along $u_1$ encapsulate the breakdown of subprocesses in
  the two branches ($Q_{\textsf{add}}$ and $Q_{\textsf{neg}}$). An abstraction
  in the branch has the same structure as the breakdown of a value
  $\abs{y:\shotup{C}}{P}$ in \Cref{mst:t:bdowncore}: it is a composition of a
  control trio and the breakdown of a subprocess; the generated propagators are
  restricted.
  % We denote these propagators with the subscript $V$ (resp. $W$) and apply the appropriate substitution $\sigma_V$ (resp. $\sigma_W$) to avoid name clashes.  
  In the first branch the server needs three actions to perform the
  session, and in the second branch the server needs to perform two actions.
  Because of that the first abstraction binds three names $y_1,y_2,y_3$, and the
  second abstraction binds two names $y_1,y_2$.

  In the bodies of the abstractions we break down $Q_{\textsf{add}}$ and $Q_{\textsf{neg}}$, but not before adjusting the names on which the broken down processes provide the sessions.
  For this,  we substitute $u$ with $y_1$ in both processes, ensuring that the broken down names are bound by the abstractions.
  By binding decomposed names in abstractions we account for different 
  session types of the original name in branches, while preserving typability: 
  this way the decomposition of different branches can use (i)~the same names
  but typed with different minimal types and (ii)~a different number of names, as it is the
  case in this example. 
  
  The decomposition of the client process $R$, which implements the selection, is as follows: 
  \begin{align*}
    \B{3}{\epsilon}{R\sigma_2} = ~& 
    \news{u_2, u_3, u_4}
    \propinp{3}{}
    {\bsel{\dual{u_1}}
    {\textsf{add}}{\binp{\dual{u_1}}{z}\propout{4}{}\appl{z}{(u_2, u_3, u_4)}}} 
    %\\
    % & \Par \news{u_2, u_3, u_4}(\propinp{4}{y}{\appl{y}{(u_2, u_3, u_4)}} 
    %&\ \qquad \qquad \qquad \qquad 
    \Par
    \B{4}{\epsilon}{\bout{\dual u_2}{\mathbf{16}}\bout{\dual u_2}{\mathbf{26}} \binpt{\dual u_2}{r}}
    \end{align*}
    \noindent where: 
    \begin{align*}
      \B{4}{\epsilon}{\bout{\dual u_2}{\mathbf{16}}\bout{\dual u_2}{\mathbf{26}} \binpt{\dual u_2}{r}} = 
      \propinp{4}{}\bout{\dual u_2}{\mathbf{16}}\apropout{5}{} \Par 
    \propinp{5}{}\bout{\dual u_3}{\mathbf{26}}\apropout{6}{} \Par 
    \propinp{6}{}\binp{\dual u_4}{r}\apropout{7}{} \Par 
    \apropinp{7}{}
    \end{align*}
    After receiving the context on $\prop_3$ (empty in this case), the selection action on $u_1$ is mimicked; then, an abstraction (an encapsulation of the selected branch) is received and applied to $(u_2, u_3, u_4)$, which are locally bound. 
    The intention is to use these names to connect the received abstraction and the continuation of a selection process: the subprocess encapsulated within the abstraction  will use
    $(u_2, u_3, u_4)$, while the dual names $(\dual u_2, \dual u_3, \dual
    u_4)$ are present in the breakdown of the continuation. 
    
    For simplicity, we defined $\B{3}{\epsilon}{R\sigma_2}$ using a subprocess with four prefixes. 
    Alternatively, we could have used a decomposition that relies on two trios, by introducing abstraction passing as in \Cref{s:opt}.%   on propagator $\prop_4$, as follows: 
    % \begin{align*}
    %   % \B{3}{\epsilon}{R\sigma_2} = 
    %   ~& \propinp{3}{}\apropout{4}{\abs{(y_1, y_2, y_3)}{\bsel{\dual{u_1}}{\textsf{add}}{\binp{\dual{u_1}}{z}\propout{5}{}
    %   \appl{z}{(y_1, y_2, y_3)}}}} \\
    %   & \qquad \qquad \Par 
    %   \news{u_2, u_3, u_4}(\propinp{4}{y}{\appl{y}{({u_2, u_3, u_4})}} \Par
    %   \B{5}{\epsilon}{\bout{\dual u_2}{\mathbf{16}}\bout{\dual u_2}{\mathbf{26}} \binpt{\dual u_2}{r}}
    %   \end{align*}
   
    We will now examine the reductions of the decomposed process $\D{P}$.
    First, $\prop_1$, $\prop_2$, and $\prop_3$ will synchronize. 
    We have $\D{P} \red^4 D_1$, where
    \begin{align*}
      D_1 &= \news{\prop_4 \ldots \prop_7} \news{u_1} \big(
     u_1 \brases \{  \textsf{add}: 
        ~u_1 
        \outses \big\langle V
%        \abs{(y_1, y_2, y_3)}{
%        \news{\prop^V_1 \ldots \prop^V_4} \apropoutv{^V_1}{} 
%        \Par \B{1}{\epsilon}{ Q_{\textsf{add}} \subst{y_1}{u}} \sigma_V
%      }
       \big\rangle,~ 
%      \\
%      & \qquad \qquad \qquad \qquad \qquad \qquad 
      \textsf{neg}:  
      ~u_1 \outses \big\langle 
      W
%            \abs{(y_1, y_2)}{
%                \news{\prop^W_{1} \ldots \prop^W_3} \apropoutv{^W_1}{} 
%        \Par \B{1}{\epsilon}{ Q_{\textsf{neg}} \subst{y_1}{u}}\sigma_W
%      } 
      \big\rangle  \} \\ 
        & \qquad \qquad \qquad \qquad  \Par
        \news{u_2, u_3, u_4}
        {\appl{(\abs{(y_1, y_2,
        y_3)}{\bsel{\dual u_1}{\textsf{add}}{\binp{\dual u_1}{z}\propout{4}{}
        \appl{z}{(y_1,y_2,y_3)}}})}{{(u_2, u_3, u_4)}}}  \Par 
        \\
        &  \qquad \qquad \qquad \qquad \qquad \qquad  
        \B{4}{\epsilon}{  \bout{\dual u_2}{\mathbf{16}}\bout{\dual u_2}{\mathbf{26}} \binpt{\dual u_2}{r}}  \big )
        \end{align*} 
        \noindent In $D_1$, $(u_2,u_3,u_4)$ will be applied to the abstraction; after that,
        the process chooses the label $\textsf{add}$ on $u_1$.
        Process $D_1$ will reduce further as
      $D_1 \red^2 D_2 \red^2 D_3$, where: 
      \begin{align*}
        D_2 &= \news{\prop_4 \ldots \prop_7} \news{u_1} \big(
          ~u_1 
          \outses \big\langle 
          V
%          \abs{(y_1, y_2, y_3)}{
%          \news{\prop_{1} \ldots \prop_4} \apropout{1}{} 
%          \Par \B{1}{\epsilon}{ Q_{\textsf{add}} \subst{y_1}{u_1}}
%        } 
        \big\rangle
%        \\ 
%        % \end{align*}
%        % \begin{align*}
%          & \qquad \qquad \qquad \qquad 
           \Par
          \news{u_2, u_3, u_4}
          ({{{{\binp{\dual u_1}{z}\propout{4}{}
          \appl{z}{(u_2, u_3, u_4)}}}}
          }  \Par 
          \B{4}{\epsilon}{\bout{\dual u}{\mathbf{26}}
          \binpt{\dual u}{r}}) \big)
          \\
            D_3 &= \news{\prop_4 \ldots \prop_7}  \news{u_1, u_2, u_3, u_4} \big(
%              \news{u_2, u_3, u_4}
              {{{{\propout{4}{}
              \appl{
              V
%                \abs{(y_1, y_2, y_3)}{
%              \news{\prop_{1} \ldots \prop_4} \apropout{1}{} 
%              \Par \B{1}{\epsilon}{ Q_{\textsf{add}} \subst{y_1}{u_1}}}
              }{(u_2, u_3, u_4)}}
              }}
              }  \Par 
              \\
              &  \qquad \qquad \qquad \qquad \qquad  
              % \B{5}{\epsilon}{\bout{\dual u}{\mathbf{26}}
              % \binp{\dual u}{r} \inact}  \big )
              \propinp{4}{}\bout{\dual u_2}{\mathbf{16}}\apropout{5}{} \Par 
              \propinp{5}{}\bout{\dual u_3}{\mathbf{26}}\apropout{6}{} \Par 
              \propinp{6}{}\binp{\dual u_4}{r}\apropout{7}{} \Par 
              \apropinp{7}{} \big )
              \end{align*} 
        % \begin{align*}
        %   P''' = \news{\prop_5 \prop_6} \big(  \news{u_2}{\propout{5}{}(
        %       \news{\prop_1 \prop_2}\propout{1}{} \inact
        %       \Par \B{1}{\epsilon}{\bout{u_2}{V_+} \inact) }
        %       \Par
        %     \B{5}{\epsilon}{ \binp{u_2}{x} \appl{x}{(\text{16,26})} })} \big ) 
        %   \end{align*}
              \noindent Then $D_3$ reduces as $D_3 \red D_4 \red D_5$, where: 
      \begin{align*}
        D_4 &= 
        \news{\prop_5 \ldots \prop_7} 
          \news{u_2, u_3, u_4}
          \big(
          \news{\prop^V_{1} \ldots \prop^V_{4}}
          (
          \apropoutv{^V_1}{} 
          \Par 
          \propinpv{^V_1}{}\binp{u_2}{a}\apropoutv{^V_2}{a} 
          \Par 
          \\ 
          & \qquad \qquad \qquad \qquad \qquad \qquad \qquad \qquad \qquad 
          \propinpv{^V_2}{a}\binp{u_3}{b}\apropoutv{^V_3}{a,b} 
          \Par 
          \\ 
          & \qquad \qquad \qquad \qquad \qquad \qquad \qquad \qquad \qquad 
          \propinpv{^V_3}{a,b}\bout{u_4}{a+b}\apropoutv{^V_4}{} \Par 
          \apropinpv{^V_4}{}
          ) 
          \Par 
          \\
          &  \qquad \qquad \qquad \qquad \qquad \qquad  
         \bout{\dual u_2}{\mathbf{16}}\apropout{5}{} \Par 
          \propinp{5}{}\bout{\dual u_3}{\mathbf{26}}\apropout{6}{} \Par 
          \propinp{6}{}\binp{\dual u_4}{r}\apropout{7}{} \Par 
          \apropinp{7}{} \big) 
%          = D_4 
        	\\
        	 D_5 &= 
        	  \news{\prop_5 \ldots \prop_7} 
          \news{u_2, u_3, u_4}
          \big(
          \news{\prop^V_{2} \ldots \prop^V_{4}} 
          (
          \binp{u_2}{a}\apropoutv{^V_2}{a} 
          \Par 
          \propinpv{^V_2}{a}\binp{u_3}{b}\apropoutv{^V_3}{a,b} 
          \Par 
          \\ 
          & \qquad \qquad \qquad \qquad \qquad \qquad \qquad \qquad \qquad 
          \propinpv{^V_3}{a,b}\bout{u_4}{a+b}\apropoutv{^V_4}{} \Par 
          \apropinpv{^V_4}{} 
          )
          \Par 
          \\
          &  \qquad \qquad \qquad \qquad \qquad \qquad  
         \bout{\dual u_2}{\mathbf{16}}\apropout{5}{} \Par 
          \propinp{5}{}\bout{\dual u_3}{\mathbf{26}}\apropout{6}{} \Par 
          \propinp{6}{}\binp{\dual u_4}{r}\apropout{7}{} \Par 
          \apropinp{7}{} \big) 
      \end{align*}
      Now, process $D_5$ can mimic the original 
      transmission of the integer $16$ on channel $u_2$ as follows: 
      \begin{align*}
      	D_5 &\red 
      	   \news{\prop_5 \ldots \prop_7}
          \news{u_2, u_3, u_4}
           \big(
          \news{\prop^V_{2} \ldots \prop^V_{4}} 
          (
          \apropoutv{^V_2}{\mathbf{16}} 
          \Par 
          \propinpv{^V_2}{a}\binp{u_3}{b}\apropoutv{^V_3}{a,b} 
          \Par 
          \\ 
          & \qquad \qquad \qquad \qquad \qquad \qquad \qquad \qquad \qquad 
          \propinpv{^V_3}{a,b}\bout{u_4}{a+b}\apropoutv{^V_4}{} \Par 
          \apropinpv{^V_4}{} 
          )
          \Par 
          \\ 
          &  \qquad \qquad \qquad \qquad \qquad \qquad  
       \apropout{5}{} \Par 
          \propinp{5}{}\bout{\dual u_3}{\mathbf{26}}\apropout{6}{} \Par 
          \propinp{6}{}\binp{\dual u_4}{r}\apropout{7}{}  \Par 
          \apropinp{7}{} \big) = D_6
      \end{align*}
      Finally, process $D_6$ reduces to $D_7$ in three steps, as follows: 
      \begin{align*}
      	D_6 &\red^3 
      	   \news{\prop_5 \ldots \prop_7} 
          \news{u_4}
          \big(
          \news{\prop^V_4}
         (\bout{u_4}{\mathbf{16}+\mathbf{26}}\apropoutv{^V_4}{} \Par 
          \apropinpv{^V_4}{}) 
          \Par 
        %  \\
         % &  \qquad \qquad \qquad \qquad \qquad \qquad  
      \binp{\dual u_4}{r}\apropout{7}{} \Par 
          \apropinp{7}{} \big) = D_7
      \end{align*}
      Clearly, process  $D_7$ correctly simulates the synchronizations of the process $P'$. 
%       This suggest for the operational correspondence 
%       between process $P$ involving labelled choice constructs 
%      and its decomposition $\D{P}$. 
 \hspace*{\fill}
 	$\lhd$
\end{example}

\subsection{The Interplay of Selection/Branching and Recursion}
\label{mst:ss:selbranchrec}
Now, we discuss by example how recursive session types involving branching/selection
are broken down. For simplicity, we consider recursive types without nested recursion and in which the recursive step is followed immediately by branching or selection, without any intermediate actions, i.e. types of the following form: 
$$\trec{t}\btbra{l_i:S_i}_{i\in I}
\qquad \trec{t}\btsel{l_i:S_i}_{i\in I}
$$
\noindent where none of $S_i$ contain branching/selection or recursion.

In this case, the decomposition of branching recursive types should be defined differently than for tail-recursive types: a  type such as $\trec{t}\btbra{l_i:S_i}_{i\in I}$ does not
  necessarily describe a channel with an infinite behavior, because some of the branches  $S_i$ can result in termination.
In such case, decomposing all actions in the type $\btbra{l_i:S_i}_{i\in I}$ as their own recursive types using the $\Rt{-}$ function would be incorrect.

Instead, we decompose the body of the recursive type with $\Gt{-}$ itself:
  \begin{align*} 
    \Gt{\trec{t}\btbra{l_i : S_i}_{i \in I}} &= 
   % \trec{t} \Gt{\btbra{l_i:\btout{\lhot{{S_i}}} \tinact}_{i \in I}}  
   % = 
    \trec{t} {\btbra{l_i:
    \btoutt{\lhot{\Gt{S_i}}}}_{i \in I}} 
    \\
      \Gt{\trec{t}\btsel{l_i : S_i}_{i \in I}} &=
     % \trec{t}\Gt{\btsel{l_i:\btinp{\lhot{{\dual {S_i}}}} \tinact}_{i \in I}} 
     % = 
      \trec{t}\btsel{l_i:\btinpt{\lhot{\Gt{\dual {S_i}}}}}_{i \in I} 
  \end{align*} 
\noindent 
If some branch $S_i$ contains the recursion variable $\tvar{t}$, then it will appear in $\Gt{S_i}$, because $\Gt{\tvar{t}}=\tvar{t}$.
That is, recursion variables will appear as part of the abstraction $\lhot{\Gt{\dual {S_i}}}$.
That  means that the decomposition of a tail-recursive type form can  produce a minimal \emph{non}-tail-recursive types.

% Recall that $\Gt{\tvar{t}}=\tvar{t}$. Thus , 
%   we remark that the decomposition of a tail-recursive type form  
%   produces a minimal non-tail-recursive types. 
  Now, we illustrate this decomposition on the level of processes.
  \begin{example}  
  We consider  a process $P$ with a name $r$ that is typed as follows: 
$$S = \trec{t}\btbra{l_1:\btinp{\mathsf{str}}\btout{\mathsf{int}}\vart{t},\  
l_2:\tinact}.$$
For simplicity, we give $P$ in \HOp (which includes \HO with recursion as sub-calculus): 
	\begin{align*}
	    P &= R \Par Q \\
		% S &= \trec{t}\btbra{l_1:\btinp{\mathsf{str}}\btout{\mathsf{int}}\vart{t},
		% 		l_2:\tinact} \\
		R &= \recp{X}\bbra{r}{l_1:\binp{r}{t}\bout{r}{\mathtt{len}(t)}X,\ l_2:\inact} \\
		Q &= \bsel{\dual r}{l_1}
    \bout{\dual r}{{{\text{``Hello''}}}}
    \binp{\dual r}{a_1}
    \bsel{\dual r}{l_1}
    \bout{\dual r}{{\text{``World''}}}
    \binp{\dual r}{a_2}
    \bsel{\dual r}{l_2}\inact
	\end{align*}
That is, $P$ contains a server $R$ which either accepts a new request to calculate a length of a string, or to terminate.
Dually, $P$ contains a client $Q$, which uses the server twice before terminating.

We can give an equivalent process in \HO by encoding the recursion (as done in \cite{DBLP:conf/esop/KouzapasPY16}):
  % $\map{P} = \map{R} \Par Q$ where: 
  \begin{align*}
    \map{P} &= \map{R} \Par Q \\ 
    \map{R} &= \news{s}(\appl{V}{(r,s)} \Par \about{\bar{s}}{V}) \\
		V &= \abs{(x_r,x_s)}\binp{x_s}{y}\bbra{x_r}
    {
      l_1:\binp{x_r}{t}\bout{x_r}{\mathtt{len}(t)}
		\news{s}{(\appl{y}{(x_r,s)} \Par \about{\dual s}{y})},\ 
		l_2: \inact
    } 
  \end{align*}
  The decomposition of $S$, denoted $M^*$, is the following minimal session type:
  \begin{align*}
    M^* = \Gt{S} = 
    \trec{t}\btbra{l_1:
    \btoutt{\lhot{(\btinpt{\mathsf{str}},\ 
    \btoutt{\mathsf{int}},\ \vart{t})}},\ 
    l_2:\tinact}	
  \end{align*}

  \noindent
As in the previous example (\Cref{mst:ex:selbra}), the continuation of a selected branch will be packed in an abstraction and sent over.
This abstraction binds names on which the session actions should be performed.
In addition, if a branch contains a recursive call, then the last argument of the abstraction will be a name on which the next instance of the recursion will be mimicked.
  We illustrate this mechanism by  giving the decomposition of $\map{P}$ and inspecting its reductions. 
  \begin{align*}
    \D{\map{P}} &= 
    \news{\prop_1, \ldots, \prop_{12}} \apropout{1}{} \Par 
    \propinp{1}{}\propout{2}{}\apropout{5}{} \Par 
    \B{2}{\epsilon}{\map{R}} \Par  \B{5}{\epsilon}{Q}\\
      \B{2}{\epsilon}{\map{R}} &= 
      \news{s_1}
      (\propinp{2}{}\propout{3}{}\apropout{4}{} 
      \Par \propinp{3}{}\appl{\V{2}{\epsilon}{V}}{(r_1, s_1)} 
      \Par \propinp{4}{}\about{\dual s_1}{\V{4}{\epsilon}{V}})
      \\
      \V{2}{\epsilon}{V} &= 
        \abs{(x_{r_1}, x_{s_1})}
        \news{\prop^V_{1}\prop^V_{2}}
        \apropoutv{^V_1}{} 
        \Par \propinpv{^V_1}{}
        \binp{x_{s_1}}{y}\apropoutv{^V_2}{y}
        \Par 
        \propinpv{^V_2}{y}\bbra{x_{r_1}}
          {l_1: \about{x_{r_1}}{W},\ l_2: \inact} 
          \\
      W &= 
            \abs{(z_1, z_2, z_3)}
              \news{\prop^W_{1} \ldots \prop^W_{5}}
              \apropoutv{^W_1}{} \Par 
              \propinpv{^W_1}{}\binp{z_1}{t}\apropoutv{^W_2}{t} \Par 
              \propinpv{^W_2}{t}\bout{z_2}{\mathtt{len}(t)}\apropoutv{^W_3}{} \\ 
          & \qquad  \Par  
          \news{s_1}(\propinpv{^W_3}{}\propoutv{^W_4}{}\apropoutv{^W_5}{} 
          \Par \propinpv{^W_4}{}\appl{y}{(z_3, s_1)} \Par 
            \propinpv{^W_5}{}\about{\dual s_1}{y}) 
            \\ 
      \B{5}{}{Q} &= 
        \news{r_2 : \btinpt{\mathsf{str}},\ 
        r_3 : \btoutt{\mathsf{int}},\ 
      r_4 : {M^*}} \\
      & \qquad 
         \propinp{5}{}\bsel{\dual r_1}
          {l_1}\propout{6}{}\binp{\dual r_1}{y}\appl{y}{(r_2, r_3, r_4)} \Par 
          \\ & \qquad 
          \propinp{6}{}\bout{\dual r_2}{{\text{``Hello''}}}\apropout{7}{} \Par 
          \propinp{7}{}\binp{\dual r_3}{t}\apropout{11}{} 
           \Par \\
          & \qquad \quad 
          \news{r_5 : \btinpt{\mathsf{str}},\ 
          r_6 : \btoutt{\mathsf{int}},\ 
        r_7 : {M^*}} \\
        & \qquad \quad \quad 
          \propinp{8}{}\bsel{\dual r_4} 
          {l_1}\propout{12}{}
           \binp{\dual r_4}{y}\appl{y}{(r_5, r_6, r_7)} \Par 
            \propinp{9}{}\bout{\dual r_5}{\text{``World''}}\apropout{10}{} \Par
            \\  
            & \qquad \quad \quad 
          \propinp{10}{}\binp{\dual r_6}{t}\apropout{11}{} \Par 
          %   \\
          % & \qquad \quad \quad \qquad 
          \propinp{11}{}\bsel{\dual r_7}{l_2}\apropout{12}{} \Par 
          \apropinp{12}{}  
    \end{align*}
    In the process $\B{5}{}{Q}$, the restricted names $(r_2, r_3, r_4)$ are the decomposition of the name $r$ for the branch $l_1$.
  To calculate their types, we unfold $S$:  
    \begin{align*}
      S &= \trec{t}\btbra{l_1:\btinp{\mathsf{str}}\btout{\mathsf{int}}\vart{t},\  
l_2:\tinact} 
\equiv 
\btbra{l_1:\btinp{\mathsf{str}}\btout{\mathsf{int}}\vart{t},\  
l_2:\tinact} \subst{S}{\vart{t}} \\ 
&= 
\btbra{l_1:\btinp{\mathsf{str}}\btout{\mathsf{int}}S,\  
l_2:\tinact},
    \end{align*}
and we look at the decomposition of the type corresponding to the branch $l_1$: 
  \begin{align*}
    \Gt{\btinp{\mathsf{str}}\btout{\mathsf{int}}S} = 
    (\btinpt{\mathsf{str}},\ \btoutt{\mathsf{int}},\ M^*)
  \end{align*}

Now we inspect a few reductions of $\D{P}$. 
First, we have synchronizations on $\prop_1,\ldots,\prop_4$. 
This is followed by the application of the exchanged value 
$\V{}{\epsilon}{V}$ to names $r_1, s_1$: 
% After synchronizations on $\prop_1,\ldots,\prop_4$ 
% we have: 
\begin{align*}
  \D{\map{P}} \red^{*}&
  % \abs{(x_{r_1}, x_{s_1})}
        \news{\prop^V_{1}\prop^V_{2}}
        \apropoutv{^V_1}{} 
        \Par \propinpv{^V_1}{}
        \binp{{s_1}}{y}\apropoutv{^V_2}{y}
        \Par 
        \propinpv{^V_2}{y}\bbra{{r_1}}
          {l_1: \about{x_{r_1}}{W},\ l_2: \inact} 
          \\
          & \Par \about{\dual s_1}{\V{}{\epsilon}{V}}
          \Par \apropout{5}{} \Par \B{5}{\epsilon}{Q} = D_1 
\end{align*}
Then, after synchronizations on $\prop^V_1,s_1$, and $\prop^V_2$ in $D_1$ 
we have the following: 
% ... 
  % After synchronizations on $\prop_1,\ldots,\prop_5$ 
  % we have: 
\begin{align*}
  D_1 \red^{*}&
  \news{\prop_6,\ldots,\prop_{12}}
   \bbra{{r_1}}{l_1: \about{{r_1}}{W\subst{\V{}{\epsilon}{V}}{y}},\ 
   l_2: \inact} \Par 
   \\ 
   & 
   \news{r_2 : \btinpt{\mathsf{str}},\ 
   r_3 : \btoutt{\mathsf{int}},\ 
 r_4 : {M^*}} \\
 & \qquad 
    \bsel{\dual r_1}{l_1}
     \propout{6}{}\binp{\dual r_1}{y}\appl{y}{(r_2, r_3, r_4)} \Par 
     \\ & \qquad 
     \propinp{6}{}\bout{\dual r_2}{\text{``Hello''}}\apropout{7}{} \Par 
     \propinp{7}{}\binp{\dual r_3}{t}\apropout{11}{} 
      \Par \\
     & \qquad \quad 
     \news{r_5 : \btinpt{\mathsf{str}},\ 
     r_6 : \btoutt{\mathsf{int}},\ 
   r_7 : {M^*}} \\
   & \qquad \quad \quad 
     \propinp{8}{}\bsel{\dual r_4} 
     {l_1}\propout{12}{}
      \binp{\dual r_4}{y}\appl{y}{(r_5, r_6, r_7)} \Par 
       \propinp{9}{}\bout{\dual r_5}{\text{``World''}}\apropout{10}{} \Par
       \\  
       & \qquad \quad \quad 
     \propinp{10}{}\binp{\dual r_6}{t}\apropout{11}{} \Par 
     %   \\
     % & \qquad \quad \quad \qquad 
     \propinp{11}{}\bsel{\dual r_7}{l_2}\apropout{12}{} \Par 
     \apropinp{12}{}  = D_2
\end{align*}
\noindent $D_2$ can mimic a silent select action on $r_1$; this is followed 
by a reception of value $W\subst{\V{}{\epsilon}{V}}{y}$ on name $\dual r_1$,
which is then applied to names $(r_2,r_3,r_4)$. The resulting 
process is as follows: 
\begin{align*}
  D_2  \red^{*}&
\news{\prop_8,\ldots,\prop_{12}}
    \news{r_2 : \btinpt{\mathsf{str}},\ 
    r_3 : \btoutt{\mathsf{int}},\ 
  r_4 : {M^*}} 
  \\
   &\qquad 
     % value here 
              \news{\prop^W_{1} \ldots \prop^W_{5}}
              \apropoutv{^W_1}{} \Par 
              \propinpv{^W_1}{}\binp{r_2}{t}\apropoutv{^W_2}{t} \Par 
              \propinpv{^W_2}{t}\bout{r_3}{\mathtt{len}(t)}\apropoutv{^W_3}{} \\ 
          & \qquad  \Par  
          \news{s_1}(\propinpv{^W_3}{}\propoutv{^W_4}{}\apropoutv{^W_5}{} 
          \Par \propinpv{^W_4}{}\appl{\V{}{\epsilon}{V}}{(r_4, \dual s_1)} \Par 
            \propinpv{^W_5}{}\about{s_1}{\V{}{\epsilon}{V}})
      \Par 
      \\ & \qquad 
      \bout{\dual r_2}{\text{``Hello''}}\apropout{7}{} \Par 
      \propinp{7}{}\binp{\dual r_3}{t}\apropout{11}{} 
       \Par \\
      & \qquad \quad 
      \news{r_5 : \btinpt{\mathsf{str}},\ 
      r_6 : \btoutt{\mathsf{int}},\ 
    r_7 : {M^*}} \\
    & \qquad \quad \quad 
      \propinp{8}{}\bsel{\dual r_4} 
      {l_1}\propout{12}{}
       \binp{\dual r_4}{y}\appl{y}{(r_5, r_6, r_7)} \Par 
        \propinp{9}{}\bout{\dual r_5}{\text{``World''}}\apropout{10}{} \Par
        \\  
        & \qquad \quad \quad 
      \propinp{10}{}\binp{\dual r_6}{t}\apropout{11}{} \Par 
      %   \\
      % & \qquad \quad \quad \qquad 
      \propinp{11}{}\bsel{\dual r_7}{l_2}\apropout{12}{} \Par 
      \apropinp{12}{}  = D_3 
\end{align*}

%\noindent Now, we can see in $D_2$ silent actions on $r$ are mimicked by 
%$r_2$ and $r_3$ as expected. 
\noindent 
The next interesting process emerges once silent actions on $r$ are mimicked by 
$r_2$ and $r_3$: 
\begin{align*}
  D_3  \red^{*}&
  \news{\prop_8, \ldots, \prop_{12}}
    \news{r_4 : {M^*}} 
     % value here 
	\news{s_1} \appl{\V{}{\epsilon}{V}}{(r_4, s_1)} \Par 
           \about{\dual s_1}{\V{}{\epsilon}{V}})
      \Par 
      \\ & \qquad 
      \news{r_5 : \btinpt{\mathsf{str}},\ 
      r_6 : \btoutt{\mathsf{int}},\ 
    r_7 : {M^*}} \\
    & \qquad \quad \quad 
    \propinp{8}{}\bsel{\dual r_4} 
    {l_1}\propout{12}{}
     \binp{\dual r_4}{y}\appl{y}{(r_5, r_6, r_7)} \Par 
      \propinp{9}{}\bout{\dual r_5}{\text{``World''}}\apropout{10}{} \Par
      \\  
      & \qquad \quad \quad 
    \propinp{10}{}\binp{\dual r_6}{t}\apropout{11}{} \Par 
    %   \\
    % & \qquad \quad \quad \qquad 
    \propinp{11}{}\bsel{\dual r_7}{l_2}\apropout{12}{} \Par 
    \apropinp{12}{} = D_4 
\end{align*}

\noindent In $D_4$, name $r_4$ with type $M^*$, is applied to the abstraction $\V{}{\epsilon}{V}$, which encapsulates a ``new instance'' of the  recursive branch process. After application, we obtain the following process: 
\begin{align*}
  D_4  \red^{*}&
  \news{\prop_8,\ldots,\prop_{12}}
    \news{r_4 : {M^*}} 
     % value here 
	\news{s_1} 
  % \abs{(x_{r_1}, x_{s_1})}
  \news{\prop^V_{1}\prop^V_{2}}
  \apropoutv{^V_1}{} 
  \Par \propinpv{^V_1}{}
  \binp{s_1}{y}\apropoutv{^V_2}{y}
  \Par 
  \\
  &
  \qquad \qquad \qquad 
  \propinpv{^V_2}{y}
  \bbra{r_4}
    {l_1: \about{r_4}{V_1},\ l_2: \inact} 
  \Par 
           \about{s_1}{\V{}{\epsilon}{V}})
      \Par 
      \\ & \qquad 
%      \bout{\dual r_2}{{``Hello''}}\apropout{10}{} \Par 
%      \propinp{10}{}\binp{\dual r_3}{t}\apropout{11}{} 
%        \Par \\
%      & \qquad \quad 
      \news{r_5 : \btinpt{\mathsf{str}},\ 
      r_6 : \btoutt{\mathsf{int}},\ 
    r_7 : {M^*}} \\
    & \qquad \quad \quad 
    \propinp{8}{}\bsel{\dual r_4} 
    {l_1}\propout{12}{}
     \binp{\dual r_4}{y}\appl{y}{(r_5, r_6, r_7)} \Par 
      \propinp{9}{}\bout{\dual r_5}{\text{``World''}}\apropout{10}{} \Par
      \\  
      & \qquad \quad \quad 
    \propinp{10}{}\binp{\dual r_6}{t}\apropout{11}{} \Par 
    %   \\
    % & \qquad \quad \quad \qquad 
    \propinp{11}{}\bsel{\dual r_7}{l_2}\apropout{12}{} \Par 
    \apropinp{12}{}  = D_5 
\end{align*}

\noindent Thus, we can see that after 
few administrative reductions (on $\prop^V_1$, $s_1$, and $\prop^V_2$) 
the process is able to mimic the a next selection on $r$ 
on name $r_4$. 
As the process again selects $l_1$, we can see that the
next selection will occur on name $r_7$, again 
typed with $M^*$. 
\hspace*{\fill}
	$\lhd$
\end{example}

We would like to finish this subsection with the following remark.
So far we have only considered recursive types which did not contain any actions between recursion and branching/selection.
However, types with prefixed branching

$$\trec{t}\alpha_1.\ldots.\alpha_n.\btbra{l_i:S_i}_{i\in I},$$
\noindent where $\alpha_1,\ldots,\alpha_n$ are some session prefixes, can also be accommodated in
the same framework, as these types can be written equivalently without prefixed branching: 
$$\alpha_1.\ldots.\alpha_n.\trec{t}\btbra{l_i:S_i\subst{\alpha_1....\alpha_n.\vart{t}}{\vart{t}}}_{i\in I}.$$

\subsection{Adapting Operational Correspondence}
\label{mst:ss:selbranchproof}
We briefly remark on how to adapt the operational correspondence result from
\Cref{ss:dynamic}. 
For the operational correspondence result, and the related lemmas, we must enforce additional constraints on the processes that we break down.
These concerns arise from the following fact.
When a type $\btbra{l_i : S_i}_{i \in I}$ is broken down as 
$$\Gt{\btbra{l_i : S_i}_{i \in I}} = \btbra{l_i: \btoutt{\lhot{\Gt{S_i}}}}_{i \in I},$$
an additional action gets introduced on the level of MST processes.
After performing the branching, an abstraction needs to be sent out.
This additional action will be matched by a corresponding abstraction-input action on the side of selection, if present.
However, this abstraction-sending action does not correspond to any action of the source process.

Therefore, to show the operational correspondence between the source term and its decomposition, we need to restrict our attention to processes in which branching and selection types are both present in (matching) pairs.
Specifically, we assume the following conditions on the source process $P$: 
\begin{itemize}
  \item $P$ is a well-typed, that is $\Gamma; \Delta;\Lambda \proves P \hastype \Proc$ with $\balan{\Delta}$;
  \item for any name $u$, $u \in \fn{P}$ with $u:S$ such that $S$ involves 
  selection or branching constructs if and only if 
  $\dual u \in \fn{P}$.
\end{itemize} 
\noindent Intuitively, these two conditions ensure that 
every branching action in $P$ has its complement (and vice-versa).
Note that for closed typeable processes both the balancedness condition and the second condition on names are vacuously true.

With this condition in place, we need to extend the relation $\relS$ in order to account for silent actions that are introduced by the breakdown of selection and branching constructs.
That is, when matching the original silent action involving selection/branching, the corresponding broken down process need to perform several silent actions, in order to be able to mimic the process continuation.

\section{Related Work}
\label{s:rw}

% \section{Related Work}
% \label{s:rw}

Here we discuss the positioning of our fragment of minimal session types with respect to (i)~other type systems for the $\pi$-calculus; (ii)~prior comparisons between session types and other type systems; and (iii)~Parrow's seminal work on trio decompositions of untyped processes.

\paragraph{Other type systems for the $\pi$-calculus}
\plscheck{The syntax of minimal session types (\Cref{mst:d:mtypesi}) contain constructs already known in the literature: while our types 
$\btout{U} \tinact$ and $\btinp{U} \tinact$ correspond, respectively, to the linear types $!^1[U]$ and $?^1[U]$ in the linearly-typed $\pi$-calculus by Kobayashi et al.~\cite{LinearPi} (where `$1$' indicates a linear multiplicity), our recursive (non-linear) types 
$ \trec{t}{\btout{U} \vart{t}}$ and $ \trec{t}{\btinp{U} \vart{t}}$ are reminiscent of the unlimited types in~\cite{LinearPi} (i.e.,  types with multiplicity $\omega$) and correspond to the recursive sortings in the simply-typed $\pi$-calculus of Pierce and Sangiorgi~\cite{PierceSangiorgi95}. We find it insightful that our identified fragment of session types for \HO contains known forms of types, with a decomposition function that satisfies static and dynamic correctness properties; this means that minimal session types stand between simple and linear types and standard session types. } 

\paragraph{Comparisons between session types and other type systems}
\plscheck{Our approach is broadly related to works that relate session types with other type systems 
for the $\pi$-calculus (cf.~\cite{DBLP:conf/unu/Kobayashi02,DBLP:conf/ppdp/DardhaGS12,DBLP:journals/iandc/DardhaGS17,DBLP:conf/concur/DemangeonH11,DBLP:journals/corr/GayGR14}).
As such, these works %target the \emph{relative expressiveness} of session-typed process languages, by 
establish formal relationships between two \emph{different type systems}; instead, here we relate the session type system for \HO with its \emph{own fragment} based on minimal session types.
%Thus, by explaining session types in terms of themselves, our work emerges as the first study of \emph{absolute expressiveness} in the context of session types.
}

\plscheck{Most related to our developments are formal connections between session types and linear types developed by Kobayashi~\cite{DBLP:conf/unu/Kobayashi02} and Dardha et al.~\cite{DBLP:conf/ppdp/DardhaGS12,DBLP:journals/iandc/DardhaGS17}. 
Kobayashi~\cite{DBLP:conf/unu/Kobayashi02} shows how to encode a finite session $\pi$-calculus into 
a $\pi$-calculus with linear types with usages (but without sequencing);
Dardha et al.~\cite{DBLP:conf/ppdp/DardhaGS12,DBLP:journals/iandc/DardhaGS17} establish the properties of Kobayashi's encoding.
This approach relies on two encodings, one  for processes and one for types: the former uses a freshly generated linear name to mimic each session action (thereby codifying a session name using multiple linear channels); the latter codifies sequencing in sessions by nesting payload types, relying on linear types extended with variant types.}

\plscheck{The difference between the Kobayashi-Dardha et al.'s encoding and our trio-based decomposition is that  we ``slice'' the $n$ actions occurring in a session $s : S$ along indexed names $s_1, \ldots, s_n$---$n$ slices of $S$---with each indexed name characterized by a minimal session type. 
Hence, while their encoding  could be characterized as codifying sequencing in a ``dynamic style'', via the freshly generated names, we follow a ``static style'' using names that are indexed according to the corresponding session type.}

\paragraph{MSTs and Parrow's trios}
We draw inspiration from prior work by Parrow~\cite{DBLP:conf/birthday/Parrow00}, who showed that every process 
in the untyped, summation-free $\pi$-calculus with replication is weakly bisimilar to its decomposition into trios (i.e., $P \approx \D{P}$). 
As already mentioned, our technical setting is different: our decomposition treats processes from a calculus without name-passing but with higher-order concurrency (abstraction-passing), supports recursive types, and can accommodate labeled choices.
Our goals are also different:
for us, trios are a relevant instrument for defining processes typable with  minimal session types, but they are not an end in themselves.
Still, we retain the definitional style and terminology for trios introduced by Parrow~\cite{DBLP:conf/birthday/Parrow00}, which are elegant and clear.

Our main results connect the typability and the behavior of a process with its decomposition, as witnessed by the static and dynamic correctness theorems.
Static correctness was not considered by Parrow, as he worked in an untyped setting.
As for dynamic correctness, a similar result was established in~\cite{DBLP:conf/birthday/Parrow00}, linking the process and its decomposition through weak bisimilarity.
In our setting we had to use a different, typed notion of bisimilarity.
A challenge here is that known notions of typed bisimilarity for session-typed processes, such as those given by Kouzapas et al.~\cite{KouzapasPY17}, only relate processes typed under the \emph{same} typing environments. 
To that extent, our notion of equivalence (MST bisimulation) is more flexible than prior related notions as it (i)~relates processes typable under different environments (e.g., $\Delta$ and $\Gt{\Delta}$) and (ii)~admits that actions along a name $s$ from $P$ can be  matched by $\D{P}$ using actions along indexed names $s_k$, for some $k$  (and viceversa).

\paragraph{Other related works}
 Jacobs~\cite{DBLP:conf/ecoop/Jacobs22} developed a small programming calculus with a single fork-like construct and a linear type system, which can be used to encode session-typed communications. 
His system can be seen as a distillation of Wadler's GV \cite{DBLP:conf/icfp/Wadler12} which is, in essence, a $\lambda$-calculus with session-based concurrency; in contrast, \HO can be seen as a $\pi$-calculus in which abstractions can be exchanged. 
While similar in spirit, our work and the developments by Jacobs are technically distant; we observe that the operational correspondences  in~\cite{DBLP:conf/ecoop/Jacobs22} are strictly simpler than our dynamic correspondence result (\Cref{mst:t:dyncorr}) although they are mechanized in the Rocq proof assistant.

Finally, we elaborate further on our choice of \HO as source language for minimal session types.
\HO is one of the sub-calculi of \HOp, a higher-order process calculus with recursion and both name- and abstraction-passing.
The basic theory of \HOp was studied by Kouzapas et al.~\cite{DBLP:conf/esop/KouzapasPY16,KouzapasPY17} as a hierarchy of session-typed calculi based on relative expressiveness. 
Our results enable us to place \HO with minimal session types firmly within this hierarchy.
Still, the definition of minimal session types does not rely on having \HO as source language, as they can be defined on top of other process languages. 
In fact, in separate work we have 
defined minimal session types on top of the first-order sub-calculus of \HOp~\cite{DBLP:conf/ppdp/ArslanagicP021}. 
This development attests that minimal session types admit meaningful formulations independently from the kind of communicated objects 
(abstractions or names).

\section{Concluding Remarks}
\label{s:concl}

% \section{Concluding Remarks}
% \label{s:concl}
We have studied \emph{minimal session types}, a fragment of session types, one of the most studied classes of {behavioral types} for message-passing  programs.
\plscheck{This fragment makes a very limited use of sequencing at the level of types; as such, it stands between linear types and standard session types}.
We relate standard and minimal session types through a \emph{decomposition} of session-typed processes, adopting the higher-order process calculus \HO as target language. 
Following Parrow~\cite{DBLP:conf/birthday/Parrow00}, we defined the decomposition of a process $P$, denoted  $\D{P}$, as a collection of \emph{trios} (processes with three sequential actions) that trigger each other mimicking the sequencing in the original process.
We proved that typability of $P$ using standard session types implies the typability of $\D{P}$ in the minimal fragment; we also established that  ${P}$ and $\D{P}$ are behaviorally equivalent through an {MST bisimulation}. 
Our results hold for all session types constructs, including labeled choices and recursive types.

From a foundational standpoint, our study of minimal session types is a conceptual contribution to the theory of behavioral types, in that we study the status of sequencing in theories of session types.
There are many session types variants, and their expressivity often comes at the price of an involved underlying theory.
Our work contributes in the opposite direction, as we identified a simple fragment of an existing session-typed framework~\cite{DBLP:conf/esop/KouzapasPY16,KouzapasPY17} that is quite expressive. %, which allows us to justify session types in terms of themselves.
Understanding further the underlying theory of minimal session types (e.g., notions such as type-based compatibility) is an exciting direction for future work.

As mentioned above,  one insight derived from our results is that sequencing in session types is convenient but not indispensable.
Convenience is an important factor in the design of type systems for message-passing programs, because types are abstract specifications of communication structures.  
By identifying sequencing as a source of redundancy, our minimal formulation of session types does not contradict or invalidate the prior work on standard session types and their extensions; rather, it contributes to our understanding of the sources of convenience of those advanced type systems.

In formulating minimal session types we have committed to a specific notion of minimality, tied to sequencing constructs in types---arguably the most distinctive feature in session types.
There could be other notions of minimality, unrelated to sequencing but worth exploring nevertheless. 
Consider, for instance,  \emph{context-free} session types~\cite{DBLP:conf/icfp/ThiemannV16}, which extend standard session types by allowing sequencing of the form $S; T$.
This form of sequential composition is quite powerful, and yet it could be seen as achieving a form of minimality different from the one we studied here: as illustrated in~\cite[Section 5]{DBLP:conf/icfp/ThiemannV16}, context-free session types allow to describe the communication of tree-structured data while minimizing the need for channel creation and avoiding channel passing.

Our work can also be seen as a new twist on Parrow's decomposition results in the \emph{untyped} setting~\cite{DBLP:conf/birthday/Parrow00}.
While Parrow's work indeed does not consider types, in fairness we must observe that when Parrow's work appeared (1996) the study of types (and typed behavioral equivalences) for the $\pi$-calculus was rather incipient (for instance, the widely known formulation of binary session types, given in~\cite{honda.vasconcelos.kubo:language-primitives}, appeared in 1998).
That said, we would like to stress that our results are not merely an extension of Parrow's work with session types, for types in our setting drastically narrow down the range of conceivable decompositions.
Additionally, in this work we exploit features not supported in~\cite{DBLP:conf/birthday/Parrow00}, most notably higher-order concurrency (cf. \Cref{s:opt}).

Finally, from a practical standpoint, we believe that our
approach paves a new avenue to the integration of session types in programming
languages whose type systems lack sequencing, such as Go. It is natural to
envision program analysis tools which, given a message-passing program that
should conform to protocols specified as session types, exploit our
decomposition as an intermediate step in the verification of communication
correctness. Remarkably, our decomposition lends itself naturally to an
implementation---in fact, we generated our examples automatically using \misty,
an associated artifact written in Haskell~\cite{DBLP:journals/darts/ArslanagicPV19}.

\paragraph{Acknowledgments}
We are grateful to Erik Voogd, who as a BSc student was one of the authors in the conference version of this paper~\cite{APV19}. 
We also thank the anonymous reviewers of previous versions of this paper for their comments and suggestions.

%%
%% Bibliography
%%

%% Please use bibtex,
%\pagebreak
%\bibliography{lipics-v2018-sample-article}

%\bibliographystyle{abbrv}% the recommended bibstyle
\bibliographystyle{abbrv}

\bibliography{session}

\newpage
\tableofcontents
\newpage

\appendix

% !TEX root = mst.tex

\section{Appendix to \secref{mst:ss:core}}
\label{app:core}

\subsection{Auxiliary Results}
\begin{remark}
	\label{r:polyadic-rules}
	We derive polyadic rules for typing \HOp as the expected extension of the typing rules for \HO:

  \begin{prooftree}
    \AxiomC{}
    \LeftLabel{\scriptsize PolyVar}
    \UnaryInfC{$\Gamma, \wtd x : \wtd U_x;\wtd y :
                  \wtd U_y; \es \proves
                  \wtd x \wtd y :\wtd U_x \wtd U_y$}
  \end{prooftree}

  \begin{prooftree}
  	\AxiomC{}
  	\LeftLabel{\scriptsize PolySess}
  	\UnaryInfC{$\Gamma;\es;\wtd u:\wtd S
  				\proves \wtd u \hastype \wtd S$}
  \end{prooftree}

	\begin{prooftree}
		\AxiomC{$\Gamma;\Lambda_1;\Delta \cat u:S \proves P \hastype \Proc$}
		\AxiomC{$\Gamma;\Lambda_2;\es \proves \wtd x \hastype \wtd U$}
		%\AxiomC{$u \in \Delta$}
		\LeftLabel{\scriptsize PolyRcv}		\BinaryInfC{$\Gamma \setminus \wtd x;\Lambda_1 \setminus \Lambda_2; \Delta \cat
		u:\btinp{\wtd U}S
		\proves
		\binp{u}{\wtd x}P \hastype \Proc
		$}
	\end{prooftree}

	\begin{prooftree}
		\AxiomC{$u:S \in \Delta$} 
		\AxiomC{$\Gamma;\Lambda_1;\Delta \proves P$}
		\AxiomC{$\Gamma;\Lambda_2;\es \proves
						\wtd x \hastype \wtd U$}
		\LeftLabel{\scriptsize PolySend}
		\TrinaryInfC{$\Gamma;\Lambda_1 \cat \Lambda_2;
							(\Delta \setminus u:S) \cat
							u:\btout{\wtd U}S \proves
							\bout{u}{\wtd x}P$}
	\end{prooftree}
	\begin{prooftree}
		\AxiomC{$\leadsto \in \{\multimap,\rightarrow \}$}
		\AxiomC{$\Gamma;\Lambda;\Delta_1 \proves V \hastype \wtd C \leadsto \diamond$}
		\AxiomC{$\Gamma;\es;\Delta_2 \proves \wtd u \hastype \wtd C$}
		\LeftLabel{\scriptsize PolyApp}
		\TrinaryInfC{$\Gamma;\Lambda;\Delta_1 \cat \Delta_2 \proves
						\appl{V}{\wtd u}$}
	\end{prooftree}

	\begin{prooftree}
		\AxiomC{$\Gamma;\Lambda;\Delta_1 \proves P \hastype \Proc$}
		\AxiomC{$\Gamma;\es;\Delta_2 \proves \wtd x \hastype \wtd C$}
		\LeftLabel{\scriptsize PolyAbs}
		\BinaryInfC{$\Gamma \setminus \wtd x;\Lambda;\Delta_1 \setminus \Delta_2 \proves \abs{\wtd x}{P} \hastype \lhot{\wtd C}$}
	\end{prooftree}

	\begin{prooftree}
		\AxiomC{$\Gamma \cat \wtd a : \wtd {\chtype{U}};\Lambda;\Delta \proves
					P$}
		%\AxiomC{$\wtd S_1 \ \dualof \ \wtd S_2$}
		\LeftLabel{\scriptsize PolyRes}
		\UnaryInfC{$\Gamma;\Lambda;\Delta \proves
						\news{\wtd a}{P}$}
	\end{prooftree}

	\begin{prooftree}
		\AxiomC{$\Gamma;\Lambda;\Delta \cat \wtd s:\wtd S_1
					\cat \dual{\wtd s}:\wtd S_2 \proves
					P$}
		\AxiomC{$\wtd S_1 \ \dualof \ \wtd S_2$}
		\LeftLabel{\scriptsize PolyResS}
		\BinaryInfC{$\Gamma;\Lambda;\Delta \proves
						\news{\wtd s}{P}$}
	\end{prooftree}

\end{remark}

\begin{lemm}[Substitution Lemma~\cite{DBLP:conf/esop/KouzapasPY16}]
  \label{lem:subst}
  $\Gamma;\Lambda;\Delta \cat x : S \proves P \hastype \Proc$
  and $u \notin \text{dom}(\Gamma,\Lambda,\Delta)$ imply
  $\Gamma;\Lambda;\Delta \cat u : S \proves P \subst{u}{x} \hastype \Proc$.
\end{lemm}

\begin{lemm}[Weakening, shared environment]
	\label{lem:weaken}
	If $\Gamma;\Lambda;\Delta \proves P \hastype \Proc$
	then
	$\Gamma \cat x:\shot{C};\Lambda;\Delta \proves P \hastype \Proc$ and
	$\Gamma \cat u:\chtype{U};\Lambda;\Delta \proves P \hastype \Proc$.
\end{lemm}

\begin{lemm}[Strengthening, shared environment]
	\label{lem:strength}
	The following hold:
	\begin{itemize}
	\item	If $\Gamma;\Lambda;\Delta \proves P \hastype \Proc$ and
	$x \notin \fv{P}$ then $\Gamma \setminus x;\Lambda;\Delta \proves P \hastype \Proc$.
	\item If $\Gamma;\Lambda;\Delta \proves P \hastype \Proc$ and
	$u \notin \fn{P}$ then $\Gamma \setminus u;\Lambda;\Delta \proves P \hastype \Proc$.
	\end{itemize}
\end{lemm}

% \section{Appendix to \secref{ss:extii}}
\newpage
\section{Appendix to \secref{mst:ss:staticcorr}}
\label{app:extii}

We use the following auxiliary lemma:

% \begin{lemma}
% 	\label{lem:indexcor}
% 	Let $\wtd z$ be tuple of channel names, $U$ a higher-order type, and $S$ a recursive 	session type.
% 	If $\wtd z : \Rts{}{s}{\btout{U}S}$ and $k = f(\btout{U}S)$
% 	then $z_k:\trec{t}\btout{\Gt{U}}\vart{t}$.
% \end{lemma}
\begin{restatable}[]{lemm}{lindexcor}
	\label{lem:indexcor}
	Let $\wtd z$ be tuple of channel names, $U$ a higher-order type, and $S$ a recursive 	session type.
	If $\wtd z : \Rts{}{s}{\btout{U}S}$ and $k = \indT{\btout{U}S}$
	then $z_k:\trec{t}\btout{\Gt{U}}\vart{t}$.
\end{restatable}

\subsection{Proof of \lemref{mst:t:typecore}}
\label{mst:app:typecore}

\typerec*

\begin{proof}
  By mutual induction on the structure of $P$ and $V$. 
  %The proof of the output case in Part~(1) relies on Parts~(2) and (3), whereas the proof of Parts~(2) and (3) relies on Part~(1).
%   Here, we analyze only Part~(1) of the theorem, as
%   Part~(2) %and Part~(3) 
%   is proven similarly: 
 \begin{enumerate}
 \item By assumption, $\Gamma;\Lambda;\Delta \cat \envR \proves P \hastype \Proc$. 
 %There are four cases, depending on the shape of $P$. 
 We consider four representative cases; the rest are similar. 
\begin{enumerate}
	\item Case $P=\inact$. The only rule that can be applied here is \textsc{Nil}.
	By inversion of this rule, we have:
	$\Gamma;\es;\es \proves \inact$.
	We shall then prove the following judgment:
	\begin{align}
		\Gt{\Gamma};\es;\Theta \proves \B{k}{\tilde x}{\inact} \hastype \Proc
	\end{align}

	\noindent where $\wtd x = \fv{\inact}=\es$ and $\Theta =
	\{\prop_k:\btinpt{\lhot{\tinact}}\}$. By \tabref{mst:t:bdowncore}, 
	$\B{k}{\epsilon}{\inact} =
	\propinp{k}{} \inact$. By convention  	$\propinp{k}{} \inact$ stands for $\propinp{k}{y} \inact$, with
	$\prop_k : \btinpt{\shot{\tinact}}$.
	The following tree proves this case:
	\def\proofSkipAmount{\vskip 1.2ex plus.8ex minus.4ex}
	\begin{prooftree}
	\AxiomC{}
	\LeftLabel{\scriptsize Nil}
		\UnaryInfC{$\Gamma';\es;\es \proves \inact \hastype \Proc$}
		\AxiomC{$\prop_k \notin \dom{\Gamma}$}
		\LeftLabel{\scriptsize End}
		\BinaryInfC{$\Gamma';\es;\prop_k : \tinact \proves \inact \hastype \Proc$}
		\AxiomC{}
		\LeftLabel{\scriptsize LVar}
		\UnaryInfC{$\Gt{\Gamma};y \hastype \lhot{\tinact};\es \proves y \hastype \lhot{\tinact}$}
		\LeftLabel{\scriptsize EProm}
		\UnaryInfC{$\Gamma';\es;\es \proves y \hastype \lhot{\tinact}$}
		\LeftLabel{\scriptsize Prom}
		\UnaryInfC{$\Gamma';\es;\es \proves y \hastype \shot{\tinact}$}
		\LeftLabel{\scriptsize Rcv}
		\BinaryInfC{$\Gt{\Gamma};\es;\Theta \proves \propinp{k}{} \inact \hastype \Proc$}
	\end{prooftree}
	\noindent where $\Gamma' = \Gt{\Gamma} \cat y : \shot{\tinact}$. 
	We know $\prop_k \notin \dom{\Gamma}$ since we use reserved names for 	propagators. %In the continuation of the proof 

	\item Case $P = \bout{u_i}{V}P'$. 
	We distinguish three
	sub-cases: \rom{1} $u_i \in \dom{\Delta}$, 
		\rom{2} $u_i \in \dom{\Gamma}$, and \rom{3} 
		$u_i \in \dom{\envR}$. 
	
	\noindent We consider sub-case \rom{1} first.
	% \begin{itemize} 
		% \item 	
		In this case Rule~\textsc{Send} can be applied:
		\begin{align}
			\label{pt:outputitr}
			\AxiomC{$\Gamma;\Lambda_1;\Delta_1 \cat {\envR}_1
			\proves P' \hastype \Proc$}
			\AxiomC{$\Gamma;\Lambda_2;\Delta_2 \cat {\envR}_2
			\proves V \hastype U$}
			\AxiomC{$u_i:S \in \Delta_1 \cat \Delta_2$}
			\LeftLabel{\scriptsize Send}
			\TrinaryInfC{$\Gamma;\Lambda_1 \cat \Lambda_2; 
			((\Delta_1 \cat \Delta_2)
						\setminus \{ u_i:S \})
						\cat u_i:\btout{U}S 
						\cat {\envR}_1 \cat {\envR}_2
						\proves \bout{u_i}{V}P' \hastype \Proc$}
			\DisplayProof
		\end{align}
			  
Let $\wtd w = \fv{P'}$. Also, let $\Gamma'_1 =
\Gamma \setminus \wtd w$ and $\Theta_1$ be a balanced environment such
that
\begin{align*} 
	\dom{\Theta_1}=\{\prop_{k+1},\ldots,\prop_{k+\plen{P'}}\}\cup
	\{\dual{\prop_{k+2}},\ldots,\dual{\prop_{k+\plen{P'}}}\}
\end{align*} 
and $\Theta_1(\prop_{k+1})=\btinpt{\wtd M_1}$ where
$\wtd M_1 = (\Gt{\Gamma},\Gt{\Lambda_1})(\wtd w)$.
We define:
\begin{align}
	\label{eq:prod_env}
	{\envPropR}_i = \prod_{r \in {\dom{{\envR}_i}}} \prop^r :
	\chtype{\lhot{\Rts{}{s}{{\envR}_i(r)}}} \ \text{for} \ i \in \{1,2\}
\end{align}

Then, by IH on the first assumption of \eqref{pt:outputitr} we have:
\begin{align}
	\label{eq:output-ih-1}
	\Gt{\Gamma'_1} \cat {\envPropR}_1;\es;
	\Gt{\Delta_1} \cat \Theta_1 \proves
	\B{k+1}{\tilde z}{P'} \hastype \Proc
\end{align}

Let $\wtd y = \fv{V}$ and $\Gamma'_2 = \Gamma \setminus \wtd y$.
Then, by IH (Part 2)
on the second assumption of \eqref{pt:outputitr} we have:
\begin{align}
	\label{eq:output-ih-2}
	\Gt{\Gamma} \cat {\envPropR}_2;\Gt{\Lambda_2};\Gt{\Delta_2} 
	 \proves
	\V{k+1}{\tilde y}{V} \hastype \Gt{U}
\end{align}

%Since further proof developments for this subcase are similar,
We may notice that if $U = \shot{C}$  then $\Lambda_2 = \es$ and $\Delta_2 = \es$.
Let $\wtd x = \fv{P}$.
We define $\Theta = \Theta_1 \cat \Theta'$, where:
\begin{align*}
	\Theta' = \prop_{k}: \btinpt{\wtd M}  \cat
	\dual{\prop_{k+1}}:\btoutt{\wtd M_2} 
\end{align*}

\noindent with $\wtd M = (\Gt{\Gamma}\cat \Gt{\Lambda_1 \cat
\Lambda_2})(\wtd x)$. By \defref{mst:def:sizeproc}, we know $\plen{P} =
\plen{P'} + 1$, so
$$
\dom{\Theta}=\{\prop_k,\ldots,\prop_{k+\plen{P}-1}\} \cup
\{\dual{\prop_{k+1}},\ldots,\dual{\prop_{k+\plen{P}-1}}\}
$$
	By construction $\Theta$ is balanced since 
	$\Theta(\prop_{k+1}) \dualof \Theta(\dual{\prop_{k+1}})$ 
	and
		$\Theta_1$ is balanced.
By \tabref{mst:t:bdowncore}, we have:
\begin{align*}
	\B{k}{\tilde x}{\bout{u_i}{V}P'} =
	\propinp{k}{\wtd x}
	\bbout{u_i}{\V{k+1}{\tilde y}
	{V\incrname{u}{i}}} \propout{k+1}{\wtd w}
	\inact \Par
	\B{k+1}{\tilde w}{P'\incrname{u}{i}}
\end{align*}
		
\noindent We know $\fv{P'} \subseteq \fv{P}$ and $\fv{V} \subseteq
\fv{P}$, i.e., $\wtd w \subseteq \wtd x$ and $\wtd y \subseteq \wtd x$.
Let $\Gamma_1 = \Gamma \setminus \wtd x$ and $\envPropR =
{\envPropR}_1 \cat {\envPropR}_2$. We shall prove the following
judgment:
\begin{align}
	\AxiomC{$\Gt{\Gamma_1} \cat {\envPropR}; \es; \Gt{((\Delta_1\cat\Delta_2)\setminus \{ u_i:S \})
	\cat u_i:\btout{U}S} \cat \Theta \proves
	\B{k}{\tilde x}{\bout{u_i}{V}P'} \hastype \Proc$}
	\DisplayProof
\end{align}
		
Let $\sigma = \incrname{u}{i}$. To type the left-hand side component of
$\B{k}{\tilde x}{\bout{u_i}{V}P'}$
we use some auxiliary derivations:
\begin{align}
	\label{pt:zsend}
	\AxiomC{}
	\LeftLabel{\scriptsize Nil}
	\UnaryInfC{$\Gt{\Gamma}\cat {\envPropR};\es;\es \proves \inact \hastype \Proc$}
	\LeftLabel{\scriptsize End}
	\UnaryInfC{$\Gt{\Gamma}\cat {\envPropR};\es; 
	\dual{\prop_{k+1}} : \tinact \proves \inact \hastype \Proc$}
	\AxiomC{}
	\LeftLabel{\scriptsize PolyVar}
	\UnaryInfC{$\Gt{\Gamma}\cat {\envPropR};\Gt{\Lambda_1};\es
				\proves \wtd w \hastype \wtd M_2$}
	\LeftLabel{\scriptsize PolySend}
	\BinaryInfC{$\Gt{\Gamma}\cat {\envPropR};\Gt{\Lambda_1};
				\dual{\prop_{k+1}}:\btout{\wtd M_2} \tinact
				\proves \bout{\dual {\prop_{k+1}}}{\wtd w} \inact
				\hastype \Proc$}
	\LeftLabel{\scriptsize End}
	\UnaryInfC{$\Gt{\Gamma}\cat {\envPropR};\Gt{\Lambda_1};
				\dual{\prop_{k+1}}:\btout{\wtd M_2} \tinact
				\cat u_i:\tinact
				\proves \bout{\dual {\prop_{k+1}}}{\wtd w} \inact
				\hastype \Proc$}
	\DisplayProof
\end{align}
		
\begin{align}
	\label{pt:output-weaken}
\AxiomC{\eqref{eq:output-ih-2}}
\LeftLabel{\scriptsize (\lemref{lem:subst}) with 
	$\subst{\tilde n}{\tilde u}$}
\UnaryInfC{$\Gt{\Gamma}\cat {\envPropR}_2;\Gt{\Lambda_2};
\Gt{\Delta_2\sigma}  \proves
			\V{k+1}{\tilde y}{V\sigma} \hastype \Gt{U}$}
\LeftLabel{\scriptsize (\lemref{lem:strength}) with 
	${\envPropR}_1$}
\UnaryInfC{$\Gt{\Gamma}\cat {\envPropR};\Gt{\Lambda_2};
\Gt{\Delta_2\sigma} 
				 \proves
			\V{k+1}{\tilde y}{V\sigma} \hastype \Gt{U}$}
\DisplayProof
\end{align}

\begin{align}
	\label{pt:uisend}
	\AxiomC{$u_i:\btout{\Gt{U}}\tinact  \in \Gt{\Delta_2\sigma}
	\cat u_i:\btout{\Gt{U}}\tinact \cat
	\dual{\prop_{k+1}}:\btout{\wtd M_2}\tinact$} 
	\AxiomC{\eqref{pt:zsend}}
	\AxiomC{\eqref{pt:output-weaken}}
	\LeftLabel{\scriptsize Send}
	\TrinaryInfC{\begin{tabular}{c}
				$\Gt{\Gamma} \cat {\envPropR};\Gt{\Lambda_1\cat\Lambda_2};
				\Gt{\Delta_2\sigma}
						\cat u_i:\btout{\Gt{U}}\tinact \cat
						\dual{\prop_{k+1}}:\btout{\wtd M_2}\tinact
						\proves$ \\ 
						$\bbout{u_i}{\V{k+1}{\tilde y}{V\sigma}}
					\propout{k+1}{\wtd w} \inact \hastype \Proc$
					\end{tabular}}
		\LeftLabel{\scriptsize End}
	\UnaryInfC{\begin{tabular}{c}
				$\Gt{\Gamma} \cat {\envPropR};\Gt{\Lambda_1\cat\Lambda_2};
				\Gt{\Delta_2\sigma}
						\cat u_i:\btout{\Gt{U}}\tinact \cat
						\dual{\prop_{k+1}}:\btout{\wtd M_2}\tinact \cat \prop_k:\tinact
						\proves$ \\
						$\bbout{u_i}{\V{k+1}{\tilde y}{V\sigma}}
					\propout{k+1}{\wtd w} \inact \hastype \Proc$
					\end{tabular}}
	\DisplayProof
		\end{align}
		
\begin{align}
\label{pt:stvalue}
\AxiomC{\eqref{pt:uisend}}
\AxiomC{}
\LeftLabel{\scriptsize PolyVar}
\UnaryInfC{$\Gt{\Gamma} \cat {\envPropR};\Gt{\Lambda_2};\es \proves
				\wtd x : \wtd M$}
\LeftLabel{\scriptsize PolyRcv}
\BinaryInfC{%\begin{tabular}{c}
				$\Gt{\Gamma_1} \cat {\envPropR};\es;\Gt{\Delta_2\sigma} \cat
				u_i:\btout{\Gt{U}} \tinact 
				\cat \Theta'
				\proves \propinp{k}{\wtd x}
				\bbout{u_i}{\V{k+1}{\tilde y}{V\sigma}}
				\propout{k+1}{\wtd w} \inact
				\hastype \Proc$
				%\end{tabular}
				}
\DisplayProof
\end{align}
The following tree proves this case:
\begin{align}
	\label{pt:output1}
	\AxiomC{\eqref{pt:stvalue}}
	\AxiomC{\eqref{eq:output-ih-1}}
	\LeftLabel{\scriptsize (\lemref{lem:subst} with $\subst{\tilde n}
	{\tilde u}$)}
	\UnaryInfC{$\Gt{\Gamma'_1} \cat {\envPropR}_1;\es;\Gt{\Delta_1\sigma} 
	\cat \Theta_1
				\proves \B{k+r+1}{\tilde w}{P'\sigma} \hastype \Proc$}
	\LeftLabel{\scriptsize (\lemref{lem:strength} with
	$\tilde x \setminus \tilde w$ and ${\envPropR}_2$)}
	\UnaryInfC{$\Gt{\Gamma_1} \cat {\envPropR};\es;\Gt{\Delta_1\sigma} 
	\cat
				\Theta_1 \proves
				\B{k+r+1}{\tilde w}{P'\sigma} \hastype \Proc$}
	\LeftLabel{\scriptsize Par}
	\BinaryInfC{$\Gt{\Gamma_1} \cat {\envPropR};\es; 
	\Gt{((\Delta_1\cat\Delta_2)
				\setminus \{ u_i:S \})
				\cat u_i:\btout{U}S} \cat \Theta \proves
				\B{k}{\tilde x}{\bout{u_i}{V}P'} \hastype \Proc$}
	\DisplayProof
\end{align}
\noindent where $\wtd n =
(u_{i+1},\ldots,u_{i+\len{\Gt{S}}})$ and $\wtd u =
(u_i,\ldots,u_{i+\len{\Gt{S}}-1})$. 
This concludes sub-case~\rom{1}. 
% \item 

\smallskip

\noindent Now, we consider sub-case \rom{2}. 
In this sub-case Rule~\textsc{Req} can be applied:
	\begin{align}
		\label{pt:outputitr-2}
		\AxiomC{$\Gamma;\es;\es \proves u \hastype \chtype{U}$}
		\AxiomC{$\Gamma;\Lambda;\Delta_1 \cat {\envR}_1 
		\hastype P' \hastype \Proc$}
	\AxiomC{$\Gamma;\es;\Delta_2 \cat {\envR}_2
		\proves V \hastype U$}
	\LeftLabel{\scriptsize Req}
		\TrinaryInfC{$\Gamma;\Lambda;\Delta_1 \cat \Delta_2 
		\cat {\envR}_1 \cat {\envR}_2
		\proves \bout{u}{V}P' \hastype \Proc$}
		\DisplayProof
	\end{align}
		
Let $\wtd w = \fv{P'}$. Further, let
$\Gamma'_1 = \Gamma \setminus \wtd w$ and let $\Theta_1$,
 ${\envPropR}_1$, and ${\envPropR}_2$ be environments
defined as in sub-case \rom{1}.
By IH on the second assumption of \eqref{pt:outputitr-2} we have:
\begin{align}
	\label{eq:output-ih-1-2}
	\Gt{\Gamma_1'} \cat {\envPropR}_1;\es;\Gt{\Delta_1} 
	\cat \Theta_1 \proves
	\B{k+1}{\tilde w}{P'} \hastype \Proc
\end{align}

	Let $\wtd y = \fv{V}$. 
	%and $\Gamma'_2 = \Gamma \setminus \wtd y$.
By IH on the second assumption of \eqref{pt:outputitr} we have:
\begin{align}
	\label{eq:output-ih-2-2}
	\Gt{\Gamma} \cat {\envPropR}_2;\es;\Gt{\Delta_2}  \proves
	\V{k+1}{\tilde y}{V} \hastype \Gt{U}
\end{align}

% need to construct \Theta
Let $\wtd x = \fv{P}$ and $\Gamma_1 = \Gamma \setminus \wtd x$.
We define $\Theta = \Theta_1  \cat \Theta'$, where:
\begin{align*}
	\Theta' = \prop_{k}: \btinpt{\wtd M}  \cat
	\dual{\prop_{k+1}}:\btoutt{\wtd M_2}  
\end{align*}
	\noindent with
$\wtd M = (\Gt{\Gamma}\cat \Gt{\Lambda})(\wtd x)$.
By \defref{mst:def:sizeproc}, we know $\plen{P} = \plen{P'} + 1$, so
$$\dom{\Theta}=(\prop_k,\ldots,\prop_{k+\plen{P}-1}) \cup
(\dual{\prop_{k+1}},\ldots,\dual{\prop_{k+\plen{P}-1}})$$
By construction $\Theta$ is balanced since $\Theta(\prop_{k+1}) \dualof 
\Theta(\dual{\prop_{k+1}})$ and
		$\Theta_1$ is balanced.
By \tabref{mst:t:bdowncore}, we have:
\begin{align*}
\B{k}{\tilde x}{\bout{u_i}{V}P'} =
\propinp{k}{\wtd x}
\bbout{u_i}{\V{k+1}{\tilde y}{V}} \propout{k+1}{\wtd w} \inact \Par
\B{k+1}{\tilde w}{P'}
\end{align*}
\noindent We know
$\fv{P'} \subseteq \fv{P}$ and $\fv{V} \subseteq \fv{P}$, i.e., 
$\wtd w \subseteq \wtd x$ and $\wtd y \subseteq \wtd x$. 
%side conditions are met.
% TODO continue proof check
		
To prove 
$$\Gt{\Gamma_1} \cat {\envPropR};\es; \Gt{\Delta_1\cat\Delta_2} \cat \Theta
				\proves \B{k}{\tilde x}{\bout{u_i}{V}P'}$$
we use some auxiliary derivations:
\begin{align}
	\label{pt:zsend-2}
	\AxiomC{}
	\LeftLabel{\scriptsize Nil}
	\UnaryInfC{$\Gt{\Gamma} \cat {\envPropR};\es;\es \proves \inact \hastype \Proc$}
	\LeftLabel{\scriptsize End}
	\UnaryInfC{$\Gt{\Gamma} \cat {\envPropR};\es;\dual{\prop_{k+1}}:\tinact \proves \inact \hastype \Proc$}
	\AxiomC{}
	\LeftLabel{\scriptsize PolyVar}
	\UnaryInfC{$\Gt{\Gamma} \cat {\envPropR};\Gt{\Lambda};\es
				\proves \wtd w : \wtd M_2$}
	\LeftLabel{\scriptsize PolySend}
	\BinaryInfC{$\Gt{\Gamma} \cat {\envPropR};\Gt{\Lambda};
		\dual{\prop_{k+1}}:\btout{\wtd M_2} \tinact
		\proves \bout{\dual {\prop_{k+1}}}{\wtd w} 
		\inact \hastype \Proc$}
	\DisplayProof
\end{align}

\begin{align}
	\label{pt:uisend-2}
	\AxiomC{\eqref{pt:zsend-2}}
	\AxiomC{\eqref{eq:output-ih-2-2}}
	\LeftLabel{\scriptsize (\lemref{lem:strength} with ${\envPropR}_1$)}
	\UnaryInfC{$\Gt{\Gamma} \cat {\envPropR};\es;\Gt{\Delta_2} 
	 \proves
	\V{k+1}{\tilde y}{V} \hastype \Gt{U}$}
	\LeftLabel{\scriptsize Req}
	\BinaryInfC{$\Gt{\Gamma} \cat {\envPropR};\Gt{\Lambda};\Gt{\Delta_2}
			 \cat
			\dual{\prop_{k+1}}:\btout{\wtd M_2} \tinact
			\proves \bbout{u_i}{\V{k+1}{\tilde y}{V}}
			\propout{k+1}{\wtd w} \inact \hastype \Proc$}
	\DisplayProof
\end{align}

\begin{align}
\label{pt:stvalue-2}
\AxiomC{\eqref{pt:uisend-2}}
\AxiomC{}
\LeftLabel{\scriptsize PolyVar}
\UnaryInfC{$\Gt{\Gamma} \cat {\envPropR};\Gt{\Lambda};\es \proves
			\wtd x \hastype \wtd M$}
\LeftLabel{\scriptsize PolyRcv}
\BinaryInfC{$\Gt{\Gamma_1} \cat {\envPropR};\es;\Gt{\Delta_2} 
		\cat \Theta'
		\proves \propinp{k}{\wtd x}
		\bbout{u_i}{\V{k+1}{\tilde y}{V}}
			\propout{k+1}{\wtd w} \inact \hastype \Proc$}
\DisplayProof
\end{align}
		
The following tree proves this case:
\begin{align}
	\label{pt:output1-2}
	\AxiomC{\eqref{pt:stvalue-2}}
	\AxiomC{\eqref{eq:output-ih-1-2}}
	\LeftLabel{\scriptsize (\lemref{lem:strength} with $\tilde x 
	\setminus \tilde w$ and ${\envPropR}_2$)}
	\UnaryInfC{$\Gt{\Gamma_1} \cat {\envPropR};\es;\Gt{\Delta_1} \cat 
	\Theta_1 \proves
	\B{k+1}{\tilde w}{P'} \hastype \Proc$}
	\LeftLabel{\scriptsize Par}
	\BinaryInfC{$\Gt{\Gamma_1} \cat \envPropR;\es; 
	\Gt{\Delta_1\cat\Delta_2} 
	\cat \Theta
			\proves \B{k}{\tilde x}{\bout{u_i}{V}P'} \hastype \Proc$}
	\DisplayProof
\end{align}
This concludes sub-case~\rom{2}. 

\smallskip

\noindent Now, we consider sub-case~\rom{3}.  
Here we know $P = \bout{u_i}{V}P'$ and $u_i:S \in \envR$. In this
case Rule~\textsc{Send} can be applied:
\begin{align}
\label{pt:r-outputitr}
\AxiomC{$\Gamma;\Lambda_1;\Delta_1 \cat {\envR}_1 \proves 
P' \hastype \Proc$}
\AxiomC{$\Gamma;\Lambda_2;\Delta_2 \cat {\envR}_2 \proves 
V \hastype U$}
\AxiomC{$u_i:S' \in {\envR}_1 \cat {\envR}_2$}
\LeftLabel{\scriptsize Send}
\TrinaryInfC{$\Gamma;\Lambda_1 \cat \Lambda_2;
				\Delta_1 \cat \Delta_2 \cat
				(({\envR}_1\cat{\envR}_2)
				\setminus \{ u_i:S' \})
				\cat u_i:\btout{U}S'
				\proves \bout{u_i}{V}P' \hastype \Proc$}
\DisplayProof
\end{align}
		Let $\wtd w = \fv{P'}$. Let $\Theta_1$,
		 $\Theta'$, $\Theta$, ${\envPropR}_1$, and ${\envPropR}_2$ be defined as in the 
		previous sub-case. Also, let
		$\Gamma'_1 = \Gamma \setminus \wtd w$.
		Then, by IH on the
		first assumption of \eqref{pt:r-outputitr} we have:
			\begin{align}
			\label{eq:r-case1-ih}
				\Gt{\Gamma'_1}\cat {\envPropR}_1;\es;\Gt{\Delta_1}\cat\Theta_1
				\proves \B{k+1}{\tilde w}{P'} \hastype \Proc
			\end{align}
			Let $\Gamma'_2 = \Gamma \setminus \wtd y$.
			Then, by IH (Part 2) on the
		second assumption of \eqref{pt:r-outputitr} we have:
		\begin{align}
			\Gt{\Gamma'_2}\cat{\envPropR}_2;\Gt{\Lambda_2};\Gt{\Delta_2}
			\proves \V{k+1}{\tilde y}{V} \hastype \Gt{U}
		\end{align}
				
		%We define $\envPropR = {\envPropR}_1 \cat {\envPropR}_2$.
		%Let $k = f(S)$. 
		By \tabref{mst:t:bdowncore} we have:
		\begin{align*}
			\B{k}{\tilde x}{P} = \propinp{k}{\wtd x}
					& \bbout{\prop^u}{N_V}  \inact \Par
			\B{k+1}{\tilde w}{P'} \\
			& \text{where} \ N_V = \abs{\wtd z}
					{\bout{z_{\indT{S}}}{\V{k+1}{\tilde y}{V}}}
					\big(\apropout{k+1}{\wtd w} \Par 
					\recprovx{u}{x}{\wtd z}
					% \binp{\prop^r}{b}(\appl{b}{\wtd z}) 
				\big) 
		\end{align*}
	
We notice that $u_i \in \rfn{V}\cat\rfn{P}$ since $u_i$ has tail-recursive type $S$.
	Hence, by \eqref{eq:prod_env} we know 
	$(\envPropR_1,\envPropR_2)(\prop^u)=\chtype{\lhot{\Rts{}{s}{S}}}$.
	Further, we know that $S = \btout{U}S'$ and by 
	\defref{mst:def:typesdecomp},
	$\Rts{}{s}{S}=\Rts{}{s}{S'}$.
	So we define
	$\envPropR = {\envPropR}_1 \cat {\envPropR}_2$. Let $\Gamma_1 = \Gamma \setminus \wtd x$ where $\wtd x = \fv{P}$. 
We shall prove the following judgment:
\begin{align*}
	\Gt{\Gamma_1} \cat \envPropR;\es; \Gt{\Delta_1 \cat \Delta_2} 
	\cat \Theta
				\proves \B{k}{\tilde x}{\bout{u_i}{V}P'} \hastype \Proc
\end{align*}

		We use some auxiliary derivations:
		\begin{align}
			\label{pt:r-send3}
			\AxiomC{}
			\LeftLabel{\scriptsize LVar}
			\UnaryInfC{\begin{tabular}{c}$\Gt{\Gamma_1} \cat \envPropR;
				\recpvar{x}:\lhot{\Rts{}{s}{S}};\es \proves$ \\
							$\recpvar{x} \hastype \lhot{\Rts{}{s}{S}}$
							\end{tabular}}
			\AxiomC{}
			\LeftLabel{\scriptsize PolySess}
			\UnaryInfC{
			\begin{tabular}{c}
			$\Gt{\Gamma_1} \cat \envPropR;\es;\wtd z:\Rts{}{s}{S}
						\proves$ \\ $\wtd z \hastype \Rts{}{s}{S}$
						\end{tabular}}
			\LeftLabel{\scriptsize PolyApp}
			\BinaryInfC{$\Gt{\Gamma_1} \cat \envPropR;
			\recpvar{x}:\lhot{\Rts{}{s}{S}};
						\wtd z: \Rts{}{s}{S} \proves
							\appl{\recpvar{x}}{\wtd z} \hastype \Proc$}
			\DisplayProof
		\end{align}
		
		\begin{align}
			\label{pt:r-send2}
		\AxiomC{\eqref{pt:r-send3}}
			\AxiomC{}
		\LeftLabel{\scriptsize Sh}
			\UnaryInfC{\begin{tabular}{c}$\Gt{\Gamma} \cat \envPropR;\es;\es \proves$  \\
						$\prop^u \hastype \chtype{\lhot{\Rts{}{s}{S}}}$
						\end{tabular}}
			\AxiomC{}
			\LeftLabel{\scriptsize LVar}
			\UnaryInfC{\begin{tabular}{c}$\Gt{\Gamma} \cat 
				\envPropR;\recpvar{x}:\lhot{\Rts{}{s}{S}};\es \proves$ \\
						$\recpvar{x} \hastype \lhot{\Rts{}{s}{S}}$
						\end{tabular}
						}
			\LeftLabel{\scriptsize Acc}
			\TrinaryInfC{$\Gt{\Gamma} \cat \envPropR;\es;\Theta' \cat
						\wtd z: \Rts{}{s}{S} \proves
						\recprovx{u}{x}{\wtd z}
						% \binp{\prop^r}{b}(\appl{b}{\wtd z}) 
						\hastype \Proc$}
			\DisplayProof
		\end{align}
	
		\begin{align}
			\label{pt:r-send5-0}
			\AxiomC{} 
			\LeftLabel{\scriptsize Nil}
			\UnaryInfC{$\Gt{\Gamma} \cat \envPropR; \es; \es \proves 
			\inact \hastype \Proc$}
			\AxiomC{$\prop_{k+1} \not\in \dom{\Gamma, \envPropR}$}
			\LeftLabel{\scriptsize End}
			\BinaryInfC{$\Gt{\Gamma} \cat \envPropR; \es; \prop_{k+1}:\tinact \proves 
			\inact \hastype \Proc$}
			\DisplayProof 
		\end{align}
	
		\begin{align}
			\label{pt:r-send4}
			\AxiomC{$\prop_{k+1} :\btout{\wtd M_2}\tinact \in \Theta'$} 
			% \AxiomC{} 
			% \LeftLabel{\scriptsize Nil}
			% \UnaryInfC{$\Gt{\Gamma} \cat \envPropR; \es; \es \proves 
			% \inact \hastype \Proc$}
			\AxiomC{\eqref{pt:r-send5-0}} 
			\AxiomC{}
			\LeftLabel{\scriptsize PolyVar}
			\UnaryInfC{$\Gt{\Gamma} \cat \envPropR;\Gt{\Lambda_2};\es
						\proves \wtd w \hastype \wtd M_2$}
			\LeftLabel{\scriptsize PolySend}
			\TrinaryInfC{$\Gt{\Gamma} \cat \envPropR;\Gt{\Lambda_1};
			\dual{\prop_{k+1}}:\btoutt{\wtd M_2}  
			% \cat
						% \wtd z : \Rts{}{s}{S}
						\proves \apropout{k+1}{\wtd w} 
						 \hastype \Proc$}
			\DisplayProof
		\end{align}

		\begin{align}
			\label{pt:r-send-5}
			\AxiomC{\eqref{pt:r-send4}} 
			\AxiomC{\eqref{pt:r-send2}} 
			\LeftLabel{\scriptsize Par}
			\BinaryInfC{$\Gt{\Gamma} \cat \envPropR;\Gt{\Lambda_1};
			\dual{\prop_{k+1}}:\btoutt{\wtd M_2}  \cat
						\wtd z : \Rts{}{s}{S}
						\proves \apropout{k+1}{\wtd w} 
						\Par 
						\recprovx{u}{x}{\wtd z}
						% \binp{\prop^r}{b}(\appl{b}{\wtd z}) 
						\hastype \Proc$}
			\DisplayProof 
		\end{align}

		\begin{align}
			\label{pt:r-rcv1}
			\AxiomC{\eqref{pt:r-send-5}}
			\AxiomC{\eqref{eq:r-case1-ih}} 
			\LeftLabel{\scriptsize (\lemref{lem:weaken}) with ${\envPropR}_1$}
			\UnaryInfC{\begin{tabular}{c}$\Gt{\Gamma'_2} \cat \envPropR;\Gt{\Lambda_2};
				\es 
						\proves$$\V{k+1}{\tilde y}{V} \hastype \Gt{U}$
						\end{tabular}}
		\LeftLabel{\scriptsize (\lemref{lem:weaken}) with $\tilde z$}
			\UnaryInfC{$\Gt{\Gamma} \cat \envPropR;\Gt{\Lambda_2};\es 
						\proves \V{k+1}{\tilde y}{V} \hastype \Gt{U}$}
			\LeftLabel{\scriptsize Send}
				\BinaryInfC{
				\begin{tabular}{c}
	$\Gt{\Gamma} \cat \envPropR;\Gt{\Lambda_1 \cat \Lambda_2};
						\Gt{\Delta_2} \cat \wtd z:\Rts{}{s}{S} \cat
						\dual{\prop_{k+1}}:\btoutt{\wtd M_2}   \proves$ \\
						${\bout{z_{\indT{S}}}{\V{k+1}{\tilde y}{V}}}
						\big(\apropout{k+1}{\wtd w} \Par 
						\recprovx{u}{x}{\wtd z}
						% \binp{\prop^r}{b}(\appl{b}{\wtd z}) 
						\big)$
						\end{tabular}}
			\DisplayProof
		\end{align}
		By \lemref{lem:indexcor} we know
			that if $\wtd z:\Rts{}{s}{S}$ then
			$z_{\indT{S}}:\trec{t}\btout{\Gt{U}}\vart{t}$.
	
		\begin{align}
			\label{pt:r-abs}
			\AxiomC{\eqref{pt:r-rcv1}}
			\AxiomC{}
			\LeftLabel{\scriptsize PolySess}
			\UnaryInfC{$\Gt{\Gamma} \cat \envPropR;\es;\wtd z : \Rts{}{s}{S}
						\proves \wtd z \hastype \Rts{}{s}{S}$}
			\LeftLabel{\scriptsize PolyAbs}
			\BinaryInfC{$\Gt{\Gamma} \cat \envPropR; \Gt{\Lambda_1 \cat \Lambda_2};
						\Gt{\Delta_2} \cat \dual{\prop_{k+1}}:\btoutt{\wtd M_2}  
				\proves N_V \hastype \lhot{\Rts{}{s}{S}}$}
			\DisplayProof
		\end{align}

		\begin{align}
			\label{pt:r-rec}
			\AxiomC{}
			\LeftLabel{\scriptsize LVar}
			\UnaryInfC{$\Gt{\Gamma} \cat \envPropR;\es;\es
						\proves \prop^u \hastype \chtype{\lhot{\Rts{}{s}{S}}})$}
		\AxiomC{}
		\LeftLabel{\scriptsize Nil}
			\UnaryInfC{$\Gt{\Gamma} \cat \envPropR;\es;\es
						\proves \inact \hastype \Proc$}
			\AxiomC{\eqref{pt:r-abs}}
			\LeftLabel{\scriptsize Req}
			\TrinaryInfC{$\Gt{\Gamma} \cat \envPropR;\Gt{\Lambda_1 \cat \Lambda_2};
						\Gt{\Delta_2} \cat \dual{\prop_{k+1}}:\btoutt{\wtd M_2}  
						\proves \bbout{\prop^u}{N_V} \inact \hastype \Proc$}
			\DisplayProof
		\end{align}

		\begin{align}
			\label{pt:r-rcv}
			\AxiomC{\eqref{pt:r-rec}}
			\AxiomC{}
			\LeftLabel{\scriptsize PolyVar}
			\UnaryInfC{$\Gt{\Gamma} \cat \envPropR; \Gt{\Lambda_1 \cat \Lambda_2};\es
						\proves \wtd x \hastype \wtd M$}
			\LeftLabel{\scriptsize PolyRcv}
			\BinaryInfC{$\Gt{\Gamma_1} \cat \envPropR;\es;
						\Gt{\Delta_2}  \cat \Theta'
						\proves \propinp{k}{\wtd x}
							\bbout{\prop^u}{N_V} \inact \hastype \Proc$}
				\DisplayProof
		\end{align}

		The following tree proves this case:
		\begin{align}
			\AxiomC{\eqref{pt:r-rcv}}
				\AxiomC{\eqref{eq:r-case1-ih}}
				\LeftLabel{\scriptsize (\lemref{lem:weaken}) with ${\envPropR}_2$}
				\UnaryInfC{$\Gt{\Gamma'_1}\cat {\envPropR};\es;
					\Gt{\Delta_1}\cat\Theta_1
						\proves \B{k+1}{\tilde w}{P'} \hastype \Proc$}
				\LeftLabel{\scriptsize (\lemref{lem:strength}) with $\tilde y$}
			\UnaryInfC{$\Gt{\Gamma_1}\cat {\envPropR};\es;
					\Gt{\Delta_1}\cat\Theta_1
						\proves \B{k+1}{\tilde w}{P'} \hastype \Proc$}
			\LeftLabel{\scriptsize Par}
			\BinaryInfC{$\Gt{\Gamma_1} \cat \envPropR;\es; \Gt{\Delta_1 \cat \Delta_2}
							\cat \Theta
						\proves \B{k}{\tilde x}{\bout{r}{V}P'} \hastype \Proc$}
			\DisplayProof
		\end{align}
		This concludes the analysis for the output case $P = \bout{u_i}{V}P'$. 
		We remark that the proof for the case when $V = y$ is a special case of the above proof, where $\tilde y = \fv{y} = y$, 
		$\V{k+1}{\tilde y}{y} = y$ and 
		it holds that $y \sigma = y$. 
	% \end{itemize} 

\item Case $P = \binp{u_i}{y}P'$. We distinguish two sub-cases: 
\rom{1} $u_i \in \dom{\Delta}$, \rom{2} $u_i \in \dom{\Gamma}$, 
and~\rom{3} $u_i \in \dom{\envR}$.
We consider sub-cases
\rom{1} and \rom{2}; we omit sub-case~\rom{3} as it follows the same 
reasoning as the corresponding sub-case of the previous (Send) case.

\noindent 
We consider sub-case \rom{1} first. For this case Rule \textsc{Rcv}
can be applied:
	\begin{align}
	\label{pt:inputInv}
	\AxiomC{$\Gamma; \Lambda_1; \Delta' \cat u_i : S
	\cat {\envR} \proves P' \hastype \Proc$}
	\AxiomC{$\Gamma; \Lambda_2; \es \proves y \hastype U$}
	\LeftLabel{\scriptsize Rcv}
	\BinaryInfC{$
			\Gamma \setminus y;\Lambda_1 \setminus \Lambda_2;
					\Delta' \cat u_i: \btinp{U}S 
					\cat {\envR} 
					\proves \binp{u_i}{y}P' \hastype \Proc
			$}
	\DisplayProof
	\end{align}

	Let $\wtd x = \fv{P}$ and $\wtd w = \fv{P'}$.  Also, let
		$\Gamma_1'=\Gamma \setminus \wtd w$ and $\Theta_1$ be a
		balanced environment such that
		$$
		\dom{\Theta_1}=\{\prop_{k+1},\ldots,\prop_{k+\plen{P'}}\}
		\cup \{\dual{\prop_{k+2}},\ldots,\dual{\prop_{k+\plen{P'}}}\}
		$$
		and ${\Theta_1(\prop_{k+1})}=\btinpt{\wtd M'}$
		where $\wtd M' = (\Gt{\Gamma},\Gt{\Lambda_1})(\wtd w)$.
		We define:
		\begin{align}
			\label{eq:prod_env2}
			{\envPropR} = \prod_{r \in {\dom{{\envR}}}} \prop^r :
			\chtype{\lhot{\Rts{}{s}{{\envR}_i(r)}}} 
		\end{align}

	Then, by IH on the first assumption of \eqref{pt:inputInv} we know:
	\begin{align}
		\label{eq:input-ih}
		\Gt{\Gamma'_1} \cat {\envPropR};\es;\Gt{\Delta' \cat u_i:S} \cat \Theta_1 \proves
		\B{k+1}{\tilde w}{P'} \hastype \Proc
	\end{align}
	
	By \Cref{mst:def:typesdecomp,def:typesdecompenv} and the second
	assumption of \eqref{pt:inputInv} we have:
	\begin{align}
		\label{eq:input-ih2}
		\Gt{\Gamma};\Gt{\Lambda_2};\es \proves y \hastype \Gt{U}
	\end{align}
	
	We define $\Theta = \Theta_1 \cat \Theta'$, where
	\begin{align*}
	\Theta' = \prop_k:\btinpt{\wtd M}   \cat \dual {\prop_{k+1}}:\btoutt{\wtd M'} 
	\end{align*}
	with $\wtd M = (\Gt{\Gamma},\Gt{\Lambda_1 \setminus \Lambda_2})(\wtd x)$.
	By \defref{mst:def:sizeproc},
	$\plen{P} = \plen{P'} + 1$ so
	$$\dom{\Theta} = \{\prop_k,\ldots,\prop_{k+\plen{P}-1}\}
		\cup \{\dual{\prop_{k+1}},\ldots,\dual{\prop_{k+\plen{P}-1}}\}$$
		and $\Theta$ is balanced since $\Theta(\prop_{k+1}) \dualof \Theta(\dual{\prop_{k+1}})$ and
			$\Theta_1$ is balanced.
	By \tabref{mst:t:bdowncore}:
		\begin{align*}
		\B{k}{\tilde x}{\binp{u_i}{y}P'} = \propinp{k}{\wtd
		x}\binp{u_i}{y}\propout{k+1}{\wtd w} \inact \Par \B{k+1}{\tilde
		w}{P'\incrname{u}{i}}
	\end{align*}
	
	Let $\Gamma_1 = \Gamma \setminus \wtd x$.
	We shall prove the following judgment:
	\begin{align*}
	\Gt{\Gamma_1 \setminus y} \cat {\envPropR};\es;\Gt{\Delta' \cat u_i:\btinp{U}S}
		\cat \Theta
		\proves
			\B{k}{\tilde x}{\binp{u_i}{y}P'}
	\end{align*}
	
	%\noindent where $\Gamma_1 = \Gamma \setminus \wtd x$.
	
	The left-hand side component of $\B{k}{\tilde x}{\binp{u_i}{y}P'}$ is typed using
		some auxiliary derivations:
	\begin{align}
		%\gamma
		\label{st:input1}
		\AxiomC{}
		\LeftLabel{\scriptsize Nil}
		\UnaryInfC{$\Gt{\Gamma} \cat {\envPropR};\es;\es \proves \inact \hastype \Proc$}
		\LeftLabel{\scriptsize End}
		\UnaryInfC{$\Gt{\Gamma} \cat {\envPropR};\es;\dual{\prop_{k+1}}:\tinact \proves \inact \hastype \Proc$}
		\AxiomC{}
		\LeftLabel{\scriptsize PolyVar}
		\UnaryInfC{$\Gt{\Gamma} \cat {\envPropR};\Gt{\Lambda_1};\es
					\proves \wtd w \hastype \wtd M'$}
		\LeftLabel{\scriptsize PolySend}
		\BinaryInfC{$\Gt{\Gamma} \cat {\envPropR};\Gt{\Lambda_1}\cat\Gt{\Lambda_2};
					\dual {\prop_{k+1}}:\btout{\wtd M'} \tinact
					\proves \propout{k+1}{\wtd w} \inact \hastype \Proc$}
		\LeftLabel{\scriptsize End}
		\UnaryInfC{$\Gt{\Gamma} \cat {\envPropR};\Gt{\Lambda_1}\cat\Gt{\Lambda_2};
					\dual {\prop_{k+1}}:\btout{\wtd M'} \tinact
					\cat u_i:\tinact
					\proves \propout{k+1}{\wtd w} \inact \hastype \Proc$}
		\DisplayProof
	\end{align}
	\begin{align}
		%beta
		\label{st:input-rcv}
		\AxiomC{\eqref{st:input1}}
		\AxiomC{\eqref{eq:input-ih2}}
		\LeftLabel{\scriptsize Rcv}
		\BinaryInfC{
		\begin{tabular}{c}
			$\Gt{\Gamma \setminus y} \cat {\envPropR};\Gt{\Lambda_1 \setminus \Lambda_2}; u_i:\btinp{\Gt{U}} \tinact \cat
					\dual {\prop_{k+1}}:\btout{\wtd M'} \tinact \proves$ \\
					$\binp{u_i}{y}\propout{k+1}{\wtd w} \inact \hastype \Proc$
		\end{tabular}
	}
		\LeftLabel{\scriptsize End}
		\UnaryInfC{ 
		\begin{tabular}{c}
		$\Gt{\Gamma \setminus y} \cat {\envPropR};\Gt{\Lambda_1 \setminus \Lambda_2}; u_i:\btinp{\Gt{U}} \tinact \cat
					\dual {\prop_{k+1}}:\btout{\wtd M'} \tinact
					\cat \prop_k: \tinact \proves$ \\ 
					$\binp{u_i}{y}\propout{k+1}{\wtd w} \inact \hastype \Proc$
		\end{tabular}}
		\DisplayProof
	\end{align}
	
	\begin{align}
		\label{st:input3}
		\AxiomC{\eqref{st:input-rcv}}
		\AxiomC{}
		\LeftLabel{\scriptsize PolyVar}
		\UnaryInfC{$\Gt{\Gamma \setminus y} \cat {\envPropR};\Gt{\Lambda_1};\es
		\proves \wtd x \hastype \wtd M$}
		\LeftLabel{\scriptsize PolyRcv}
		\BinaryInfC{
			$
		\Gt{\Gamma_1 \setminus y} \cat {\envPropR};\es;u_i:\btinp{\Gt{U}} \tinact
		\cat \Theta' \proves \propinp{k}{\wtd
		x}\binp{u_i}{y}\propout{k+1}{\wtd w} \inact \hastype \Proc
		$
		}
	\DisplayProof
	\end{align}
	
	The following tree
	proves this case:
	\begin{align}
		\label{pt:input}
		\AxiomC{\eqref{st:input3}}
		\AxiomC{\eqref{eq:input-ih}}
		\LeftLabel{\scriptsize (\lemref{lem:subst}) with $\subst{\tilde n}{\tilde u}$}
		\UnaryInfC{
		\begin{tabular}{c}
		$\Gt{\Gamma_1 \setminus y} \cat {\envPropR};\es; \Gt{\Delta' \cat u_{i+1}:S} \cat \Theta_1 \proves$ \\
		$\B{k+1}{\tilde w}{P'\incrname{u}{i}} \hastype \Proc$
		\end{tabular}}
		\LeftLabel{\scriptsize Par}
		\BinaryInfC{\begin{tabular}{c}
					$\Gt{\Gamma_1 \backslash y} \cat {\envPropR};\es;
					\Gt{\Delta' \cat u_i:\btinp{U}S}
					\cat \Theta \proves$ \\
					$\propinp{k}{\wtd x}
					\binp{u_i}{y}\propout{k+1}{\wtd w} \inact \Par \B{k+1}
					{\tilde w}{P'\incrname{u}{i}} \hastype \Proc$
					\end{tabular}}
		\DisplayProof
	\end{align}
		\noindent
		where $\wtd n =
		(u_{i+1},\ldots,u_{i+\len{\Gt{S}}})$ and $\wtd u =
		(u_i,\ldots,u_{i+\len{\Gt{S}}-1})$.
	%  	We may notice that $\Gamma'_1 = \Gamma_1 \setminus y$. 
		We may notice that
		if $y \in \fv{P'}$ then $\Gamma'_1 = \Gamma_1 \setminus y$. On the other hand,
		when $y \notin \fv{P'}$ then $\Gamma'_1 = \Gamma_1$ so
		we need to apply \lemref{lem:strength} with $y$
			after \lemref{lem:subst} to \eqref{eq:input-ih}
			in \eqref{pt:input}.
		Note that we have used the following for the right assumption of \eqref{pt:input}:
		\begin{align*}
			\Gt{\Delta' \cat u_i:S}\subst{\tilde n}{\tilde u} &=
			\Gt{\Delta' \cat u_{i+1}:S} \\
			\B{k+1}{\tilde w}{P'}\subst{\tilde n}{\tilde u} &=
			\B{k+1}{\tilde w}{P'\incrname{u}{i}}
		\end{align*}

		This concludes sub-case \rom{1}. We now consider sub-case \rom{2}, i.e.,
		$u_i \in \dom{\Gamma}$.
		Here Rule~\textsc{Acc} can be applied:
		\begin{align}
			\label{pt:inputptr-subcase2}
			\AxiomC{$\Gamma;\es;\es \proves u_i \hastype \chtype{U}$}
			\AxiomC{$\Gamma;\Lambda_1;\Delta \cat {\envR} \proves P' \hastype \Proc$}
		\AxiomC{$\Gamma;\Lambda_2;\es \proves y \hastype U$}
		\LeftLabel{\scriptsize Acc}
			\TrinaryInfC{$
			\Gamma \setminus y;\Lambda_1 \setminus \Lambda_2;\Delta 
			\cat {\envR}
							\proves \binp{u_i}{y}P' \hastype \Proc
							$}
			\DisplayProof
		\end{align}
		  
		Let $\wtd x = \fv{P}$ and $\wtd w =\fv{P'}$. Furthermore, let
		$\Theta_1$, $\Theta$, $\Gamma_1$, $\Gamma_1'$, and ${\envPropR}$ 
		be defined as in sub-case
		\rom{1}. By IH on the second assumption of \eqref{pt:inputptr-subcase2} we have:
		\begin{align}
		\label{eq:input-ih-2}
		\Gt{\Gamma'_1} \cat {\envPropR};\es;\Gt{\Delta} \cat 
		\Theta_1 \proves \B{k+1}{\tilde w}{P'} \hastype \Proc
		\end{align}
	
		By \Cref{mst:def:typesdecomp,def:typesdecompenv} and the first
			assumption of \eqref{pt:inputptr-subcase2} we have:
		\begin{align}
			\label{eq:input-ih2-2}
			\Gt{\Gamma} \cat {\envPropR};\es;\es \proves u_i \hastype \chtype{\Gt{U}}
		\end{align}
	
		By \Cref{mst:def:typesdecomp,def:typesdecompenv} and the third
			assumption of \eqref{pt:inputptr-subcase2} we have:
		\begin{align}
			\label{eq:input-ih3-2}
			\Gt{\Gamma} \cat {\envPropR};\Gt{\Lambda_2};\es \proves y \hastype \Gt{U}
		\end{align}
		  
			   By \tabref{mst:t:bdowncore}, we have:
			   \begin{align}
					\B{k}{\tilde x}{\binp{u_i}{y}P'} = \propinp{k}{\wtd x}
					\binp{u_i}{y}\propout{k+1}{\wtd w} \inact \Par 
					\B{k+1}{\tilde w}
					{P'\incrname{u}{i}}
			   \end{align}
		  
			   We shall prove the following judgment:
			   \begin{align}
				   \Gt{\Gamma_1 \setminus y} \cat {\envPropR};\es;\Gt{\Delta} \cat \Theta \proves
				   \B{k}{\tilde x}{\binp{u_i}{y}P'} \hastype \Proc
			   \end{align}
		  
			  To this end, we use some auxiliary derivations:
			\begin{align}
			  %\gamma
			  \label{st:input1-2}
			  \AxiomC{}
			  \LeftLabel{\scriptsize Nil}
			  \UnaryInfC{$\Gt{\Gamma} \cat {\envPropR};\es;\es \proves \inact \hastype \Proc$}
			  \LeftLabel{\scriptsize End}
			  \UnaryInfC{$\Gt{\Gamma} \cat {\envPropR};\es;\dual{\prop_{k+1}}:\tinact \proves \inact \hastype \Proc$}
			  \AxiomC{}
			  \LeftLabel{\scriptsize PolyVar}
			  \UnaryInfC{$\Gt{\Gamma} \cat {\envPropR};\Gt{\Lambda_1};\es
						  \proves \wtd w \hastype \wtd M'$}
			  \LeftLabel{\scriptsize PolySend}
			  \BinaryInfC{$\Gt{\Gamma} \cat {\envPropR};\Gt{\Lambda_1};
							\dual {\prop_{k+1}}:\btout{\wtd M'} \tinact
							\proves \propout{k+1}{\wtd w} \inact \hastype \Proc$}
			  \DisplayProof
			\end{align}
		  
	\begin{align}
		%beta
		\label{st:input2-2}
		\AxiomC{\eqref{eq:input-ih2-2}}
		\AxiomC{\eqref{st:input1-2}}
		\AxiomC{\eqref{eq:input-ih3-2}}
		\LeftLabel{\scriptsize Acc}
		\TrinaryInfC{$\Gt{\Gamma \setminus y} \cat {\envPropR};\Gt{\Lambda_1};
					\dual{\prop_{k+1}}:\btout{\wtd M'}
					\tinact \proves 
					\binp{u_i}{y}\propout{k+1}{\wtd w} \inact \hastype \Proc$}
		\LeftLabel{\scriptsize End}
		\UnaryInfC{$\Gt{\Gamma \setminus y} \cat {\envPropR};
		\Gt{\Lambda_1 \setminus \Lambda_2};
					\dual{\prop_{k+1}}:\btout{\wtd M'}
					\tinact \cat \prop_k:\tinact 
					\proves \binp{u_i}{y}\propout{k+1}{\wtd w} 
					\inact \hastype \Proc$}
		\DisplayProof
	\end{align}
		  
	\begin{align}
		\label{st:input3-2}
		\AxiomC{\eqref{st:input2-2}}
		\AxiomC{}
		\LeftLabel{\scriptsize PolyVar}
		\UnaryInfC{$\Gt{\Gamma \setminus y} \cat {\envPropR};\Gt{\Lambda_1 \setminus \Lambda_2};\es \proves
					\wtd x \hastype \wtd M$}
		\LeftLabel{\scriptsize PolyRcv}
		\BinaryInfC{$\Gt{\Gamma_1 \setminus y} \cat {\envPropR};\es;\Theta' \proves
						\propinp{k}{\wtd x}
						\binp{u_i}{y}\propout{k+1}{\wtd w} \inact
						\hastype \Proc$}
		\DisplayProof
	\end{align}
	
	The following tree
	proves this sub-case:
	\begin{align}
		\label{pt:input-2}
		\AxiomC{\eqref{st:input3-2}}
		\AxiomC{\eqref{eq:input-ih-2}}
	%    \LeftLabel{\scriptsize (\lemref{lem:strength}) with $y$}
	%    \UnaryInfC{$\Gt{\Gamma_1 \setminus y};\es;\Gt{\Delta} \cat \Theta_1 \proves
	%                	\B{k+1}{\tilde x'}{P'}$}
		\LeftLabel{\scriptsize Par}
		\BinaryInfC{$\Gt{\Gamma_1 \setminus y} \cat {\envPropR};\es;\Gt{\Delta} \cat \Theta \proves
				\propinp{k}{\wtd x}\binp{u_i}{y}\propout{k+1}{\wtd w}
						\inact \Par \B{k+1}{\tilde w}{P'} \hastype \Proc$}
		\DisplayProof
	\end{align}
	As in sub-case \rom{1}, we may notice that if $y \in \fv{P'}$ then $\Gamma'_1 = \Gamma_1 \setminus y$.
	On the other hand, if $y \notin \fv{P'}$ then $\Gamma'_1 = \Gamma_1$ so we need to apply
	\lemref{lem:strength} with $y$ to \eqref{eq:input-ih-2} in \eqref{pt:input-2}.
	This concludes the analysis for the input case $P = \binp{u_i}{y}P'$.
	This concludes sub-case \rom{2}. 
	
	Now, we consider sub-case \rom{3}. 
	Here we know $P = \binp{u_i}{V}P'$ and $u_i:S \in \envR$.
	\begin{align}
		\label{pt:inputInv-tr}
		\AxiomC{$\Gamma; \Lambda_1; \Delta' \cat 
		 {\envR} 
		 \cat u_i : S'
		 \proves P' \hastype \Proc$}
		\AxiomC{$\Gamma; \Lambda_2; \es \proves y \hastype U$}
		\LeftLabel{\scriptsize Rcv}
		\BinaryInfC{$
				\Gamma \setminus y;\Lambda_1 \setminus \Lambda_2;
						\Delta' \cat 
						 {\envR} \cat u_i: \btinp{U}S' 
						\proves \binp{u_i}{y}P' \hastype \Proc
				$}
		\DisplayProof
		\end{align}
	\noindent Let $\wtd w = \fv{P'}$. Let $\Theta_1$,$\Theta_2$, $\Theta'$, and 
	$\envPropR$ be defined as in the sub-case~\rom{1}. 
	Also, let
	$\Gamma'_1 = \Gamma \setminus \wtd w$.
	Then, by IH on the
		first assumption of~\eqref{pt:inputInv-tr} we have:
	\begin{align}
		\label{eq:input-ih-tr}
		\Gt{\Gamma'_1} \cat {\envPropR};\es;\Gt{\Delta'} \cat \Theta_1 
		\proves
		\B{k+1}{\tilde w}{P'} \hastype \Proc
	\end{align}

	Further, by IH on the second assumption of~\eqref{pt:inputInv-tr} we have: 
	\begin{align}
		\label{eq:input-ih-tr-2}
		\Gt{\Gamma} \cat \envPropR; \Gt{\Lambda_2};\es 
					\proves y \hastype \Gt{U}
	\end{align}

	By~\tabref{mst:t:bdowncore} we have:
		\begin{align*}
			\B{k}{\tilde x}{P} &= \propinp{k}{\wtd x}\abbout{\prop^u}
			{N_y}
			% \binp{\prop^u}{b}(\appl{b}{\wtd z}) 
			\Par \B{k+1}{\tilde w}{P'} \\ 
			& \qquad \qquad \qquad \text{where} \ 
			N_y = \abs{\wtd z}{\binp{z_{\indT{S}}}{y}
			\big(\apropout{k+1}{\wtd w} \Par 
			\recprov{u}{x}{\wtd z} \big)}
		\end{align*}
	\noindent Notice that $u_i \in \rfn{P}$ as $\tr(u_i)$. Hence, by~\eqref{eq:prod_env2} 
    we know $\envPropR(\prop^u)=\chtype{\lhot{\Rts{}{s}{S}}}$. 
	Further, we know that $S = \btinp{U}S'$ and by~\defref{mst:def:typesdecomp},
	$\Rts{}{s}{S}=\Rts{}{s}{S'}$.
	Let $\Gamma_1 = \Gamma \setminus \wtd x$ where $\wtd x = \fv{P}$. 
	Thus, we shall prove the following judgment:
	\begin{align*}
		\Gt{\Gamma_1 \setminus y} \cat \envPropR;\es; \Gt{\Delta'} 
		\cat \Theta
					\proves \B{k}{\tilde x}{\binpt{u_i}{y}P'} \hastype \Proc
	\end{align*}

	We use auxiliary derivations:

	\begin{align}
		\label{pt:rcv-tr-9}
		\AxiomC{}
		\LeftLabel{\scriptsize LVar}
		\UnaryInfC{\begin{tabular}{c}$\Gt{\Gamma_1} \cat \envPropR;
			\recpvar{x}:\lhot{\Rts{}{s}{S}};\es \proves$ \\
						$\recpvar{x} \hastype \lhot{\Rts{}{s}{S}}$
						\end{tabular}}
		\AxiomC{}
		\LeftLabel{\scriptsize PolySess}
		\UnaryInfC{
		\begin{tabular}{c}
		$\Gt{\Gamma_1} \cat \envPropR;\es;\wtd z:\Rts{}{s}{S}
					\proves$ \\ $\wtd z \hastype \Rts{}{s}{S}$
					\end{tabular}}
		\LeftLabel{\scriptsize PolyApp}
		\BinaryInfC{$\Gt{\Gamma_1} \cat \envPropR;
		\recpvar{x}:\lhot{\Rts{}{s}{S}};
					\wtd z: \Rts{}{s}{S} \proves
						\appl{\recpvar{x}}{\wtd z} \hastype \Proc$}
		\DisplayProof
	\end{align}

	\begin{align}
		\label{pt:rcv-tr-8}
	\AxiomC{\eqref{pt:rcv-tr-9}}
		\AxiomC{}
	\LeftLabel{\scriptsize Sh}
		\UnaryInfC{\begin{tabular}{c}$\Gt{\Gamma} \cat \envPropR;\es;\es \proves$  \\
					$\prop^u \hastype \chtype{\lhot{\Rts{}{s}{S}}}$
					\end{tabular}}
		\AxiomC{}
		\LeftLabel{\scriptsize LVar}
		\UnaryInfC{\begin{tabular}{c}$\Gt{\Gamma} \cat 
			\envPropR;\recpvar{x}:\lhot{\Rts{}{s}{S}};\es \proves$ \\
					$\recpvar{x} \hastype \lhot{\Rts{}{s}{S}}$
					\end{tabular}
					}
		\LeftLabel{\scriptsize Acc}
		\TrinaryInfC{$\Gt{\Gamma} \cat \envPropR;\es; 
					\wtd z: \Rts{}{s}{S} \proves
					\recprovx{u}{x}{\wtd z}
					% \binp{\prop^r}{b}(\appl{b}{\wtd z}) 
					\hastype \Proc$}
		\DisplayProof
	\end{align}

	\begin{align}
		\label{pt:rcv-tr-7}
		\AxiomC{} 
		\LeftLabel{\scriptsize Nil}
		\UnaryInfC{$\Gt{\Gamma} \cat \envPropR; \es; \es \proves 
		\inact \hastype \Proc$}
		\AxiomC{$\prop_{k+1} \not\in \dom{\Gamma, \envPropR}$}
		\LeftLabel{\scriptsize End}
		\BinaryInfC{$\Gt{\Gamma} \cat \envPropR; \es; \prop_{k+1}:\tinact \proves 
		\inact \hastype \Proc$}
		\DisplayProof 
	\end{align}

	\begin{align}
		\label{pt:rcv-tr-6}
		\AxiomC{$\prop_{k+1} :\btout{\wtd M'}\tinact \in \Theta'$} 
		% \AxiomC{} 
		% \LeftLabel{\scriptsize Nil}
		% \UnaryInfC{$\Gt{\Gamma} \cat \envPropR; \es; \es \proves 
		% \inact \hastype \Proc$}
		\AxiomC{\eqref{pt:rcv-tr-7}} 
		\AxiomC{}
		\LeftLabel{\scriptsize PolyVar}
		\UnaryInfC{$\Gt{\Gamma} \cat \envPropR;\Gt{\Lambda_1};\es
					\proves \wtd w \hastype \wtd M'$}
		\LeftLabel{\scriptsize PolySend}
		\TrinaryInfC{$\Gt{\Gamma} \cat \envPropR;\Gt{\Lambda_1};
		\dual {\prop_{k+1}}:\btoutt{\wtd M'} 
		% \cat
					% \wtd z : \Rts{}{s}{S}
					\proves \apropout{k+1}{\wtd w} 
					 \hastype \Proc$}
		\DisplayProof
	\end{align}

	\begin{align}
		\label{pt:rcv-tr-5}
		\AxiomC{\eqref{pt:rcv-tr-6}} 
		\AxiomC{\eqref{pt:rcv-tr-8}} 
		\LeftLabel{\scriptsize Par}
		\BinaryInfC{$\Gt{\Gamma} \cat \envPropR;\Gt{\Lambda_1};\Theta' \cat
					\wtd z : \Rts{}{s}{S}
					\proves \apropout{k+1}{\wtd w} 
					\Par 
					\recprovx{u}{x}{\wtd z}
					% \binp{\prop^r}{b}(\appl{b}{\wtd z}) 
					\hastype \Proc$}
		\DisplayProof 
	\end{align}

	\begin{align}
		\label{pt:rcv-tr-4}
		\AxiomC{\eqref{pt:rcv-tr-5}}
		\AxiomC{\eqref{eq:input-ih-tr-2}} 
		\LeftLabel{\scriptsize Rcv}
			\BinaryInfC{
			\begin{tabular}{c}
$\Gt{\Gamma \setminus y} \cat \envPropR; 
\Gt{\Lambda_1 \setminus \Lambda_2}; 
\dual {\prop_{k+1}}:\btoutt{\wtd M'} 
\cat \wtd z:\Rts{}{s}{S}  \proves$ 
% \\
					$\binp{z_{\indT{S}}}{y}
					\big(\apropout{k+1}{\wtd w} \Par 
					\recprovx{u}{x}{\wtd z}
					% \binp{\prop^r}{b}(\appl{b}{\wtd z}) 
					\big)$
					\end{tabular}}
		\DisplayProof
	\end{align}

	By \lemref{lem:indexcor} we know
	that if $\wtd z:\Rts{}{s}{S}$ then
	$z_{\indT{S}}:\trec{t}\btinp{\Gt{U}}\vart{t}$.

	\begin{align}
		\label{pt:rcv-tr-3}
		\AxiomC{\eqref{pt:r-rcv1}}
		\AxiomC{}
		\LeftLabel{\scriptsize PolySess}
		\UnaryInfC{$\Gt{\Gamma \setminus y} \cat \envPropR;\es;\wtd z : \Rts{}{s}{S}
					\proves \wtd z \hastype \Rts{}{s}{S}$}
		\LeftLabel{\scriptsize PolyAbs}
		\BinaryInfC{$
		\Gt{\Gamma \setminus y} \cat \envPropR; 
		\Gt{\Lambda_1 \setminus \Lambda_2}; 
		\dual {\prop_{k+1}}:\btoutt{\wtd M'} 
			\proves N_y \hastype \lhot{\Rts{}{s}{S}}
			$}
		\DisplayProof
	\end{align}

	\begin{align}
		\label{pt:r-rcv-inp2}
		\AxiomC{}
		\LeftLabel{\scriptsize LVar}
		\UnaryInfC{$\Gt{\Gamma \setminus y} \cat \envPropR;\es;\es
					\proves \prop^u \hastype \chtype{\lhot{\Rts{}{s}{S}}}$}
	\AxiomC{}
	\LeftLabel{\scriptsize Nil}
		\UnaryInfC{$\Gt{\Gamma \setminus y} \cat \envPropR;\es;\es
					\proves \inact \hastype \Proc$}
		\AxiomC{\eqref{pt:rcv-tr-3}}
		\LeftLabel{\scriptsize Req}
		\TrinaryInfC{$\Gt{\Gamma \setminus y} \cat \envPropR;
		\Gt{\Lambda_1 \setminus \Lambda_2}; 
		\dual {\prop_{k+1}}:\btoutt{\wtd M'} 
					\proves \bbout{\prop^u}{N_y} \inact \hastype \Proc$}
		\DisplayProof
	\end{align}

	\begin{align}
		\label{pt:r-rcv-inp}
		\AxiomC{\eqref{pt:r-rcv-inp2}}
		\AxiomC{}
		\LeftLabel{\scriptsize PolyVar}
		\UnaryInfC{$\Gt{\Gamma \setminus y} \cat \envPropR; 
		\Gt{\Lambda};\es
					\proves \wtd x \hastype \wtd M$}
		\LeftLabel{\scriptsize PolyRcv}
		\BinaryInfC{$\Gt{\Gamma_1 \setminus y} \cat \envPropR;\es;
					 \Theta'
					\proves \propinp{k}{\wtd x}
						\bbout{\prop^u}{N_y} \inact \hastype \Proc$}
			\DisplayProof
	\end{align}

	The following tree proves this case:
		\begin{align}
			\AxiomC{\eqref{pt:r-rcv-inp}}
				\AxiomC{\eqref{eq:input-ih-tr}}
				% \LeftLabel{\scriptsize (\lemref{lem:weaken}) with ${\envPropR}_2$}
				% \UnaryInfC{$\Gt{\Gamma'_1}\cat {\envPropR};\es;
				% 	\Gt{\Delta'}\cat\Theta_1
				% 		\proves \B{k+1}{\tilde w}{P'} \hastype \Proc$}
				\LeftLabel{\scriptsize (\lemref{lem:strength}) with $\tilde y$}
			\UnaryInfC{$\Gt{\Gamma_1 \setminus y}\cat {\envPropR};\es;
					\Gt{\Delta'}\cat\Theta_1
						\proves \B{k+1}{\tilde w}{P'} \hastype \Proc$}
			\LeftLabel{\scriptsize Par}
			\BinaryInfC{$
			\Gt{\Gamma_1 \setminus y} \cat \envPropR;\es; \Gt{\Delta'} 
		\cat \Theta
					\proves \B{k}{\tilde x}{\binpt{u_i}{y}P'} \hastype \Proc
						$}
			\DisplayProof
		\end{align}
		This concludes sub-case~\rom{3}. 
	
	\item Case $P = \appl{V}{(\wtd r, u_i)}$. We assume a fixed order in the tuple $(\wtd r, u_i)$: names in $\wtd r$ have 
	recursive session types $\wtd r = (r_1,\ldots,r_n) : (S_1,\ldots,S_n)$, and $u_i$ has non-recursive session type $u_i : C$. 
%	$\wtd u = (u_1,\ldots,u_r,\ldots,u_n)$, $\wtd u: 
%	(S_1,\ldots,S_r,\ldots,S_n)$ and $(S_1,\ldots,S_r)$ are recursive session
%	 types. 
	 We distinguish two sub-cases: \rom{1} $V : \lhot{\wtd S C}$ 
	 and \rom{2} $V : \shot{\wtd S C}$. We will consider only sub-case 
	 \rom{1} since the other is similar. 
	 In this case, Rule~\textsc{PolyApp} can be applied: 
	 \begin{align}
	 	\label{pt:r-app}
	 	\AxiomC{$\Gamma;\Lambda;\Delta_1 \cat {\envR}_1 \proves 
	 				V \hastype \lhot{\wtd S C}$}
	 	\AxiomC{$\Gamma;\es;\Delta_2 \cat {\envR}_2 \proves 
	 				(\wtd r, u_i) \hastype \wtd S C$}
	 	\LeftLabel{\scriptsize PolyApp}
	 	\BinaryInfC{$\Gamma;\Lambda;\Delta_1 \cat \Delta_2 \cat {\envR}_1 \cat 
	 	{\envR}_2 \proves \appl{V}{(\wtd r, u_i)}$}
	 	\DisplayProof
	 \end{align}
	 
	 Let $\wtd x = \fv{V}$ and $\Gamma_1 \setminus \wtd x$. 
	 Let $\wtd x =
\fv{V}$ and let $\Theta_1$ be a balanced environment such that
$$
\dom{\Theta_1}=\{\prop_{k+1},\ldots,\prop_{k+\plen{V}}\}
\cup \{\dual{\prop_{k+1}},\ldots,\dual{\prop_{k+\plen{V}}}\}
$$
and $\Theta_1(\prop_{k+1}) = \btinpt{\wtd M}$
and
$\Theta_1(\dual{\prop_{k+1}}) = \btoutt{\wtd M} $ where $\wtd
M = (\Gt{\Gamma}\cat\Gt{\Lambda})(\wtd x)$.

	We define: 
	\begin{align}
		{\envPropR}_1 = \prod_{r \in {\dom{{\envR}_1}}} c^r:\chtype{\lhot{\Rts{}{s}{{\envR}_1(r)}}}
	\end{align}
	Then, by IH (Part 2) on the first assumption of \eqref{pt:r-app}
	we have: 
	\begin{align}
		\label{eq:r-app-ih}
		\Gt{\Gamma_1}\cat {\envPropR}_1;\es; \Gt{\Delta_1}  \cat \Theta_1
		\proves \V{k+1}{\tilde x}{V}  \hastype \lhot{\Gt{\wtd S C}}
	\end{align}
	By \Cref{def:typesdecompenv,mst:def:typesdecomp} and 
	the second assumption of \eqref{pt:r-app} we have: 
	\begin{align}
		\label{eq:r-app-ih2}
		\Gt{\Gamma};\es;\Gt{\Delta_2} \cat \Gt{{\envR}_2} \proves 
		(\wtd r_1,\ldots,\wtd r_n, \wtd m):\Gt{\wtd S C}
	\end{align}
	\noindent where 
	$\wtd r_i = (r^i_{i},\ldots,r^i_{i+\len{\Gt{S_i}}-1})$
	for $i \in \{1,\ldots,n\}$ and $\wtd m = (u_i,\ldots,u_{i+\len{G(C)}-1})$. 
	
	We define $\envPropR = {\envPropR}_1 \cat {\envPropR}_2$ where:
	\begin{align*}
		 {\envPropR}_2 = \prod_{r \in {\dom{{\envR}_2}}} c^r:\chtype{\lhot{\Rts{}{s}{{\envR}_2(r)}}}
	\end{align*}
	
	We define $\Theta = \Theta_1 \cat \prop_k : \btinpt{\wtd M}$.
	
	We will first consider the case where $n=3$; the proof is then generalized for any $n \geq 1$: 

\begin{itemize} 
    \item 
    %% Case n =3
	If $n=3$ then $P = \appl{V}{(r_1,r_2,r_3,u_i)}$. 
	By \Cref{mst:t:bdowncore}, we have:
	\begin{align*}
		\B{k}{\tilde x}{\appl{V}{(r_1,r_2,r_3,u_i)}} = 
		\propinp{k}{\tilde x}\bout{\prop^{r_1}}{\abs{\wtd z_1}
		{\bout{\prop^{r_2}}{\abs{\wtd z_2}{
		\bout{\prop^{r_3}}{\abs{\wtd z_3}{Q}}\inact
		}}}\inact}\inact
	\end{align*}
	\noindent where 
	$Q = \appl{\V{k+1}{\tilde x}{V}}{(\wtd z_1,\ldots,\wtd z_n, \wtd m)}$;
	$\wtd{z_i} = (z^i_1,\ldots,z^i_{\len{\Gt{S_i}}})$
	for $i = \{1,2,3\}$; $\wtd m =(u_{i},\ldots,u_{i+\len{\Gt{C}}-1})$.

	We shall prove the following judgment: 
	\begin{align}
			\Gt{\Gamma}\cat \envPropR;\es;\Gt{\Delta_1 \Delta_2} \cat \Theta 
			\proves \B{k}{\tilde x}{\appl{V}{(\wtd r, u_i)}} \hastype \Proc
	\end{align}
	
	We use some auxiliary derivations: 
	\begin{align}
		\label{pt:r-app-ih1}
		\AxiomC{\eqref{eq:r-app-ih}}
		\LeftLabel{\scriptsize (\lemref{lem:weaken} with $\envPropR_2$)}
		\UnaryInfC{\begin{tabular}{c}
			$\Gt{\Gamma}\cat \envPropR;\Gt{\Lambda};
						\Gt{\Delta_1} 
						\cat \Theta_1  
						\proves$  $\V{k+1}{\tilde x}{V} 
						\hastype \lhot{\Gt{\wtd S C}}$
						\end{tabular}
						}
		\DisplayProof
	\end{align}

	\begin{align}
	\label{pt:r-app-ih2}
		\AxiomC{\eqref{eq:r-app-ih2}} 
	\LeftLabel{\scriptsize (\lemref{lem:subst} with $\sigma$)}
		\UnaryInfC{\begin{tabular}{c}$\Gt{\Gamma}\cat \envPropR; \es; 
					\Gt{\Delta_2} \cat \Gt{{\envR}_2} 
					\proves$  $(\wtd z_1,\wtd z_2, \wtd z_3, \wtd m) 
					\hastype \Gt{\wtd S C}$
					\end{tabular}}
		\DisplayProof
	\end{align}
	
	\begin{align}
	\label{pt:r-app-3-6}
		\AxiomC{\eqref{pt:r-app-ih1}}
	\AxiomC{\eqref{pt:r-app-ih2}} 
		\LeftLabel{\scriptsize PolyApp}
		\BinaryInfC{$\Gt{\Gamma}\cat \envPropR;\Gt{\Lambda};
						\Gt{\Delta_1 \cat \Delta_2} 
						\cat \Theta_1 \cat \Gt{{\envR}_2} 
						\proves 
						\appl{\V{k+1}{\tilde x}{V}}
						{(\wtd z_1,\wtd z_2, \wtd z_3, \wtd m)}
						$}	
		\DisplayProof
	\end{align}
	\noindent where $\sigma = \subst{\wtd n_1}{\wtd z_1} 
	 \cdot \subst{\wtd n_2}{\wtd z_2} \cdot 
	 \subst{\wtd n_3}{\wtd z_3}$
	with $\wtd n_i = (r^i_{i},\ldots,r^i_{i+\len{\Gt{S_i}}-1})$
	for $i = \{1,2,3\}$. 
	
	\begin{align}
		\label{pt:r-app-3-5}
		\AxiomC{\eqref{pt:r-app-3-6}} 
		\AxiomC{}
		\LeftLabel{\scriptsize PolySess}
		\UnaryInfC{$\Gt{\Gamma}\cat \envPropR;\es;
						\wtd z_3 : \Gt{S_3} \proves 
		 				\wtd z_3 \hastype \Gt{S_3}$}
		\LeftLabel{\scriptsize PolyAbs}
		\BinaryInfC{$\Gt{\Gamma}\cat \envPropR;\Gt{\Lambda};
						\Gt{\Delta_1 \Delta_2} 
						\cat \Theta_1 \proves
						\abs{\wtd z_3}{Q} \hastype 
						\lhot{\Gt{S_3}}$}
		\DisplayProof
	\end{align}
	
	\begin{align}
		\label{pt:r-app-3-4}
		\AxiomC{\eqref{pt:r-app-3-5}}
		\AxiomC{} 
		\LeftLabel{\scriptsize Nil}
		\UnaryInfC{$\Gt{\Gamma}\cat \envPropR;\es;\es \proves \inact$}
		\AxiomC{}
		\LeftLabel{\scriptsize LVar}
		\UnaryInfC{$\Gt{\Gamma}\cat \envPropR;\es;\es \proves
					\prop^{r_3} \hastype \chtype{\lhot{\Gt{S_3}}}$} 
		\LeftLabel{\scriptsize Req}
		%\AxiomC{}
		%\LeftLabel{\scriptsize Nil}
		%\UnaryInfC{$\Gt{\Gamma},\envPropR;\es;\es \proves \inact \hastype \Proc$}
		\TrinaryInfC{$\Gt{\Gamma}\cat \envPropR;\Gt{\Lambda};
						\Gt{\Delta_1 \Delta_2} \cat 
						\Theta_1 \cat 
						\wtd z_1 : \Gt{S_1}
						\cat 
						\wtd z_2 : \Gt{S_2} 
						 \proves 
						\bout{\prop^{r_3}}{\abs{\wtd z_3}{Q}}\inact
						 \hastype 
						\Proc$}
		\DisplayProof
	\end{align}

%		\begin{align}
%		\label{pt:r-app-3-4}
%		\AxiomC{\eqref{pt:r-app-3-5}} 
%		\AxiomC{}
%		\LeftLabel{\scriptsize LVar}
%		\UnaryInfC{$\Gt{\Gamma}\cat \envPropR;\es;\es \proves
%					\prop^{r_3} \hastype \chtype{\lhot{\Gt{S_3}}}$} 
%		\LeftLabel{\scriptsize Req}
%		%\AxiomC{}
%		%\LeftLabel{\scriptsize Nil}
%		%\UnaryInfC{$\Gt{\Gamma},\envPropR;\es;\es \proves \inact \hastype \Proc$}
%		\BinaryInfC{$\Gt{\Gamma}\cat \envPropR;\Gt{\Lambda};
%						\Gt{\Delta_1 \Delta_2} \cat 
%						\Theta_1 \cat 
%						\wtd z_1 : \Gt{S_1}
%						\cat 
%						\wtd z_2 : \Gt{S_2} 
%						 \proves 
%						\bout{\prop^{r_3}}{\abs{\wtd z_3}{Q}}\inact
%						 \hastype 
%						\Proc$}
%		\DisplayProof
%	\end{align}
	
	\begin{align}
		\label{pt:r-app-3-3}
		\AxiomC{\eqref{pt:r-app-3-4}} 
		\AxiomC{}
		\LeftLabel{\scriptsize PolySess}
		\UnaryInfC{$\Gt{\Gamma}\cat \envPropR;\es;
						\wtd z_2 : \Gt{S_2} \proves 
		 				\wtd z_2 \hastype \Gt{S_2}$}
		\LeftLabel{\scriptsize PolyAbs}
		\BinaryInfC{$\Gt{\Gamma}\cat \envPropR;\Gt{\Lambda};
						\Gt{\Delta_1 \Delta_2} 
						\cat \Theta_1 \cat 
						\wtd z_1 : \Gt{S_1}\proves
						\abs{\wtd z_2}{
						\bout{\prop^{r_3}}{\abs{\wtd z_3}{Q}}\inact
						} \hastype 
						\lhot{\Gt{S_2}}$}
		\DisplayProof
	\end{align}
	
	\begin{align}
		\label{pt:r-app-3-2}
		\AxiomC{\eqref{pt:r-app-3-3}} 
		\AxiomC{}
		\LeftLabel{\scriptsize Nil}
		\UnaryInfC{$\Gt{\Gamma}\cat \envPropR;\es;\es \proves \inact$}
		\AxiomC{}
		\LeftLabel{\scriptsize LVar}
		\UnaryInfC{$\Gt{\Gamma}\cat \envPropR;\es;\es \proves
					\prop^{r_2} \hastype \chtype{\lhot{\Gt{S_2}}}$} 
		\LeftLabel{\scriptsize Req}
		%\AxiomC{}
		%\LeftLabel{\scriptsize Nil}
		%\UnaryInfC{$\Gt{\Gamma},\envPropR;\es;\es \proves \inact \hastype \Proc$}
		\TrinaryInfC{$\Gt{\Gamma}\cat \envPropR;\Gt{\Lambda};
						\Gt{\Delta_1 \Delta_2} \cat \Theta_1 \cat 
						\wtd z_1 : \Gt{S_1}
						  \proves 
						{\bout{\prop^{r_2}}{\abs{\wtd z_2}{
						\bout{\prop^{r_3}}{\abs{\wtd z_3}{Q}}\inact
						}}}\inact \hastype 
						\Proc$}
		\DisplayProof
	\end{align}	

	\begin{align}
		\label{pt:r-app-3-1}
		\AxiomC{\eqref{pt:r-app-3-2}} 
		\AxiomC{}
		\LeftLabel{\scriptsize PolySess}
		\UnaryInfC{$\Gt{\Gamma}\cat \envPropR;\es;
						\wtd z_1 : \Gt{S_1} \proves 
		 				\wtd z_1 \hastype \Gt{S_1}$}
		\LeftLabel{\scriptsize PolyAbs}
		\BinaryInfC{$\Gt{\Gamma}\cat \envPropR;\Gt{\Lambda};
						\Gt{\Delta_1 \Delta_2} 
						\cat \Theta_1 \proves
						\abs{\wtd z_1}
						{\bout{\prop^{r_2}}{\abs{\wtd z_2}{
						\bout{\prop^{r_3}}{\abs{\wtd z_3}{Q}}\inact
						}}}\inact \hastype 
						\lhot{\Gt{S_1}}$}
		\DisplayProof
	\end{align}
	
	\begin{align}
		\label{pt:r-app-3}
		\AxiomC{\eqref{pt:r-app-3-1}} 
		\AxiomC{}
		\LeftLabel{\scriptsize Nil}
		\UnaryInfC{$\Gt{\Gamma}\cat \envPropR;\es;\es \proves \inact$}
		\AxiomC{} 
		\LeftLabel{\scriptsize LVar}
		\UnaryInfC{$\Gt{\Gamma}\cat \envPropR;\es;\es \proves
					\prop^{r_1} \hastype \chtype{\lhot{\Gt{S_1}}}$} 
		\LeftLabel{\scriptsize Req}
		\TrinaryInfC{$\Gt{\Gamma}\cat \envPropR;\Gt{\Lambda};
						\Gt{\Delta_1 \Delta_2} 
						\cat \Theta_1 \proves 
						\bout{\prop^{r_1}}{\abs{\wtd z_1}
						{\bout{\prop^{r_2}}{\abs{\wtd z_2}{
						\bout{\prop^{r_3}}{\abs{\wtd z_3}{Q}}\inact
						}}}\inact}\inact$}
		\DisplayProof
	\end{align}
	The following tree proves this case: 
	\begin{align*}
		\AxiomC{\eqref{pt:r-app-3}} 
		\AxiomC{} 
		\LeftLabel{\scriptsize PolyVar}
		\UnaryInfC{$\Gt{\Gamma_1}\cat \envPropR;\Gt{\Lambda};\es \proves 
					\wtd x \hastype \wtd M$}
		\LeftLabel{\scriptsize PolyRcv}
		\BinaryInfC{$\Gt{\Gamma_1}\cat \envPropR;\es;\Gt{\Delta_1 \cat \Delta_2} 
						\cat \Theta \proves 
						\B{k}{\tilde x}{\appl{V}{(\wtd r, u_i)}} 
						\hastype \Proc$}
		\DisplayProof	
	\end{align*}
\item 
	
	%%% Case n >1
	Now we consider the general case for any $n \geq 1$.
	By \tabref{mst:t:bdowncore}, we have: 
	\begin{align*}
%	\B{k}{\tilde x}{\appl{V}{\wtd u}} = 
%	 	\propinp{k}{\wtd x} \bout{\prop^{u_1}}{\abs{\wtd z_1}{
%	 	\ldots \bout{\prop^{u_r}}{\abs{\wtd z_r}{f}}\inact}\ldots} \inact \\
	 	\B{k}{\tilde x}{\appl{V}{\wtd r}} = \propinp{k}{\wtd x}\overbracket{\prop^{r_1}!\langle \lambda \wtd z_1. \ldots 
	 	\prop^{r_n}!\langle \lambda \wtd z_n.}^{n = \len{\wtd r}}
	 	Q
	 \rangle \ldots \rangle
	\end{align*}
	\noindent where 
	$Q = \appl{\V{k+1}{\tilde x}{V}}{(\wtd r_1,\ldots,\wtd r_n, \wtd m)}$ with:
	$\wtd{z_i} = (z^i_1,\ldots,z^i_{\len{\Gt{S_i}}})$
	for $i = \{1,\ldots,n\}$; and 
	$\wtd m =(u_{i},\ldots,u_{i+\len{\Gt{C}}-1})$.

	We shall prove the following judgment: 
	\begin{align}
			\Gt{\Gamma}\cat \envPropR;\es;\Gt{\Delta_1 \Delta_2} \cat \Theta 
			\proves \B{k}{\tilde x}{\appl{V}{(\wtd r, u_i)}} \hastype \Proc
	\end{align}

	We construct auxiliary derivations parametrized by $k$ and denoted by $d(k)$. 
	 If $k=n$, derivation $d(n)$ is defined as: 		\begin{align} 
		\label{pt:r-app-b2}
		\AxiomC{\eqref{eq:r-app-ih}}
			\LeftLabel{\scriptsize (\lemref{lem:weaken} with ${\envPropR}_2$)} 
			\UnaryInfC{$\Gt{\Gamma}\cat \envPropR;\Gt{\Lambda};
						\Gt{\Delta_1} \cat \Theta_1 \proves \V{k+1}{\tilde x}{V}
						\hastype \shot{\Gt{\wtd S C}} 
						$}
			\DisplayProof
	\end{align}

	\begin{align}
			\label{pt:r-app-b}
			\AxiomC{\eqref{pt:r-app-b2}}
			\AxiomC{\eqref{eq:r-app-ih2}}
			\LeftLabel{\scriptsize (\lemref{lem:subst} with $\sigma$)}
			\UnaryInfC{$\Gt{\Gamma}\cat \envPropR;\es;\Gt{{\envR}_2} 
							\proves (\wtd r_1,\ldots,\wtd r_n, \wtd m)
							\hastype \Gt{\wtd S C}$} 
			\LeftLabel{\scriptsize PolyAbs}
			\BinaryInfC{$\Gt{\Gamma}\cat \envPropR;\Gt{\Lambda};
						\Gt{\Delta_1 \cat \Delta_2} 
						\cat$ 
						$\Theta_1 \cat \Gt{{\envR}_2} \proves
						\abs{\wtd z_n}{Q} \hastype 
						\lhot{\Gt{S_n}}$}
			\DisplayProof
	\end{align}
	
	\noindent where $\sigma = \prod_{i \in \{1,\ldots,n\}}\subst{\wtd n_i}{\wtd z_i}$
	with $\wtd n_i = (r^i_{i},\ldots,r^i_{i+\len{\Gt{S_i}}-1})$
	for $i = \{1,\ldots,n\}$. 
	
	\begin{align}
	\AxiomC{\eqref{pt:r-app-b}} 
	\AxiomC{}
		\LeftLabel{\scriptsize Nil}
		\UnaryInfC{$\Gt{\Gamma}\cat \envPropR;\es;\es \proves \inact$}
	\AxiomC{}
	\LeftLabel{\scriptsize LVar}
	\UnaryInfC{$\Gt{\Gamma}\cat \envPropR;\es;\es \proves 
				\prop^{r_n} \hastype \chtype{\lhot{\Gt{S_n}}}$} 
		\LeftLabel{\scriptsize Req}
			\TrinaryInfC{$\Gt{\Gamma}\cat \envPropR;\Gt{\Lambda};
						\Gt{\Delta_1 \cat \Delta_2} 
						\cat$ 
						$\Theta_1 \cat \Gt{{\envR}_2} \proves \bout{\prop^{r_n}}{\abs{\wtd z_n}{Q}}\inact$}
			\DisplayProof
	\end{align}

	Otherwise, if $k \in \{1,\ldots,n-1\}$, the derivation $d(k)$ is as follows: 
	\begin{align}
		\label{pt:r-abs-k}
		\AxiomC{$d(k+1)$} 
%		\AxiomC{\begin{tabular}{c}$\Gt{\Gamma}\cat \envPropR;\Gt{\Lambda};
%						\Gt{\Delta_1 \Delta_2} 
%						\cat $ \\
%						$\Theta \cat (\wtd z_1,\ldots,\wtd z_{k}) : 
%						(\Gt{S_1},\ldots,\Gt{S_{k}})
%						\proves$  \\
%						$\overbracket{\prop^{u_{k+1}}!\langle \lambda 
%						\wtd z_{k+1}. \ldots 
%	 	\prop^{u_r}!\langle \lambda \wtd z_r.}^{r-k} 
%	 	Q
%	 	\overbracket{\rangle.\inact \ldots \rangle.\inact}^{r-k} 
%	 	\hastype \Proc 
%						$\end{tabular}}
	\AxiomC{}
	\LeftLabel{\scriptsize PolyVar}
	\UnaryInfC{$\Gt{\Gamma}\cat \envPropR;\es;\wtd z_k : \Gt{S_k} 
				\proves \wtd z_k \hastype \Gt{S_k}$} 
	\LeftLabel{\scriptsize PolyAbs}
	\BinaryInfC{\begin{tabular}{c}$\Gt{\Gamma}\cat \envPropR;\Gt{\Lambda};
						\Gt{\Delta_1 \cat \Delta_2} 
						\cat \Theta_1 
						\cat 
						(\wtd z_1,\ldots,\wtd z_{k-1}) : 
						(\Gt{S_1},\ldots,\Gt{S_{k-1}})
						\proves$ \\
						$\abs{\wtd z_k}\overbracket{\prop^{r_{k+1}}!\langle \lambda 
						\wtd z_{k+1}. \ldots 
	 	\prop^{r_n}!\langle \lambda \wtd z_n.}^{n-k} 
	 	Q
	 	\overbracket{\rangle.\inact \ldots \rangle.\inact}^{n-k}
						\hastype \lhot{\Gt{S_1}}$
			\end{tabular}} 
			\DisplayProof 
	\end{align}
	
	\begin{align}
	\label{pt:r-acc-k}
	\AxiomC{\eqref{pt:r-abs-k}}
		\AxiomC{$\Gt{\Gamma}\cat \envPropR;\es; \es \proves \prop^{r_1} \hastype 
						\chtype{\lhot{\Gt{S_1}}}$}
	\LeftLabel{\scriptsize Acc}
		\BinaryInfC{
		\begin{tabular}{c}
		$\Gt{\Gamma}\cat \envPropR;\Gt{\Lambda};
						\Gt{\Delta_1 \cat \Delta_2} 
						\cat \Theta_1 
						\cat (\wtd z_1,\ldots,\wtd z_{k-1}) : 
						(\Gt{S_1},\ldots,\Gt{S_{k-1}})
						\proves$  \\
						$\overbracket{\prop^{r_k}!\langle \lambda 
						\wtd z_k. \ldots 
	 	\prop^{r_n}!\langle \lambda \wtd z_r.}^{n-k+1} 
	 	Q
	 	\overbracket{\rangle.\inact \ldots \rangle.\inact}^{n-k+1}
						$
				\end{tabular}}
		\DisplayProof	
	\end{align}

%	\begin{align}
%	\label{pt:r-app-send}
%	\AxiomC{$\Gt{\Gamma}\cat \envPropR;\Gt{\Lambda};
%						\Gt{\Delta_1 \Delta_2} 
%						\cat  
%						\Theta \cat \wtd z_1 : \Gt{S_1}
%						\proves 
%						\overbracket{\prop^{u_2}!\langle \lambda 
%						\wtd z_2. \ldots 
%	 	\prop^{u_r}!\langle \lambda \wtd z_r.}^{r-1} 
%	 	Q
%	 	\overbracket{\rangle.\inact \ldots \rangle.\inact}^{r-1} 
%	 	\hastype \Proc 
%						$}
%	\AxiomC{}
%	\LeftLabel{\scriptsize PolyVar}
%	\UnaryInfC{$\Gt{\Gamma}\cat \envPropR;\es;\wtd z_1 : \Gt{S_1} 
%				\proves \wtd z_1 \hastype \Gt{S_1}$} 
%	\LeftLabel{\scriptsize PolyAbs}
%	\BinaryInfC{\begin{tabular}{c}$\Gt{\Gamma}\cat \envPropR;\Gt{\Lambda};
%						\Gt{\Delta_1 \Delta_2} 
%						\cat \Theta
%						\proves$ \\
%						$\abs{\wtd z_1}\overbracket{\prop^{u_2}!\langle \lambda 
%						\wtd z_2. \ldots 
%	 	\prop^{u_r}!\langle \lambda \wtd z_r.}^{r-1} 
%	 	Q
%	 	\overbracket{\rangle.\inact \ldots \rangle.\inact}^{r-1}
%						\hastype \lhot{\Gt{S_1}}$
%			\end{tabular}} 
%	\AxiomC{$\Gt{\Gamma}\cat \envPropR;\es; \es \proves \prop^{u_1} \hastype 
%						\chtype{\lhot{\Gt{S_1}}}$}
%	\LeftLabel{\scriptsize Acc}
%		\BinaryInfC{$\Gt{\Gamma}\cat \envPropR;\Gt{\Lambda};
%						\Gt{\Delta_1 \Delta_2} 
%						\cat \Theta \proves 
%						\overbracket{\prop^{u_1}!\langle \lambda 
%						\wtd z_1. \ldots 
%	 	\prop^{u_r}!\langle \lambda \wtd z_r.}^{r} 
%	 	Q
%	 	\overbracket{\rangle.\inact \ldots \rangle.\inact}^{r}
%						$}	
%	\DisplayProof
%	\end{align}

	The following tree proves this case: 
	\begin{align*}
		\AxiomC{$d(1)$} 
		\AxiomC{} 
		\LeftLabel{\scriptsize PolyVar}
		\UnaryInfC{$\Gt{\Gamma_1}\cat \envPropR;\Gt{\Lambda};\es \proves 
					\wtd x \hastype \wtd M$}
		\LeftLabel{\scriptsize PolyRcv}
		\BinaryInfC{$\Gt{\Gamma_1}\cat \envPropR;\es;\Gt{\Delta_1 \cat \Delta_2} 
						\cat \Theta \proves 
						\B{k}{\tilde x}{\appl{V}{(\wtd r, u_i)}} 
						\hastype \Proc$}
		\DisplayProof	
	\end{align*}
\end{itemize}
\end{enumerate}

This concludes the analysis for \Cref{mst:t:typecore}(1).

\item This part of \Cref{mst:t:typecore} concerns values. 
Without a loss of generality we assume $\wtd T = \wtd S, C$ with 
 $\wtd S = (S_1, \ldots, S_n)$ such that $\tr(S_i)$ for $i \in \{1,\ldots,n\}$.  
We distinguish two sub-cases: 
\rom{1} $V=y$ and \rom{2} $V=\abs{\wtd y, z}{P}$. 

We first consider sub-case~\rom{1}.  
By assumption $\Gamma;\Lambda;\Delta \proves y \hastype \slhot{C}$.

	Further, we can distinguish two sub-sub-cases~(a) $\leadsto = \multimap$ and
	(b) $\leadsto = \rightarrow$. In sub-sub-case~(a), when $\leadsto =
	\multimap$, only Rule~\textsc{LVar} can be applied; by inversion  
	$\Lambda = \{y : \lhot{\wtd T}\}$ and $\Delta = \es$.  By
	\tabref{mst:t:bdowncore}, we have $\V{k}{\tilde x}{y} =y$ and by
	\Cref{mst:def:typesdecomp,def:typesdecompenv} we have
	$\Gt{\Delta} = \{ \Gt{\lhot{\wtd T}}\}$. Hence, we prove the following judgment
	by applying Rule~\textsc{LVar}:
	\begin{align*} 
	 \Gt{\Gamma};\Gt{\Delta};\es \proves \V{k}{\tilde x}{y} \hastype 
	 \Gt{\lhot{\wtd T}}
	 \LeftLabel{\scriptsize LVar}
	\end{align*} 
	\noindent In sub-sub-case~(b) only Rule~\textsc{Sh} can be applied; 
	by inversion we have $\Gamma = \{y:\shot{\wtd T}\}$, $\Lambda = \es$, and 
	$\Delta=\es$. Similarly to~(a),  by \Cref{mst:def:typesdecomp,def:typesdecompenv} we have $\Gt{\Gamma} = \{ \Gt{\shot{\wtd T}}\}$.
	Hence, we prove the following judgment by applying Rule~\textsc{SH}:
	\begin{align*} 
	 \Gt{\Gamma};\es;\es \proves \V{k}{\tilde x}{y} \hastype 
	 \Gt{\shot{\wtd T}}
	 \LeftLabel{\scriptsize LVar}
	\end{align*}  
	This concludes sub-case~\rom{1}. 
	
	Now, we consider sub-case~\rom{2}. 
This is the second sub-case concerning values when 
$V = \abs{\wtd y, z}P$ where $\wtd y = y_1,\ldots,y_n$. 
By assumption we have  
$\Gamma;\Lambda;\Delta \cat \envR \proves V \hastype \slhot{\wtd S, C}$. Also here 
we distinguish two sub-sub-cases (a) $\leadsto = \multimap$ and (b) 
$\leadsto = \rightarrow$: 
\begin{itemize} 
    \item $\leadsto = \multimap$. 
By assumption, $\Gamma;\Lambda;\Delta \cat \envR \proves V \hastype 
\lhot{\wtd S, C}$.
In this case Rule~\textsc{Abs}
can be applied. Firstly, we $\alpha$-convert value $V$ as follows:
\begin{align}
	V \equiv_{\alpha} \abs{\wtd y, z_1}P\subst{z_1}{z}
\end{align}
	For this case only Rule~\textsc{Abs} can be applied:
	\begin{align}
		\label{pt:absitr}
		\AxiomC{$\Gamma;\Lambda;\Delta_1 \cat {\envR}_1 
		\proves P\subst{z_1}{z}
		         \hastype \Proc$}
		\AxiomC{$
			\Gamma;\es;\Delta_2 \cat {\envR}_2 \proves 
			\wtd y, z_1 \hastype \wtd S, C$}
		\LeftLabel{\scriptsize Abs}
		\BinaryInfC{$
		\Gamma \setminus z_1;\Lambda;
		\Delta_1 \setminus \Delta_2 \cat 
		{\envR}_1 \setminus {\envR}_2
		\proves
			\abs{\wtd y, z_1}{P\subst{z_1}{z}} \hastype \lhot{\wtd S, C}$}
		\DisplayProof
	\end{align}
Let $\wtd x = \fv{P}$ and $\Gamma_1 = \Gamma \setminus \wtd x$. Also, let $\Theta_1$ be a balanced environment such that
$$\dom{\Theta_1}=\{\prop_1,\ldots,\prop_{\plen{P}}\} \cup
\{\dual{\prop_{2}},\ldots,\dual{\prop_{\plen{P}}}\}$$
and $\Theta_1(\prop_1) = \btinp{\wtd M}\tinact$ with
$\wtd M = (\Gt{\Gamma \setminus y_1}, \Gt{\Lambda})(\wtd x)$.
We define: 
\begin{align*}
	{\envPropR}_i = \prod_{r \in {\dom{{\envR}_i}}} 
	c^r:\chtype{\lhot{\Rts{}{s}{{\envR}_i(r)}}} \text{ for } 
	i \in \{1,2 \}
\end{align*}
 Then, by IH (Part 1) on the first assumption of \eqref{pt:absitr} we have:
 \begin{align}
 		\label{eq:abs-ih-1}
 		\Gt{\Gamma_1} \cat {\envPropR}_1;\es;\Gt{\Delta_1} \cat \Theta_1
 		\proves \B{1}{\tilde x}{P\subst{y_1}{y}} \hastype \Proc
 \end{align}

%We apply weakening to \eqref{eq:abs-ih-1}:
% \begin{align}
% 		\label{eq:abs-ih-1-weaken}
% 		\Gt{\Gamma};\varnothing;\Gt{\Delta_1 \cdot u_1:C} \cdot \Theta_1
% 		\proves \B{k}{\tilde x}{P\subst{u_1}{u}} \hastype \Proc
% \end{align}

Let $\wtd T = \Gt{S_1}, \ldots, \Gt{S_n}, \Gt{C}$. 
By \Cref{mst:def:typesdecomp,def:typesdecompenv} and the
second assumption of
\eqref{pt:absitr} we have:
\begin{align}
	\label{eq:abs-ih-2}
	\Gt{\Gamma};\es;\Gt{\Delta_2}\proves \wtd y^1,\ldots,\wtd y^n, 
	\wtd z
	 \hastype \wtd T 
\end{align}
\noindent where $\wtd z = (z_1,\ldots,z_{\len{\Gt{C}}})$ 
and $\wtd{y^i} = (y^i_1,\ldots,y^i_{\len{\Gt{S_i}}}))$ for 
$i \in \{1,\ldots,n\}$.

We define $\Theta = \Theta_1 \cat \dual {\prop_k}:
	\btoutt{\wtd M}$. By \tabref{mst:t:bdowncore}, we have:
% \begin{align*}
% 	\V{k}{\tilde x}{\abs{y}{P}} = 
% 	\abs{\wtd{y}:\Gt{C}}{\propout{k}{\wtd x} \inact
%   	\Par \B{k}{\tilde x}{P \subst{y_1}{y}}}
% \end{align*}
\begin{align*}
	\V{k}{\tilde x}{\abs{y_1, \ldots, y_n, z}{P}} = 
	\abs{(\wtd{y^1},\ldots,\wtd{y^n}, \wtd z):
		\slhotup{(\wtd{T})}}{N}
\end{align*}
\noindent where 
\begin{align*}
	N = \news{\wtd \prop} \news{\wtd \prop_r}
     \prod_{i \in \len{\wtd y}}
	 (\recprovx{y_i}{x}{\wtd y^i})
	%  (\binp{\prop^{y_i}}{b}(\appl{b}{\wtd y^i})) 
	 \Par 
	 \apropout{1}{\wtd x} \Par \B{1}{\tilde x}{P \subst{z_1}{z}}
\end{align*}
\noindent with $\wtd \prop = (\prop_1,\ldots,\prop_{\plen{P}})$
and $\widetilde{\prop_r} = \bigcup_{r\in \tilde{y}}\prop^r$. 

%  By construction $\Theta$ is balanced since $\Theta(\prop_{k}) \dualof \Theta(\dual{\prop_{k}})$ and $\Theta_1$ is balanced.
 We use an auxiliary derivation:
 \begin{align}
 	\label{pt:abs-send}
 	\AxiomC{}
	\LeftLabel{\scriptsize Nil}
	\UnaryInfC{$\Gt{\Gamma} \cat {\envPropR}_1;\es;\es
				       \proves \inact \hastype \Proc$}
	\LeftLabel{\scriptsize End}
	\UnaryInfC{$\Gt{\Gamma} \cat {\envPropR}_1;\es;\prop_k:\tinact
				       \proves \inact \hastype \Proc$}
	\AxiomC{}
	\LeftLabel{\scriptsize PolyVar}
	\UnaryInfC{$\Gt{\Gamma} \cat {\envPropR}_1;\Gt{\Lambda};\es
				        \proves \wtd x \hastype \wtd M$}
	\LeftLabel{\scriptsize Send}
	\BinaryInfC{$\Gt{\Gamma} \cat {\envPropR}_1;\Gt{\Lambda}; \dual {\prop_1}:
				        \btout{\wtd M} \tinact
				        \proves \propout{1}{\wtd x} \inact \hastype \Proc$}
	\DisplayProof
 \end{align}

 \begin{align} 
	\label{pt:r-d1-value1}
	\AxiomC{}
	\LeftLabel{\scriptsize LVar}
	\UnaryInfC{ 
		$\Gt{\Gamma}\cat {\envPropR}_1;\recpvar{x}:\lhot{\Rts{}{s}{{\envR}_2(y^i)}};\es 
		\proves \recpvar{x} \hastype \lhot{\Rts{}{s}{{\envR}_2(y^i)}}$}
	\DisplayProof 
 \end{align} 

 \begin{align}
	\label{pt:r-d1-value}
	% \AxiomC{}
	% \LeftLabel{\scriptsize LVar}
	% \UnaryInfC{\begin{tabular}{c}
	% 	$\Gt{\Gamma}\cat {\envPropR}_1;b:\lhot{\Rts{}{s}{{\envR}_2(y^i)}};\es 
	% 	\proves$\\
	% 		$b \hastype \lhot{\Rts{}{s}{{\envR}_2(y^i)}}$\end{tabular}}
	\AxiomC{\eqref{pt:r-d1-value1}}
		\AxiomC{}
		\LeftLabel{\scriptsize PolySess}
	\UnaryInfC{\begin{tabular}{c}$\Gt{\Gamma} \cat {\envPropR}_1;\es;
		\wtd y^i:\Rts{}{s}{{\envR}_2(y^i)} \proves$ 
		\\ $\wtd y^i \hastype \Rts{}{s}{{\envR}_2(y^i)}$
	\end{tabular}}
	\LeftLabel{\scriptsize PolyApp}
		\BinaryInfC{$\Gt{\Gamma}\cat {\envPropR}_1;
		\recpvar{x} : \lhot{\Rts{}{s}{{\envR}_2(y^i)}};
		\wtd y^i:\Rts{}{s}{{\envR}_2(y^i)} \proves
					(\appl{b}{\wtd y^i})$}
				\DisplayProof
	\end{align}

	\begin{align}
		\label{pt:r-d-value}
		\AxiomC{\eqref{pt:r-d1-value}} 
		\AxiomC{}
		\LeftLabel{\scriptsize Sh} 
		\UnaryInfC{\begin{tabular}{c}
			$\Gt{\Gamma}\cat {\envPropR}_1;\es;\es
			 \proves \prop^{y^i} \hastype$
   \\  				$\chtype{\lhot{\Rts{}{s}{\envR(y^i)}}}$\end{tabular}}
		\AxiomC{}
		\LeftLabel{\scriptsize LVar}
		\UnaryInfC{\begin{tabular}{c}$\Gt{\Gamma}\cat \
			{\envPropR}_1; \recpvar{x}:\lhot{\Rts{}{s}{\envR(y^i)}};\es$ \\ 
					$\proves \recpvar{x} \hastype \lhot{\Rts{}{s}{\envR(y^i)}}$
					\end{tabular}} 
		\LeftLabel{\scriptsize Acc}
		\TrinaryInfC{$\Gt{\Gamma \sigma}\cat {\envPropR}_1;\es;
		\wtd y^i:\Rts{}{s}{{\envR}_1(y^i)} \proves
		\recprovx{y_i}{x}{\wtd y^i}
		% \binp{\prop^{y^i}}{b}(\appl{b}{\wtd y^i})
		$}
		\DisplayProof
	\end{align}

 \begin{align}
	\label{pt:abspar1}
	\AxiomC{for $y^i \in \wtd y$} 
	\AxiomC{\eqref{pt:r-d-value}} 
	\LeftLabel{\scriptsize Par ($\len{\wtd y}-1$ times)}
	 \BinaryInfC{
		 $
		 \Gt{\Gamma} \cat {\envPropR}_1;
	 \Gt{\Lambda};\Gt{\Delta_1} \cat \Gt{{\envR}_2}
						  \proves
	  \prod_{i \in \len{\wtd y}}
	  (\recprovx{y_i}{x}{\wtd y^i})
	%   (\binp{\prop^{y_i}}{b}(\appl{b}{\wtd y^i})) 
				\hastype \Proc
				$
				} 
	 \DisplayProof
 \end{align}

\begin{align}
	\label{pt:abspar}
	\AxiomC{\eqref{pt:abspar1}}
	\AxiomC{\eqref{pt:abs-send}} 
	\AxiomC{\eqref{eq:abs-ih-1}}
	\LeftLabel{\scriptsize (\lemref{lem:weaken} with $\tilde x$)}
	\UnaryInfC{
		$
			\Gt{\Gamma} \cat {\envPropR}_1;\es;\Gt{\Delta_1} 
			\cat \Theta_1
				\proves \B{1}{\tilde x}{P\subst{z_1}{z}} \hastype \Proc
		 $}
	\LeftLabel{\scriptsize Par}
	\BinaryInfC{
		$
			\Gt{\Gamma} \cat {\envPropR}_1;
			\Gt{\Lambda};\Gt{\Delta_1} 
			\cat \Theta
								\proves
			\apropout{1}{\wtd x} \Par \B{1}{\tilde x}{P \subst{z_1}{z}}
					\hastype \Proc
			   $}
	\BinaryInfC{
		\begin{tabular}{l} 
		$\Gt{\Gamma} \cat {\envPropR}_1;
	\Gt{\Lambda};\Gt{\Delta_1} \cat  \Gt{{\envR}_2} \cat \Theta
					     \proves$
						 \\
     $\prod_{i \in \len{\wtd y}}
	 (\recprovx{y_i}{x}{\wtd y^i})
	%  (\binp{\prop^{y_i}}{b}(\appl{b}{\wtd y^i})) 
	 \Par 
	 \apropout{1}{\wtd x} \Par \B{1}{\tilde x}{P \subst{z_1}{z}}
			   \hastype \Proc$
		\end{tabular}}
	\LeftLabel{\scriptsize PolyRes}
			   \UnaryInfC{
				   \begin{tabular}{l}
				   $\Gt{\Gamma} \cat 
			   {\envPropR}_1 \setminus {\envPropR}_2;
			   \Gt{\Lambda};\Gt{\Delta_1} \cat  \Gt{{\envR}_2} \cat \Theta
									\proves$ \\
				$\news{\wtd \prop_r} 
				\prod_{i \in \len{\wtd y}}
				(\recprovx{y_i}{x}{\wtd y^i})
				% (\binp{\prop^{y_i}}{b}(\appl{b}{\wtd y^i})) 
				\Par 
				\apropout{1}{\wtd x} \Par \B{1}{\tilde x}{P \subst{z_1}{z}}
						  \hastype \Proc$
				   \end{tabular} }
			   % \AxiomC{} 
			   \LeftLabel{\scriptsize PolyResS}
	% \AxiomC{} 
	\LeftLabel{\scriptsize PolyResS}
	\UnaryInfC{
		$
		\Gt{\Gamma}
		\cat {\envPropR}_1 \setminus {\envPropR}_2
		;\Gt{\Lambda};\Gt{\Delta_1}
		\cat \Gt{{\envR}_2}
					     \proves
			 N 
			   \hastype \Proc
		$
	}
  \DisplayProof
\end{align}

The following tree proves this part:
\begin{align}
	\label{pt:value1}
	\AxiomC{\eqref{pt:abspar}}
	\AxiomC{\eqref{eq:abs-ih-2}}
	%\LeftLabel{\scriptsize PolySess}
	%\UnaryInfC{$\Gt{\Gamma};\varnothing; \wtd u : \Gt{C}
	%				     \proves \wtd u \hastype \Gt{C}$}
	\LeftLabel{\scriptsize Abs}
	\BinaryInfC{$\Gt{\Gamma \setminus z_1}
	\cat {\envPropR}_1 \setminus {\envPropR}_2
	;\Gt{\Lambda};
	\Gt{\Delta_1 \setminus \Delta_2} 
	% \Gt{{\envR}_1 \setminus {\envR}_2}
	              \proves
				 \abs{(\wtd{y^1},\ldots,\wtd{y^n}, \wtd z):
				 \lhotup{(\wtd{T})}}{N} \hastype \Proc$}
	\DisplayProof
\end{align}
\item $\leadsto = \rightarrow$.
%This is the second case concerning values.
By assumption, $\Gamma;\es;\es \proves V \hastype \shot{C}$.
%Now, we consider sub-case \rom{2}. For
In this case Rule~\textsc{Prom} can be applied:
\begin{align}
	\label{eq:abs-sub2-ih-1}
	\AxiomC{$\Gamma;\es;\Delta \proves P\subst{y_1}{y} \hastype \Proc$}
		\AxiomC{$\Gamma;\es;\Delta \proves y_1 \hastype C$}
		\LeftLabel{\scriptsize Abs}
	\BinaryInfC{$\Gamma \setminus z_1;\es;\es
				\proves\abs{\wtd y, z_1}P\subst{z_1}{z} \hastype \lhot{C}$}
	\LeftLabel{\scriptsize Prom}
	\UnaryInfC{$\Gamma \setminus z_1;\es;\es
					     \proves \abs{\wtd y, z_1}P\subst{z_1}{z}
					     \hastype \shot{C}$}
	\DisplayProof
\end{align}
% \noindent This sub-case is similar to the previous sub-case 
% $\leadsto = \multimap$. Notice that in this sub-case we have 
% $\Delta_1 \setminus \Delta_2 = \emptyset$ 
% (resp. ${\envR}_1 \setminus {\envR}_2 = \emptyset$), 
% that is $\Delta_1 = \Delta_2$ (resp. ${\envR}_1 = {\envR}_2$) and 
% $\Lambda = \emptyset$. 
% Finally, the following tree proves this case: 
\noindent Now, we can see that we can specialize 
the previous sub-case by taking 
$\Delta_1 \setminus \Delta_2 = \emptyset$ 
(resp. ${\envR}_1 \setminus {\envR}_2 = \emptyset$), 
that is $\Delta_1 = \Delta_2$ (resp. ${\envR}_1 = {\envR}_2$) and 
$\Lambda = \emptyset$.  
Subsequently, we have ${\envPropR}_1 \setminus {\envPropR}_2 = \emptyset$, 
$\Gt{\Lambda} = \emptyset$, and $\Gt{\Delta_1 \setminus \Delta_2} = \emptyset$. 
Thus, we can apply 
Rule~\textsc{Prom} to \eqref{pt:value1} to prove this sub-case, as follows:  
\begin{align*}
	\AxiomC{\eqref{pt:value1}} 
	\LeftLabel{\scriptsize Prom}
	\UnaryInfC{$\Gt{\Gamma \setminus z_1}; 
	% \cat {\envPropR}_1 \setminus {\envPropR}_2
	% ;\Gt{\Lambda};
	\es; \es 
	% \Gt{\Delta_1 \setminus \Delta_2} 
	% \Gt{{\envR}_1 \setminus {\envR}_2}
	              \proves
				 \abs{(\wtd{y^1},\ldots,\wtd{y^n}, \wtd z):
				 \shotup{(\wtd{T})}}{N} \hastype \Proc$}
	\DisplayProof
\end{align*}
\end{itemize}
This concludes this part (and the proof of \Cref{mst:t:typecore}). 
 \end{enumerate}
\end{proof}

\subsection{Proof of \thmref{mst:t:decompcore}}
\label{mst:app:decompcore}

\decompcore* 
% \typerec* 

\begin{proof} 
By assumption $\Gamma;\es;\Delta \circ \envR \proves P \hastype \Proc$.  Then,
	by applying \lemref{lem:subst} we have:
	\begin{align}
		\label{eq:typerec-subst-ih}
		\Gamma\sigma;\es;\Delta \sigma \circ \envR \sigma \proves P \sigma \hastype \Proc
	\end{align}
	
  \noindent 
  By \lemref{mst:t:typecore} on \eqref{eq:typerec-subst-ih} we have: 
  \begin{align}
  \label{eq:typerec-ih}
  	\Gt{\Gamma_1 \sigma}\cat \envPropR;\es;\Gt{\Delta \sigma} \cat \Theta	\proves 
  	\B{k}{\epsilon}{P\sigma} \hastype \Proc 
  \end{align}
	\noindent where $\Theta$ is balanced with
	$\dom{\Theta} = \{\prop_k,\prop_{k+1},\ldots,\prop_{k+\plen{P}-1} \}
	\cup \{\dual{\prop_{k+1}},\ldots,\dual{\prop_{k+\plen{P}-1}} \}$, and
	$\Theta(\prop_k)=\btinpt{\cdot}$, and 
	  $\envPropR = \prod_{r \in \dom{\envR}} c^r:\chtype{\lhot{\Rts{}{s}{\envR(r)}}}$.
	 By assumption,  $\fv{P} = \es$.
  
  By \defref{mst:def:decomp}, we
	shall prove the following judgment:
	\begin{align*}
		\Gt{\Gamma \sigma};\es;\Gt{\Delta \sigma} \circ \Gt{\envR \sigma} \proves \news{\wtd \prop}
		\news{\wtd c_r}\Big(
  		\prod_{r \in \tilde v} 
		\recprovx{r}{x}{\wtd r}
		% \binp{\prop^r}{b}\appl{b}{\wtd r} 
		\Par
  		\propout{k}{} \inact \Par \B{k}{\epsilon}{P\sigma}\Big)
	\end{align*}
  \noindent where: $k >0$;  
  $\wtd v = \rfn{P}$; $\wtd r = 
  (r_1,\ldots,r_{\len{\Gt{S}}})$ for each $r \in \wtd v$.
  
We know $\dom{\envR} = \wtd v$.  We assume that
recursive session types are unfolded. 
  By \Cref{def:typesdecompenv,mst:def:typesdecomp}, for
  $r \in \dom{\envR}$ we have:
   $$\Gt{\envR}(r) = \Rt{\envR(r)} = \Rts{}{s}{\envR(r)}$$
 We use a family of auxiliary derivations parametrized by $r \in \wtd v$. 
  
  \begin{align}
  \label{pt:r-d1}
  \AxiomC{}
  \LeftLabel{\scriptsize LVar}
  \UnaryInfC{\begin{tabular}{c}$\Gt{\Gamma \sigma}\cat \envPropR;
	\recpvar{x}:\lhot{\Rts{}{s}{\envR(r)}};\es \proves$\\
  		$\recpvar{x} \hastype \lhot{\Rts{}{s}{\envR(r)}}$\end{tabular}}
  	\AxiomC{}
  	\LeftLabel{\scriptsize PolySess}
  \UnaryInfC{\begin{tabular}{c}$\Gt{\Gamma \sigma}\cat \envPropR;\es;\wtd r:\Rts{}{s}{\envR(r)} \proves$ \\ $\wtd r \hastype \Rts{}{s}{\envR(r)}$
  \end{tabular}}
  \LeftLabel{\scriptsize PolyApp}
  	\BinaryInfC{$\Gt{\Gamma \sigma}\cat \envPropR;
	  \recpvar{x}:\lhot{\Rts{}{s}{\envR(r)}};
  	\wtd r:\Rts{}{s}{\envR(r)} \proves
  				(\appl{\recpvar{x}}{\wtd r})$}
  			\DisplayProof
  \end{align}

  \begin{align}
  	\label{pt:r-d}
  	\AxiomC{\eqref{pt:r-d1}} 
  	\AxiomC{}
  	\LeftLabel{\scriptsize Sh} 
  	\UnaryInfC{\begin{tabular}{c}$\Gt{\Gamma \sigma}\cat \envPropR;\es;\es \proves \prop^r \hastype$
 \\  				$\chtype{\lhot{\Rts{}{s}{\envR(r)}}}$\end{tabular}}
  	\AxiomC{}
  	\LeftLabel{\scriptsize LVar}
  	\UnaryInfC{\begin{tabular}{c}$\Gt{\Gamma \sigma}\cat \envPropR;
		\recpvar{x}:\lhot{\Rts{}{s}{\envR(r)}};\es$ \\ 
  				$\proves \recpvar{x} \hastype \lhot{\Rts{}{s}{\envR(r)}}$
  				\end{tabular}} 
  	\LeftLabel{\scriptsize Acc}
  	\TrinaryInfC{$\Gt{\Gamma \sigma}\cat \envPropR;\es;\wtd r:\Rts{}{s}{\envR(r)} \proves 
	  \recprovx{r}{x}{\wtd r}
	% \binp{\prop^r}{b}(\appl{b}{\wtd r})
	$}
  	\DisplayProof
  \end{align}
  
  We will then use: 
  \begin{align}
  \label{pt:prod}
  \AxiomC{for $r \in \wtd v$} 
  \AxiomC{\eqref{pt:r-d}}
  \LeftLabel{\scriptsize Par ($\len{\wtd v}-1$ times)}
  \BinaryInfC{$\Gt{\Gamma \sigma}\cat \envPropR;\es;\Gt{\envR \sigma} \proves \prod_{r \in \tilde v} 
  \recprovx{r}{x}{\wtd r}
%   \binp{\prop^r}{b}(\appl{b}{\wtd r})
  $}
  \DisplayProof
  \end{align}
  
  where we apply Rule~\textsc{Par} $\len{\wtd v}-1$ times and for every $r
  \in \wtd v$ we apply derivation \eqref{pt:r-d}. Notice that by
  \Cref{def:typesdecompenv,mst:def:typesdecomp} we have $\Gt{\envR
  \sigma} = \prod_{r \in \tilde v}\wtd r : \Rts{}{}{\envR(r)}$.

  \begin{align}
  \label{pt:r-par}
  \AxiomC{}
  \LeftLabel{\scriptsize Sess}
      \UnaryInfC{$\Gt{\Gamma \sigma}\cat \envPropR;\es;\dual{\prop_k}:\btout{\cdot}\tinact
  			\proves \propout{k}{} \inact$}
  \AxiomC{\eqref{eq:typerec-ih}}
  \LeftLabel{\scriptsize Par}
  \BinaryInfC{$\Gt{\Gamma \sigma}\cat \envPropR;\es;\Gt{\Delta \sigma} \cat 
  				\Theta \cat \dual{\prop_k}:\btout{\cdot}\tinact \proves \propout{k}{} \inact \Par \B{k}{\epsilon}{P\sigma}$}
  				\DisplayProof 
  \end{align}
  
  The following tree proves this case: 
  \begin{align}
  \AxiomC{\eqref{pt:prod}} 
  \AxiomC{\eqref{pt:r-par}}
  \LeftLabel{\scriptsize Par}
  	\BinaryInfC{
  	\begin{tabular}{c}
  	$\Gt{\Gamma \sigma}\cat \envPropR;\es;\Gt{\Delta \sigma}  \cat
  				\Theta \cat \dual{\prop_k}:\btout{\cdot}\tinact 
				  \circ \Gt{\envR \sigma}
				\proves$  \\
  		$\prod_{r \in \tilde v} 
		  (\recprovx{r}{x}{\wtd r})
		% \binp{\prop^r}{b}(\appl{b}{\wtd r}) 
		\Par
  		\propout{k}{} \inact \Par \B{k}{\epsilon}{P\sigma}$
  		\end{tabular}} 
  	\LeftLabel{\scriptsize PolyRes}
  	\UnaryInfC{
  	\begin{tabular}{c}
  	$\Gt{\Gamma \sigma};\es;\Gt{\Delta \sigma}  \cat
  				\Theta \cat \dual{\prop_k}:\btout{\cdot}\tinact 
				  \circ \Gt{\envR \sigma}
				\proves$ \\ 
  				$\news{\wtd c_r}(
  		\prod_{r \in \tilde v} 
		  (\recprovx{r}{x}{\wtd z})
		% \binp{\prop^r}{b}(\appl{b}{\wtd r}) 
		\Par
  		\propout{k}{} \inact \Par \B{k}{\epsilon}{P\sigma})$
  		\end{tabular}}
  	\LeftLabel{\scriptsize PolyResS}
  	\UnaryInfC{$\Gt{\Gamma \sigma};\es;\Gt{\Delta \sigma} \circ 
  	\Gt{\envR \sigma}
  	\proves \news{\wtd \prop}\news{\wtd c_r}(
  		\prod_{r \in \tilde v} 
		 (\recprovx{r}{x}{\wtd r})
		% \binp{\prop^r}{b}(\appl{b}{\wtd r}) 
		\Par
  		\propout{k}{} \inact \Par \B{k}{\epsilon}{P\sigma})$}	
  	\DisplayProof
  \end{align}

\end{proof}

\newpage 
\section{Appendix to \secref{ss:dynamic}}
\label{app:dynamic}

\subsection{Proof of \lemref{l:c-prop-closed}}
\label{mst:app:tauclosed}

\tauclosed* 

\begin{proof}
    By the induction on the structure of $P_1$. We consider two base cases: 
    \begin{itemize}
      \item Case $P_1=\inact$.
      Let $\wtd B$ such that $\wtd W \vrelate \wtd B$. 
      Then, the elements of $\Cb{\tilde W}{\tilde x}{\inact}$ are 
      $R_1 = \news{\prop_k}\apropout{k}{\wtd B} \Par \propinp{k}{\wtd x}\inact$
      and $R_2 = \inact$. Clearly, $R_1 \by{\tau} R_2$.

      \item Case $P_1 = \appl{V_1}{(\wtd r, u)}$. 
      % Let 
      Let $n = \len{\wtd r}$. 
      Then, we have 
      \begin{align*}
        \Cb{\tilde W}{\tilde x}{P_1} =  N_1 \cup N_4 
        \cup \big(
        \cup_{1 \leq l \leq n} N^l_2  \big) 
        \cup \big(
        \cup_{1 \leq l \leq n - 1} N^l_3 \big) 
      \end{align*}
      % $$\Cb{\tilde W}{\tilde x}{P_1} = N_1 \cup N_2 \cup N_5 
      % \bigcup_{1 \leg l \leg n} N^l_3 
      % \bigcup_{1 \leq l \leq n - 1} N^l_4$$ 
      % \{ R^l_3: 1 \leq l \leq n\}
      % \cup \{R^l_4 : 1 \leq l \leq n - 1 \}$$
      % $$\Cb{\tilde W}{\tilde x}{P_1} = \{R_1, R_2, R_5 \} \cup \{ R^l_3: 1 \leq l \leq n\}
      % \cup \{R^l_4 : 1 \leq l \leq n - 1 \}$$
       where 
      \begin{align*} 
        % N_1 &=  \{ \Rb{\tilde v, \tilde r \setminus r'} \Par \propout{k}{\wtd B}
        % \binp{\prop^{r'}}{b}(\appl{b}{\wtd r'})
        % \Par \B{k}{\tilde x}{P_1} 
        % : \wtd W \vrelate \wtd B \}
        % \\
        N_1 &= \{ \Rb{\tilde v, \tilde r} \Par \apropout{k}{\wtd B}
        \Par \B{k}{\tilde x}{P_1} : \wtd W \vrelate \wtd B \} 
        \\ 
        N^l_2 &= \{
        \Rb{\tilde v, r_l, \ldots, r_n} \Par
                     \overbracket{\prop^{r_l}!\big\langle 
                    \lambda \wtd z_l. \prop^{r_{l+1}}!\langle\lambda \wtd z_{l+1}.\cdots. 
                 \prop^{r_n}!\langle \lambda \widetilde z_n.}^{|\tilde r| - l +1} 
                 Q^{V_2}_l \rangle \,\rangle \big\rangle 
                 : V_1 \subst{\tilde W}{\tilde x} \vrelate V_2 \}
                 \\
        N^l_3 &= \{
        \Rb{\tilde v, r_{l+1}, \ldots, r_n} \Par \lambda \wtd z_l.
                 \overbracket{
                     \prop^{r_{l+1}}!\langle\lambda \wtd z_{l+1}.\cdots. 
                 \prop^{r_n}!\langle \lambda \wtd z_n.}^{|\tilde r| - l} 
                 Q^{V_2}_l \rangle \, \big\rangle \ {\wtd r_l} 
                 : V_1 \subst{\tilde W}{\tilde x} \vrelate V_2  \}
                 \\
        N_4 &= \{
        \Rb{\tilde v} \Par \appl{V_2}{(\wtd r_1,\ldots,\wtd r_{n}, \wtd m)}
        :  V_1 \subst{\tilde W}{\tilde x} \vrelate V_2 \}
        \end{align*}
    
      \noindent with 
      \begin{align*} 
      Q^{V_2}_l = \appl{V_2}{(\wtd r_1,\ldots,\wtd r_{l-1}, \wtd z_{l}, \ldots, 
               \wtd z_n, \wtd m)}
      \end{align*} 
      We can see that for $R_1 \in N_1$ there exist $R^0_2 \in N^0_2$ such that
      $R_1 \by{\tau} R^{0}_2$. Further, we observe that for $R^l_2 \in N^l_2$
      there is  $R^l_3 \in N^l_3$ such that $R^l_2 \by{\tau} R^l_3$. Now, we can
      see that for $R^l_3 \in N^l_3$ there is $R^{l+1}_2 \in N^{l+1}_2$ such that
      $R^l_3 \by{\tau} R^{l+1}_2$. Finally, we have for 
      $$R^{n}_3 \in N^n_3 = \{ \Rb{\tilde v}  \Par 
      \appl{\abs{\wtd z_n}{Q^{V_2}_n}}{\wtd r_n} :
       V_1 \subst{\tilde W}{\tilde x} \vrelate V_2 \}$$ 
      
      there is $R_4 \in N_4$ such that $R^{n}_3 \by{\tau} R_4$. 
    \end{itemize}
    
      \noindent We consider two inductive cases as remaining cases are similar: 
      \begin{itemize}
        \item Case $P_1 = \bout{u_i}{V_1}P_2$. 
        % We dist
        We distinguish two sub-case: \rom{1} $\neg \tr(u_i)$ 
        and \rom{2} $\tr(u_i)$. 
        In both sub-cases, we  distinguish two kinds of an object value $V_1$: 
        (a)~ $V_1 \equiv x$, such that  
        $\subst{V_x}{x}\in \subst{\tilde W}{\tilde  x}$ 
        and (b)~ $V_1 = \abs{y:C}P'$, that is $V_1$ is a pure abstraction. 
      % We distinguish two sub-cases: \rom{1} $V_1 \equiv x$, such that  
      %   $\subst{V_x}{x}\in \subst{\tilde W}{\tilde  x}$ and 
      %   \rom{2} $V_1 = \abs{y:C}P'$, that is $V_1$ is a pure abstraction.
    
        First, we consider sub-case \rom{1}. 
        Let $\wtd y = \fv{V_1}$, $\wtd w = \fv{P_2}$, $\wtd W_1$, and $\wtd W_2$ such that 
        $\subst{\wtd W}{\wtd x} = \subst{\wtd W_1}{\wtd y}\cdot
         \subst{\wtd W_2}{\wtd w}$. 
        Further, let $\wtd v = \rfn{P_1\subst{\tilde W}{\tilde x}}$ and 
        $\sigma = \subst{u_{i+1}}{u_i}$. 
        Then,  by the definition of $\Cb{-}{-}{-}$ (\tabref{mst:t:tablecd}),  
        we have that
        $\Cb{\tilde W}{\tilde x}{\bout{u_i}{V_1}P_2} = N_1 \cup N_2$ where: 
        \begin{align*}
          N_1 &= \{   
          \Rb{\tilde v} \Par \apropout{k}{\wtd B}
          \Par \B{k}{\tilde x}{P_1} : \wtd W \vrelate \wtd B \} 
          \\
          % N_2 &= \{ \Rb{\tilde v \setminus r} \Par \propout{k}{\wtd B}
          % \binp{\prop^r}{b}(\appl{b}{\wtd r})
          % \Par \B{k}{\tilde x}{P_1} : r \in \wtd v,\ \wtd W \vrelate \wtd B \} 
          %   \\
             N_2 &= \{
               \Rb{\tilde v} \Par 
             \bout{u_i}{V_2} \apropout{k}{\wtd B_2} \Par 
             \B{k_2}{\tilde w}{P_2\sigma} 
             :  V_1\sigma\subst{\wtd W_1}{\wtd y} \vrelate V_2,\ \wtd W_2 \vrelate \wtd B_2 \}  
          \end{align*}
          
          \noindent 
          We can see that in both cases of $V_1$, a variable or a pure abstraction,
          we have that for $R_1 \in N_1$ there is $R_2 \in N_2$ such that $R_1
          \by{\tau} R_2$. 
        % \noindent First, we could see that $R_1 \by{\tau} R_3$. 
        % Further, for $R_2 \in N_2$ there is 
        % $R_3 \in N_3$ such that $R_2 \by{\tau} R_3$ as 
        % by \defref{mst:def:rec-providers} we have 
        % $\Rb{\tilde v} = \Rb{\tilde v \setminus u} \Par 
        % \binp{\prop^u}{b}(\appl{b}{\wtd u})$. 
        In the sub-case
          (a)~by $\subst{V_x}{x}\in \subst{\tilde W}{\tilde x}$ we have  $V_2 \in
          \wtd B$, such that $V_x \vrelate V_2$, 
        that by a $\tau$-move substitutes $x$ in $\B{k}{\tilde x}{P_1}$. 
        Thus, by  
         $x \sigma \subst{\tilde W_1}{\tilde y}=V_x$ we have 
         $V_1 \sigma \subst{\tilde W_1}{\tilde y} \vrelate V_2$. 
        In the sub-case (b), by definition of $\B{-}{-}{-}$
        we have that 
        by a $\tau$-move $\wtd B_1 \subseteq \wtd B$
        substitute $\wtd y$ in $\V{k+1}{\tilde y}{V_1\sigma}$ so we have  
        $R_1 \by{\tau} R_2$ where 
    % \begin{align*}
    %   \V{k+1}{\tilde
    %      y}{V_1}\subst{\tilde B_1}{\tilde y} = 
    % \end{align*}
         $$V_2 = \V{k+1}{\tilde y}{V_1\sigma}\subst{\tilde B_1}{\tilde y}$$ 
         
         Now, we may notice that by \tabref{mst:t:tablecd} we have 
         \begin{align*}
          V_2 \in \Vb{\tilde W_1}{\tilde y}{V_1\sigma}
         \end{align*}
    
         Hence, by \Cref{mst:def:valuesset,mst:def:vrelate} we have 
         $V_1\sigma \subst{\tilde W_1}{\tilde y} \vrelate V_2$. 
         Further, we may notice that there is no  
         $\tau$-transition involving non-essential prefixes in $R_2$. 
        This concludes this sub-case. 
    
        Now, we consider sub-case \rom{2}. 
        Let $\wtd y = \fv{V_1}$, $\wtd w = \fv{P_2}$, $\wtd W_1$, and $\wtd W_2$ such that 
        $\subst{\wtd W}{\wtd x} = \subst{\wtd W_1}{\wtd y}\cdot
         \subst{\wtd W_2}{\wtd w}$. 
        Then, we have 
        $\Cb{\tilde W}{\tilde x}{P_1} = N_1 \cup N_2 \cup N_3 \cup N_4$ 
        where 
        \begin{align*}
          N_1 &= \{ \Rb{\tilde v} \Par \apropout{k}{\wtd B}
          \Par \B{k}{\tilde x}{P_1} : \wtd W \vrelate \wtd B \} 
          % \\
          % N_2 &= \{ \Rb{\tilde v \setminus r} \Par \propout{k}{\wtd B}
          % \binp{\prop^r}{b}(\appl{b}{\wtd r})
          % \Par \B{k}{\tilde x}{P_1} : \wtd W \vrelate B \} 
          \\
          N_2 &= 
          \{ \Rb{\tilde v} \Par \abbout{\prop^u}{M^{\tilde B_2}_{V_2}} \Par 
          \B{k}{\tilde w}{P_2} 
          : V_1 \subst{\tilde W_1}{\tilde y} \vrelate V_2,\ \wtd W_2 \vrelate \wtd B_2 \} 
          \\
          N_3 &= \{ \Rb{\tilde v \setminus u} \Par 
          \appl{M^{\tilde B_2}_{V_2}}{\wtd u} \Par \B{k}{\tilde w}{Q}
          : V_1 \subst{\tilde W_1}{\tilde y} \vrelate V_2,\ \wtd W_2 \vrelate \wtd B_2 \} \\ 
          N_4 &= \{ \Rb{\tilde v \setminus u} \Par {\bbout{u_{\indT{S}}}{V_2}}{} 
          \propout{k} {\wtd B_2} \binp{\prop^u}{b}
      (\appl{b}{\wtd u}) \Par \B{k}{\tilde w}{P_2} 
      : V_1 \subst{\tilde W_1}{\tilde y} \vrelate V_2,\ \wtd W_2 \vrelate \wtd B_2 \} 
        \end{align*}
    where 
        \begin{align*}
          M^{\tilde B_2}_{V_2} = \abs{\wtd z}
                {\bbout{z_{\indT{S}}}{V_2}}{}  \apropout{k}{\wtd B_2} 
                \Par \binp{\prop^u}{b}(\appl{b}{\wtd z})
        \end{align*}
        \noindent 
        As in the previous sub-case, we can see that in both cases of $V_1$, a
        variable or a pure abstraction, we have that for $R_1 \in N_1$ there is 
        $R_2 \in N_2$ such that $R_1 \by{\tau} R_2$, 
       for appropriate choice 
        of $\wtd B$, $\wtd B_2$, and $k$.
        Similarly, by the communication 
          on shared name $\prop^u$ for $R_2 \in N_2$ 
          there is $R_3 \in N_3$ such that $R_2 \by{\tau} R_3$.
        Finally, for $R_3
        \in N_3$ there is $R_4 \in N_4$ such that $R_3 \by{\tau} R_4$ by the
        application. This concludes the output case.

        \item 
        Case $P_1 = Q_1 \Par Q_2$. 
        Let $\wtd y = \fv{Q_1}$, $\wtd w = \fv{Q_2}$, $\wtd W_1$, and $\wtd W_2$ 
        such that 
        $\subst{\wtd W}{\wtd x} = 
        \subst{\tilde W_1}{\tilde y} \cdot \subst{\tilde W_2}{\tilde w}$. 
        Then,  by definition of $\Cb{-}{-}{-}$ (\tabref{mst:t:tablecd}),  we have that
        $\Cb{\tilde W}{\tilde x}{Q_1 \Par Q_2} = N_1 \cup N_2 \cup N_3$ where: 
        \begin{align*}
          N_1 &= 
          \{ 
            \apropout{k}{\wtd B} \Par \B{k}{\tilde x}{Q_1 \Par Q_2}
          : \wtd W \vrelate \wtd B \}
          \\
        N_2 &= \{ \propout{k}{\wtd B_1}
        \apropout{k+l}{\wtd B_2} \Par \B{k}{\tilde y}{Q_1} \Par \B{k+l}{\tilde w}{Q_2} 
        : \wtd W_1 \vrelate \wtd B_1,\ \wtd W_2 \vrelate \wtd B_2 \} \\
        N_3 &= \{ R_1 \Par R_2 : R_1 \in \Cb{\tilde W_1}{\tilde y}{Q_1},\ 
        R_2 \in \Cb{\tilde W_2}{\tilde w}{Q_2} \}
        \end{align*}
        \noindent with $l = \plen{Q_1}$. 
        
        We show that the thesis holds for processes in each of these three sets:
        \begin{enumerate}
            \item 
        
        Clearly, by picking 
        appropriate $\wtd B$, $\wtd B_1$, and $\wtd B_2$, 
        for any $R^1_1 \in N_1$ there is $R^1_2$ such that 
        $R^1_1 \by{\tau} R^1_2$ and $R^1_2 \in N_2$.
        
        \item 
        Now, we consider set $N_3$. 
        Let us pick $R^3 = R_1 \Par R_2 \in N_3$, for some  
        $R_1 \in \Cb{\tilde W_1}{\tilde y}{Q_1}$ 
        and $R_2 \in \Cb{\tilde W_2}{\tilde w}{Q_2}$. 
        By the definition of $\Cb{-}{-}{-}$ (\tabref{mst:t:tablecd}) we know
        all propagator names are restricted element-wise in $R^3$, and so there
        is no communication between $R_1$ and $R_2$ on propagator prefixes. 
        This ensures that any $\tau$-actions emanating from $R^3$ arise from $R_1$ or $R_2$ separately, not from their interaction. 
        The thesis then follows by IH, for we know that if $R_1 \by{\tau} R'_1$ then $R'_1 \in \Cb{\tilde
        W_1}{\tilde y}{Q_1}$; similarly,  
         if $R_2 \by{\tau} R'_2$ then $R'_2 \in \Cb{\tilde W_2}{\tilde w}{Q_2}$.
        Thus, by the definition $R'_1 \Par R'_2 \in N_3$. 
    
        \item  Finally, we  show that for any $R_2 \in N_2$ if $R_2 \by{\tau} R'_2$
        then $R'_2 \in \Cb{\tilde W}{\tilde x}{Q_1 \Par Q_2}$. We know that $R_2 =
        \propout{k}{\wtd B_1} \apropout{k+l}{\wtd B_2} \Par \B{k}{\tilde y}{Q_1}
        \Par \B{k+l}{\tilde w}{Q_2}$, where $\wtd B_i$ are such that $\wtd W_i
        \vrelate \wtd B_i$ for $i \in \{1,2 \}$. Clearly, by a synchronization on
        $c_k$, we have 
        $$R_2 \by{\tau} \apropout{k+l}{\wtd B_2} \Par B_{Q_{1}} \Par \B{k+l}{\tilde
        w}{Q_2} = R'_2$$ where $B_{Q_{1}}$ stands for the derivative of
        $\B{k}{\tilde y}{Q_1}$ after the synchronization (and substitution of $\wtd
        B_1$).
        
        To show that $R'_2$ is already in $\Cb{\tilde W}{\tilde x}{Q_1 \Par Q_2}$,
        we consider an  $R^3 \in N_3$ such that  
        \begin{align*}
          R^3 = \apropout{k}{\wtd B_1}
          \Par  \B{k}{\tilde y}{Q_1} \Par 
          \apropout{k+l}{\wtd B_2} \Par \B{k+l}{\tilde w}{Q_2}
        \end{align*}
       Note that there is a $\tau$-transition on $\prop_k$ such that 
        $R^3 \by{\tau} R'_2$. 
        Because processes in $N_3$ satisfy the thesis (cf. the previous
        sub-sub-case), we have that 
    $R'_2 \in N_3$. 
    Therefore, $R'_2 \in \Cb{\tilde W}{\tilde x}{Q_1 \Par Q_2}$, as desired.
        This concludes parallel composition case.  
        \end{enumerate}
          \end{itemize}
      This concludes the proof of \lemref{l:c-prop-closed}. 
    \end{proof}

\subsection{Proof of \lemref{lemma:mstbs1}}
\label{mst:app:mstbs1}

For the proof we will use the following syntactic sugar.
\begin{definition}[Function  $\Cbs{-}{-}{-}$ ]
  \label{d:cbs-ho}
  Let $P$ be an \HO process, $\rho$ be a value substitution, and $\sigma$ be an indexed 
  name substitution. We define $\Cbs{\rho}{\sigma}{P}$ as follows: 
  \begin{align*}
    \Cbs{\rho}{\sigma}{P_1}  &= \Cb{\tilde W \sigma}{\tilde x}{P_1\sigma} 
    \qquad  \quad \text{with } \rho = \subst{\tilde W}{\tilde x}
  \end{align*}
  % \noindent where $\wtd u : \wtd C$. 
\end{definition}

\mstbs* 

\begin{proof} 
    By transition induction. Let $\rho_1 = \subst{\tilde W}{\tilde x}$. 
    By inversion of $P_1 \rho_1 \relS Q_1$
    we know there is $\sigma_1 \in \indices{\fn{P\rho_1}}$
    such that $Q_1 \in \Cbs{\rho_1}{\sigma_1}{P_1}$.
    Then,  we need the following assertion on the index substitution. 
     If 
    $P_1\rho_1 \by{\ell} P_2 \rho_2$  
    and $\subj{\ell}=n$ such that $\neg\tr(n)$ then 
   there exists 
    $Q_2$ such that $Q_1 \By{\iname \ell} Q_2$ 
    with $\subj{\iname \ell} = n_i$ and 
    $Q_2 \in \Cbs{\rho_2}{\sigma_2}{P_2}$ 
    such that $\sigma_2 \in \indices{P_2\rho_2}$ and $\subsqn{n_i} \in \sigma_2$. 
     
    First, we consider three base cases: Rules~\texttt{Snd}, \texttt{Rv}, and
    \texttt{App}. Then, we distinguish five inductive cases and analyze three cases
    (as cases $\texttt{Par}_R$ and $\texttt{Par}_L$, and \texttt{New} and
    \texttt{Res} are similar).  
    Thus, in total we consider six cases: 
   
     \begin{enumerate}
       \item Case $\rulename{Snd}{}$. 
       Then $P_1 = \bout{n}{V_1}P_2$. We first consider the case when $P_1$ is not
       a trigger collection, and then briefly discuss the case when it is a trigger
       collection. 
       
       We distinguish two sub-cases: \rom{1} $\neg \tr(n)$ 
       and \rom{2} $\tr(n)$. 
       In both sub-cases, we  distinguish two kinds of an object value $V_1$: 
       (a)~ $V_1 \equiv x$, such that  
       $\subst{V_x}{x}\in \subst{\tilde W}{\tilde  x}$ 
       and (b)~ $V_1 = \abs{y:C}P'$, that is $V_1$ is a pure abstraction. 
       Next, we consider two sub-cases: 
       \begin{enumerate}[i)]
         \item Sub-case $\neg \tr(n)$. 
         Let $\wtd W_1$, $\wtd W_2$, $\wtd y$, and $\wtd w$ such that 
         $$P_1\subst{\tilde W}{\tilde x} = 
         \bout{n}{V_1\subst{\tilde W_1}{\tilde y}}P_2 
         \subst{\tilde W_2}{\tilde w}$$
        
         We have the following transition: 
         \begin{align*} 
           \AxiomC{} 
           \LeftLabel{\scriptsize $\rulename{Snd}{}$}
           \UnaryInfC{$P_1\subst{\tilde W}{\tilde x} 
           \by{\bactout{n}{V_1 \subst{\tilde W_1}{\tilde y}}} 
           P_2 \subst{\tilde W_2}{\tilde w}
           $}
           \DisplayProof
         \end{align*}
     
         Let    
         $\sigma_1 \in \indices{\wtd u}$ where 
         $\wtd u = \fn{P_1\subst{\tilde W}{\tilde x}}$ such that $\subst{n_i}{n} \in \sigma_1$.  
          Also, let $\sigma_2 = \sigma_1 \cdot \subsqn{n_i}$. 
          % Also, let $\sigma_2 = \sigma_1 \cdot \subsqn{n_i} \cdot \subst{t_j}{t}$ for some $j > 0$. 
          % Further, let 
          % $\wtd v = \rfn{P_1}$, 
          Further, let 
          $\wtd v = \rfn{P_1\subst{\tilde W}{\tilde x}}$, 
       $\wtd \prop_r = \cup_{r \in \tilde v} \prop^r$, 
       $\wtd \prop_k = (\prop_k, \ldots, \prop_{k+\plen{P_1}-1})$, and $\wtd
       \prop_{k+1} = (\prop_{k+1}, \ldots, \prop_{k+\plen{P_1}-1})$.
          When $P_1$ is not a
         trigger, by definition of $\relS$ (\tabref{mst:t:tablecd}), for both
         sub-cases, 
         we have 
         $Q_1 \in N_1 \cup N_2$ where: 
         \begin{align*}
           N_1 &= \big\{   
             \news{\wtd \prop_r} \news{\wtd \prop_k} \Rb{\tilde v} \Par 
           \apropout{k}{\wtd B}
           \Par \B{k}{\tilde x}{P_1\sigma_1} : \wtd W \vrelate \wtd B \big\} 
           % \\
           % N_2 &= \big\{ 
           %   \news{\wtd \prop_r} \news{\wtd \prop_k}  
           % \Rb{\tilde v \setminus r} \Par \propout{k}{\wtd B}
           % \binp{\prop^r}{b}(\appl{b}{\wtd r})
           % \Par \B{k}{\tilde x}{P_1\sigma_1} : r \in \wtd v,\ 
           % \wtd W \vrelate \wtd B \big\} 
             \\
              N_2 &= \big\{
               \news{\wtd \prop_r}
               \news{\wtd \prop_{k+1}}
                \Rb{\tilde v} \Par 
              \bout{n_i}{V_2} \apropout{k+1}{\wtd B_2} \Par 
              \B{k+1}{\tilde w}{P_2\sigma_2} 
              :  
              V_1\sigma\subst{\wtd W_1}{\wtd y} \vrelate V_2,\ 
              \wtd W_2 \vrelate \wtd B_2 \big\}  
           \end{align*}
     
       \noindent 
      %  with $\wtd \prop_k = (\prop_k, \ldots, \prop_{k+\plen{P_1}-1})$ and $\wtd
      %  \prop_{k+1} = (\prop_{k+1}, \ldots, \prop_{k+\plen{P_1}-1})$.
        If $Q_1 \in N_1$, 
       then there is some $Q_2 \in N_2$ such that $Q_1$ reduces to
       $Q_2$ through communication on non-essential prefixes. By
       \lemref{l:c-prop-closed} it is then sufficient to consider the situation
       when $Q_2 \in N_2$. 
      %  Let $\wtd \prop_{r_1}= \fcr{P_2}$ and $\wtd \prop_{r_2} = \fcr{V_2}$. 
      % Let $\wtd v_1=\rfn{P_2}$, 
      % $\wtd v_2=\rfn{V_2}$, 
      Let $\wtd v_1=\rfn{P_2\subst{\tilde W_2}{\tilde z}}$, 
      $\wtd v_2=\rfn{V_2\subst{\tilde W_1}{\tilde y}}$, 
      $\wtd \prop_{r_1} = \cup_{r \in \tilde v_1} \prop^r$, 
      and 
      $\wtd \prop_{r_2} = \cup_{r \in \tilde v_2} \prop^r$.
       By \defref{mst:def:rec-providers} and the assumption that 
       $P_1\subst{\tilde W}{\tilde x}$ is well-typed, 
        we have  
       % $\Rb{\tilde v} =  \Rb{\tilde v_1} \Par \Rb{\tilde v_2}$ and 
       $\wtd \prop_r = \wtd \prop_{r_1} \cdot \wtd \prop_{r_2}$ 
       and $\wtd \prop_{r_1} \cap \wtd \prop_{r_2} = \es$. 
       In that case we have the following transition:
       \begin{align}
       \label{pt:snd1}
       \AxiomC{} 
       \LeftLabel{\scriptsize $\rulename{Snd}{}$}
         \UnaryInfC{$\bout{n_i}{V_2}\apropout{k+1}{\wtd B_2} 
         \by{\bactout{n_i}{V_2}} 
         \apropout{k+1}{\wtd B_2}$} 
         \AxiomC{}
         \LeftLabel{\scriptsize $\rulename{Par}{_L}$}
         \BinaryInfC{$\bout{n_i}{V_2}\apropout{k+1}{\wtd B_2} \Par 
         \B{k+1}{\tilde w}{P_2\sigma_2}
         \by{\bactout{n_i}{V_2}} R_1$} 
         \AxiomC{}
         \LeftLabel{\scriptsize  $\rulename{Par}{_R}$}
           \BinaryInfC{
             $\Rb{\tilde v} \Par 
             \bout{n_i}{V_2} \apropout{k+1}{\wtd B_2} \Par 
             \B{k+1}{\tilde w}{P_2\sigma_2} \by{\bactout{n_i}{V_2}} R_1$
         }
         \AxiomC{$
         \wtd \prop_r \cdot 
         \wtd \prop_{k+1} \cap \fn{\bactout{n_i}{V_2}} = \emptyset$} 
         \LeftLabel{\scriptsize $\rulename{New}{}$}
         \BinaryInfC{
           $
           \news{\wtd \prop_r}
           \news{\wtd \prop_{k+1}}
           (\Rb{\tilde v} \Par 
         \bout{n_i}{V_2} \apropout{k+1}{\wtd B_2} \Par 
         \B{k+1}{\tilde z}{P_2\sigma_2}) \by{\news{\wtd \prop_{r_2}}\bactout{n_i}{V_2}} Q'_2$
         }
         \DisplayProof
       \end{align}
     
     \noindent where 
     $Q_2' = 
     \news{\wtd \prop_{r_1}}
     \news{\wtd \prop_{k+1}}\Rb{\tilde v} \Par 
     \apropout{k+1}{\wtd B_2} 
     \Par \B{k+1}{\tilde w}{P_2\sigma_2}$ with 
     $V_1 \subst{\tilde W_1}{\tilde y} \sigma_2 \vrelate V_2$ and 
     $\wtd W_2 \sigma_1 \vrelate \wtd B_2$. 
     Then, we shall show the following:  
     \begin{align}
       \label{eq:snd-goal1}
       ({P_2 \parallel \htrigger{t}{V_1}})
         \subst{\tilde W}{\tilde x}
                ~\relS~
        \news{\wtd \prop_{r_2}} (Q'_2 \parallel \htrigger{\tni}{V_2}) 
      \end{align}
   
  % Let $\iname t = t_j$ for some $j > 0$. 
  %  Let $\wtd v_1=\rfn{P_2}$ and $\wtd v_2=\rfn{V_2}$. 
   By assumption that $P_1\subst{\tilde W}{\tilde x}$ is well-typed, 
   we know $\wtd v = \wtd v_1 \cdot \wtd v_2$ and $\wtd v_1 \cap \wtd v_2 = \emptyset$. 
   Thus, 
   by \defref{mst:def:rec-providers} we know 
   $\Rb{\tilde v} =  \Rb{\tilde v_1} \Par \Rb{\tilde v_2}$, i.e., 
   \begin{align*}
     &\news{\wtd \prop_{r_2}} (Q'_2 \parallel \htrigger{\tni}{V_2})  
     \\
     & \qquad \equiv 
     \news{\wtd \prop_{r_1}} \news{\wtd \prop_{r_2}}     
     {\newsp{\wtd\prop_{k+1}}
     {\Rb{\tilde v_1} 
     \Par \apropout{k+1} {\wtd B} \Par \B{k+1}{\tilde w}{P_2\sigma_1} 
     \Par \Rb{\tilde v_2} \parallel 
     \htrigger{\tni}{V_2}}} 
     \\ 
     & \qquad \equiv 
     \news{\wtd \prop_{r}}   
     {\newsp{\wtd\prop_{k+1}}
     {\Rb{\tilde v_1} 
     \Par \apropout{k+1} {\wtd B} \Par \B{k+1}{\tilde w}{P_2\sigma_1} 
     \Par \Rb{\tilde v_2} \parallel 
     \htrigger{\tni}{V_2}}} 
     \\
     & \qquad =  
     \news{\wtd \prop_{r}}   
     \news{\wtd\prop_{k+1}} R 
   \end{align*}
      
     %By $\sigma_2 = \sigma_1 \cdot \subsqn{n_i}$ and \defref{mst:d:indexedsubstitution}, we know 
     %\begin{align*}
     %  \sigma_2 \in 
     %  \indices{\fn{(P_2 \parallel  \htrigger{t}{V_1}) \subst{\tilde W}{\tilde x}} }
     %\end{align*}
      
     % 
     %\noindent We know there is 
     %\begin{align*}
     %  \sigma'_2 \in 
     %  \indices{\fn{(P_2 \parallel  \htrigger{t}{V_1}) \subst{\tilde W}{\tilde x}} }
     %\end{align*}
     %
     %such that 
     %$(P_2 \parallel  \htrigger{t}{V_1}) \subst{\tilde W}{\tilde x} \sigma_2 =
     %(P_2 \parallel  \htrigger{t}{V_1}) \subst{\tilde W}{\tilde x} \sigma_2'$. 
     
     From the definition of $\relS$ (\defref{mst:d:relation-s}) we have that $n \not\in \fn{\wtd W_2}$ and thus
     $\wtd W_2 \sigma_1 = \wtd W_2 \sigma_2$. 
     Thus, by the definition of $\wtd B_2$ ($\wtd W_2 \sigma_2 \vrelate \wtd B_2$)
     we can see that  
     \begin{align}
       \label{eq:ml-snd-p1}
       \Rb{\tilde v_1} 
       \Par \apropout{k+1} {\wtd B} \Par \B{k+1}{\tilde w}{P_2\sigma_1} 
       \Par \Rb{\tilde v_2} \in \Cb{\tilde W_2\sigma_2}{\tilde w}{P_2\sigma_2}
     \end{align}
     Now, we can see that assertion $\subsqn{n_i} \in \sigma_2$ holds, as by
     definition $\sigma_2 = \sigma_1 \cdot \subsqn{n_i}$.
     Let $\sigma_2' = \sigma_2 \cdot \subst{\tni}{t}$. 
     Then, we have  
     \begin{align*}
       \Cb{\tilde W_1 \sigma'_2}{\tilde y}{\htrigger{\tni}{V_1\sigma_2}} =  
       \{P': (\htrigger{\tni}{V_1})\subst{\tilde W_1}{\tilde y} \sigma_2 
       \processrelate P' \}
     \end{align*}
     
     As 
     $(\htrigger{\tni}{V_1}) \subst{\tilde W_1}{\tilde y} \sigma_2
     =\htrigger{\tni}{V_1 \subst{\tilde W_1}{\tilde y} \sigma_2}$
     and $V_1 \subst{\tilde W_1}{\tilde y} \sigma_2 \vrelate V_2$ 
     we have 
     \begin{align}
       \label{eq:ml-snd-v1}
       (\htrigger{\tni}{V_1}) \subst{\tilde W_1}{\tilde y} \sigma_2 \processrelate 
       (\htrigger{\tni}{V_2})  
     \end{align}
     that is 
     $$
     \Rb{\tilde v_2} \parallel 
     \htrigger{\tni}{V_2} \in \Cb{\tilde W_1\sigma_2'}{\tilde y}{\htrigger{\tni}{V_1\sigma_2}}
     $$

     Finally, by \tabref{mst:t:tablecd}  we have 
     \begin{align*}
      \Cb{\tilde W\sigma_2'}{\tilde x}{P_2\sigma_2 \parallel \htrigger{\tni}{V_1\sigma_2}} 
      = \big\{ R_1\parallel R_2 :  
      R_1 \in \Cb{\tilde W_2\sigma_2'}{\tilde w}{P_2\sigma_2},\ 
      R_2 \in \Cb{\tilde W_1\sigma_2'}{\tilde y}{\htrigger{\tni}{V_1\sigma_2}} \big\}	
     \end{align*}
     So, by this, \eqref{eq:ml-snd-p1}, and \eqref{eq:ml-snd-v1}
     we have:  
      \begin{align*}
        R 
        \in \Cb{\tilde W \sigma_2'}{\tilde x}{P_2\sigma_2' \parallel \htrigger{\tni}{V_1\sigma_2}}
      \end{align*}
      
      Further, by the assumption $\dual n \not\in \wtd u$ we know 
      $\subst{\dual n_i}{\dual n} \not\in \sigma_1$. 
      Thus, by $\sigma_2' = \sigma_1 \cdot \subsqn{n_i} \cdot \subst{t_j}{t}$ and \defref{mst:d:indexedsubstitution}, 
      we have  
     \begin{align*}
       \sigma_2' \in 
       \indices{\fn{(P_2 \parallel  \htrigger{t}{V_1}) \subst{\tilde W}{\tilde x}} }
     \end{align*} 
     
      Hence, the goal \eqref{eq:snd-goal1} follows. This concludes sub-case~$\neg\tr(n)$.
         \item Sub-case $\tr(n)$. 
       Let $\wtd W_1$, $\wtd W_2$, $\wtd y$, and $\wtd w$ be such that 
         $$P_1\subst{\tilde W}{\tilde x} = 
         \bout{n}{V_1\subst{\tilde W_1}{\tilde y}}P_2 \subst{\tilde W_2}{\tilde w}$$
      
       The transition inference tree is as follows: 
   
       \begin{align*} 
         \AxiomC{} 
         \LeftLabel{\scriptsize $\rulename{Snd}{}$}
         \UnaryInfC{$P_1\subst{\tilde W}{\tilde x} 
         \by{\bactout{n}{V_1 \subst{\tilde W_1}{\tilde y}}} 
         P_2 \subst{\tilde W_2}{\tilde w}
         $}
         \DisplayProof
       \end{align*}
   
       Let $\sigma_1 \in \indices{\wtd u}$ 
       where $\wtd u = \fn{P_1\subst{\tilde W}{\tilde x}}$. 
      %  Let $\sigma_1 = \sigma' \cdot \subst{n_{\indT{S}}}{n}$ 
      %  where $\sigma' = \indices{\wtd u}$ with 
      %  $\wtd u = \fn{P_1\subst{\tilde W}{\tilde x}}$ and $n:S$.
       Also, let $\wtd v = \rfn{P_1}$, 
       $\wtd \prop_r = \cup_{r \in \tilde v} \prop^r$, 
       $\wtd \prop_k = (\prop_k, \ldots, \prop_{k+\plen{P_1}-1})$,  
       $\wtd \prop_{k+1} = (\prop_{k+1}, \ldots, \prop_{k+\plen{P_1}-1})$, 
       and let $S$ be such that $n:S$. 
       Then, by the definition of $\relS$ (\defref{mst:d:relation-s})  
       we have $Q_1 \in N_1 \cup N_2 \cup N_3 \cup N_4$ where 
   \begin{align*} 
       % N_1 &= \{ 
       % \news{\wtd \prop_r}  
       % \news{\wtd \prop_k} \Rb{\tilde v \setminus r} \Par \propout{k}{\wtd B}
       % \binp{\prop^{r}}{b}(\appl{b}{\wtd r})
       % \Par \B{k}{\tilde x}{P_1\sigma_1} 
       % : r \in \rfn{P_1},\ \wtd W \sigma_1 \vrelate \wtd B \} 
       % \\
       N_1 &= \{
         \news{\wtd \prop_r}    
       \news{\wtd \prop_k} \Rb{\tilde v} \Par \apropout{k}{\wtd B}
       \Par \B{k}{\tilde x}{P_1\sigma_1}
       : \wtd W \sigma_1 \vrelate \wtd B \} 
       \\ 
       N_2 &= \{ 
         \news{\wtd \prop_r}    
       \news{\wtd \prop_{k+1}} \Rb{\tilde v} \Par
       \abbout{\prop^n}{M^{\tilde B_2}_{V_2}} \Par \B{k+1}{\tilde w}{P_2\sigma_1} 
       : \wtd W_2 \sigma_1 \vrelate \wtd B_2  \}
       \\
       N_3 &=  \{ 
         \news{\wtd \prop_r}  
         \news{\wtd \prop_{k+1}} \Rb{\tilde v} \Par 
       \appl{M^{\tilde B_2}_{V_2}}{\wtd n} \Par 
       \B{k+1}{\tilde w}{P_2\sigma_1} 
       : V_1\subst{\tilde W_1}{\tilde y} \sigma_1 \vrelate V_2,\ 
       \wtd W_2 \sigma_1 \vrelate \wtd B_2 \} 
       \\
       N_4 &= \big\{ 
         \news{\wtd \prop_r}    
       \news{\wtd \prop_{k+1}} \Rb{\tilde v \setminus n} \Par 
       {\bbout{n_{\indT{S}}}{V_2}}{} 
       (\apropout{k+1}{\wtd B_2} \Par 
       \recprovx{n}{x}{\wtd z}
      %  \binp{\prop^n}{b})(\appl{b}{\wtd n}) 
       \Par \B{k+1}{\tilde w}{P_2\sigma_1} \\ 
     & \qquad  \qquad\qquad\qquad\qquad\qquad\qquad\qquad\qquad\qquad
     : V_1\subst{\tilde W_1}{\tilde y} \sigma_1 \vrelate V_2,\ 
     \wtd W_2 \sigma_1 \vrelate \wtd B_2 \big\} 
   \end{align*}
       \noindent with 
       \begin{align*}
         M^{\tilde B_2}_{V_2} = \abs{\wtd z}
           {\bbout{z_{\indT{S}}}{V_2}}{} \big(\apropout{k+1}{\wtd B_2} 
           \Par 
           \recprovx{n}{x}{\wtd z}
          %  \binp{\prop^n}{b}(\appl{b}{\wtd z}) 
           \big) 
     \end{align*}
  %  %
  %  where $\wtd \prop_k = (\prop_k, \ldots, \prop_{k+\plen{P_1}-1})$ 
  %  and $\wtd \prop_{k+1} = (\prop_{k+1}, \ldots, \prop_{k+\plen{P_1}-1})$.
  %  %
   If $Q_1 \in N_1 \cup N_2 \cup N_3$, then $Q_1$ reduces to some $Q_4
   \in N_4$ through communication on non-essential prefixes. By
   \lemref{l:c-prop-closed} it suffices to consider the case when $Q_4 \in N_4$. 
    %  Let $\wtd \prop_{r_1}= \fcr{\B{k+1}{\tilde w}{P_2\sigma_1}}$ and $\wtd \prop_{r_2} = \fcr{V_2}$. 
    Let $\wtd v_1=\rfn{P_2}$, 
      $\wtd v_2=\rfn{V_2}$, 
      $\wtd \prop_{r_1} = \cup_{r \in \tilde v_1} \prop^r$, 
      and 
      $\wtd \prop_{r_2} = \cup_{r \in \tilde v_2} \prop^r$.
     By \defref{mst:def:rec-providers} and the assumption that $P_1\subst{\tilde W}{\tilde x}$ is well-typed 
      we have  
     $\wtd \prop_r = \wtd \prop_{r_1} \cdot \wtd \prop_{r_2}$ 
     and $\wtd \prop_{r_1} \cap \wtd \prop_{r_2} = \es$. 
   We
   then infer the following transition:    
   \begin{align*}
     Q_4 \by{\bactout{\news{\wtd \prop_{r_2}}n_{\indT{S}}}{V_2}} Q'_4
   \end{align*}
   where $Q'_4 =  \news{\wtd \prop_{r_1}}
   \news{\wtd \prop_{k+1}}\Rb{\tilde v \setminus n} \Par \apropout{k+1} {\wtd B_2}
   \Par 
   \binp{\prop^n}{b} (\appl{b}{\wtd n}) \Par \B{k+1}{\tilde w}{P_2\sigma_1}$. 
   Then, we shall show the following 
   \begin{align}
       ({P_2 \parallel \htrigger{t}{V_1}})
           \subst{\tilde W}{\tilde x}
                        \ \relS\
       {{
         \news{\wtd \prop_{r_2}} \big( 
           Q'_4
       %     \news{\wtd \prop_{r_1}}
       %   % \news{\wtd \prop_r}    
       % \news{\wtd \prop_{k+1}}R' 
       \parallel \htrigger{\tni}{V_2}}} \big) 
       \label{eq:snd-rec-goal1}
    \end{align}
   
  %  Let $\wtd v_1=\rfn{P_2}$ and $\wtd v_2=\rfn{V_2}$. 
   By assumption that $P_1\subst{\tilde W}{\tilde x}$ is well-typed, 
   we know $\wtd v = \wtd v_1 \cdot \wtd v_2$ and $\wtd v_1 \cap \wtd v_2 = \emptyset$. 
   Hence, 
   by \defref{mst:def:rec-providers} we know 
   $\Rb{\tilde v} =  \Rb{\tilde v_1} \Par \Rb{\tilde v_2}$, that is 
   \begin{align*}
     Q'_4 &\equiv 
     \news{\wtd \prop_{r_1}}
   \news{\wtd \prop_{k+1}}\Rb{\tilde v \setminus n} \Par
   \recprovx{n}{x}{\wtd n}
  %  \binp{\prop^n}{b} (\appl{b}{\wtd n})
   \Par \apropout{k+1} {\wtd B_2} 
    \Par \B{k+1}{\tilde w}{P_2\sigma_1} \\ 
   &= 
   \news{\wtd \prop_{r_1}}
   \news{\wtd \prop_{k+1}}
   \Rb{\tilde v}  
   \Par \apropout{k+1} {\wtd B_2} 
    \Par \B{k+1}{\tilde w}{P_2\sigma_1}
    \\ 
    &= 
    \news{\wtd \prop_{r_1}}
   \news{\wtd \prop_{k+1}}
   \Rb{\tilde v_1} \Par \Rb{\tilde v_2}
   \Par \apropout{k+1} {\wtd B_2} 
    \Par \B{k+1}{\tilde w}{P_2\sigma_1}
   \end{align*}
   
   That is, we have 
   \begin{align*}
     &\news{\wtd \prop_{r_2}} \big(  Q'_4 \parallel \htrigger{\tni}{V_2} \big)  
     \\ 
     & \qquad \equiv \news{\wtd \prop_{r_2}} \big(   \news{\wtd \prop_{r_1}}
     \news{\wtd \prop_{k+1}}
     \Rb{\tilde v_1} \Par \Rb{\tilde v_2}
     \Par \apropout{k+1} {\wtd B_2} 
      \Par \B{k+1}{\tilde w}{P_2\sigma_1} \parallel \htrigger{\tni}{V_2} \big)  
      \\ 
     & \qquad \equiv 
     \news{\wtd \prop_{r_1}} \news{\wtd \prop_{r_2}}     
     {\newsp{\wtd\prop_{k+1}}
     {\Rb{\tilde v_1} 
     \Par \apropout{k+1} {\wtd B} \Par \B{k+1}{\tilde w}{P_2\sigma_1} 
     \Par \Rb{\tilde v_2} \parallel 
     \htrigger{\tni}{V_2}}} 
     \\ 
     & \qquad \equiv 
     \news{\wtd \prop_{r}}   
     {\newsp{\wtd\prop_{k+1}}
     {\Rb{\tilde v_1} 
     \Par \apropout{k+1} {\wtd B} \Par \B{k+1}{\tilde w}{P_2\sigma_1} 
     \Par \Rb{\tilde v_2} \parallel 
     \htrigger{\tni}{V_2}}} =  
     \news{\wtd \prop_{r}}   
     \news{\wtd\prop_{k+1}} R 
   \end{align*}
   
   % \begin{align*}
   %   Q_2' \equiv 
   %   \news{\wtd \prop_{r_1}} \news{\wtd \prop_{r_2}}     
   %   {\newsp{\wtd\prop_{k+1}}
   %   {\Rb{\tilde v_1} 
   %   \Par \apropout{k+1} {\wtd B} \Par \B{k+1}{\tilde w}{P_2\sigma_1} 
   %   \Par \Rb{\tilde v_2} \parallel 
   %   \htrigger{t}{V_2}}}
   % \end{align*}
   
   Now, by \defref{mst:def:trigger} we may notice that $\wtd v_2 = \rfn{\htrigger{t}{V_2}}$. 
   Let $\sigma_1' = \sigma_1 \cdot \subst{\tni}{t}$.
   So, by  \tabref{mst:t:tablecd} we have 
   $$\Rb{\tilde v_1} \Par \apropout{k+1} {\wtd B} \Par \B{k+1}{\tilde w}{P_2\sigma_1} 
   \in \Cb{\tilde W_2\sigma_1}{\tilde w}{P_2\sigma_1}$$
   
   \noindent and 
   \begin{align} 
   \label{p:snd-trigger-v-c}
   \Rb{\tilde v_2} \parallel \htrigger{\tni}{V_2} \in 
   \Cb{\tilde W_1\sigma_1}{\tilde y}{\htrigger{\tni}{V_1\sigma'_1}}
   \end{align} 
   Thus, by the definition of the parallel composition 
   case of $\Cb{-}{-}{-}$ (\tabref{mst:t:tablecd}) we have 
   \begin{align*}
   	  R \in 
     \Cb{\tilde W\sigma'_1}{\tilde x}{(P_2 \parallel \htrigger{t}{V_1})\sigma'_1}
   \end{align*}
   Now, we can notice that $\wtd \prop_r = \fcr{R}$ and $\wtd \prop_{k+1} = \fpn{R}$. Further, we have 
   \begin{align*}
   	 \sigma_1' \in \indices{\fn{(P_2 \parallel  \htrigger{t}{V_1}) \subst{\tilde W}{\tilde x}} }
   \end{align*}

   Thus,
   %  ${\news{\wtd \prop_{k+1}}R_1 
   % \parallel \htrigger{t}{V_2}} \subst{\tilde W}{\tilde x} \relS Q_2'$\
   \eqref{eq:snd-rec-goal1}
   follows.  
    This concludes sub-case~$\tr(n)$ of case~$\rulename{Snd}{}$. 
   \end{enumerate}

   Finally, we briefly analyze the case when $P_1$ is a trigger collection. 
   Let $\sigma_1$ be defined as above. 
   % Let $\sigma_1 =\sigma' \cdot \subst{n_i}{n}$ where   
   %     $\sigma' \in \indices{\wtd u}$ with 
   %     $\wtd u = \fn{P_1\subst{\tilde W}{\tilde x}} \setminus n$, and
   %  with $i > 0$.  
   Let $H'_1$ be such that $P_1  \subst{\tilde W}{\tilde x} \sigma_1 \processrelate H'_1$.
   Then, by \tabref{mst:t:tablecd},   $Q_1$ 
   is as follows: 
  %  has the following shape: 
   \begin{align*}
     Q_1 = \Rb{\tilde v} \parallel H_1'
   \end{align*}
   \noindent where $\wtd v = \rfn{P_1\subst{\tilde W}{\tilde x}}$.
   Further, by  \defref{mst:def:valuesprelation} we know 
   $H'_1 = \bout{n_i}{V_2}H'_2$
   such that $V_1 \subst{\tilde W_1}{\tilde y} \sigma_2 \vrelate V_2$ and $P_2\subst{\tilde W_2}{\tilde w}\sigma_2
   \processrelate H'_2$ with $\sigma_2 = \sigma_1 \cdot \subsqn{n_i}$.
   We can see that 
   \begin{align*}
     H'_1 \by{\bactout{n_i}{V_2}} H'_2
   \end{align*}
   
   So, we should show  
   \begin{align}
   \label{pt:goal-snd-trigger}
     (P_2 \parallel \htrigger{t}{V_1}) \subst{\tilde W}{\tilde x} ~\relS~ 
     (\Rb{\tilde v} \parallel  H'_2 \Par \htrigger{\tni}{V_2})
   \end{align}
   
   Similarly to previous sub-cases, we have  
   \begin{align*}
   	\Rb{\tilde v} \parallel  H'_2 \Par \htrigger{\tni}{V_2} 
   	\equiv 
   	\Rb{\tilde v_1} \parallel H'_2 \parallel 
   	\Rb{\tilde v_2} \parallel \htrigger{\tni}{V_2}
   \end{align*}

   By 
   $P_2\subst{\tilde W_2}{\tilde w}\sigma_2
   \processrelate H'_2$
   and $\wtd v_1 =\rfn{P_2\subst{\tilde W_2}{\tilde w}}$
   we have 
   \begin{align}
   \label{p:snd-trigger-h-c}
   	\Rb{\tilde v_1} \parallel 
   	H'_2 \in \Cb{\tilde W_2\sigma_2}{\tilde
   w}{P_2\sigma_2}
   \end{align}
   
   Let $\sigma'_2 = \sigma_2 \cdot \subst{\tni}{t}$. 
   By  \eqref{p:snd-trigger-h-c}, \eqref{p:snd-trigger-v-c}, 
   and definition of $\Cb{-}{-}{-}$ (\tabref{mst:t:tablecd}) for the parallel composition case we have 
   \begin{align*}
   	 	\Rb{\tilde v_1} \parallel H'_2 \parallel 
   	\Rb{\tilde v_2} \parallel \htrigger{\tni}{V_2}
 		\in \Cb{\tilde W \sigma_2'}{\tilde x}{P_2 \parallel \htrigger{t}{V_1}}
   \end{align*}
   
   Thus, we reach goal~\eqref{pt:goal-snd-trigger}.
   This concludes case~$\rulename{Snd}{}$. 
    
   \item Case $\rulename{Rv}{}$. 
    In this case we know $P_1 = \binp{n}{y}P_2$. 
    We first consider cases when $P_1$ is not a trigger collection, and then
    briefly discuss the case when it is a trigger collection. As in the previous
    case, we distinguish two sub-cases: \rom{1} $\neg \tr(n)$ and \rom{2} $\tr(n)$: 
       \begin{enumerate}[i)]
         \item Sub-case $\neg \tr(n)$. We have the following transition: 
   \begin{align*}
       \AxiomC{} 
       \LeftLabel{\scriptsize $\rulename{Rv}{}$}
     \UnaryInfC{$(\binp{n_i}{y}P_2) \subst{\tilde W}{\tilde x} \by{\abinp{n}{V_1}} 
     P_2 \subst{\tilde W}{\tilde x}\subst{V_1}{y}$}
       \DisplayProof  	
   \end{align*}
   
   Here we assume $y \in \fv{P_2}$. 
   Let   
   $\sigma_1 \in \indices{\wtd u}$ with $\wtd u = \fn{P_1\subst{\tilde W}{\tilde x}}$
   such that $\initname{n}{i} \in \sigma_1$.  
   Also, let 
   $\sigma_2 = \sigma_1 \cdot \subsqn{n_i}$. 
   Further, let 
   $\wtd v = \rfn{P_1\subst{\tilde W}{\tilde x}}$, 
$\wtd \prop_r = \cup_{r \in \tilde v} \prop^r$, 
$\wtd \prop_k = (\prop_k, \ldots, \prop_{k+\plen{P_1}-1})$, and $\wtd
\prop_{k+1} = (\prop_{k+1}, \ldots, \prop_{k+\plen{P_1}-1})$.   By the definition of $\relS$ (\tabref{mst:t:tablecd}) we have 
   $Q_1 \in N_1 \cup N_2$ where 
   \begin{align*}
     N_1 &= \big\{   
       \news{\wtd \prop_r} \news{\wtd \prop_k} \Rb{\tilde v} \Par 
     \apropout{k}{\wtd B}
     \Par 
     \B{k}{\tilde x}{P_1\sigma_1} : 
     \wtd W \sigma_1 \vrelate \wtd B \big\} 
     \\
        N_2 &= \big\{ 
       \news{\wtd \prop_r}
      \news{\wtd \prop_{k+1}}
      \Rb{\tilde v} \Par 
      \binp{n_i}{y} 
        \apropout{k+1}{\wtd By}  
   \Par   \B{k+1}{\tilde xy}{P_2 \sigma_2}
        : 
        \wtd W \sigma_1 \vrelate \wtd B \big\}
   \end{align*}
   %
  %  where $\wtd \prop_k = (\prop_k, \ldots, \prop_{k+\plen{P_1}-1})$ 
  %  and $\wtd \prop_{k+1} = (\prop_{k+1}, \ldots, \prop_{k+\plen{P_1}-1})$.
   Similar to the other cases, if $Q_1 \in N_1$, then $Q_1$ reduces to
   some $Q'_1 \in N_2$ through communication on non-essential prefixes.
   
  %  Now, we pick $V_2$ such that its names indices 
  %  are consistent with  substitution $\sigma_2$. 
  %  First, 
   We may notice that $\sigma_2 \in \indices{\wtd u}$
   as by the assumption 
   $\dual n \not\in \wtd u$ we know    
   $\subst{\dual n_i}{\dual n} \not\in \sigma_1$. 
  %  that is we have 
  %  $\sigma_2 \in \indices{\wtd u}$. 
  Now, we pick $V_2$ such that $V_1\sigma_v\cdot\sigma_2 \vrelate V_2$ 
  where 
  $\sigma_v \in \indices{\fn{V} \setminus \wtd u}$
  such that 
  % $\subst{\dual n_{i+1}}{\dual n} \in \sigma_v$ iff 
  % $\dual n \in \fn{V}$.
\begin{align}
  \label{eq:rcv-subst-1}
  \sigma_v \cdot \sigma_2 
    \in 
    \indices{ \fn{P_2 \subst{\tilde W V_1}{\tilde x y}}}
\end{align}
  %  Now, let $\sigma_v \in \indices{\fn{V}}$ and 
  %  $\sigma = \sigma_v \cdot \sigma_1$.
    By
   \lemref{l:c-prop-closed} it suffices to consider the case when $Q'_1 \in N_2$,
   under which we have the following transition:
   \begin{align}
   \AxiomC{} 
   \LeftLabel{\scriptsize $\rulename{Rv}{}$}
    \UnaryInfC{$\binp{n_i}{y}\apropout{k+1}{\wtd By} 
    \by{\abinp{n_i}{V_2}}
    \apropout{k+1}{\wtd B V_2}$} 
    \DisplayProof
    \label{pt:inp1}
   \end{align}

    \begin{align}
    \AxiomC{\eqref{pt:inp1}} 
    \AxiomC{$\bn{\abinp{n_i}{V_2}} \cap \fn{\B{k+1}{\tilde xy}{P_2\sigma}} = \emptyset$}
    \LeftLabel{\scriptsize $\rulename{Par}{_L}$}
    \BinaryInfC{$\binp{n_i}{y}\apropout{k+1}{\wtd By} 
            \Par 
            \B{k+1}{\tilde xy}{P_2\sigma_2} \by{\abinp{n_i}{V_2}}
            \apropout{k+1}{\wtd B V_2} 
            \Par 
            \B{k+1}{\tilde xy}{P_2\sigma_2}$}
    \DisplayProof	
    \label{pt:rv1}
    \end{align}
     \begin{align*} 
         \AxiomC{\eqref{pt:rv1}} 
     %        \AxiomC{$
     %        \wtd \prop_r \cdot 
     %        \wtd \prop_{k+1} \cap \fn{\abinp{n_i}{V_2}} = \emptyset$} 
     %     \LeftLabel{\scriptsize $\rulename{Par}{_R}$}
     %     \BinaryInfC{$
     %     \news{\wtd \prop_r}
     %     \news{\wtd \prop_{k+1}}\binp{n_i}{y} 
     %     \apropout{k+1}{\wtd By}  
     %  \Par   \B{k+1}{\tilde xy}{P_2 \sigma_2}
     %     \by{\abinp{n_i}{V_2}} R$}
     %     \LeftLabel{\scriptsize $\rulename{Res}{}$}
         \AxiomC{$
         \wtd \prop_r \cdot 
         \wtd \prop_{k+1} \cap \fn{\abinp{n_i}{V_2}} = \emptyset$} 
      \LeftLabel{\scriptsize $\rulename{Res}{}$} 
         \BinaryInfC{
           $\news{\wtd \prop_r}
           \news{\wtd \prop_{k+1}}
           \Rb{\tilde v} \Par 
           \binp{n_i}{y} 
           \apropout{k+1}{\wtd By}  
        \Par   \B{k+1}{\tilde xy}{P_2 \sigma_2}
           \by{\abinp{n_i}{V_2}} R$
         }
         \DisplayProof
       \end{align*}
     
       \noindent where 
    $R = 
    \news{\wtd \prop_r}
    \news{\wtd \prop_{k+1}}
    \apropout{k+1}{\wtd BV_2} 
           \Par \B{k+1}{\tilde xy}{P_2 \sigma_2 }$ 
           with $\wtd W \sigma_1 \vrelate
           \wtd B$.
          %  and $V_1 \sigma \vrelate V_2$ for some $\sigma \in
          %  \indices{\fn{V_1}}$. 
           We can see that assertion $\subsqn{n_i} \in \sigma_2$ holds by the definition. 
           We should show that 
   \begin{align}
       P_2 \subst{\tilde W V_1}{\tilde x y} ~\relS~ R
       \label{eq:goal-rv1}
   \end{align}
    
   We know $n \not\in \fn{\wtd W}$ and $n \not\in \fn{V_1}$. 
  %  Thus, $\wtd W V_1
  %  \sigma_1 \cdot \sigma = \wtd W V_1 \sigma_2 \cdot \sigma$. 
  %  That is, we may
  %  notice that $\wtd W V_1 \sigma_2 \cdot \sigma \vrelate \wtd B V_2$. 
   % 
   Thus, $\wtd W V_1
   \sigma_v \cdot \sigma_1 = \wtd W V_1 \sigma_v \cdot \sigma_2$. 
   That is, we may
   notice that $\wtd W V_1 \sigma_v \cdot \sigma_2 \vrelate \wtd B V_2$.
   Further, by the definition of $\sigma_v$, we have 
   $P_2\sigma_2 = P_2 \sigma_v \cdot \sigma_2$. 
   Thus, we have   
   \begin{align*}
       R \in 
       \Cb{\tilde WV_1 \sigma_v \cdot \sigma_2}{\tilde xy}
       {P_2 \sigma_v \cdot \sigma_2} 	
   \end{align*}

%    Now, by the assumption $\dual n \not\in \wtd u$ we know 
%    $\subst{\dual n_i}{\dual n} \not\in \sigma_1$. 
%    Thus, by $\sigma_2' = \sigma_1 \cdot \subsqn{n_i} \cdot \subst{t_j}{t}$ and \defref{mst:d:indexedsubstitution}, 
%    we have  
%   \begin{align*}
%     \sigma_2' \in 
%     \indices{\fn{(P_2 \parallel  \htrigger{t}{V_1}) \subst{\tilde W}{\tilde x}} }
%   \end{align*} 
% ... 
   
   Finally, by this and \eqref{eq:rcv-subst-1}
   the goal \eqref{eq:goal-rv1} follows. 
   This concludes sub-case~$\neg\tr(n)$. 
   \item Sub-case $\tr(n)$. 
   The transition inference tree is as follows: 
   \begin{align*}
       \AxiomC{} 
       \LeftLabel{\scriptsize $\rulename{Rv}{}$}
     \UnaryInfC{$(\binp{n}{y}P_2) \subst{\tilde W}{\tilde x} \by{\abinp{n}{V_1}} 
     P_2 \subst{\tilde W}{\tilde x}\subst{V_1}{y}$}
       \DisplayProof  	
   \end{align*}
   
   Let $\sigma_1 = \indices{\wtd u}$
   where $\wtd u = \fn{P_1\subst{\tilde W}{\tilde x}}$. 
  %  Let  $\wtd \prop_k = (\prop_k, \ldots, \prop_{k+\plen{P_1}-1})$ 
  %  and $\wtd \prop_{k+1} = (\prop_{k+1}, \ldots, \prop_{k+\plen{P_1}-1})$. 
   Also, let $\wtd v = \rfn{P_1}$, 
   $\wtd \prop_r = \cup_{r \in \tilde v} \prop^r$, 
   $\wtd \prop_k = (\prop_k, \ldots, \prop_{k+\plen{P_1}-1})$,  
   $\wtd \prop_{k+1} = (\prop_{k+1}, \ldots, \prop_{k+\plen{P_1}-1})$, 
   and let $S$ be such that $n:S$. 
   Then, by the definition of $\relS$ (\defref{mst:d:relation-s})  
       we have $Q_1 \in N_1 \cup N_2 \cup N_3 \cup N_4$ where  
   \begin{align*} 
     N_1 &= \{ 
      \news{\wtd \prop_r}
       \news{\wtd \prop_k} 
     \Rb{\tilde v} \Par \apropout{k}{\wtd B}
     \Par \B{k}{\tilde x}{P_1\sigma_1}
     : \wtd W \sigma_1 \vrelate \wtd B\} 
     \\ 
     N_2 &= \{ 
      \news{\wtd \prop_r}
       \news{\wtd \prop_k} 
       \Rb{\tilde v} \Par 
               \abbout{\prop^n}{M^{\tilde B}_{V}} \Par 
               \B{k+1}{\tilde x y}{P_2\sigma_1} 
               : \wtd W \sigma_1 \vrelate \wtd B \} 
               \\
     N_3 &=  \{ 
      \news{\wtd \prop_r}
       \news{\wtd \prop_k}
       \Rb{\tilde v \setminus n} \Par 
               \appl{M^{\tilde B}_V}{\wtd n} \Par 
               \B{k+1}{\tilde x y}{P_2\sigma_1} 
               : \wtd W \sigma_1 \vrelate \wtd B \} 
               \\
     N_4 &= \{
      \news{\wtd \prop_r}
        \news{\wtd \prop_{k+1}}
     {\Rb{\tilde v \setminus n} \Par \binp{n_{\indT{S}}}{y}
     (\apropout{k+1}{\wtd By} \Par 
     \recprovx{n}{x}{\wtd n}
    %  \binp{\prop^n}{b}(\appl{b}{\wtd n}))
     \Par \B{k+1}{\tilde xy}{P_2\sigma_1}} :  
     \wtd W \sigma_1 \vrelate \wtd B \} 
   \end{align*}
     \noindent with  
     \begin{align*}
         M^{\tilde B}_V = \abs{\wtd z}
             {\binp{z_{\indT{S}}}{y}}{} \big(\apropout{k+1}{\wtd B y} 
             \Par 
             \recprovx{n}{x}{\wtd z}
            %  \binp{\prop^n}{b}(\appl{b}{\wtd z}) 
             \big) 
     \end{align*}
  %  %
  %    with $\wtd \prop_k = (\prop_k, \ldots, \prop_{k+\plen{P_1}-1})$ 
  %   and $\wtd \prop_{k+1} = (\prop_{k+1}, \ldots, \prop_{k+\plen{P_1}-1})$. 
  %  %
    Similar to the other cases, if $Q_1 \in N_1 \cup N_2 \cup N_3$, then
    there exists some $Q'_1 \in N_4$ such that $Q_1$ reduces to $Q'_1$ through
    communication on non-essential prefixes. By \lemref{l:c-prop-closed} it suffices
    to consider the case when $Q_1 \in N_4$. We infer the following transition:
     \begin{align*}
       Q_1  \by{\abinp{n_{\indT{S}}}{V_2}}         
       \news{\wtd \prop_r}
       \news{\wtd \prop_{k+1}}  R
     \end{align*} 
     \noindent where $R = \Rb{\tilde v \setminus n} \Par \apropout{k+1} {\wtd B V_2}
     \Par \binp{\prop^n}{b} (\appl{b}{\wtd n}) \Par 
     \B{k+1}{\tilde xy}{P_2\sigma_1}$ 
     and $V_1 \sigma \vrelate V_2$ for some $\sigma$. 
     We should show that 
     \begin{align}
      \label{eq:rv-tr-goal}
       P_2 \subst{\tilde W V_1}{\tilde x y}
                \ \relS\
                \news{\wtd \prop_r}
         \news{\wtd \prop_{k+1}}R
      \end{align}
      We may notice that 
      we have $\wtd v =\rfn{P_1} = \rfn{P_2}$ 
      and as $\tr(n)$ we have 
      $n \in \wtd v$. Thus, we have the following structural equivalence 
      \begin{align*}
      R \equiv \Rb{\tilde v} \Par \apropout{k+1} {\wtd B V_2} \Par 
      \B{k+1}{\tilde xy}{P_2\sigma_1}
     \end{align*}
    %  $n \in \rfni{R}$, 
    %  $\tilde x y = \fv{P_2}$, \
     Further, we have $V_1 \sigma \vrelate V_2$, 
     $\wtd W \sigma_1 \vrelate \wtd B$. Thus by 
      the definition of $\Cb{-}{-}{-}$ (\tabref{mst:t:tablecd}) 
      the goal~\eqref{eq:rv-tr-goal} follows. 
    This concludes sub-case~$\tr(n)$ of case~$\rulename{Rv}{}$. 
   \end{enumerate}
   
   Now, we briefly consider the case when $P_1$ is a trigger collection. 
   Let $\sigma_1$ be defined as above. 
   Let $H'_1$ be such that 
   $P_1\subst{\tilde W}{\tilde x} \sigma_1 \processrelate H'_1$. 
   Then, by definition of $\relS$, we know $Q_1$ has the 
   following shape: 
   \begin{align*}
     Q_1 =  \Rb{\tilde v} \parallel H_1
   \end{align*}
   \noindent where $\wtd v = \rfn{P_1\subst{\tilde W}{\tilde x}}$. 
   Further, by \defref{mst:def:valuesprelation} 
   we know $H_1' = \binp{n_i}{y}H_2'$
   such that $P_2\subst{\tilde W}{\tilde x} \sigma_2 \processrelate H_2'$, 
   where $\sigma_2 =\sigma_1 \cdot \subsqn{n_i}$. 
   Now, let $V_2$ be such that $V_1 \sigma \vrelate V_2$, for some $\sigma$. 
   We could see that 
   $$H_1' \by{\abinp{n_i}{V_2}} H_2'\subst{V_2}{y}$$
   We should show that 
   \begin{align}
   \label{pt:rcv-goal-trigger}
     P_2 \subst{\tilde W}{\tilde x}\subst{V_1}{y} ~\relS~ 
     \Rb{\tilde v} \parallel H_2\subst{V_2}{y}
   \end{align}
  
  By 
  $P_2\subst{\tilde W}{\tilde x} \sigma_2 \processrelate H_2'$ 
  and noticing that  $\processrelate$ is closed 
   under the substitution of $\vrelate$-related values 
   we have 
   \begin{align*}
   	P_2 \subst{\tilde W}{\tilde x}\subst{V_1}{y} \sigma_2 
   	\processrelate H_2' \subst{V_2}{y}
   \end{align*}
   
   Thus, goal~\eqref{pt:rcv-goal-trigger} follows. 
%  
%   This immediately follows by 
%    Rule~\texttt{IPSnd} (\defref{mst:def:valuesprelation}) 
%    and noticing that  $\processrelate$ is closed 
%   under substitution of $\vrelate$-related values. That is, 
%   if $P_2\subst{\tilde W}{\tilde x} \processrelate H_2 \subst{\tilde W}{\tilde x}$
%   and $V_1 \sigma \vrelate V_2$  then 
%   $P_2\subst{\tilde W}{\tilde x}\subst{V_1}{y} \processrelate H_2
%   \subst{\tilde W}{\tilde x}\subst{V_2}{y}$. 
   This concludes case~$\rulename{Rv}{}$. 
   
   \item Case~$\rulename{App}{}$. 
    Here we know $P_1=\appl{V_1}{(\wtd r, u)}$ 
   where $\wtd r = (r_1,\ldots,r_n)$.
   %  Similarly to the corresponding when an object name in action is 
   %  not tail-recursive, 
   We distinguish two sub-cases:
   \rom{1} $V_1 = x$ 
   where $\subst{V_x}{x} \in \subst{\tilde W}{\tilde x}$ and 
   \rom{2} $V_1$ is an abstraction. 
   Let  $V_1\subst{\tilde W}{\tilde x} = \abs{(\wtd y, z)}P_2$ 
   where $\wtd y = \wtd y^1, \ldots, \wtd y^n$. 
   % with $n = \len{\wtd r}$. 
   The inference tree is as follows: 
   \begin{align*}
   \AxiomC{} 
   \LeftLabel{\scriptsize $\rulename{App}{}$}
   \UnaryInfC{$(\appl{V_1}{(\wtd r, u)})\subst{\tilde W}{\tilde x} 
    \by{\tau} P_2\subst{\tilde r, u}{\tilde y, z}$} 	
   \DisplayProof 
   \end{align*}
   
   \noindent Let $\sigma_1 = \indices{\fn{P_1\rho_1}}$ 
   such that $\subst{u_i}{u} \in \sigma_1$.
   Further, 
   let 
  %  $r' \in \rfn{P_1}$, 
   $\wtd v = \rfn{V_1}$,  
   $\wtd \prop_{vr} = \bigcup_{r \in \tilde v, \tilde r}\prop^r$, 
   $\wtd \prop_k = (\prop_k, \ldots, \prop_{k+\plen{P_1}-1})$, 
   and $\wtd \prop_{k+1} = (\prop_{k+1}, \ldots, \prop_{k+\plen{P_1}-1})$. 
   Also, let $\wtd m = (u_i, \ldots, u_{i+\len{\Gt{C}}-1})$ 
  with $u_i:C$ 
   and 
   $\wtd{r_i} = (r^i_1,\ldots,r^i_{\len{\Rts{}{s}{S_i}}}))$
   with $r_i : S_i$
   for $i = \{1,\ldots,n\}$.
   % $\wtd r^i =()$
   Then, by the definition of $\relS$ 
   we have $Q_1 \in N$ where $N$ is defined as follows: 
    \begin{align*}
     N = N_1 \cup N_4 
      \cup \big(
      \cup_{1 \leq l \leq n} N^l_2  \big) 
      \cup \big(
      \cup_{1 \leq l \leq n - 1} N^l_3 \big) 
    \end{align*}
    \noindent where 
   \begin{align*} 
    N_1 &= \{ 
      \news{\wtd \prop_k} 
      \news{\wtd \prop_{vr}}  
    \Rb{\tilde v, \tilde r} \Par \apropout{k}{\wtd B}
    \Par \B{k}{\tilde x}{P_1\sigma_1} : \wtd W \sigma_1 \vrelate \wtd B \} 
    \\ 
    N^l_2 &= \{
      \news{\wtd \prop_{k+1}} 
      \news{\wtd \prop_{vr}}
    \Rb{\tilde v, r_l, \ldots, r_n} \Par
                 \overbracket{\prop^{r_l}!\big\langle 
                \lambda \wtd z_l. \prop^{r_{l+1}}!\langle\lambda \wtd z_{l+1}.\cdots. 
             \prop^{r_n}!\langle \lambda \widetilde z_n.}^{|\tilde r| - l +1} 
             Q^{V_2}_l \rangle \,\rangle \big\rangle 
             \\ 
             & \qquad 
             : V_1 \subst{\tilde W}{\tilde x} \sigma_1 \vrelate V_2 \}
             \\
    N^l_3 &= \{
      \news{\wtd \prop_{k+1}} 
      \news{\wtd \prop_{vr}}
    \Rb{\tilde v, r_{l+1}, \ldots, r_n} \Par \lambda \wtd z_l.
             \overbracket{
                 \prop^{r_{l+1}}!\langle\lambda \wtd z_{l+1}.\cdots. 
             \prop^{r_n}!\langle \lambda \wtd z_n.}^{|\tilde r| - l} 
             Q^{V_2}_l \rangle \, \big\rangle \ {\wtd r_l} 
             \\ 
             & \qquad : V_1 \subst{\tilde W}{\tilde x} \sigma_1 \vrelate V_2  \}
             \\
    N_4 &= \{
      \news{\wtd \prop_{k+1}}
     \news{\wtd \prop_{vr}}
    \Rb{\tilde v} \Par \appl{V_2}{(\wtd r_1,\ldots,\wtd r_{n}, \wtd m)}
    :  V_1\subst{\tilde W}{\tilde x} \sigma_1  \vrelate V_2 \}
    \end{align*}

    \noindent where 
    \begin{align*} 
    Q^{V_2}_l = \appl{V_2}{(\wtd r_1,\ldots,\wtd r_{l-1}, \wtd z_{l}, \ldots, 
             \wtd z_n, \wtd m)}
    \end{align*} 
    % \noindent with 
    %  $\wtd \prop_k = (\prop_k, \ldots, \prop_{k+\plen{P_1}-1})$ 
    % and $\wtd \prop_{k+1} = (\prop_{k+1}, \ldots, \prop_{k+\plen{P_1}-1})$.
   Note that for any
    $$Q_1 \in N_1 \cup N_4 
    \cup \big(
    \cup_{1 \leq l \leq n} N^l_2  \big) 
    \cup \big(
    \cup_{1 \leq l \leq n - 1} N^l_3 \big)$$
   there exists a $Q'_1 \in N_5$ such that 
   $Q_1$ reduces to  $Q'_1$ through communication on non-essential prefixes.
   By \lemref{l:c-prop-closed} it then suffices to consider the case $Q'_1 \in N_5$.
   
    Let $V_2 = \abs{\wtd y^1, \ldots, \wtd y^n, \wtd z}{Q_2}$ 
    where $\wtd z = (z_1,\ldots,z_{\len{\Gt{C}}})$.
    Then, we have the following transition: 
    \begin{align*}
      Q'_1  \by{\tau} \news{\wtd \prop_{k+1}}\news{\wtd \prop_r}R
    \end{align*} 
    \noindent where 
    $R = \Rb{\tilde v} \Par 
    Q_2\subst{\tilde r_1, \ldots, \tilde r_n, \tilde m}
    {\tilde y^1, \ldots, \tilde y^n, \tilde z}$. 
    We should show that 
    \begin{align}
      \label{eq:app-rec-s}
      P_2\subst{\tilde r, u}{\tilde y, z} \ \relS \ 
      \news{\wtd \prop_{k+1}} \news{\wtd \prop_{vr}} R
    \end{align} 
   
    \noindent 
    By 
    $V_1 \rho_1 \sigma_1 \vrelate V_2$ 
    (with $\rho_1 = \subst{\tilde W}{\tilde x}$)
    and 
    \defref{mst:def:vrelate} either 
    $V_2 \in \Cb{}{}{V_1\rho_1\sigma_1}$ or 
    $V_1 \rho_1 \valuesrelate V_2$. 
    In the former case, we know 
   $V_2 \in \Cbs{\rho_1'}{\sigma'_1}{V'_1}$
    where 
    $\rho_1'=\subst{\tilde W'}{\tilde x'}$ is such that 
    $V'_1 \rho_1'\sigma'_1 = V_1 \rho_1\sigma_1$. 
    % where $V'_1\rho_1'\sigma_1=V_1\rho_1$.  

    % By \defref{mst:def:vrelate} either $V_1 \rho_1 \valuesrelatev
    % V_2$ or $V'_1 \rho_1 \valuesrelate V_2$. 
    % If $V_1 \rho_1 \valuesrelatev V_2$ then 
    % $V_2 \in \Cbs{\rho_1}{\sigma_1}{V_1}$, that 
    % is $V_2 \in \Cbs{\rho_1'}{\sigma_1}{V'_1}$
    % where 
    % $\rho_1'=\subst{\tilde W'}{\tilde x'}$ is such that 
    % $V'_1 \rho_1' = V_1 \rho_1$. 
   Let $\wtd B'$ be such that $\wtd W' \sigma_1' \vrelate \wtd B'$. 
    By \tabref{mst:t:tablecd} we have 
   $$V_2 = \V{k+1}{\tilde x'}{\abs{(\wtd y, z)}P'_2}\subst{\tilde B'}{\tilde x'} = 
   \abs{(\wtd{y^1},\ldots,\wtd{y^n}, \wtd z)}Q_2 \subst{\tilde B'}{\tilde x'}
   $$
   
   \noindent where 
   \begin{align*} 
   Q_2 = 
   % \news{\wtd \prop_{k+1}}
   \news{ \wtd \prop_y}
   \prod_{i \in \len{\wtd y}}(\binp{\prop^{y_i}}{b}(\appl{b}{\wtd y^i})) \Par 
   \apropout{k+1}{\wtd B'}
   \Par
   \B{k+1}{\tilde x'}{P_2' \subst{z_1}{z}}
   \end{align*}
   
   \noindent with $\wtd \prop_y = \bigcup_{i \in \len{\tilde y}} c^{y_i}$ 
   and $P_2'$ is such that 
   \begin{align*}
    % \label{e:app-p2'}
    P_2'\subst{\tilde W'}{\tilde x'} = P_2
   \end{align*}
   
   Thus, we know  
   \begin{align*}
     R \equiv \Rb{\tilde v} \Par \news{ \wtd \prop_y}
    %  \news{\wtd \prop_{k+1}}
   \prod_{i \in \len{\wtd y}}(\binp{\prop^{y_i}}{b}(\appl{b}{\wtd r_i})) \Par 
   \apropout{k+1}{\wtd B'}
   \Par
   \B{k+1}{\tilde x'}{P_2' \subst{z_1}{z}}
   \end{align*}

   Now, we know $\news{\wtd \prop_{vr}} R \equiv \news{\wtd \prop_v} R $ 
   where $\wtd \prop_v = \bigcup_{r \in \tilde v}\prop^r$ since 
   $\wtd \prop_{vr} \setminus \wtd \prop_v \not\subseteq \fn{R}$. 
   Further, by renaming bound names we have 
   $$R\equiv
   \Rb{\tilde v} \Par \news{ \wtd \prop_r}
   \prod_{r \in \wtd r}(\binp{\prop^{r}}{b}(\appl{b}{\wtd r})) \Par 
   \apropout{k+1}{\wtd B'}
   \Par
   \B{k+1}{\tilde x'}{P_2' \subst{z_1}{z}} \subst{\tilde \prop_r}{\tilde \prop_y}
   $$
   
   \noindent where $\wtd \prop_r = \bigcup_{r \in \tilde r}\prop^r$. 
   Now, 
   by the definition of $\Rb{\tilde v}$ (\defref{mst:def:rec-providers}) 
   we know 
   \begin{align*}
    \Rb{\tilde v, \tilde r} = 
    %  \prod_{i \in \len{\wtd y}}(\binp{\prop^{y_i}}{b}(\appl{b}{\wtd r_i})) \equiv 
     \Rb{\tilde v} \Par \prod_{r \in \wtd r}
     (\recprovx{n}{x}{\wtd r})
    %  (\binp{\prop^{r}}{b}(\appl{b}{\wtd r}))
   \end{align*}
   Thus, we have 
   \begin{align*}
    R \equiv \news{\wtd \prop_r} 
    \Rb{\tilde v, \tilde r} \Par  \apropout{k+1}{\wtd B'}
    \Par
    \B{k+1}{\tilde x'}{P_2' \subst{z_1}{z}}\subst{\tilde \prop_r}{\tilde \prop_y} 
    = \news{\wtd \prop_r} R' 
   \end{align*}
   
   and by the definition we have $R' \in \Cb{\tilde W'}{\tilde x'}{P_2'}$.  
   We may notice that $\wtd v, \wtd r = \rfn{P_2'}$ 
   and $\news{\prop_{vr}}R \equiv \news{\prop_v}\news{\prop_r}R'$. 
   
   The later case, when $V_1 \rho_1 \valuesrelate V_2$, 
   follows by the fact that bodies of 
   characteristic and triggers values are $\processrelate$-related 
   to their minimal counterparts as shown in~\lemref{mst:lemm:processrelate-triggers}
   and that relation $\processrelate$ is closed under name substitutions. 
  %  That is, we have: 
  %  \begin{align*}
  %   P_2 \subst{}{} \processrelate 
  %  \end{align*}
   So, the goal \eqref{eq:app-rec-s} follows. 

   This concludes case~$\rulename{App}{}$ (and base cases) of the proof of \lemref{lemma:mstbs1}. Next, we consider 
  the inductive cases.

   \item Case~$\rulename{Par}{_L}$.
   In this case we 
   distinguish two sub-cases: \rom{1} 
   $P_1 \rho = P_1' \rho'_1 \Par P_1'' \rho''_1$ 
   and \rom{2} 
   $P_1 \rho = P_1' \rho'_1 \parallel P_1'' \rho''_1$
   where $P_1'$ is a trigger collection. 
    The final rule in the inference tree is: 
   \begin{align*} 
       \AxiomC{$P_1' \rho_1' \by{\ell} P_2' \rho_2' $} 
       \AxiomC{$\bn{\ell} \cap \fn{P_1''} = \emptyset$} 
       \LeftLabel{\scriptsize $\rulename{Par}{_L}$}
       \BinaryInfC{$P_1' \rho_1' \Par P_1'' \rho_1''
       \by{\ell} 
       P_2' \rho_2' \Par P_1'' \rho_1''$} 	
   \DisplayProof 
   \end{align*}
   
   \noindent 
   Let $\sigma_1 \in \indices{\wtd u}$
   where $\wtd u = \fn{P_1\subst{\tilde W}{\tilde x}}$. 
   Further, we know $\rho_1' = \subst{\tilde W_1}{\tilde y}$ 
   and $\rho_1'' = \subst{\tilde W_2}{\tilde z}$. 
   % So, let $\wtd B$ and $\wtd B_i$ be such that $\wtd W\sigma_1 \vrelate
   % \wtd B$ and $\wtd W_i \sigma_1 \vrelate \wtd B_i$
   % for $i \in \{1,2 \}$. 
   Further, let $\sigma'_1$ and $\sigma_1''$
   such that 
   $$P_1 \rho_1 \sigma_1 = 
   P_1' \rho_1' \sigma'_1 \Par P''_1 \rho''_1 \sigma''_1$$
   
   % Further, let $\sigma'_1 = \subst{\indices{\tilde u'}}{\tilde u'}$
   % and $\sigma''_1 = \subst{\indices{\tilde u''}}{\tilde u''}$ 
   % where $\wtd u' = \fn{P_1'}$ and $\wtd u'' = \fn{P_1''}$. 
   In sub-case \rom{1}, by the definition of $\relS$
   (\tabref{mst:t:tablecd}) we have $Q_1 \in N_1 \cup N_2 \cup N_3$ where 
   % there are the following possibilities for $Q_1$: 
   \begin{align*}
     N_1 =& \{ 
       \news{\widetilde \prop_k}
       \news{\wtd \prop_r}
       (\apropout{k}{\widetilde B} 
     \Par \B{k}{\tilde x}{P_1'\sigma_1' \Par P_1''\sigma_1''}) : 
     \wtd W \sigma_1 \vrelate \wtd B \} \\
     N_2 =& \{ 
       \news{\widetilde \prop_{k+1}}
       \news{\wtd \prop_r}
       (\propout{k+1}{\wtd B_1}
     \apropout{k+2}{\wtd B_2} 
     \Par \B{k+1}{\tilde y}{P_1'\sigma_1'} \Par \B{k+2}{\tilde z}{P_1''\sigma_1''}) 
     \\ 
     & \qquad : \wtd W_i \sigma_1 \vrelate \wtd B_i, i \in \{1,2 \}\}
     \\
       N_3 =& 
        \{ 
       \news{\widetilde \prop}
      \news{\wtd \prop_r}  
      R'_1 \Par R''_1 : R'_1 \in
        \Cb{\tilde W_1}{\tilde y}{P_1'\sigma_1'},\ 
        R''_1 \in \Cb{\tilde W_2}{\tilde z}{P_1''\sigma_1''}\} 
   \end{align*}
   \noindent where  $\wtd \prop_k = (\prop_k, \ldots, \prop_{k+\plen{P_1}-1})$, 
    $\wtd \prop_{k+1} = (\prop_{k+1}, \ldots, \prop_{k+\plen{P_1}-1})$, 
   and $\wtd \prop = \fpn{R_1' \Par R''_1}$.
   Note that for $Q_1 \in N_1 \cup N_2$ there exists some $Q'_1$ and $Q''_1$
   such that $Q_1$ reduces to $Q'_1\Par Q''_1 \in N_3$ through communication on non-essential prefixes,
   with 
   \begin{align}
     Q'_1 &\in \Cbs{\rho'_1}{\sigma'_1}{P_1'}
   \\
     Q''_1 &\in \Cbs{\rho''_1}{\sigma''_1}{P_1''}
     \label{eq:parl-q''_1}
   \end{align}
   
   Then, by \lemref{l:c-prop-closed} it suffices to consider the case of $Q'_1 \Par Q''_1 \in N_3$.
   By the definition of $\relS$ we have 
   \begin{align} 
     &P_1' \rho_1' \relS Q'_1 
     \label{parl-q'1}
     \\
     &P_1'' \rho_1'' \relS Q''_1
     \label{parl-q''1}
   \end{align}
   
    To apply IH we do a case analysis on the action $\ell$: 
   
   \begin{itemize}
   
   \item Sub-case $\ell \not\equiv \news{\wtd m_1}\about{n}{V_1}$. 
   By \eqref{parl-q'1} and IH we know 
   there is $Q_2'$ such that 
   $Q_1' \By{\ell} Q_2'$ and 
   \begin{align}
   \label{eq:parl-ih2}
   P_2' \rho_2' \ \relS \ Q_2' 	
   \end{align}
   
   \noindent We should show that 
   \begin{align}
     \label{eq:parl-goal1}
       P_2' \rho_2' \Par 
       P''_1 \rho''_1  ~\relS~  Q'_2 \Par Q''_1
   \end{align}
   
   We know that there is $R'$ such that 
   \begin{align} 
   Q'_1 \By{\tau} R' \by{\iname \ell} Q'_2
   \end{align}

   Thus, by Rule~$\rulename{Par}{_L}$ we can infer the following: 
   \begin{align*} 
     Q'_1 \Par Q''_1 \By{\tau} R' \Par Q''_1
   \end{align*} 
   % $$Q'_1 \Par Q''_1 \By{\tau} R' \Par Q''_1$$
   
   Further, we can infer 
   \begin{align*}
     \AxiomC{$R' \by{\iname \ell} Q'_2$} 
     \AxiomC{} 
     \LeftLabel{\scriptsize $\rulename{Par}{_L}$}
     \BinaryInfC{$R' \Par Q_1'' \by{\iname \ell} Q'_2 \Par Q''_1$} 
     \DisplayProof
   \end{align*}
   
   Then, by the IH \eqref{eq:parl-ih2} and the definition of $\relS$ 
   (\defref{mst:d:relation-s}) we know 
   % $$ Q'_2 \in \Cb{\tilde W'_1 \sigma_1' \cdot \sigma_2'}{\tilde y'}
   % {P_2'\sigma_1' \cdot \sigma'_2}$$
   $$ Q'_2 \in \Cbs{\rho'_2}{\sigma_1' \cdot \sigma_2'}{P_2'}$$
   
   So, we may notice that 
   \begin{align*}
     \Cb{\rho_1''}{\sigma_1'' \cdot \sigma_2'}{P_1''} 
     =
     \Cb{\rho_1''}{\sigma_1''}{P_1''} 
   \end{align*}
   
   % \begin{align*}
   %   \Cb{\tilde W_2'\sigma_1'' \cdot \sigma_2'}{\tilde z}{P_1'' \sigma_1'' \cdot \sigma_2'} 
   %   =
   %   \Cb{\tilde W_2'\sigma_1''}{\tilde z}{P_1''\sigma_1''} 
   % \end{align*}
   
   % Further, we know $\sigma_2' =$
   
   % \todo[inline]{todo}
   
   % Now, we may notice that 
   % $\sigma_2' \in \indices{\wtd u}$ where 
   % $\wtd u = \fn{P_2' \subst{\tilde W'_1}{\tilde y'} \Par P_1'' \subst{\tilde W_2}{\tilde z}}$, 
   % and by assumption that $P_1\subst{\tilde W}{\tilde x}$ is well-typed we have 
   % $(P_2' \Par P_1'') \sigma_1 \cdot \sigma_2' = P_2' \sigma'_2 \Par P_1''\sigma_1$ 
   % and $\wtd W_1' \sigma_2' \wtd W_2 \sigma_2 = \tilde W'_1 \tilde W_2  \sigma_1 \cdot \sigma_2'$. 
   So, by \eqref{eq:parl-q''_1} and  definition of $\Cb{-}{-}{-}$ we have 
   \begin{align*} 
     Q'_2 \Par Q''_1 \in 
       \Cbs{ \rho_2' \cdot \rho_1''}
       {\sigma_2' \cdot \sigma_1''}
       {P_2'  \Par P_1''}
   \end{align*} 
   
   By IH and assertion, we know that if $\ell = \bactinp{n}{V_1}$
   then $\subsqn{n_i} \in \sigma_2'$. So, assertion 
   $\subsqn{n_i} \in \sigma_2' \cdot \sigma_1''$ holds in this case. 
   Now, by \eqref{parl-q'1} we have $\sigma_1'' \in \indices{\fn{P_1''\rho_1''}}$ and 
   by \eqref{eq:parl-ih2} we have $\sigma_2' \in \indices{\fn{P_2'\rho_2'}}$. 
   We may notice that if $\ell = \bactinp{n}{V_1}$, by 
   transition rule~$\rulename{SRv}{}$ we have 
   $\bar n \not\in \fn{P_2'\rho_2' \Par P_1''\rho_1''}$ so 
   by \ref{mst:d:indexedsubstitution} we have $\subst{\dual n_j}{\dual n} \not\in \sigma_1''$ for 
   any $j > 0$. 
   So, we have  
   $$\sigma_2' \cdot \sigma_1'' \in \indices{\fn{P_2'\rho_2' \Par P_1''\rho_1''}}$$
   % Further, we have $\sigma'_2 \in \indices{}$.
   Thus, the goal \eqref{eq:parl-goal1} follows. 
   This concludes this sub-case.

       \item Sub-case $\ell \equiv \news{\wtd m_1}\about{n}{V_1}$. This sub-case 
   follows the essential steps of the previous sub-case. 
   By \eqref{parl-q'1} and IH we know there is $Q_2'$ such that 
   $Q_1' \By{\news{\tilde m_2}\about{n_i}{V_2}} Q _2'$ 
   and 
   \begin{align}
       \label{eq:parl-out-ih1}
     \news{\wtd m_1}(P_2' \parallel \htrigger{t}{V_1}) \rho_1' 
     \ \relS \
       \news{\wtd m_2}(Q_2' \parallel \htrigger{t}{V_2})
   \end{align}
   
   % In this sub-case we take $\wtd W_2 = \wtd W$ and $\wtd y_2 = \wtd y$. 
   We should show that 
   \begin{align}
     \news{\wtd m_1}(P_2' \Par P_1'' \parallel \htrigger{t}{V_1}) 
    \rho_1
     \ \relS \ 
       \news{\wtd m_2} 
     (Q_2' \Par Q''_1 \parallel \htrigger{t}{V_2})
     \label{pr:parlout-goal}
   \end{align}
   
   \noindent We pick $R'$ as in the previous sub-case. 
   So, we can infer the following transition: 
   \begin{align*} 
   \AxiomC{$R' 
   \by{\news{\tilde m_2}\about{n_i}{V_2}} 
   Q_2' $} 
   \LeftLabel{\scriptsize $\rulename{Par}{_L}$}
   \UnaryInfC{$R'\Par Q''_1
          \by{\news{\tilde m_2}\about{n_i}{V_2}}
            Q_2' \Par Q''_1$}
   \DisplayProof  	
   \end{align*}
   
   Next, by \tabref{mst:t:tablecd} we can infer the following. 
   First, by \eqref{eq:parl-out-ih1} we know 
   \begin{align*}
       Q_2' \parallel \htrigger{t}{V_2} \in 
       \Cbs{\rho_1'}{\sigma_2'}
       {P_2' \parallel \htrigger{t}{V_1}}
   \end{align*}
   
   By IH and assertion, we know $\subsqn{n_i} \in \sigma_2'$. 
   So, assertion $\subsqn{n_i} \in \sigma_2' \cdot \sigma_1''$ holds. 
   Similarly to the previous sub-case, we have 
   \begin{align*}
     \Cbs{\rho_1''}{\sigma_1'' \cdot \sigma_2'}{P_1''} 
     =
     \Cb{\rho_1''}{\sigma_1''}{P_1''} 
   \end{align*}

   Thus, we have  
   \begin{align*}
     \news{\wtd m_2} 
     (Q_2' \Par Q''_1 \parallel \htrigger{t}{V_2}) \in 
     \Cbs{\rho_1}{\sigma_1'' \cdot \sigma_2'}
     {\news{\wtd m_1}(P_2' \Par P''_1 \parallel \htrigger{t}{V_1})}
   \end{align*}
   
   Now, by \eqref{parl-q'1} we have $\sigma_1'' \in \indices{\fn{P_1''\rho_1''}}$ and 
    by \eqref{eq:parl-out-ih1} we have $\sigma_2' \in \indices{P_2'\rho_2'}$. 
    We may notice that if $\ell = \bactout{n}{V_1}$, by 
   transition rule~\textsc{SSnd} we have 
   $\bar n \not\in \fn{P_2'\rho_2' \Par P_1''\rho_1''}$ so 
   by \defref{mst:d:indexedsubstitution} we have $\subst{\dual n_j}{\dual n} \not\in \sigma_1''$ for 
   any $j > 0$. 
   So, we have 
   $$\sigma_2' \cdot \sigma_1'' \in \indices{\fn{P_2'\rho_2' \Par P_1''\rho_1''}}$$
   Thus, \eqref{pr:parlout-goal} follows.

   \end{itemize}
   This concludes case~$\rulename{Par}{_L}$. 
   
   \item Case~$\rulename{Tau}{}$. We distinguish two sub-cases: 
   \rom{1}
   $P_1 \rho_1  = P_1' \rho_1'
   \Par P_1'' \rho_1''$ 
   and \rom{2} 
   $P_1 \rho_1  = P_1' \rho_1' \parallel P_1'' \rho_1''$ 
   where one of parallel components is a trigger collection.
   Without loss of generality, we assume $\ell_1 = \news{\wtd m_1}\about{\dual
   n}{V_1}$ and $\ell_2 = \abinp{n}{V_1}$.
    The final rule 
   in the inference tree is then as follows: 
   \begin{align*}
   \AxiomC{$
       P_1' \rho_1' \by{\ell_1} P_2' \rho_2'$} 
   \AxiomC{$
       P_1'' \rho_1'' \by{\ell_2} P_2'' \rho_2''$} 
   \AxiomC{$\ell_1 \asymp \ell_2$} 
   \LeftLabel{\scriptsize $\rulename{Tau}{}$}
   \TrinaryInfC{$
         P_1' \rho'_1 \Par P_1'' \rho''_1
         \by{\tau} 
         \news{\wtd m_1}(P_2'  \rho_2' 
         \Par P_2'' \rho_2'')$}	
   \DisplayProof
   \end{align*}

   Let $\sigma_1 = \indices{\wtd u}$ where $\wtd u =
   \fn{P_1\subst{\tilde W}{\tilde x}}$. 
   Further, let $\sigma'_1$ and $\sigma_1''$
   such that 
   $$P_1 \rho_1 \sigma_1 = 
   P_1' \rho_1' \sigma'_1 \Par P''_1 \rho''_1 \sigma''_1$$
   
   We know $\rho_1' = \subst{\tilde W_1}{\tilde y}$ 
   and $\rho_1'' = \subst{\tilde W_2}{\tilde w}$.
   % Further, let $\wtd B$ and $\wtd B_i$ be
   % such that $\wtd W \sigma_1 \vrelate \wtd B$ and $\wtd W_i \sigma_1 \vrelate \wtd
   % B_i$ for $i \in \{1,2\}$. 
   By the definition of $\relS$
   (\tabref{mst:t:tablecd}) we have 
   $Q_1 \in N_1 \cup N_2 \cup N_3$ where
   \begin{align*}
     N_1 &= \{  \news{\wtd \prop_k}(\apropout{k}{\wtd B} 
     \Par \B{k}{\tilde x}{P_1'\sigma_1' \Par P_1''\sigma_1''}) 
     : \wtd W \sigma_1 \vrelate \wtd B \}
     \\
     N_2 &= \{ \news{\wtd \prop_{k+1}}(\propout{k+1}{\wtd B_1}
     \apropout{k+l+1}{\wtd B_2} 
     \Par \B{k+1}{\tilde y}{P_1'\sigma_1'} \Par \B{k+l+1}{\tilde w}{P_1''\sigma_1''})
     \\ 
     & \qquad : \wtd W_i \sigma_1 \vrelate \wtd B_i, i \in \{1,2 \}\}
     \\
       N_3 &= 
        \{R'_1 \Par R''_1 : R'_1 \in
        \Cb{\tilde W_1}{\tilde y}{P_1'\sigma_1'},\
        R''_1 \in \Cb{\tilde W_2}{\tilde w}{P_1''\sigma_1''} \} 
   \end{align*}
  
   By \lemref{l:c-prop-closed}, for $Q^1_1 \in N_1$ there exists 
   $Q^2_1 \in N_2$
   such that 
   \begin{align*}
     Q^1_1 \By{\tau} Q^2_1 \By{\tau} Q'_1 \Par Q''_1
     \end{align*}
     \noindent where $Q'_1 \Par Q''_1 \in N_3$, that is 
   \begin{align}
     Q'_1 &\in \Cbs{\rho'_1}{\sigma'_1}{P_1'}
     \label{eq:tau-q'_1} 
     \\
     Q''_1 &\in \Cbs{\rho''_1}{\sigma''_1}{P_1''}
     \label{eq:tau-q''_1}
   \end{align}

   Thus, in both cases we only consider how $Q_1' \Par Q_1''$ evolves. 
   By the definition of $\relS$ we have 
   \begin{align} 
     P_1' \rho_1' &\relS Q'_1 
     \label{tau-q'1}
     \\
     P_1'' \rho''_1  &\relS Q''_1
     \label{tau-q''1}
   \end{align}
   
   We have the following IH: 
   \begin{enumerate}
     \item 
   By \eqref{tau-q'1} and IH there is 
   $Q_2'$ such that 
   $Q'_1 \By{\news{\wtd m_1}\about{\dual n_i}{V_2}} Q_2'$ and 
   \begin{align}
   \news{\wtd m_1}(P_2' \parallel \htrigger{t}{V_1}) \rho_1'
    \ \relS \ 
   \news{\wtd m_2}(Q_2'  \parallel \htrigger{\tni}{V_2})
   \label{eq:tau-ih-pair-1}
   \end{align}

   By \defref{mst:d:relation-s} we know there is 
   $\sigma_v \cdot \sigma'_2 \in \indices{\fn{(P_2' \parallel \htrigger{t}{V_1})\rho_1'}}$
   such that  $\sigma_v \in \indices{\fn{\htrigger{t}{V_1}\rho_1'}}$ and 
   $\sigma_2' \in \indices{\fn{P_2'\rho_2'}}$ and  
   \begin{align*}
     Q_2'  \parallel \htrigger{\tni}{V_2} \in
      \Cbs{\rho_1'}{\sigma_v \cdot \sigma_2'}
      {P_2' \parallel \htrigger{\tni}{V_1}}
   \end{align*}
    
   So, we can infer  
   \begin{align}
     \label{eq:tau-ih1}
     Q'_2 \in \Cbs{\rho_2'}{\sigma'_2}{P_2'}
   \end{align}
   
   \item 
   By \eqref{tau-q''1} and IH there is 
   $Q''_2$ such that 
   $Q''_1 \By{\abinp{n_j}{V_2}} Q_2''$ and 
   \begin{align}
   P_2'' \rho_2''  \ \relS \ Q_2''
   \label{eq:tau-ih-pair-2}
   \end{align}

   By \defref{mst:d:relation-s} and \eqref{eq:tau-ih-pair-2} we know there is 
   $\sigma''_2 \in \indices{\fn{P_2''\rho_2''}}$ such that  
   \begin{align}
     \label{eq:tau-ih2}
     Q_2'' \in \Cb{\rho''_2}{\sigma''_2}{P_2''}
   \end{align}
   
   \end{enumerate}

   Similarly to the $\texttt{Par}_\texttt{L}$ case, we know there is 
   $R'$ such that 
   \begin{align*}
     Q'_1 \By{\tau} R' \by{\news{\wtd m_2}\about{\dual n_i}{V_2}} Q'_2
   \end{align*}
   \noindent where $\iname \ell_1 =  \news{\wtd m_2}\about{\dual n_i}{V_2}$. 
   Further, there is 
    $R''$ such that 
   \begin{align*}
     Q''_1 \By{\tau} R'' \by{\abinp{n_j}{V_2}} Q''_2
   \end{align*}
   
   \noindent where $\iname \ell_2 = \abinp{n_j}{V_2}$. 
   By Rule~$\textsc{Par}_\textsc{L}$ and Rule~$\textsc{Par}_\textsc{R}$ 
   we can infer the following: 
   $$Q'_1 \Par Q''_1 \By{\tau} R' \Par R''$$

   Now, to proceed we must show $\iname \ell_1 \asymp \iname \ell_2$, 
   which boils down to showing that indices of $\dual n_i$ and $n_j$ match. 
   For this, we distinguish two sub-cases: \rom{1} $\neg \tr(\dual n_i)$ and 
   $\neg \tr(n_j)$ 
   and \rom{2} $\tr(\dual n_i)$ and $\tr(n_j)$. In the former sub-case,  
     we have 
   $\subst{\dual n_i}{n} \in \sigma_1$ and $\subst{n_j}{n} \in \sigma_1$, 
   where $\sigma_1 = \indices{\wtd u}$. 
   Further, by this and 
   and \defref{mst:d:indexedsubstitution} 
   we know that $i =j$. Now, we consider the later case. 
   By assumption that $P_1\subst{\tilde W}{\tilde x}$ is well-typed, 
   we know there $\Gamma_1, \Lambda_1$, and $\Delta_1$ such that 
    $\Gamma_1;\Lambda_1;\Delta_1 \proves P_1\subst{\tilde W}{\tilde x} 
    \hastype \Proc$ with $\balan{\Delta_1}$, Thus,  we have 
    $n:S \in \Delta_1$ and $\dual n:T \in \Delta_1$ such that 
    $S ~\textbf{dual}~ T$. 
    Hence, by the definition of $\indT{-}$ (\Cref{mst:def:indexfunction}) we have 
     $i = \indT{S} = \indT{T} = j$. 
   %  Second, we may notice that 
   %  $\sigma_1' = \sigma_2'$ and 
   %  $\sigma_1'' = \sigma_2''$ as, by \defref{mst:d:indexedsubstitution}, 
   %  these substitutions do not include substitutions on tail-recursive names.
   Hence, we can infer the following transition: 
   \begin{align*}
   \AxiomC{$R' \by{\news{\wtd m_2}\about{\dual n_i}{V_2}} Q_2'$} 
   \AxiomC{$R'' \by{\abinp{n_i}{V_2}} Q_2''$} 
   \AxiomC{$\iname \ell_1 \asymp \iname \ell_2$} 
   \LeftLabel{\scriptsize $\rulename{Tau}{}$}
   \TrinaryInfC{$(R' \Par R'') \by{\tau} \news{\wtd m_2}(Q_2' \Par Q_2'')$}
   \DisplayProof	
   \end{align*}
   
   Now, we should show that 
   \begin{align}
     \news{\wtd m_1}(P_2' \rho_2' \Par P_2'' \rho_2'')  
     \ \relS \ 
     \news{\wtd m_2}(Q_2' \Par Q_2'')
     \label{eq:tau-goal2}
    \end{align}
    Further, we have 
    $(P_2' \Par P_2'') \rho'_2 \cdot \rho''_2 = P_2' \rho'_2 \Par P_2'\rho''_2$.  
    So, by \eqref{eq:tau-ih1} and \eqref{eq:tau-ih2} we have 
   \begin{align*} 
     \news{\wtd m_2}(Q_2' \Par Q_2'') \in 
     \Cbs{\rho_2' \cdot \rho_2''}{\sigma'_2 \cdot \sigma''_2} 
     { \news{\wtd m_1}
     (P_2' \rho_2' \sigma'_2 \Par P_2'' \rho_2'' \sigma''_2)}
   \end{align*} 
   
   Finally, we need to show 
   $\sigma'_2 \cdot \sigma''_2 \in \indices{\wtd u}$, where 
    $\wtd u = \fn{P_2' \rho_2' \Par P_2'' \rho_2''}$. 
    By IH we have 
    $\sigma_2' \in \indices{\fn{P_2'\rho_2'}}$ and 
    $\sigma_2'' \in \indices{\fn{P_2''\rho_2''}}$. 
    Further, as $P$ is well-typed, we have 
    $\dual n \in \fn{P_2'\rho_2'}$,  
    $n \not\in \fn{P_2'\rho_2'}$, 
    $\dual n \in \fn{\fn{P_2''\rho_2''}}$, 
    and  $n \not\in {\fn{P_2''\rho_2''}}$. 
    Thus,  by \defref{mst:d:indexedsubstitution} 
    in sub-case $\neg\tr(n)$ we only need to show that 
    for some $k > 0$ we have 
    $\subst{\dual n_k}{\dual n} \in \sigma_2'$ and 
    $\subst{n_k}{n} \in \sigma_2''$. 
    This follows by the assertion 
   as we know $\subsqn{\dual n_i} = \incrname{\dual n}{i} \in \sigma_2'$ 
   and $\subsqn{n_i} = \incrname{n}{i} \in \sigma_2''$.  
  The sub-case $\tr(n)$ follows directly by the \Cref{mst:d:indexedsubstitution}
 as  we have $\subst{\dual n_1}{\dual n} \in \sigma_2'$ 
  and $\subst{n_1}{n} \in \sigma_2''$.  
  So, we have $\sigma'_2 \cdot \sigma''_2 \in \indices{\wtd u}$. 
    Thus, the goal \eqref{eq:tau-goal2} follows. 
   This concludes case~$\rulename{Tau}{}$.

   \item Case~$\rulename{New}{}$. In this case we know  $P_1= \news{m:C}P_1'$.  
       The final rule in the transition inference tree is as follows: 
       \begin{align}
         \AxiomC{$P_1'\rho_1 
         \by{\bactout{\news{\wtd n_1}u}{V_1}}
         P_2 \rho_2 $} 
         \AxiomC{$m \in \fn{V_1}$} 
             \LeftLabel{\scriptsize $\rulename{New}{}$}
             \BinaryInfC{$\news{m}P_1' \rho_1
             \by{\bactout{\news{m \cdot \wtd n_1}u}{V_1}} P_2 \rho_1$} 
         \DisplayProof
       \end{align}
   \noindent
   Let $\sigma_1 = \indices{\fn{P_1 \rho_1 }}$.
   % Let $\wtd B$ be such that $\wtd W\sigma_1 \vrelate \wtd B$.  
   By the
   definition of $\relS$ (\tabref{mst:t:tablecd})
   we have $Q_1 \in N_1$ where 
   %  there are following
   % possibilities for $Q_1$: 
   \begin{align*}
     N_1 =& \{\news{\wtd \prop_r}\news{\wtd \prop}\news{\wtd m_2}
     \news{\tilde \prop^{m}}{R} : 
     R \in \Cb{\tilde W\sigma_1}{\tilde x}{P_1'\sigma_1 \cdot 
     \subst{m_1 \dual m_1}{m \dual m}} \}
   \end{align*}
   \noindent where $\wtd \prop_r = \fcr{R}$, 
   $\wtd m_2 = (m_1,\ldots,m_{\len{\Gt{C}}})$,
    and  $\tilde \prop^m = \prop^m \cdot \prop^{\dual m}$ if $\tr(C)$, 
   otherwise $\tilde \prop^m = \epsilon$.  
   By IH, if
   $P_1'\rho_1 \relS Q'_1$ 
   there are $Q_2$ and $V_2$ such that 
   $$Q'_1 \By{\bactout{\news{\wtd n_2}u_i}{V_2}} Q_2$$
   and 
   \begin{align}
     \label{eq:news-send-s-ih}
     \news{\wtd n_1}(P_2 \parallel \htrigger{t}{V_1}) \rho_1 
     \ \relS \ 
     \news{\wtd n_2}(Q_2 \parallel \htrigger{t}{V_2})
   \end{align}
   For $Q_1 \in N_1$ we should show that 
   \begin{align}
     Q_1 \By{\bactout{\news{\wtd m_2 \cdot \wtd n_2}u_i}{V_2}} Q_2
     \label{pt:new-goal1}
   \end{align}
   such that  
   \begin{align}
     \label{eq:news-send-s-goal}
     \news{m \cdot \wtd n_1}(P_2 \parallel \htrigger{t}{V_1}) \rho_1 
    \ \relS \ 
   \news{\wtd \prop'_r \cdot \wtd m_2 \cdot \wtd  n_2}(Q_2 \parallel \htrigger{t}{V_2})
   \end{align}
   \noindent where $\wtd \prop'_r = \fcr{V_2}$. 
   Note that by the definition we have $\fpn{V_2} = \emptyset$.
   By \lemref{l:c-prop-closed}, we know there is $R$ such that 
   $P_1' \rho_1 \relS R$ and 
   \begin{align}
     \label{eq:news-q1-r}
     Q_1' \By{\tau} R \by{\bactout{\news{\tilde n_2}u_i}{V_2}} Q_2
   \end{align}
   
   Now, by rule~$\rulename{New}{}$ we have 
   % Now, we notice that by 
   % \lemref{l:c-prop-closed}, for $Q_1 \in N_1$ we have  
   $$Q_1 \By{\tau} \news{\wtd \prop'_r}\news{\wtd m_2}R$$
   
   Now, we need to apply Rule~$\rulename{New}{}$ $\len{\wtd m_2}$ times 
   to \eqref{eq:news-q1-r} 
   to infer the following: 
   \begin{align*}
     \label{pt:news-r-q2}
     \AxiomC{$R' \by{\bactout{\news{\wtd n_2}u_i}{V_2}} Q_2$} 
     \AxiomC{$\wtd m_2 \subseteq \fn{V_2}$} 
     \LeftLabel{\scriptsize $\rulename{New}{}$}
     \BinaryInfC{
       $
       \news{\wtd m_2}R
       \by{\bactout{\news{\wtd m_2 \cdot \wtd n_2}u_i}{V_2}} Q_2
       $
     } 
     \DisplayProof
   \end{align*}
   
   Therefore, the sub-goal \eqref{pt:new-goal1} follows. 
   Now, by \eqref{eq:news-send-s-ih} 
   and by \defref{mst:d:relation-s} we can infer the following: 
   $$\news{\wtd n_2}(Q_2 \parallel \htrigger{t}{V_2}) \equiv 
   \news{\wtd \prop'_r}
   \news{\wtd n}
   ( \Rb{\tilde v} \Par R_2 \Par  \Rb{\tilde w} \parallel \htrigger{t}{V_2})
   $$
   \noindent where $\wtd v =  \rfn{R_2}$,  $\wtd w = \rfn{V_2}$, 
   $\wtd \prop'_r = \fcr{V_2}$, 
   and 
   $$\news{\wtd n}
   ( \Rb{\tilde v} \Par R_2 \Par \Rb{\tilde w} \parallel \htrigger{t}{V_2}) \in 
   \Cbs{\rho_1}{\sigma_1}{\news{\wtd n_1}(P_2 \parallel \htrigger{t}{V_1})}$$
   
   By this and the definition of $\Cb{-}{-}{-}$ (\tabref{mst:t:tablecd}) 
   we have 
   $$\news{\wtd m}
   \news{\tilde \prop^m}
   \news{\wtd n}
   (\Rb{\tilde v} \Par R_2 \Par \Rb{\tilde w} \parallel \htrigger{t}{V_2})
   \in \Cbs{\rho_1}{\sigma_1}
   {\news{m}\news{\wtd n_1}(P_2 \parallel \htrigger{t}{V_1})}$$
   
   Further, we may notice 
   \begin{align*}
     \news{\wtd m_2 \cdot \wtd  n_2}(Q_2 \parallel \htrigger{t}{V_2})
     \equiv
     \news{\wtd \prop_r} 
     \news{\wtd m_2}
     \news{\wtd n}(\Rb{\tilde v} \Par R_2 \Par \Rb{\tilde w} \parallel 
     \htrigger{t}{V_2})
   \end{align*}
   
   Therefore, the sub-goal \eqref{eq:news-send-s-goal} follows. 
   This concludes case~$\rulename{New}{}$ and the proof of \Cref{lemma:mstbs1}.
     \end{enumerate}
   \end{proof}

\subsection{Proof of \lemref{mst:lemma:mstbs3}}
\label{mst:app:mstbs3}

\mstbsrev*

\begin{proof}[Proof (Sketch)] 
    By transition induction. 
    First, we analyze the case of non-essential prefixes, which induce  $\tau$-actions 
    that do not correspond to actions in $P_1$. 
    This concerns the sub-case \rom{1} of Part 3. 
    This directly follows by \lemref{l:c-prop-closed}, that is by the fact that 
     $\Cb{\tilde W}{\tilde x}{P_1}$ is closed under transitions on 
    non-essential prefixes.

    \noindent Now, assume $Q_1 \by{\ell} Q_2$ when $\ell$ is an essential prefix.
    This is mainly the converse of the proof of \lemref{lemma:mstbs1} noting that
    there are no essential actions in $\Cb{\tilde W}{\tilde x}{P_1}$ not matched
    in $P_1$. We consider only one case: 
    \begin{itemize}
      \item Case $\rulename{Snd}{}$. In this case we know $P_1 =\bout{n}{V_1}P_2$. 
      We distinguish two sub-cases: \rom{1} $\neg \tr(u_i)$ 
      and \rom{2} $\tr(u_i)$. 
      In both sub-cases, we  distinguish two kinds of an object value $V_1$: 
      (a)~ $V_1 \equiv x$, such that  
      $\subst{V_x}{x}\in \subst{\tilde W}{\tilde  x}$ 
      and (b)~ $V_1 = \abs{y:C}P'$, that is $V_1$ is a pure abstraction.
      We only consider sub-case (a). 
  
      Let $\wtd W_1$, $\wtd W_2$, $\wtd y$, and $\wtd w$ such that 
      $$P_1\subst{\tilde W}{\tilde x} = 
      \bout{n}{V_1\subst{\tilde W_1}{\tilde y}}P_2 \subst{\tilde W_2}{\tilde w}$$
  
      Let    
      $\sigma_1 \in \indices{\wtd u}$ where 
      $\wtd u = \fn{P_1\subst{\tilde W}{\tilde x}}$ such that $\subst{n_i}{n} \in \sigma_1$.  
       Also, let $\sigma_2 = \sigma_1 \cdot \subsqn{n_i}$. 
       When $P_1$ is not a
      trigger, by the definition of $\relS$ (\tabref{mst:t:tablecd}), for both
      sub-cases, 
      we have 
      $Q_1 \in N_1$ where: 
      \begin{align*}
           N_1 &= \big\{
            \news{\wtd \prop_r}
            \news{\wtd \prop_{k+1}}
             \Rb{\tilde v} \Par 
           \bout{n_i}{V_2} \apropout{k+1}{\wtd B_2} \Par 
           \B{k+1}{\tilde w}{P_2\sigma_2} 
           :  V_1\sigma\subst{\wtd W_1}{\wtd y} \vrelate V_2,\ 
           \wtd W_2 \vrelate \wtd B_2 \big\}  
        \end{align*}
        For $Q_1 \in N_1$ we have the following transition inference tree: 
        \begin{align*} 
          \AxiomC{} 
          \LeftLabel{\scriptsize $\rulename{Snd}{}$}
          \UnaryInfC{$Q_1 \by{\bactout{n_i}{V_1}}Q_2$}
          \DisplayProof
        \end{align*}
        where 
        $$Q_2 =\news{\wtd \prop_r}
        \news{\wtd \prop_{k+1}}\apropout{k+1}{\wtd B_2} \Par 
        \B{k}{\tilde w}{P_2\sigma_2}$$
        
        \noindent We have 
        \begin{align*} 
          \AxiomC{} 
          \LeftLabel{\scriptsize $\rulename{Snd}{}$}
          \UnaryInfC{$P_1\subst{\tilde W}{\tilde x} 
          \by{\bactout{n}{V_1 \subst{\tilde W_1}{\tilde y}}} 
          P_2 \subst{\tilde W_2}{\tilde w}
          $}
          \DisplayProof
        \end{align*}
        We should show that 
        \[
          {\newsp{\wtd{m_1}}{P_2 \parallel \htrigger{t}{V_1}}}
            \ \relS\
            {\newsp{\wtd{m_2}}{Q_2 \parallel \htrigger{\tni}{V_2}}}
         \]
        
         \noindent This immediately follows by the definition of $\relS$ 
         and \tabref{mst:t:tablecd}. 
        %  and the following congruence 
        %  $${\newsp{\widetilde{m_2}\wtd \prop_V}{Q_2 \Par \htrigger{t}{V_2}}} 
        %  \equiv 
        %  {\newsp{\wtd \prop_{Q_2}\widetilde{m_2}\wtd \prop_V}{C_{P_2} \Par \htrigger{t}{V_2}}}$$
         This concludes \texttt{Snd} case.   
  \end{itemize}
  As can be seen the proof of this part is essentially the inverse of 
  the proof of
  \lemref{lemma:mstbs1}. We just need to show that 
  $\Cb{\tilde W\sigma}{\tilde x}{P_1\sigma}$ does 
  not introduce extra actions on essential prefixes not present in $P_1$.
  This is evident by the inspection of the definition of $\Cb{-}{-}{-}$.  
  Briefly, only in the case of the input and the output prefix does $\Db{-}{-}{-}$  
  introduce actions that mimic those prefixes. The remaining cases only
  introduce actions on non-essential prefixes ($\tau$-actions on propagator names). 
  %$\hfill\square$
  \end{proof}

%%% Local Variables:
%%% mode: latex
%%% TeX-master: "mst"
%%% End:

\end{document}